\documentclass[12pt]{article}
\PassOptionsToPackage{hyphens}{url}
\usepackage[pagebackref=true,colorlinks]{hyperref}
\usepackage{amsmath,amssymb,amsthm,mathtools,amsrefs}
\usepackage{marvosym}
\usepackage{mathpazo} % add possibly `sc` and `osf` options
\usepackage{eulervm}
\usepackage{csquotes}
\usepackage{tikz}
\usetikzlibrary{arrows.meta}
\usetikzlibrary{decorations.pathmorphing,shapes}
\usetikzlibrary{calc}
\usepackage{pgfplots}
\pgfplotsset{compat=1.13}
\usepackage{comment}
\usepackage{adjustbox}
\usepackage{stmaryrd}
\usepackage{fancyhdr}
\usepackage{soul}
\usepackage{dsfont}%this is for mathbb{1} but you write mathds{1}
\usepackage[normalem]{ulem}
\usepackage{longtable}
\usepackage{caption}
\usepackage{subfig} %this is to plot figures vertically
\usepackage{wrapfig}%this is to allow figures in line with text
\usepackage{pictexwd, dcpic}
\usepackage{float}
\usepackage{enumerate}
\usepackage[all]{xy} %This is to make higher categorical stuff
\usetikzlibrary{circuits.ee.IEC}
\hypersetup{
	colorlinks=true,%set true if you want colored links
    linktoc=all,%set to all if you want both sections and subsections linked
    linkcolor=blue,  %choose some color if you want links to stand out
    citecolor=blue,
    }
\makeatletter
\renewcommand*\env@matrix[1][\arraystretch]{%
  \edef\arraystretch{#1}%
  \hskip -\arraycolsep
  \let\@ifnextchar\new@ifnextchar
  \array{*\c@MaxMatrixCols c}}
\makeatother

\theoremstyle{theorem}
\newtheorem{theorem}[equation]{Theorem}
\newtheorem{lemma}[equation]{Lemma}
\newtheorem{proposition}[equation]{Proposition}
\newtheorem{corollary}[equation]{Corollary}
\theoremstyle{definition}
\newtheorem{definition}[equation]{Definition}
\newtheorem{construction}[equation]{Construction}

\newtheorem{question}[equation]{Question}
\newtheorem{problem}[equation]{Problem}
\newtheorem{example}[equation]{Example}
\newtheorem{exercise}[equation]{Exercise}
\newtheorem*{answer}{Answer}
\newtheorem*{solution}{Solution}
\newtheorem{remark}[equation]{Remark}

\newcommand\define[1]{\emph{\textbf{#1}}}%italicize and bold-face %this seems like a good alternative

\newcommand\aeequals[1]{\underset{\raisebox{0.3ex}[0pt][0pt]{\scriptsize${#1}$}}{=}}

\numberwithin{equation}{section}
\let\a=\alpha \let\b=\beta \let\g=\gamma \let\de=\delta 
  \let\q=\theta \let\i=\iota 
\let\l=\lambda \let\r=\rho
\let\s=\sigma    \let\c=\chi 
\let\w=\omega       \let\D=\Delta  
  \let\S=\Sigma   
\let\C=\Chi \let\W=\Omega

\newcommand{\be}{\begin{equation}}
\newcommand{\ee}{\end{equation}}
\def\ba{\begin{align}} %previously this was ``array''
\def\ea{\end{align}}
\newcommand{\bea}{\begin{eqnarray}}
\newcommand{\eea}{\end{eqnarray}}
\newcommand{\bx}{\begin{example}}
\newcommand{\ex}{\end{example}}
\newcommand{\bex}{\begin{exercise}}
\newcommand{\eex}{\end{exercise}}
\newcommand{\ban}{\begin{answer}}
\newcommand{\ean}{\end{answer}}
\newcommand{\bt}{\begin{theorem}}
\newcommand{\et}{\end{theorem}}
\newcommand{\bc}{\begin{corollary}}
\newcommand{\ec}{\end{corollary}}
\newcommand{\blem}{\begin{lemma}}
\newcommand{\elem}{\end{lemma}}
\newcommand{\bp}{\begin{problem}}
\newcommand{\ep}{\end{problem}}
\newcommand{\bn}{\begin{proposition}}
\newcommand{\en}{\end{proposition}}
\newcommand{\bd}{\begin{definition}}
\newcommand{\ed}{\end{definition}}
\newcommand{\bcon}{\begin{construction}}
\newcommand{\econ}{\end{construction}}
\newcommand{\bq}{\begin{question}}
\newcommand{\eq}{\end{question}}
\newcommand{\bprf}{\begin{proof}}
\newcommand{\eprf}{\end{proof}}
\newcommand{\br}{\begin{remark}}
\newcommand{\er}{\end{remark}}
\newcommand{\bs}{\begin{solution}}
\newcommand{\es}{\end{solution}}
\newcommand{\beqs}{\begin{eqnarray}}
\newcommand{\eeqs}{\end{eqnarray}}

 \let\ov=\overline

\newcommand{\<}{\langle}
\renewcommand{\>}{\rangle}

\newcommand{\id}{\mathrm{id}}

\newcommand{\mC}{\mathcal{C}}
\newcommand{\mM}{\mathcal{M}}
\newcommand{\mD}{\mathcal{D}}

\newcommand{\imp}{\qquad \Rightarrow \qquad}

\newcommand{\tr}{{\rm tr} }

\def\C{{{\mathbb C}}}

\def\N{{{\mathbb N}}}
\def\Z{{{\mathbb Z}}}

\def\B{{{\mathbb B}}}

\newcommand{\Ad}{\mathrm{Ad}}
\def\Hi{{{\mathcal{H}}}}

\newcommand{\CAlgSPU}{\mathbf{C\text{*-}AlgSPU}}

\newcommand{\fdCAlg}{\mathbf{fdC\text{*-}Alg}}
\newcommand{\fdCAlgY}{\mathbf{fdC\text{*-}Alg}_{\text{\Yinyang}}}
\newcommand{\fdCAlgP}{\mathbf{fdC\text{*-}AlgP}}
\newcommand{\fdCAlgPU}{\mathbf{fdC\text{*-}AlgPU}}
\newcommand{\fdCAlgU}{\mathbf{fdC\text{*-}AlgU}}
\newcommand{\fdCAlgUY}{\mathbf{fdC\text{*-}AlgU}_{\text{\Yinyang}}}
\newcommand{\fdCAlgCPU}{\mathbf{fdC\text{*-}AlgCPU}}
\newcommand{\fdCAlgSPU}{\mathbf{fdC\text{*-}AlgSPU}}

\def\mA{{{\mathcal{A}}}}

\def\mB{{{\mathcal{B}}}}

\def\mN{{{\mathcal{N}}}}

\newcommand{\op}{\mathrm{op}}

\newcommand{\FinStoch}{\mathbf{FinStoch}}
\newcommand{\FinMeas}{\mathbf{FinMeas}}
\newcommand{\FinMeasP}{\mathbf{FinMeas}_{+}}

\newcommand{\Stoch}{\mathbf{Stoch}}

\newcommand{\stoch}{\;\xy0;/r.25pc/:(-3,0)*{}="1";(3,0)*{}="2";{\ar@{~>}"1";"2"|(1.06){\hole}};\endxy\!}
\newcounter{sarrow}
\newcommand\xstoch[1]{%
\stepcounter{sarrow}%
\mathrel{\begin{tikzpicture}[baseline= {( $ (current bounding box.south) + (0,-0.1ex) $ )}]
\node[inner sep=.5ex] (\thesarrow) {\;$\scriptstyle #1$\;};
\path[draw,{<[scale=1.5,width=3,length=2]}-,decorate,
  decoration={snake,amplitude=0.3mm,segment length=2.1mm,pre=lineto,pre length=1pt}] 
    (\thesarrow.south east) -- (\thesarrow.south west);
\end{tikzpicture}}%
}

\newcommand{\ben}{\renewcommand{\theenumi}{\alph{enumi}} 
\renewcommand{\labelenumi}{(\theenumi)}\begin{enumerate}}
\newcommand{\een}{\end{enumerate}}

\newcommand\blfootnote[1]{%
  \begingroup
  \renewcommand\thefootnote{}\footnote{#1}%
  \addtocounter{footnote}{-1}%
  \endgroup
}

\newlength\stateheight
\newlength\minimumstatewidth
\tikzset{width/.initial=\minimummorphismwidth}
\tikzset{colour/.initial=white}

\newif\ifblack\pgfkeys{/tikz/black/.is if=black}
\newif\ifwedge\pgfkeys{/tikz/wedge/.is if=wedge}
\newif\ifvflip\pgfkeys{/tikz/vflip/.is if=vflip}
\newif\ifhflip\pgfkeys{/tikz/hflip/.is if=hflip}
\newif\ifhvflip\pgfkeys{/tikz/hvflip/.is if=hvflip}

\pgfdeclarelayer{foreground}
\pgfdeclarelayer{background}
\pgfsetlayers{background,main,foreground}

\def\thickness{0.4pt}

\makeatletter
\pgfkeys{%
  /tikz/on layer/.code={
    \pgfonlayer{#1}\begingroup
    \aftergroup\endpgfonlayer
    \aftergroup\endgroup
  },
  /tikz/node on layer/.code={
    \gdef\node@@on@layer{%
      \setbox\tikz@tempbox=\hbox\bgroup\pgfonlayer{#1}\unhbox\tikz@tempbox\endpgfonlayer\pgfsetlinewidth{\thickness}\egroup}
    \aftergroup\node@on@layer
  },
  /tikz/end node on layer/.code={
    \endpgfonlayer\endgroup\endgroup
  }
}
\def\node@on@layer{\aftergroup\node@@on@layer}
\makeatother

\makeatletter
\pgfdeclareshape{state}
{
    \savedanchor\centerpoint
    {
        \pgf@x=0pt
        \pgf@y=0pt
    }
    \anchor{center}{\centerpoint}
    \anchorborder{\centerpoint}
    \saveddimen\overallwidth
    {
        \pgf@x=3\wd\pgfnodeparttextbox
        \ifdim\pgf@x<\minimumstatewidth
            \pgf@x=\minimumstatewidth
        \fi
    }
    \savedanchor{\upperrightcorner}
    {
        \pgf@x=.5\wd\pgfnodeparttextbox
        \pgf@y=.5\ht\pgfnodeparttextbox
        \advance\pgf@y by -.5\dp\pgfnodeparttextbox
    }
    \anchor{A}
    {
        \pgf@x=-\overallwidth
        \divide\pgf@x by 4
        \pgf@y=0pt
    }
    \anchor{B}
    {
        \pgf@x=\overallwidth
        \divide\pgf@x by 4
        \pgf@y=0pt
    }
    \anchor{text}
    {
        \upperrightcorner
        \pgf@x=-\pgf@x
        \ifhflip
            \pgf@y=-\pgf@y
            \advance\pgf@y by 0.4\stateheight
        \else
            \pgf@y=-\pgf@y
            \advance\pgf@y by -0.4\stateheight
        \fi
    }
    \backgroundpath
    {
       \begin{pgfonlayer}{foreground}
        \pgfsetstrokecolor{black}
        \pgfsetlinewidth{\thickness}
        \pgfpathmoveto{\pgfpoint{-0.5*\overallwidth}{0}}
        \pgfpathlineto{\pgfpoint{0.5*\overallwidth}{0}}
        \ifhflip
            \pgfpathlineto{\pgfpoint{0}{\stateheight}}
        \else
            \pgfpathlineto{\pgfpoint{0}{-\stateheight}}
        \fi
        \pgfpathclose
        \ifblack
            \pgfsetfillcolor{black!50}
            \pgfusepath{fill,stroke}
        \else
            \pgfusepath{stroke}
        \fi
       \end{pgfonlayer}
    }
}

\makeatletter
\pgfdeclareshape{effect}
{
    \savedanchor\centerpoint
    {
        \pgf@x=0pt
        \pgf@y=0pt
    }
    \anchor{center}{\centerpoint}
    \anchorborder{\centerpoint}
    \saveddimen\overallwidth
    {
        \pgf@x=3\wd\pgfnodeparttextbox
        \ifdim\pgf@x<\minimumstatewidth
            \pgf@x=\minimumstatewidth
        \fi
    }
    \savedanchor{\upperrightcorner}
    {
        \pgf@x=.5\wd\pgfnodeparttextbox
        \pgf@y=-.5\ht\pgfnodeparttextbox %negated this to shift text up
        \advance\pgf@y by .5\dp\pgfnodeparttextbox %negated this to shift text up
    }
    \anchor{A}
    {
        \pgf@x=-\overallwidth
        \divide\pgf@x by 4
        \pgf@y=0pt
    }
    \anchor{B}
    {
        \pgf@x=\overallwidth
        \divide\pgf@x by 4
        \pgf@y=0pt
    }
    \anchor{text}
    {
        \upperrightcorner
        \pgf@x=-\pgf@x
        \ifhflip
            \pgf@y=\pgf@y %negated this to shift text up
            \advance\pgf@y by -0.4\stateheight %negated this to shift text up
        \else
            \pgf@y=\pgf@y %negated this to shift text up
            \advance\pgf@y by 0.4\stateheight %negated this to shift text up
        \fi
    }
    \backgroundpath
    {
       \begin{pgfonlayer}{foreground}
        \pgfsetstrokecolor{black}
        \pgfsetlinewidth{\thickness}
        \pgfpathmoveto{\pgfpoint{-0.5*\overallwidth}{0}}
        \pgfpathlineto{\pgfpoint{0.5*\overallwidth}{0}}
        \ifhflip
            \pgfpathlineto{\pgfpoint{0}{-\stateheight}}%negate this to flip triangle
        \else
            \pgfpathlineto{\pgfpoint{0}{\stateheight}}%negate this to flip triangle
        \fi
        \pgfpathclose
        \ifblack
            \pgfsetfillcolor{black!50}
            \pgfusepath{fill,stroke}
        \else
            \pgfusepath{stroke}
        \fi
       \end{pgfonlayer}
    }
}
\pgfdeclareshape{my ground}
{
  \inheritsavedanchors[from=rectangle ee]
  \inheritanchor[from=rectangle ee]{center}
  \inheritanchor[from=rectangle ee]{north}
  \inheritanchor[from=rectangle ee]{south}
  \inheritanchor[from=rectangle ee]{east}
  \inheritanchor[from=rectangle ee]{west}
  \inheritanchor[from=rectangle ee]{north east}
  \inheritanchor[from=rectangle ee]{north west}
  \inheritanchor[from=rectangle ee]{south east}
  \inheritanchor[from=rectangle ee]{south west}
  \inheritanchor[from=rectangle ee]{input}
  \inheritanchor[from=rectangle ee]{output}
  \anchorborder{\csname pgf@anchor@my ground@input\endcsname}

  \backgroundpath{
    \pgf@process{\pgfpointadd{\southwest}{\pgfpoint{\pgfkeysvalueof{/pgf/outer xsep}}{\pgfkeysvalueof{/pgf/outer ysep}}}}
    \pgf@xa=\pgf@x \pgf@ya=\pgf@y
    \pgf@process{\pgfpointadd{\northeast}{\pgfpointscale{-1}{\pgfpoint{\pgfkeysvalueof{/pgf/outer xsep}}{\pgfkeysvalueof{/pgf/outer ysep}}}}}
    \pgf@xb=\pgf@x \pgf@yb=\pgf@y
    \pgfpathmoveto{\pgfqpoint{\pgf@xa}{\pgf@ya}}
    \pgfpathlineto{\pgfqpoint{\pgf@xa}{\pgf@yb}}
    \pgfmathsetlength\pgf@xa{.5\pgf@xa+.5\pgf@xb}
    \pgfmathsetlength\pgf@yc{.16666\pgf@yb-.16666\pgf@ya}
    \advance\pgf@ya by\pgf@yc
    \advance\pgf@yb by-\pgf@yc
    \pgfpathmoveto{\pgfqpoint{\pgf@xa}{\pgf@ya}}
    \pgfpathlineto{\pgfqpoint{\pgf@xa}{\pgf@yb}}
    \advance\pgf@ya by\pgf@yc
    \advance\pgf@yb by-\pgf@yc
    \pgfpathmoveto{\pgfqpoint{\pgf@xb}{\pgf@ya}}
    \pgfpathlineto{\pgfqpoint{\pgf@xb}{\pgf@yb}}
  }
}
\makeatother

\tikzset{inline text/.style =
  {text height=1.2ex,text depth=0.25ex,yshift=0.5mm}}
\tikzset{arrow box/.style =
  {rectangle,inline text,fill=white,draw,
    minimum height=5mm,yshift=-0.5mm,minimum width=5mm}}
\tikzset{bubble/.style =
  {inner sep=0mm,minimum width=3mm,minimum height=3mm,
    draw,shape=circle,fill=white}}
\tikzset{dot/.style =
  {inner sep=0mm,minimum width=1mm,minimum height=1mm,
    draw,shape=circle}}
\tikzset{white dot/.style = {dot,fill=white,text depth=-0.2mm}}
\tikzset{scalar/.style = {diamond,draw,inner sep=1pt}}
\tikzset{square/.style =
  {inner sep=0mm,minimum width=2mm,minimum height=2mm,
    draw,shape=rectangle}}%%me trying to make a fakecopier

\tikzset{star/.style = {dot,fill=white,text depth=-0.2mm}}
\tikzset{copier/.style = {dot,fill,text depth=-0.2mm}}
\tikzset{fakecopier/.style = {square,fill,text depth=-0.2mm}}
\tikzset{discarder/.style = {my ground,draw,inner sep=0pt,
    minimum width=4.2pt,minimum height=11.2pt,anchor=input,rotate=90}}

\tikzset{xshiftu/.style = {shift = {(#1, 0)}}}
\tikzset{yshiftu/.style = {shift = {(0, #1)}}}
\tikzset{scriptstyle/.style={font=\everymath\expandafter{\the\everymath\scriptstyle}}}

\usepackage[framemethod=TikZ]{mdframed}
\usepackage{titlesec}
\newmdtheoremenv[%
   outerlinewidth=2,
   roundcorner=10pt,
   leftmargin=0,
   rightmargin=0,
   innertopmargin=4pt,
   backgroundcolor=green!05,
   outerlinecolor=blue!30,
   ntheorem=true,
   ]{fancydefn}{Definition}%[fancydefn]

\newmdtheoremenv[%
   outerlinewidth=2,
   roundcorner=10pt,
   leftmargin=0,
   rightmargin=0,
   innertopmargin=-4pt,
   backgroundcolor=blue!05,
   outerlinecolor=violet!30,
   ntheorem=true,
   ]{fancyex}{Example}%[fancydefn]
   
\newmdtheoremenv[%
   outerlinewidth=2,
   roundcorner=10pt,
   leftmargin=0,
   rightmargin=0,
   innertopmargin=-4pt,
   backgroundcolor=red!05,
   outerlinecolor=red!30,
   ntheorem=true,
   ]{fancyprop}{Proposition}%[fancydefn]

   \newmdtheoremenv[%
   outerlinewidth=2,
   roundcorner=10pt,
   leftmargin=0,
   rightmargin=0,
   innertopmargin=-4pt,
   backgroundcolor=violet!05,
   outerlinecolor=red!30,
   ntheorem=true,
   ]{fancythm}{Theorem}%[fancydefn]

\newenvironment{theo}[2][]{%
\refstepcounter{equation}%
\ifstrempty{#1}%
{\mdfsetup{%
frametitle={%
\tikz[baseline=(current bounding box.east),outer sep=0pt]
\node[anchor=east,rectangle,fill=red!20]
{\strut Theorem~\theequation};}}
}%
{\mdfsetup{%
frametitle={%
\tikz[baseline=(current bounding box.east),outer sep=0pt]
\node[anchor=east,rectangle,fill=red!20]
{\strut Theorem~\theequation~[#1]};}}%
}%
\mdfsetup{innertopmargin=0pt,linecolor=red!20,%
linewidth=2pt,topline=true,%
frametitleaboveskip=\dimexpr-\ht\strutbox\relax
}
\begin{mdframed}[]\relax%
\label{#2}}{\end{mdframed}}

\newenvironment{lem}[2][]{%
\refstepcounter{equation}%
\ifstrempty{#1}%
{\mdfsetup{%
frametitle={%
\tikz[baseline=(current bounding box.east),outer sep=0pt]
\node[anchor=east,rectangle,fill=red!20]
{\strut Lemma~\theequation};}}
}%
{\mdfsetup{%
frametitle={%
\tikz[baseline=(current bounding box.east),outer sep=0pt]
\node[anchor=east,rectangle,fill=red!20]
{\strut Lemma~\theequation~[#1]};}}%
}%
\mdfsetup{innertopmargin=0pt,linecolor=red!20,%
linewidth=2pt,topline=true,%
frametitleaboveskip=\dimexpr-\ht\strutbox\relax
}
\begin{mdframed}[]\relax%
\label{#2}}{\end{mdframed}}

\newenvironment{rmk}[2][]{%
\refstepcounter{equation}%
\ifstrempty{#1}%
{\mdfsetup{%
frametitle={%
\tikz[baseline=(current bounding box.east),outer sep=0pt]
\node[anchor=east,rectangle,fill=yellow!30]
{\strut Remark~\theequation};}}
}%
{\mdfsetup{%
frametitle={%
\tikz[baseline=(current bounding box.east),outer sep=0pt]
\node[anchor=east,rectangle,fill=yellow!30]
{\strut Remark~\theequation~[#1]};}}%
}%
\mdfsetup{innertopmargin=0pt,linecolor=yellow!30,%
linewidth=2pt,topline=true,%
frametitleaboveskip=\dimexpr-\ht\strutbox\relax
}
\begin{mdframed}[]\relax%
\label{#2}}{\end{mdframed}}

\newenvironment{prop}[2][]{%
\refstepcounter{equation}%
\ifstrempty{#1}%
{\mdfsetup{%
frametitle={%
\tikz[baseline=(current bounding box.east),outer sep=0pt]
\node[anchor=east,rectangle,fill=red!20]
{\strut Proposition~\theequation};}}
}%
{\mdfsetup{%
frametitle={%
\tikz[baseline=(current bounding box.east),outer sep=0pt]
\node[anchor=east,rectangle,fill=red!20]
{\strut Proposition~\theequation~[#1]};}}%
}%
\mdfsetup{innertopmargin=0pt,linecolor=red!20,%
linewidth=2pt,topline=true,%
frametitleaboveskip=\dimexpr-\ht\strutbox\relax
}
\begin{mdframed}[]\relax%
\label{#2}}{\end{mdframed}}

\newenvironment{cor}[2][]{%
\refstepcounter{equation}%
\ifstrempty{#1}%
{\mdfsetup{%
frametitle={%
\tikz[baseline=(current bounding box.east),outer sep=0pt]
\node[anchor=east,rectangle,fill=red!20]
{\strut Corollary~\theequation};}}
}%
{\mdfsetup{%
frametitle={%
\tikz[baseline=(current bounding box.east),outer sep=0pt]
\node[anchor=east,rectangle,fill=red!20]
{\strut Corollary~\theequation~[#1]};}}%
}%
\mdfsetup{innertopmargin=0pt,linecolor=red!20,%
linewidth=2pt,topline=true,%
frametitleaboveskip=\dimexpr-\ht\strutbox\relax
}
\begin{mdframed}[]\relax%
\label{#2}}{\end{mdframed}}

\newenvironment{defn}[2][]{%
\refstepcounter{equation}%
\ifstrempty{#1}%
{\mdfsetup{%
frametitle={%
\tikz[baseline=(current bounding box.east),outer sep=0pt]
\node[anchor=east,rectangle,fill=green!20]
{\strut Definition~\theequation};}}
}%
{\mdfsetup{%
frametitle={%
\tikz[baseline=(current bounding box.east),outer sep=0pt]
\node[anchor=east,rectangle,fill=green!20]
{\strut Definition~\theequation~[#1]};}}%
}%
\mdfsetup{innertopmargin=0pt,linecolor=green!20,%
linewidth=2pt,topline=true,%
frametitleaboveskip=\dimexpr-\ht\strutbox\relax
}
\begin{mdframed}[]\relax%
\label{#2}}{\end{mdframed}}

\newenvironment{conv}[2][]{%
\refstepcounter{equation}%
\ifstrempty{#1}%
{\mdfsetup{%
frametitle={%
\tikz[baseline=(current bounding box.east),outer sep=0pt]
\node[anchor=east,rectangle,fill=green!15]
{\strut Convention~\theequation};}}
}%
{\mdfsetup{%
frametitle={%
\tikz[baseline=(current bounding box.east),outer sep=0pt]
\node[anchor=east,rectangle,fill=green!15]
{\strut Convention~\theequation~[#1]};}}%
}%
\mdfsetup{innertopmargin=0pt,linecolor=green!15,%
linewidth=2pt,topline=true,%
frametitleaboveskip=\dimexpr-\ht\strutbox\relax
}
\begin{mdframed}[]\relax%
\label{#2}}{\end{mdframed}}

\newenvironment{exa}[2][]{%
\refstepcounter{equation}%
\ifstrempty{#1}%
{\mdfsetup{%
frametitle={%
\tikz[baseline=(current bounding box.east),outer sep=0pt]
\node[anchor=east,rectangle,fill=violet!20]
{\strut Example~\theequation};}}
}%
{\mdfsetup{%
frametitle={%
\tikz[baseline=(current bounding box.east),outer sep=0pt]
\node[anchor=east,rectangle,fill=violet!20]
{\strut Example~\theequation~[#1]};}}%
}%
\mdfsetup{innertopmargin=0pt,linecolor=violet!20,%
linewidth=2pt,topline=true,%
frametitleaboveskip=\dimexpr-\ht\strutbox\relax
}
\begin{mdframed}[]\relax%
\label{#2}}{\end{mdframed}}

\newenvironment{ques}[2][]{%
\refstepcounter{equation}%
\ifstrempty{#1}%
{\mdfsetup{%
frametitle={%
\tikz[baseline=(current bounding box.east),outer sep=0pt]
\node[anchor=east,rectangle,fill=blue!20]
{\strut Question~\theequation};}}
}%
{\mdfsetup{%
frametitle={%
\tikz[baseline=(current bounding box.east),outer sep=0pt]
\node[anchor=east,rectangle,fill=blue!20]
{\strut Question~\theequation~[#1]};}}%
}%
\mdfsetup{innertopmargin=0pt,linecolor=blue!20,%
linewidth=2pt,topline=true,%
frametitleaboveskip=\dimexpr-\ht\strutbox\relax
}
\begin{mdframed}[]\relax%
\label{#2}}{\end{mdframed}}

\titleformat{\section}[block]{\Large\bfseries\filcenter}{}{1em}{}%%%this centers titles of sections

\title{Inverses, disintegrations, and Bayesian inversion in quantum Markov categories}
\author{Arthur J. Parzygnat}
\date{\today}

\newcommand{\Addresses}{{% additional braces for segregating \footnotesize
  \bigskip
  \footnotesize

  A.~Parzygnat, \textsc{Institut des Hautes \'Etudes Scientifiques, 35 Route de Chartres 91440, Bures-sur-Yvette, France}\par\nopagebreak
  \textit{E-mail address}, A.~Parzygnat: \texttt{parzygnat@ihes.fr}
}}

\begin{document}
\emergencystretch 2em
%\sloppy%this forces text to not go into the margins, though some argue it is not ideal https://tex.stackexchange.com/questions/241343/what-is-the-meaning-of-fussy-sloppy-emergencystretch-tolerance-hbadness/241355#241355    &     https://tex.stackexchange.com/questions/9107/how-can-i-make-my-text-never-go-over-the-right-margin-by-always-hyphenating-or-b
\maketitle 
\vspace{-9mm}

\begin{abstract}
We introduce quantum Markov categories as a structure that refines and extends a synthetic approach to probability theory and information theory so that it includes quantum probability and quantum information theory. 
In this broader context, we analyze three successively more general notions of reversibility and statistical inference: ordinary inverses, disintegrations, and Bayesian inverses. 
We prove that each one is a strictly special instance of the latter for certain subcategories, providing a categorical foundation for Bayesian inversion as a generalization of reversing a process. We unify the categorical and $C^*$-algebraic notions of almost everywhere (a.e.) equivalence.
As a consequence, we prove many results including a universal no-broadcasting theorem for S-positive categories, a generalized Fisher--Neyman factorization theorem for a.e.\ modular categories, a relationship between error correcting codes and disintegrations, and the relationship between Bayesian inversion and Umegaki's non-commutative sufficiency. 

\blfootnote{\emph{2020 Mathematics Subject Classification.} 
62F15, %Bayesian inference
46L53 (Primary); %Non-commutative probability and statistics
81R15, %(Secondary). %operator algebraic methods in quantum mechanics
18D10, %monoidal categories (Fritz used this one) OR 18D99 = Categories with structure
%60A05, %Foundations of probability theory--Axioms; other general questions
81P45 (Secondary) %Quantum information, communication, networks 
%81P10, %Logical foundations of quantum mechanics; quantum logic 
%46M99 (Secondary)%Functional analysis: methods of category theory in functional analysis
}
\end{abstract}

\vspace{-11mm}
\tableofcontents

%%%%%%%%%%%%%%%%%%%%%%%%%%%%%%%%%%%%%%
\section[Introduction, motivation, and outline]{a000}%Introduction, motivation, and outline
\label{sec:intro}
\vspace{-12mm}
\noindent
\begin{tikzpicture}
\coordinate (L) at (-8.75,0);
\coordinate (R) at (8.75,0);
\draw[line width=2pt,orange!20] (L) -- node[anchor=center,rectangle,fill=orange!20]{\strut \Large \textcolor{black}{\textbf{1\;\; Introduction, motivation, and outline}}} (R);
\end{tikzpicture}
%\vspace{1mm}
%%%%%%%%%%%%%%%%%%%%%%%%%%%%%%%%%%%%%%

In his lectures on entropy, Gromov emphasized that concepts in mathematics should frequently be revisited due to our constantly growing and changing perspectives, which may provide new insight on old subjects~\cite{Gr14}. Probability theory is no exception, and a dramatic change in viewpoint on the structural foundations of probability theory has gained enormous momentum recently~\cite{La62,Gi82,Pa99,Go02,Ja11,BFL,Fo12,CuSt14,DDGK16,CDDG17,GaPa18,Ja18,ChJa18,Fr19,Ja19,Ja20,JaSt20,FrRi20,FGPR20}. In particular, many diagrammatic axioms have been put forward that allow one to prove a variety of interesting theorems and constructions from classical statistics using purely diagrammatic arguments~\cite{Fr19}. However, most of the guiding examples towards this perspective have come from \emph{classical} (measure-theoretic) probability theory. To further refine and expand on the synthetic approach, it is important to analyze how \emph{quantum} information theory and probability fit into this scheme, especially since we believe that nature is governed not by the classical theory, but by the laws of quantum theory. In return, an analysis of the quantum setting may uncover new structures and relations that were previously not known classically. The present work provides the first steps in this direction. 

%%%%%%%%%%%%%%%%%%%%%%%%%%%%%%%%%%%%%%
\vspace{2mm}
\noindent
\begin{tikzpicture}
\coordinate (L) at (-8.75,0);
\coordinate (R) at (8.75,0);
\draw[line width=2pt,orange!20] (L) -- node[anchor=center,rectangle,fill=orange!20] {\strut \large \textcolor{black}{{The main goals of this work}}} (R);
%\node[anchor=center,rectangle,fill=orange!20] at (0,0) {\strut \large \textcolor{black}{{The main goals of this work}}};
\end{tikzpicture}
%%%%%%%%%%%%%%%%%%%%%%%%%%%%%%%%%%%%%%

This paper serves three major purposes. 

First, we provide a categorical language that can be used to describe both classical and quantum probability. We do this by introducing quantum Markov categories. We show that the category of completely positive unital (CPU) maps between $C^*$-algebras is \emph{neither} a quantum nor a classical Markov category because it does not admit a copy map. Nevertheless, it naturally \emph{embeds} into a quantum Markov category that \emph{does} have a copy map, and the latter allows string-diagrammatic reasoning generalizing those of the CD and Markov categories of Cho--Jacobs and Fritz, respectively~\cite{ChJa18,Fr19}. 
Thus, physically implementable operations can be manipulated via string diagrams in a larger enveloping category. 

Second, though the above mentioned copy map in quantum mechanics is not a physically implementable operation (by the no-cloning theorem), it can be used to define disintegrations~\cite{PaRu19}, define Bayesian inversion~\cite{PaRuBayes}, and to extend Frtiz' many additional axioms~\cite{Fr19} to the quantum setting. We relativize these axioms to \emph{subcategories} of quantum Markov categories since the subcategories themselves need not be Markov categories for these definitions to make sense. In doing so, special care needs to be taken regarding how the axioms are formulated due to the subtleties that arise from a lack of commutativity in the quantum setting.
Nevertheless, we accomplish this, and our approach allows one to easily specialize definitions (such as Bayesian inversion) or theorems proved abstractly (such as the Fisher--Neyman factorization theorem) to the language of non-commutative statistics. 

The third, and main, purpose of this paper is to present precise theorems relating the notions of inverses, disintegrations, and Bayesian inverses, which are valid both in the classical and quantum setting. To do this, we (re)introduce the notion of an almost everywhere (a.e.) modular subcategory of a quantum Markov category (called a positive category in an earlier draft of~\cite{Fr19}). An example of such a category is the one of CPU maps. 
The simplest version of the relationship between inverses, disintegrations, and Bayesian inverses holds in any a.e.\ modular category. Briefly, every state-preserving inverse is a disintegration and every disintegration is a Bayesian inverse, but the reverse implications are not generally true. For example, if a state-preserving morphism has a disintegration, then the morphism must be a.e.\ deterministic. However, if the morphism has a Bayesian inverse, it need not be the case that the original morphism is a.e.\ deterministic. This shows that Bayesian inversion is the most general candidate (of these three) as a reversing procedure (sometimes interpreted as an inference) for CPU maps. We summarize this hierarchy of reversibility by the following implications
\[
\text{invertible} 
\imp
\text{disintegrable}
\imp
\text{Bayesian invertible}.
\] 

While accomplishing these three main goals, we discover several interesting consequences. First, the notion of non-commutative a.e.\ equivalence, which was recently introduced in the operator-algebraic framework~\cite{PaRu19}, is shown to agree with the string-diagrammatic definition of Cho and Jacobs when extended to quantum Markov categories~\cite{ChJa18}. Second, we prove a universal no-broadcasting theorem for S-positive subcategories of quantum Markov categories. Finally, our main goals relating disintegrations, a.e.\ determinism, and Bayesian inversion are shown to imply a more general version of the Fisher--Neyman factorization theorem, valid in \emph{both} the classical and quantum settings. A further detailed account of our contributions are described later in this introduction. 

The present work supplements our investigations of \emph{quantum} disintegrations and Bayesian inversion, where we have analyzed the \emph{existence} in the finite-dimensional hybrid classical and quantum setting~\cite{PaRu19,PaRuBayes}. Nevertheless, the present paper can be read completely independently of these two works. 

%%%%%%%%%%%%%%%%%%%%%%%%%%%%%%%%%%%%%%
\vspace{2mm}
\noindent
\begin{tikzpicture}
\coordinate (L) at (-8.75,0);
\coordinate (R) at (8.75,0);
\draw[line width=2pt,orange!20] (L) -- node[anchor=center,rectangle,fill=orange!20] {\strut \large \textcolor{black}{{The broader context and some motivation}}} (R);
%\node[anchor=center,rectangle,fill=orange!20] at (0,0) {\strut \large \textcolor{black}{{The broader context and some motivation}}};
\end{tikzpicture}
%%%%%%%%%%%%%%%%%%%%%%%%%%%%%%%%%%%%%%

Recent progress has shown that many concepts of classical information theory, probability theory, and information transfer can be encompassed in string diagrammatic language through the usage of CD and Markov categories~\cite{ChJa18,Fr19}. Such a synthetic reformulation of information theory divorces itself from some of the specific structures associated with measure theory. Instead, one attempts to isolate diagrammatic structures that allow one to obtain categorical analogues of theorems. There are many good reasons to do this, though here we focus mainly on one, which is the ability to transfer these concepts into other settings. This has the benefit of new potential discoveries in the new setting, but it also has the possibility to help refine the abstract theory. Here, we mainly concentrate on how these concepts are instantiated in quantum theory, though it seems reasonable that other contexts exist. 

Some of the axioms that a classical Markov category can satisfy allow one to construct conditionals, disintegrate measure-preserving maps, deduce causal inference, and perform Bayesian updating~\cite{ChJa18,Fr19,Ja19,Ja20,JaKiZa}. Other axioms can be used to reformulate and prove theorems of sufficient statistics for example~\cite{Fr19,FGPR20}. However, not all classical Markov categories satisfy these axioms, and yet some of these constructions and theorems might still hold. In fact, there are also many categories coming from quantum mechanics that satisfy analogous versions of these theorems but they do not fall within the Markov category framework, simply because Markov categories contain a copy map and further demand a commuting condition that seems incompatible with quantum observables~\cite{Bo44}. Nevertheless, at least in the case of quantum theory, almost the same string-diagrammatic manipulations can be done by embedding the category of completely positive maps in an \emph{enveloping quantum} Markov category, which does contain a (not necessarily commutative) copy map. The ambient quantum Markov category need not include only physically implementable processes, but it can nevertheless be used to \emph{formulate} useful axioms and definitions for \emph{subcategories} whose morphisms have a more direct physical meaning. 
We prove that most (though not all!) of the important axioms of sufficient statistics hold for completely positive (in fact, Schwarz positive) unital maps between unital finite-dimensional $C^*$-algebras by viewing them inside the enveloping quantum Markov category.
More generally, it is possible that relativizing axioms in this way to subcategories of classical Markov categories (or mathematical structures generalizing Markov categories, such as quantum Markov categories) may allow problems of inference and information flow to be analyzed diagrammatically, without the requirement that the class of morphisms under study form a Markov category.

An interesting aspect of disintegrations and Bayesian inversion is their interpretation as generalized inverses. If one initially begins in a category of state-preserving deterministic dynamics~\cite{BFL,Pa17}, it is rarely possible to construct state-preserving inverses. To expand the possibilities for a reversal procedure, one might hope to construct an approximate inverse in some sense. This leads to the notion of a disintegration, which is a morphism that optimally reverses  the original deterministic dynamics. But to construct disintegrations, one must enlarge the original category to include probabilistic morphisms.%
%footnote
\footnote{Although this is reminiscent of what one does in the localization of a category with respect to a class of morphisms, we have not made any explicit connection. It would be interesting to see the relationship, if one exists. The perspective we mention here is where one begins with a category of deterministic processes and uses a monad to construct a Kleisli category, whose new morphisms are thought of as describing \emph{stochastic} dynamics. For classical (quantum) systems, this categorical procedure takes us from evolution described by functions ($^*$-homomorphisms) on phase space (the algebra of observables) to evolution described by Markov kernels~\cite{La62,Gi82} (completely positive unital maps~\cite{We17}). We merely mention this as a remark, but we do not develop this abstract viewpoint here. 
}
Hence, one now has new morphisms describing \emph{stochastic} dynamics. One can then ask if these dynamics can be reversed in a similar way. In one of our main theorems here, we show that this is not possible. More precisely, if a stochastic morphism has a disintegration, then the original stochastic morphism is necessarily a.e.\ deterministic. Bayesian inversion, the third notion of reversibility that we will examine, correctly captures a more robust reversal procedure that reduces to the disintegration case when the original dynamics is a.e.\ deterministic. Although our results are proved quite generally (for certain subcategories of any quantum Markov category), we focus on 
the category of unital $C^*$-algebras and unital $^*$-homomorphisms or, more generally, completely positive unital maps (quantum operations). Most of our results are stated for Schwarz-positive unital maps, which include completely positive unital maps as a subclass. 

In defining Bayesian inversion and reversability, one must take particular care of measure zero events. More specifically, a.e.\ equivalence in classical probability theory plays an important role in uniqueness theorems. In \cite{PaRu19}, Russo and the author introduced the notion of a.e.\ equivalence for maps between $C^*$-algebras equipped with states 
to determine the uniqueness of disintegrations. The definition is simple, intuitive, and is motivated by the Gelfand--Naimark--Segal (GNS) construction. In \cite{ChJa18}, Cho and Jacobs independently introduced an elegant categorical formulation of a.e.\ equivalence valid for any Markov category. In this paper, we will show that these two notions agree when the latter is generalized to quantum Markov categories. 
This notion of a.e.\ equivalence plays an essential role in determining the uniqueness of quantum Bayesian inverses as well~\cite{PaRuBayes}.   
Although, the topic of reversibility in quantum mechanics has been studied in great depth in the literature (a small selection of references include~\cite{Pe84,Pe88,BaKn02,NaSe07,Le07,LeSp13}), the categorical approach we take here is different from alternative categorical approaches in the literature~\cite{Le06,Le07,CoSp12}, but is more closely related to a notion introduced by Accardi and Cecchini in 1982~\cite{AcCe82}.%
%footnote
\footnote{Incidentally, Accardi and Cecchini introduced \emph{two} types of Bayesian inverses in their work along with their connection to the modular group of Tomita and Takesaki, assuming that the states involved were faithful (though they did not realize the connection to Bayes' theorem at the time, as far as I am aware). Only one of these two maps seems to have been popularized---the one often called the \emph{Petz recovery map} in the literature---due to its close connection to the saturation of relative entropy under open quantum dynamics~\cite{Pe03}. Our Bayesian inverse is a generalization of the other type considered in~\cite{AcCe82}.}
%end footnote 
This latter definition is generalized appropriately to the setting of not necessarily faithful states on $C^*$-algebras.
The reason \emph{why} our analysis is done with this latter map as a form of Bayesian inversion, instead of the more standard recovery map in the literature, will be presented elsewhere~\cite{PaJeffrey}. 

Finally, it is also worthwhile mentioning that although the quantum Markov categories we define enable us to reason probabilistically via diagrammatic techniques as a form of two-dimensional algebra, this is certainly not the first of its kind. There is a large body of literature on diagrammatic reasoning in the context of categorical quantum mechanics~\cite{CoKi17,HeVi19}. The approach we take here is a technically different from these approaches, but similar in spirit. Future investigations will have to be done to properly relate the two.

%%%%%%%%%%%%%%%%%%%%%%%%%%%%%%%%%%%%%%
\vspace{2mm}
\noindent
\begin{tikzpicture}
\coordinate (L) at (-8.75,0);
\coordinate (R) at (8.75,0);
\draw[line width=2pt,orange!20] (L) -- node[anchor=center,rectangle,fill=orange!20] {\strut \large \textcolor{black}{{Topics not covered here}}} (R);
%\node[anchor=center,rectangle,fill=orange!20] at (0,0) {\strut \large \textcolor{black}{{Topics not covered here}}};
\end{tikzpicture}
%%%%%%%%%%%%%%%%%%%%%%%%%%%%%%%%%%%%%%

Although a variety of points are addressed in this work, it has not been possible to include all aspects of potential interest to some readers.  
Though the following list is not meant to be complete, here are some of the topics that we do not cover in this work. 
\begin{itemize}
\itemsep0pt
\item
The formal definition of quantum Markov categories is preliminary. It is mainly introduced to find a rigorous setting where string-diagrammatic manipulations could be applied to the quantum setting even in the absence of a physically implementable copy map. In particular, we do not explore further mathematical structures of graded monoidal categories nor do we explore examples other than the main one involving finite-dimensional unital $C^*$-algebras (for one reason, see Remark~\ref{ex:allCAlg}). Nevertheless, we speculate that it seems plausible that the notion may be extended further to braided monoidal categories, which itself may have interesting consequences for condensed matter and topological quantum computing (see~\cite{RoWa18} and references therein). 
\item
Although our main interest is in applications to quantum information theory, the present paper is meant to provide a \emph{mathematical} basis for results we plan to present elsewhere. In particular, though we try to make our results accessible to the physics community, only a few physical examples are provided. Furthermore, justification for why the Bayesian inversion presented here is a physically reasonable notion of quantum inference will be presented elsewhere~\cite{PaJeffrey}. 
\item
We do not go into great detail regarding the conditions for existence of disintegrations and Bayesian inverses, though we briefly mention very special cases. The reader is referred to~\cite{PaRu19,PaRuBayes}, which are specifically catered to these questions. 
\item
We do not review in detail the background and mathematical justification for why one should view completely positive maps as analogues of conditional probabilities. Categorical justifications are provided in earlier work~\cite{Pa17} (see also~\cite{FuJa15}).
\item
Many \emph{classical} Markov categories arise as the Kleisli category of a symmetric monoidal monad on certain categories~\cite{Fr19,FGPR20}. 
We make no attempt to describe if quantum Markov categories arise as Kleisli categories of some kind of monad. However, we mention the interesting fact that 
Westerbaan has recently shown that the category of completely positive unital maps can be viewed as the Kleisli category of a certain monad on the category of $C^*$-algebras and $*$-homorphisms~\cite{We17}. 
\end{itemize}

%%%%%%%%%%%%%%%%%%%%%%%%%%%%%%%%%%%%%%
\vspace{2mm}
\noindent
\begin{tikzpicture}
\coordinate (L) at (-8.75,0);
\coordinate (R) at (8.75,0);
\draw[line width=2pt,orange!20]  (L) -- node[anchor=center,rectangle,fill=orange!20] {\strut \large \textcolor{black}{{Detailed outline and main results}}} (R);
%\node[anchor=center,rectangle,fill=orange!20] at (0,0) {\strut \large \textcolor{black}{{Detailed outline and main results}}};
\end{tikzpicture}
%%%%%%%%%%%%%%%%%%%%%%%%%%%%%%%%%%%%%%

In Section~\ref{sec:whatisbayes}, we provide a diagrammatic reformulation of Bayes' theorem to motivate the definition we adopt later abstractly and in the quantum setting. 

In Section~\ref{sec:QMC}, we define classical and quantum Markov categories and provide the two main examples used in this work: finite sets with stochastic maps and finite-dimensional $C^*$-algebras with unital linear and conjugate-linear maps. The category of completely positive unital (CPU) maps embeds into the latter category so that string diagrammatic manipulations of CPU maps can be done in the larger category. In defining quantum Markov category, an involution morphism $*$ is introduced to allow for an axiom that resembles (but is different from) commutativity in the $C^*$-algebra setting. The morphisms that are natural with respect to this involution (called \emph{$*$-preserving}) form a subcategory with a symmetry that often simplifies many definitions. We prove that there is no quantum Markov category of all $C^*$-algebras in Remark~\ref{ex:allCAlg}, though we speculate a potential resolution in Question~\ref{ques:multilinearCAlg}. 

In Section~\ref{sec:pos}, we adapt Fritz' definition of a positive Markov category (cf.~\cite[Definition~11.22]{Fr19}) to the quantum setting. In Theorem~\ref{thm:fdcalgcpupositive}, we prove that the category of CPU maps forms a positive subcategory of the quantum Markov category of linear and conjugate-linear maps on finite-dimensional $C^*$-algebras. 
We then prove that ordinary (operator-algebraic) positivity (as opposed to complete positivity) in the quantum setting is \emph{not} enough to satisfy Fritz' categorical definition of positivity.
Since Schwarz-positive unital maps satisfy this definition, we call such subcategories S-positive instead. 
As a simple corollary, we prove a general no-broadcasting (no-cloning) theorem for S-positive subcategories in Theorem~\ref{thm:nocloning}. 

Section~\ref{sec:ae} reviews a.e.\ equivalence and contains several new results such as Theorem~\ref{thm:ncaeequivalence}, which shows that the notion of a.e.\ equivalence defined via the nullspace (which appears in the GNS construction for example) introduced in~\cite[Definition~3.16]{PaRu19} coincides with the definition of Cho--Jacobs a.e.\ equivalence~\cite[Definition~5.1]{ChJa18}. Theorem~\ref{thm:ncaeequivalence} is one of our main results relating the operator-algebraic and diagrammatic structures. It is used repeatedly to simplify many proofs in the rest of this work. 

Section~\ref{sec:aedet} is devoted to the a.e.\ determinism, which is particularly subtle in the quantum setting due to the lack of commutativity. In particular, there are subcategories of quantum Markov categories where a morphism is a.e.\ equivalent to a deterministic morphism but is itself not a.e.\ deterministic (Example~\ref{exa:Unotdetreason}). We therefore introduce the notion of a deterministically reasonable subcategory (Definition~\ref{defn:detreason}) and note that all classical Markov categories are deterministically reasonable. Furthermore, we prove that the subcategory of CPU maps between $C^*$-algebras is deterministically reasonable. We also prove a weak a.e.\ Multiplication Lemma (Lemma~\ref{lem:Attalslemmaaegeneral}) that is first used to describe what a.e.\ determinism looks like in the category of CPU maps, but it is also used to prove one of our main theorems in a later section. We conclude the section with a discussion of a.e.\ unitality. 

In Section~\ref{sec:disintbayes}, we define disintegrations and Bayesian inversion in quantum Markov categories. As an example, we show that the error map is a disintegration of the recovery map in the classical Hamming error correcting codes as well as certain quantum error correcting codes. We provide examples of Bayesian inversion for matrix algebras. We also prove some relationships among these different concepts. For example, Proposition~\ref{prop:bayesfunctorial} shows that Bayesian inversion is compositional (with no assumptions), while every $*$-preserving morphism is a Bayesian inverse of its Bayesian inverse.  Proposition~\ref{thm:Bayesianinverseofdeterministicisadisint} shows that a Bayesian inverse of an a.e.\ deterministic morphism is a disintegration. 

Section~\ref{sec:modposcaus} contains our main results, including the theorem relating disintegrations to a.e.\ determinism and Bayesian inversion. We review three of Fritz' additional axioms that classical Markov categories can satisfy. We relativize these notions to subcategories of quantum Markov categories. These three axioms are a.e.\ modularity (previously called positivity in an earlier draft of~\cite{Fr19}), strict positivity, and causality. Theorem~\ref{thm:aemodbayesdisint} provides an equivalent characterization of an a.e.\ modular category in terms of Bayesian inversion and a.e.\ determinism. Not only does this theorem provide a simpler criterion for proving that CPU maps form an a.e.\ modular category (Theorem~\ref{thm:SPUmodular}), but it is also a generalization of the Fisher--Neyman factorization theorem, which, by our result, is valid in both the classical and quantum settings (see Remark~\ref{rmk:fisherneyman} for the relation to the Fisher--Neyman factorization theorem and Example~\ref{exa:Umegakisufficientstatistic} for a relationship to Umegaki's notion of sufficiency in quantum statistical decision theory). In the process of proving Theorem~\ref{thm:SPUmodular}, we prove a relative version of the Multiplication Lemma for Schwarz-positive unital maps (Lemma~\ref{lem:relmulttheorem}) and a certain conditional expectation property (Lemma~\ref{lem:disintaehomproperty}), both of which are either well-known to operator algebraists, or may be interesting new results. Corollary~\ref{cor:binversiondisint} contains the statement that if a CPU map has a CPU disintegration, then it is automatically a Bayesian inverse and the original map is necessarily a.e.\ deterministic. The rest of the section discusses strictly positive and causal subcategories. Strict positivity is shown to fail for CPU maps (Example~\ref{exa:strictpositive}), which is why we have brought back Fritz' older axiom of a.e.\ modularity. In any case, we prove that strict positivity implies a.e.\ modularity and S-positivity in the context of quantum Markov categories (Proposition~\ref{prop:strictimplies}), illustrating that the choice of definitions in the non-commutative setting is very subtle (Remark~\ref{rmk:strictpossubtle}). The section ends by showing that CPU maps form a causal subcategory, but positive unital (PU) maps do not. Although causality is used to prove that a.e.\ equivalence classes of state-preserving morphisms form a category, the latter statement is still true for PU maps. This provides a counterexample to the suspicion made before Definition~13.8 in~\cite{Fr19}. A summary of the main results relating disintegrations and Bayesian inverses in the context of $C^*$-algebras is given in Corollary~\ref{cor:summary}.

Due to the length of this work, some tables of notation have been included in Appendix~\ref{sec:notationtables} for the reader's convenience.

%%%%%%%%%%%%%%%%%%%%%%%%%%%%%%%%%%%%%%
\section[What is Bayes' theorem?]{a001}
\label{sec:whatisbayes}
\vspace{-12mm}
\noindent
\begin{tikzpicture}
\coordinate (L) at (-8.75,0);
\coordinate (R) at (8.75,0);
\draw[line width=2pt,orange!20] (L) -- node[anchor=center,rectangle,fill=orange!20]{\strut \Large \textcolor{black}{\textbf{2\;\; What is Bayes' theorem?}}} (R);
\end{tikzpicture}
%\vspace{1mm}
%%%%%%%%%%%%%%%%%%%%%%%%%%%%%%%%%%%%%%

To provide the setting and motivation for our results, we would first like to illustrate one version of Bayes' theorem that can be described purely diagrammatically~\cite{Fo12,CuSt14,CDDG17,ChJa18,Fr19}.%
%footnote
\footnote{In most of these references, Bayes' theorem is formulated as a bijection between joint distributions and conditionals. Our emphasis is on the process of inference from conditionals, which will be used more in the non-commutative setting (the distinction between these points of view is irrelevant in the classical setting). Why this is so is explained in~\cite[Remark~5.96]{PaRuBayes}. To the best of our knowledge, the first reference that explicitly draws the diagram in Theorem~\ref{thm:classicalBayestheorem} is Fong's thesis~\cite{Fo12} (see the section ``Further Directions''), though it is formulated using string diagrams. Here, we have elevated this diagram to encapsulate what the \emph{statement} of Bayes' theorem is.}
%end footnote
 We will presently illustrate it in the case of finite sets and stochastic maps (for the reader unfamiliar with the notation, we will briefly review it after the statement of the theorem). 

\vspace{-1mm}
\begin{theo}[Bayes' theorem]{thm:classicalBayestheorem}
%\bt[Bayes' theorem]
%\label{thm:classicalBayestheorem}
Let $X$ and $Y$ be finite sets, let $\{\bullet\}\xstoch{p}X$ be a probability measure, and let $X\xstoch{f}Y$ be a stochastic map. Then there exists a stochastic map $Y\xstoch{g}X$ 
such that%
%footnote
\footnote{The equals sign in this diagram indicates that the diagram commutes. The notation is meant to be consistent with higher categorical notation. Namely, we think of this equality as the identity 2-cell. We will not comment on higher categorical generalizations in this paper.
}
%end footnote
\[
\xy0;/r.25pc/:
(0,7.5)*+{\{\bullet\}}="1";
(-25,7.5)*+{Y}="Y";
(25,7.5)*+{X}="X";
(-25,-7.5)*+{Y\times Y}="YY";
(25,-7.5)*+{X\times X}="XX";
(0,-7.5)*+{X\times Y}="XY";
{\ar@{~>}"1";"X"^{p}};
{\ar@{~>}"1";"Y"_{q}};
{\ar"Y";"YY"_{\Delta_{Y}}};
{\ar"X";"XX"^{\Delta_{X}}};
{\ar@{~>}"YY";"XY"_{g\times\id_{Y}}};
{\ar@{~>}"XX";"XY"^{\id_{X}\times f}};
(0,0)*{=\joinrel=\joinrel=};
\endxy
,
\]
where $\{\bullet\}\xstoch{q}Y$ is given by $q:=f\circ p.$ 
Furthermore, for any other $g'$ satisfying 
this condition, 
$g\underset{\raisebox{.6ex}[0pt][0pt]{\scriptsize$q$}}{=}g'$. 
%\et
\end{theo}

We quickly recall some notation to explain the theorem (see \cite{Fr19,PaRu19,Pa17} and the ``Stochastic maps'' series in~\cite{PaLinear}).

\vspace{-1mm}
\begin{defn}[Stochastic maps]{defn:stochasticmaps}
If $X$ and $Y$ are finite sets, a \define{stochastic map} 
$X\xstoch{f}Y$ is an assignment sending $x\in X$ to a probability measure $f_{x}$ on $Y$. The value of this probability measure on $y\in Y$ will be denoted by $f_{yx}$. Stochastic maps are drawn with squiggly arrows to distinguish them from deterministic maps (stochastic maps assigning Dirac delta measures), which are drawn with straight arrows $\rightarrow$.%
%footnote
\footnote{Such straight arrows correspond to functions.}
%end footnote
A single element set will be denoted by $\{\bullet\}$. Stochastic maps $X\xstoch{f}Y\xstoch{g}Z$ can be \define{composed} via the Chapman--Kolmogorov equation 
\[
(g\circ f)_{zx}:=\sum_{y\in Y}g_{zy}f_{yx}. 
\]
Given $X\xstoch{f}Y$ and $X'\xstoch{f'}Y'$, the \define{product}%
%footnote
\footnote{This is \emph{not} a categorical product (in the sense of limits).}
%end footnote
 $X\times X'\xstoch{f\times f'}Y\times Y'$ is defined by the product of probability measures
\[
(f\times f')_{(y,y')(x,x')}:=f_{yx}f_{y'x'}. 
\]
Given%
%\footnote
\footnote{A stochastic map $\{\bullet\}\xstoch{p}X$ encodes the data of a probability measure on $X$.}
%end footnote
\[
\xy0;/r.25pc/:
(-15,0)*+{\{\bullet\}}="0";
(0,0)*+{X}="X";
(15,0)*+{Y}="Y";
{\ar@{~>}"0";"X"^{p}};
{\ar@{~>}@<0.6ex>"X";"Y"^{f}};
{\ar@{~>}@<-0.6ex>"X";"Y"_{h}};
\endxy
,
\]
$f$ is \define{$p$-a.e.\ equivalent to $h$}, written $f\underset{\raisebox{.6ex}[0pt][0pt]{\scriptsize$p$}}{=}h$, whenever $f_{x}\ne h_{x}$ for all $x\in X\setminus N_{p}$, where 
\[
N_{p}:=\big\{x\in X\;:\;p_{x}=0\big\}
\]
is the \define{nullspace} of $p$. In other words, $f\underset{\raisebox{.6ex}[0pt][0pt]{\scriptsize$p$}}{=}h$ whenever the set on which $f$ and $h$ differ is a set of $p$-measure zero. 
The map $X\xrightarrow{\D_{X}}X\times X$ is determined by the function $\D_{X}(x):=(x,x)$ for all $x\in X$ and is called the \define{duplicate, copy,} or \define{diagonal} map. The morphism $g$ in Theorem~\ref{thm:classicalBayestheorem} is called a \define{Bayesian inverse} of $(f,p,q)$. The diagram in Theorem~\ref{thm:classicalBayestheorem} is called \define{the Bayes diagram}. The pair $(Y,q)$ is called a \define{finite probability space}. A stochastic map $X\xstoch{f}Y$ satisfying $q=f\circ p$ is said to be \define{probability-preserving}%
%footnote
\footnote{In the context of measure theory, where the probability measures are replaced by arbitrary measures, one often says that $f$ is \emph{measure-preserving} when it satisfies $q=f\circ p$. When using states on $C^*$-algebras (non-commutative analogues of probability measures), the notion of probability-preserving stochastic maps is replaced by that of state-preserving positive unital maps between $C^*$-algebras (cf.\ Definition~\ref{defn:disintegration}).}
%end footnote
and is often written as $(X,p)\xstoch{f}(Y,q)$. 
\end{defn}

\vspace{-1mm}
\begin{exa}[Visualizing probability-preserving functions]{exa:probpreserving}
\begin{wrapfigure}{r}{0.27\textwidth}
  \centering
    \begin{tikzpicture}[decoration=snake]
\node at (-1.75,3.25) {$X$};
\node at (-1.75,0) {$Y$};
\draw[blue,thick,fill=blue,fill opacity=0.4] (-1,3.75) circle (1.0ex);
\draw[blue,thick,fill=blue,fill opacity=0.4] (-1,3.25) circle (0.5ex);
\draw[blue,thick,fill=blue,fill opacity=0.4] (-1,2.75) circle (1.25ex);
\draw[blue,thick,fill=blue,fill opacity=0.4] (0,4.5) circle (0.25ex);
\draw[blue,thick,fill=blue,fill opacity=0.4] (0,4.0) circle (0.75ex);
\draw[blue,thick,fill=blue,fill opacity=0.4] (0,3.5) circle (1.0ex);
\draw[blue,thick,fill=blue,fill opacity=0.4] (0,3) circle (0.5ex);
\draw[blue,thick,fill=blue,fill opacity=0.4] (1,4.05) circle (1.75ex); 
\draw[blue,thick,fill=blue,fill opacity=0.4] (1,3.5) circle (0.75ex); 
\draw[blue,thick,fill=blue,fill opacity=0.4] (1,3) circle (1.0ex); 
\draw[blue,thick,fill=blue,fill opacity=0.4] (1,2.5) circle (1.25ex); 
\draw[blue,thick,fill=blue,fill opacity=0.4] (2,3.75) circle (0.75ex); 
\draw[blue,thick,fill=blue,fill opacity=0.4] (2,3.25) circle (0.25ex); 
\draw[blue,thick,fill=blue,fill opacity=0.4] (2,2.75) circle (0.5ex); 
\draw[blue,thick,fill=blue,fill opacity=0.4] (2,2.25) circle (0.75ex); 
\draw[-{>[scale=2.5,
          length=2,
          width=3]},thick] 
          (0.25,2.0) -- node[left]{$f$} (0.25,0.5);
\draw[blue,thick,fill=blue,fill opacity=0.4] (-1,0) circle (1.677ex);
\draw[blue,thick,fill=blue,fill opacity=0.4] (0,0) circle (1.3693ex);
\draw[blue,thick,fill=blue,fill opacity=0.4] (1,0) circle (2.48746ex);
\draw[blue,thick,fill=blue,fill opacity=0.4] (2,0) circle (1.19895ex);
\end{tikzpicture}
\end{wrapfigure}
In the special case that a stochastic map $X\xstoch{f}Y$ corresponds to a function $X\xrightarrow{f}Y$ (with the same letter used abusively to denote both the function and the stochastic map), it is helpful to visualize a probability-preserving map $(X,p)\xrightarrow{f}(Y,q)$ in terms of combining water droplets as in the figure on the right~\cite{Gr14},~\cite[Section~2.2]{PaRu19}. A Bayesian inverse $g$ of $(f,p,q)$ in this case is sometimes called a \emph{disintegration} of $(f,p,q)$. It is a stochastic map that splits the water droplets back into the form above. More precisely, it is a stochastic map that satisfies $p=g\circ q$ and $f\circ g\aeequals{q}\id_{Y}$. These two equations are what define a disintegration. They will be discussed in more detail in Section~\ref{sec:disintbayes}.
\end{exa}

With all this notation explained, the reader can now verify that the Bayes diagram reads
\be
\label{eq:Bayesonpoints}
g_{xy}q_{y}=f_{yx}p_{x}
\ee
for all values of $x\in X$ and $y\in Y$. This is Bayes' rule for point events.%
%footnote
\footnote{If we set $P(x|y):=g_{xy},P(y):=q_{y},P(y|x):=f_{yx},$ and $P(x):=p_{x}$,  equation~(\ref{eq:Bayesonpoints}) reads  $P(x|y)P(y)=P(y|x)P(x)$ in more standard (albeit abusive) notation.}
%end footnote
 The case of Bayes' rule for more general events is a simple consequence of this rule~\cite{PaJeffrey,PaRuBayes}. One can also show that both $f$ and $g$ are probability-preserving using the Bayes diagram (we will prove this more abstractly in Lemma~\ref{lem:Bayesianinversespreservestates}). 

Though it is more common to see Bayes' rule written more explicitly as $P(A|B)P(B)=P(B|A)P(A)$ in the probability context, the abstract diagrammatic reformulation illustrates that this is just one instantiation of Bayes' theorem. Classically, it is used to make inferences on outcomes based on evidence, such as diagnosing illnesses~\cite{MeRo55,We10} (see~\cite{Stone13} for a lucid introduction), it is the foundation of many machine learning algorithms~\cite{Ma92,Ma03}, and it drives how intelligent beings make decisions~\cite{Ja19}. The diagrammatic viewpoint allows one to use the more abstract concept as a \emph{definition} (as opposed to a theorem) in a completely new context, where an equation such as $P(A|B)P(B)=P(B|A)P(A)$ might not make any sense, but the diagram itself has meaning~\cite{PaRu19,PaRuBayes}. 

Considering how ubiquitous Bayes' rule is, it is very possible that we have only scratched the surface with its applications. What new insight can such a reformulation teach us? Where else can it be utilized, and how can we interpret it? Before we can answer these questions, we would like to to first provide an appropriate categorical structure in which Bayes' rule (and other concepts) can be defined in such a way so that it makes sense both in classical and quantum theory. 

\pagebreak
%%%%%%%%%%%%%%%%%%%%%%%%%%%%%%%%%%%%%%
\section[Quantum Markov categories]{a003}
\label{sec:QMC}
\vspace{-12mm}
\noindent
\begin{tikzpicture}
\coordinate (L) at (-8.75,0);
\coordinate (R) at (8.75,0);
\draw[line width=2pt,orange!20] (L) -- node[anchor=center,rectangle,fill=orange!20]{\strut \Large \textcolor{black}{\textbf{3\;\; Quantum Markov categories}}} (R);
\end{tikzpicture}
%\vspace{1mm}
%%%%%%%%%%%%%%%%%%%%%%%%%%%%%%%%%%%%%%

We begin by defining our main categories of study and then working through a few examples. 
The first definition (Definition~\ref{defn:gradedmonoidalcategories}) contains a few technical details and can be skipped on a first reading. These details are merely included to make rigorous sense of the string diagrams that will follow. Quantum Markov categories are defined in Definition~\ref{defn:qmc}. 
In what follows, let $\Z_{2}=\{0,1\}$ be the abelian group where $0$ is the identity and $1+1:=0$ (addition modulo $2$). Given any group $G$, let $\mathbb{B} G$ be the one object category whose set of morphisms equals $G$ with composition given by the group operation (see~\cite{Pa182dgt,BaLa11} for more details). We will always write $0$ for the identity element of $G$. The next definition of a graded monoidal category is due to Fr\"ohlich and Wall~\cite[Chapter~3]{FrWa74}. 

%\bd
%\label{defn:gradedmonoidalcategories}
%\pagebreak
%\phantom{why do I need this here?}
%\vspace{-7mm}
\begin{defn}[Graded monoidal category]{defn:gradedmonoidalcategories}
Let $G$ be a group. A category $\mC$ equipped with a functor $\mathfrak{g}:\mC\to\mathbb{B}G$ is called a \define{$G$-graded category}.%
%footnote
\footnote{This is not to be confused with the notion of graded fusion category appearing in conformal field theories and topological quantum matter for example~\cite[Section~2.3]{ENO09}, \cite{RoWa18}. A $G$-graded category says that the \emph{composite} of a pair of composable morphisms, where one is of grade $g'$ and another following it of grade $g$, is a morphism of grade $gg'$. A graded fusion category, on the other hand, says that the \emph{tensor product} of a morphism of grade $g$ and another of grade $g'$ is a morphism of grade $gg'$.}
%end footnote
The functor $\mathfrak{g}$ is called a \define{grading} on $\mC$ and $\mathfrak{g}(f)$ of a morphism $f$ in $\mC$ is called the \define{grade} of $f$. 
A grading $\mathfrak{g}$ is \define{stable} iff for any object $X$ in $\mC$ and for any $\g\in G$, there exists an isomorphism in $\mC$ with source $X$ and grade $\g$. 
A collection of morphisms that is of a single grade is said to be \define{homogeneous}. 
If $H$ is another group and $(\mathcal{D},\mathfrak{h}:\mathcal{D}\to\mathbb{B}H)$ is an $H$-graded category, a morphism of graded categories consists of a group homomorphism $\kappa:G\to H$ together with a functor $L:\mC\to\mD$ such that the grades of morphisms are preserved, i.e.\
\[
\xy0;/r.20pc/:
(-10,7.5)*+{\mC}="1";
(10,7.5)*+{\mD}="2";
(-10,-7.5)*+{\mathbb{B}G}="3";
(10,-7.5)*+{\mathbb{B}H}="4";
{\ar"1";"2"^{L}};
{\ar"3";"4"_{\mathbb{B}\kappa}};
{\ar"1";"3"_{\mathfrak{g}}};
{\ar"2";"4"^{\mathfrak{h}}};
\endxy
\]
commutes. Given two $G$-graded categories $(\mC,\mathfrak{g})$ and $(\mC',\mathfrak{g}')$, let $\mC{}_{\mathfrak{g}}\!\!\times_{\!\mathfrak{g}'}\!\mC'$ denote the (strict) pullback 
%\be
%\label{eq:pullbackGgraded}
\[
\xy0;/r.20pc/:
(-10,7.5)*+{\mC{}_{\mathfrak{g}}\!\!\times_{\!\mathfrak{g}'}\!\mC'}="1";
(10,7.5)*+{\mC'}="2";
(-10,-7.5)*+{\mC}="3";
(10,-7.5)*+{\mathbb{B}G}="4";
(-5,4)*{\lrcorner};
{\ar"1";"2"^(0.65){\pi'}};
{\ar"3";"4"_(0.45){\mathfrak{g}}};
{\ar"1";"3"_{\pi}};
{\ar"2";"4"^{\mathfrak{g}'}};
\endxy
\quad,
\]
%\ee
which is more explicitly given by the category whose objects are pairs $(X,X')$ with $X$ in $\mC$ and $X'$ in $\mC'$ and whose morphisms are pairs of morphisms $(f,f')$ with the \emph{same} grading (the $\pi$ and $\pi'$ functors are the projections onto the respective factors). 
Thus, $\mC{}_{\mathfrak{g}}\!\!\times_{\!\mathfrak{g}'}\!\mC'$ inherits a canonical $G$-grading.
A \define{$G$-monoidal category} consists of a $G$-graded category $(\mC,\mathfrak{g})$ with a stable grading, a morphism $\otimes:\mC{}_{\mathfrak{g}}\!\!\times_{\!\mathfrak{g}}\!\mC\to\mC$ of graded categories, a section $I:\mathbb{B}G\to\mC$ of $\mathfrak{g}$, and natural isomorphisms (of grade $0$) $\a:X\otimes(Y\otimes Z)\to(X\otimes Y)\otimes Z$, $c:X\otimes Y\to Y\otimes X$, and $i:I\otimes X\to X$ satisfying the usual axioms of a symmetric monoidal category. Note that $I$ also refers to the image of the single object in $\mathbb{B}G$ under the functor $I$. In this entire paper, the groups $G$ will always be either the trivial group or $\Z_{2}$. When the group is $\Z_{2}$, we will use \define{even} and \define{odd} to denote grade $0$ and grade $1$, respectively. In this case, $\mC_{\mathrm{even}}$, the collection all objects of $\mC$ and their even morphisms, is a subcategory%
%footnote
\footnote{The collection of odd morphisms is \emph{not} a subcategory!}
%end footnote
 of $\mC$. 
%\ed
\end{defn}

The idea behind a $G$-monoidal category $\mC$ is to endow $\mC$ with a \emph{partially defined} tensor product, where one is only allowed to take tensor products of morphisms of equal degrees (see also~\cite[Proposition~10.1]{FrWa74} and the discussion that follows for an alternative viewpoint). The following example illustrates how this works. 

%\bx
%\label{ex:linearandantilineartensor}
%\vspace{-1mm}
\begin{exa}[Linear and conjugate linear maps as a $\Z_{2}$-monoidal category]{ex:linearandantilineartensor}
The category of complex vector spaces together with the class of linear and conjugate-linear maps can be endowed with a $\Z_{2}$-monoidal structure. Recall, a function $V\xrightarrow{f} W$ is \define{conjugate-linear} iff $f$ is additive and $f(\l v)=\ov\l f(v)$ for all $v\in V$ and $\l\in\C$ ($\ov\l$ denotes the complex conjugate of $\l$). If we declare linear maps to be grade $0$ and conjugate-linear maps to be grade $1$, then the grade of their composites obey modular $2$ arithmetic. The tensor product of linear maps is defined in the usual way. The tensor product of conjugate-linear maps can be defined similarly~\cite[Section~9.2.1]{Uh16}. However, if $V\xrightarrow{f}W$ is linear and $X\xrightarrow{g}Y$ is conjugate-linear, then it does not make sense to define $f\otimes g$ since $(\l v)\otimes x=v\otimes (\l x)$, while $\big(\l f(v)\big)\otimes g(x)\ne f(v)\otimes\big(\ov\l g(x)\big).$ If all linear maps have grade $0\in\Z_{2}$ and all conjugate-linear maps have grade $1\in\Z_{2}$, then this shows that the tensor product is actually defined on the pullback $\mC{}_{\mathfrak{g}}\!\!\times_{\!\mathfrak{g}'}\!\mC'$ from Definition~\ref{defn:gradedmonoidalcategories}. The section $I:\mathbb{B}\Z_{2}\to\mC$ in this case sends $0$ to $\id_{\C}$ and $1$ to $*_{\C}$, the complex conjugation map from $\C$ to itself. The grading is stable because every complex vector space $V$ admits a real structure.%
%footnote
\footnote{Choose a basis $\{e_{\a}\}$ of $V$ and define the conjugate-linear map $V\to V$ uniquely determined by $\l e_{\a}\mapsto\ov\l e_{\a}$ for all $\l\in\C$ and $\a$ in the index set for the basis. This isomorphism has grade $1$.}
%end footnote 
%\ex
\end{exa}

The definition of a quantum Markov category below will use the language of string diagrams to present its axioms~\cite{Pe71,JSV96,Se10,CoKi17,HeVi19,Fr19} (and the reference closest to our specific usage of these diagrams is that of Cho--Jacobs~\cite{ChJa18} and Fritz~\cite{Fr19}). In particular, we assume the standard string diagram expressions for the composition and monoidal product of morphisms in series and in parallel, respectively. 
Afterwards, we will translate the axioms provided below in the example of $C^*$-algebras so that the reader unfamiliar with how they are manipulated should still be able to follow most of the results developed here. 

\vspace{-1mm}
\begin{conv}[Directionality of our string diagrams]{conv:directionality}
Our convention for string diagrams, regardless of the category, is that \emph{time always goes up} the page. However, the direction of \emph{composition} will \emph{either} go up or down depending on the specific category being used. Composition will go \emph{up} the page in the Schr\"odinger picture (the evolution of states). Composition will also always go up the page for all general definitions. This is to be consistent with the literature on Markov and CD categories because their emphasis is more often on the classical aspects.
In the quantum setting, however, we use the Heisenberg picture (the evolution of observables) so that composition will go \emph{down} the page when working with linear (or conjugate-linear) maps between $C^*$-algebras. 
Our two main reasons for doing this is because the definition of determinism used later (Definition~\ref{eq:deterministicmap}) is meant to agree with the classical one~\cite{Pa17,FuJa15} and because traces/partial traces need not be defined in the general setting of (possibly infinite-dimensional) $C^*$-algebras. 
To avoid potential confusion, we will \emph{never} use the Schr\"odinger picture for $C^*$-algebras. 
\end{conv}

%\bd
\vspace{-1mm}
\begin{defn}[Quantum Markov category]{defn:qmc} %uncomment this when you want to highlight green box
%\label{defn:qmc} 
A \define{quantum copy-discard (CD) category} is a $\Z_{2}$-monoidal category $\mM_{\text{\Yinyang}}$ together with a family of morphisms $\D_{X}:X\stoch X\times X$, $!_{X}:X\stoch I$, and $*_{X}:X\stoch X$, all depicted in string diagram notation as 
\[
\D_{X}\;\equiv\;
\vcenter{\hbox{%
\begin{tikzpicture}[font=\footnotesize]
\node[copier] (c) at (0,0.5) {};
\draw (c)
to[out=15,in=-90] (0.44,1.0);
\draw (c)
to[out=165,in=-90] (-0.44,1.0);
\draw (c) to (0,0);
\node at (0.25,0.15) {$X$};
\end{tikzpicture}}}
\qquad\text{,}\qquad
!_{X}\;\equiv\;
\vcenter{\hbox{%
\begin{tikzpicture}[font=\footnotesize]
\node[discarder] (d) at (0,0.25) {};
\draw (d) to (0,-0.45);
\node at (0.25,-0.3) {$X$};
\end{tikzpicture}}}
\quad,\;\;\;\text{and}\qquad
*_{X}\;\equiv\;
\vcenter{\hbox{%
\begin{tikzpicture}[font=\footnotesize]
\node[white dot] (s) at (0,0.5) {};
\draw (c)
to (0,1.0);
\draw (c) to (0,0);
\node at (0.25,0.15) {$X$};
\end{tikzpicture}}}
,
\]
for all objects $X$ in $\mM_{\text{\Yinyang}}$. 
These morphisms are required to satisfy the following conditions
\be
\label{eq:markovcatfirstconditions}
\vcenter{\hbox{%
\begin{tikzpicture}[font=\footnotesize]
\node[copier] (c) at (0,0) {};
\coordinate (x1) at (-0.3,0.3);
\node[discarder] (d) at (x1) {};
\coordinate (x2) at (0.3,0.5);
\draw (c) to[out=165,in=-90] (x1);
\draw (c) to[out=15,in=-90] (x2);
\draw (c) to (0,-0.3);
\end{tikzpicture}}}
\quad=\quad
\vcenter{\hbox{%
\begin{tikzpicture}[font=\footnotesize]
\draw (0,0) to (0,0.8);
\end{tikzpicture}}}
\quad=\quad
\vcenter{\hbox{%
\begin{tikzpicture}[font=\footnotesize]
\node[copier] (c) at (0,0) {};
\coordinate (x1) at (-0.3,0.5);
\coordinate (x2) at (0.3,0.3);
\node[discarder] (d) at (x2) {};
\draw (c) to[out=165,in=-90] (x1);
\draw (c) to[out=15,in=-90] (x2);
\draw (c) to (0,-0.3);
\end{tikzpicture}}}
\qquad\qquad
\vcenter{\hbox{%
\begin{tikzpicture}[font=\footnotesize]
\node[copier] (c2) at (0,0) {};
\node[copier] (c1) at (-0.3,0.3) {};
\draw (c2) to[out=165,in=-90] (c1);
\draw (c2) to[out=15,in=-90] (0.4,0.6);
\draw (c1) to[out=165,in=-90] (-0.6,0.6);
\draw (c1) to[out=15,in=-90] (0,0.6);
\draw (c2) to (0,-0.3);
\end{tikzpicture}}}
\quad=\quad
\vcenter{\hbox{%
\begin{tikzpicture}[font=\footnotesize]
\node[copier] (c2) at (-0.5,0) {};
\node[copier] (c1) at (-0.2,0.3) {};
\draw (c2) to[out=15,in=-90] (c1);
\draw (c2) to[out=165,in=-90] (-1,0.6);
\draw (c1) to[out=165,in=-90] (-0.5,0.6);
\draw (c1) to[out=15,in=-90] (0.1,0.6);
\draw (c2) to (-0.5,-0.3);
\end{tikzpicture}}}
\qquad\qquad
\vcenter{\hbox{%
\begin{tikzpicture}[font=\footnotesize]
\node[copier] (c) at (0,0.4) {};
\node[star] (R) at (0.25,0.7) {};
\node[star] (L) at (-0.25,0.7) {};
\draw (c)
to[out=15,in=-90] (R);
\draw (R)
to[out=90,in=-90] (-0.25,1.4);
\draw (c)
to[out=165,in=-90] (L);
\draw (L)
to[out=90,in=-90] (0.25,1.4);
\draw (c) to (0,0.1);
\end{tikzpicture}}}
\quad=\quad
\vcenter{\hbox{%
\begin{tikzpicture}[font=\footnotesize]
\node[copier] (c) at (0,0.3) {};
\node[star] (s) at (0,0) {};
\draw (c)
to[out=15,in=-90] (0.25,0.8);
\draw (c)
to[out=165,in=-90] (-0.25,0.8);
\draw (c) to (s);
\draw (s) to (0,-0.3);
\end{tikzpicture}}}
\tag{QCD1}
\ee
%
%second line
%
\be
\label{eq:markovcatsecondconditions}
\vcenter{\hbox{%
\begin{tikzpicture}[font=\footnotesize]
\node[discarder] (d) at (0,0) {};
\draw (d) to +(0,-0.5);
\node at (0.5,-0.3) {$X\otimes Y$};
\end{tikzpicture}}}
=
\vcenter{\hbox{%
\begin{tikzpicture}[font=\footnotesize]
\node[discarder] (d) at (0,0) {};
\node[discarder] (d2) at (0.6,0) {};
\draw (d) to +(0,-0.5);
\draw (d2) to +(0,-0.5);
\node at (0.2,-0.3) {$X$};
\node at (0.8,-0.3) {$Y$};
\end{tikzpicture}}}
\qquad\quad
\vcenter{\hbox{%
\begin{tikzpicture}[font=\footnotesize]
\node[discarder] (d) at (0,0) {};
\draw (d) to +(0,-0.5);
\node at (0.2,-0.3) {$I$};
\end{tikzpicture}}}
=\;
\vcenter{\hbox{%
\begin{tikzpicture}[font=\footnotesize]
\node at (0.2,-0.05) {};
\draw [gray,dashed] (0,0) rectangle (0.45,0.65);
\end{tikzpicture}}}
\qquad\quad
\vcenter{\hbox{%
\begin{tikzpicture}[font=\footnotesize]
\node[copier] (c) at (0,0.4) {};
\draw (c)
to[out=15,in=-90] (0.25,0.7);
\draw (c)
to[out=165,in=-90] (-0.25,0.7);
\draw (c) to (0,0);
\node at (0.5,0.1) {$X\otimes Y$};
\end{tikzpicture}}}
=
\vcenter{\hbox{%
\begin{tikzpicture}[font=\footnotesize]
\node[copier] (c) at (0,0) {};
\node[copier] (c2) at (0.4,0) {};
\draw (c) to[out=15,in=-90] +(0.5,0.45);
\draw (c) to[out=165,in=-90] +(-0.4,0.45);
\draw (c) to +(0,-0.4);
\draw (c2) to[out=15,in=-90] +(0.4,0.45);
\draw (c2) to[out=165,in=-90] +(-0.5,0.45);
\draw (c2) to +(0,-0.4);
\node at (-0.2,-0.3) {$X$};
\node at (0.6,-0.3) {$Y$};
\end{tikzpicture}}}
\qquad\quad
\vcenter{\hbox{%
\begin{tikzpicture}[font=\footnotesize]
\node[copier] (c) at (0,0.4) {};
\draw (c)
to[out=15,in=-90] (0.25,0.7);
\draw (c)
to[out=165,in=-90] (-0.25,0.7);
\draw (c) to (0,0);
\node at (0.2,0.1) {$I$};
\end{tikzpicture}}}
=\;
\vcenter{\hbox{%
\begin{tikzpicture}[font=\footnotesize]
\node at (0.2,-0.05) {};
\draw [gray,dashed] (0,0) rectangle (0.45,0.75);
\end{tikzpicture}}}
\tag{QCD2}
\ee
%%%%%end of second line
\be
\label{eq:lastsetformarkovcatidentities}
\vcenter{\hbox{%
\begin{tikzpicture}[font=\small]
\node[star] (s1) at (0,0.15) {};
\node[star] (s2) at (0,0.55) {};
\draw (0,-0.15) to (s1);
\draw (s1) to (s2);
\draw (s2) to (0,0.85);
\end{tikzpicture}}}
\quad=\quad
\vcenter{\hbox{%
\begin{tikzpicture}[font=\footnotesize]
\draw (0,-0.15) to (0,0.85);
\end{tikzpicture}}}
\qquad\qquad
\vcenter{\hbox{%
\begin{tikzpicture}[font=\footnotesize]
\node[star] (s) at (0,0) {};
\draw (s) to +(0,-0.5);
\draw (s) to +(0,0.5);
\node at (0.6,-0.3) {$X\otimes Y$};
\end{tikzpicture}}}
=
\vcenter{\hbox{%
\begin{tikzpicture}[font=\footnotesize]
\node[star] (s) at (0,0) {};
\node[star] (s2) at (0.6,0) {};
\draw (s) to +(0,-0.5);
\draw (s) to +(0,0.5);
\draw (s2) to +(0,-0.5);
\draw (s2) to +(0,0.5);
\node at (0.2,-0.3) {$X$};
\node at (0.8,-0.3) {$Y$};
\end{tikzpicture}}}
\qquad
\quad
\vcenter{\hbox{%
\begin{tikzpicture}[font=\footnotesize]
\node[discarder] (d) at (0,0) {};
\node at (0,0.6) {};
\node[star] (s) at (0,-0.3) {};
\draw (d) to (s);
\draw (s) to (0,-0.6);
\node at (0.25,-0.5) {$X$};
\end{tikzpicture}}}
=\;
\vcenter{\hbox{%
\begin{tikzpicture}[font=\footnotesize]
\node[discarder] (d) at (0,-0.2) {};
\node[star] (s) at (0,0.3) {};
\draw (0,-0.6) to (d);
\draw [gray,dashed] (d) to (s);
\draw [gray,dashed] (s) to (0,0.6);
\node at (0.25,-0.5) {$X$};
\end{tikzpicture}}}
\quad.
\tag{QCD3}
\ee
The morphisms $\id_{X}$, $\D_{X}$, and $!_{X}$ are declared to be even for all $X$. The involutions $*_{X}$ are declared to be odd.%
%footnote
\footnote{A $\Z_{2}$-monoidal category has the \emph{property} that the grading is stable. In a quantum Markov category, the choice of a representative $*_{X}$ is additional \emph{structure}.}
%end footnote 
The map $\D_{X}$ is sometimes called \define{copy} or \define{duplicate} and 
the map $!_{X}$ is sometimes called \define{delete} or \define{ground}. 
If there is a subcategory $\mC$ of $\mM_{\text{\Yinyang}}$ that is also a quantum CD category but satisfies, in addition, 
\be
\label{eq:classicalqmcaxiom}
\vcenter{\hbox{%
\begin{tikzpicture}[font=\footnotesize]
\node[copier] (c) at (0,0.3) {};
\node[star] (s) at (0,0) {};
\node at (-0.2,-0.15) {$X$};
\node at (-0.4,0.55) {$X$};
\node at (0.4,0.55) {$X$};
\draw (c)
to[out=15,in=-90] (0.25,0.7);
\draw (c)
to[out=165,in=-90] (-0.25,0.7);
\draw (c) to (s);
\draw (s) to (0,-0.3);
\end{tikzpicture}}}
\quad=\quad
\vcenter{\hbox{%
\begin{tikzpicture}[font=\footnotesize]
\node[copier] (c) at (0,0) {};
\node[star] (s1) at (-0.25,0.35) {};
\node[star] (s2) at (0.25,0.35) {};
\node at (-0.2,-0.15) {$X$};
\node at (-0.4,0.55) {$X$};
\node at (0.4,0.55) {$X$};
\draw (c)
to[out=15,in=-90] (s2);
\draw (c)
to[out=165,in=-90] (s1);
\draw (0,-0.3) to (c);
\draw (s1) to (-0.25,0.7);
\draw (s2) to (0.25,0.7);
\end{tikzpicture}}}
\qquad\forall\;X 
\tag{CD1}
\ee
then $\mC_{\mathrm{even}}$ is said to be a \define{classical CD subcategory} of $\mM_{\text{\Yinyang}}$. In general, a \define{classical CD category} is a symmetric monoidal category admitting all the structure above except that the grading is trivial (so that the involution is not present) and the commutativity axiom
\be
\label{eq:commutativity}
\vcenter{\hbox{%
\begin{tikzpicture}[font=\small]
\node[copier] (c) at (0,0.4) {};
\draw (c)
to[out=15,in=-90] (0.25,0.65)
to[out=90,in=-90] (-0.25,1.2);
\draw (c)
to[out=165,in=-90] (-0.25,0.65)
to[out=90,in=-90] (0.25,1.2);
\draw (c) to (0,0.1);
\end{tikzpicture}}}
\quad=\quad
\vcenter{\hbox{%
\begin{tikzpicture}[font=\small]
\node[copier] (c) at (0,0.4) {};
\draw (c)
to[out=15,in=-90] (0.25,0.7);
\draw (c)
to[out=165,in=-90] (-0.25,0.7);
\draw (c) to (0,0.1);
\end{tikzpicture}}}
\tag{CD2}
\ee
holds for all objects.  
%\ed
A \define{quantum Markov category} is a quantum CD category in which 
every even (odd) morphism $X\xstoch{f}Y$ in $\mM_{\text{\Yinyang}}$ satisfies%
%footnote
\footnote{This axiom is naturality (in the sense of natural transformations) with respect to the assignment that sends each $X$ to the morphism $!_{X}$~\cite[Equation~(2.5) in Definition~2.1]{Fr19}.
}
%end footnote
 the condition that the composite $X\xstoch{f}Y\xstoch{!_{Y}}I$ is equal to $X\xstoch{!_{X}}I$ ($X\xstoch{*_{X}}X\xstoch{!_{X}}I$). In pictures, 
\[
\vcenter{\hbox{%
\begin{tikzpicture}[font=\small]
\node[arrow box] (c) at (0,0) {$f$};
\node[discarder] (d) at (0,0.5) {};
\draw  (c) to (d);
\draw (c) to (0,-0.5);
\end{tikzpicture}}}
\;=\;
\vcenter{\hbox{%
\begin{tikzpicture}[font=\small]
\node[discarder] (d) at (0,0) {};
\node at (0,0.5) {};
\draw (d) to (0,-0.5);
\end{tikzpicture}}}
\qquad
\left(
\vcenter{\hbox{%
\begin{tikzpicture}[font=\small]
\node[arrow box] (c) at (0,0) {$f$};
\node[discarder] (d) at (0,0.5) {};
\draw  (c) to (d);
\draw (c) to (0,-0.5);
\end{tikzpicture}}}
\;=\;
\vcenter{\hbox{%
\begin{tikzpicture}[font=\small]
\node[discarder] (d) at (0,0.1) {};
\node at (0,0.5) {};
\node[star] (s) at (0,-0.2) {};
\draw (d) to (s);
\draw (s) to (0,-0.5);
\end{tikzpicture}}}
\right)
\quad.
\]
Morphisms $f$ satisfying this condition are called \define{unital}.%
%footnote
\footnote{They are often called ``causal'' in the literature~\cite{KiUi19}. We have chosen to call these maps unital to avoid potential confusion with Definition~\ref{defn:causal}.}
%end footnote
A \define{classical Markov category} is a classical CD category for which every morphism is unital. 
\end{defn}

\vspace{-1mm}
\begin{rmk}[Terminology of CD and Markov categories]{rmk:CDMarkovcausal}
The terminology `Markov category' was first used by Fritz~\cite{Fr19}. The terminology `CD category' was used earlier by Cho--Jacobs, which is also where the axioms were first provided~\cite{ChJa18}. The only distinction between the two is whether grounding is natural for every morphism. 
Similar, though not quite the same, structure appeared earlier in works of Carboni~\cite{Ca87} and Golubtsov~\cite{Go02}. A more thorough historical account can be found in~\cite{Fr19}. 
In this paper, we mostly focus on Markov categories, though on occasion we provide digressions on what happens when unitality is dropped. We also prefer the terminology `Markov category' because this sounds more appropriate for our generalization to the non-commutative context.%
%footnote
\footnote{In quantum mechanics, the operation copy (C) is not a quantum operation. Hence, if we used `non-commutative CD categories' or `quantum CD categories' in this work, this might cause some speculation from the quantum information community (cf.\ Example~\ref{ex:CPUnotQMC} and Theorem~\ref{thm:nocloning}).} 
%end footnote
\end{rmk}

\vspace{-1mm}
\begin{rmk}[Classical Markov categories and naturality of the swap map]{a004}
The usual commutativity axiom~(\ref{eq:commutativity})
of a classical CD category is a consequence of the axioms of a quantum CD category and~(\ref{eq:classicalqmcaxiom}). This follows from 
\[
\vcenter{\hbox{%
\begin{tikzpicture}[font=\small]
\node[copier] (c) at (0,0.4) {};
\draw (c)
to[out=15,in=-90] (0.25,0.65)
to[out=90,in=-90] (-0.25,1.2);
\draw (c)
to[out=165,in=-90] (-0.25,0.65)
to[out=90,in=-90] (0.25,1.2);
\draw (c) to (0,0.1);
\end{tikzpicture}}}
\quad
\overset{*^2=\id}{=\joinrel=\joinrel=}
\quad
\vcenter{\hbox{%
\begin{tikzpicture}[font=\small]
\node[star] (s1) at (0,-0.2) {};
\node[star] (s2) at (0,0.1) {};
\node[copier] (c) at (0,0.4) {};
\draw (c)
to[out=15,in=-90] (0.25,0.65)
to[out=90,in=-90] (-0.25,1.2);
\draw (c)
to[out=165,in=-90] (-0.25,0.65)
to[out=90,in=-90] (0.25,1.2);
\draw (s2) to (c);
\draw (s1) to (s2);
\draw (0,-0.4) to (s1);
\end{tikzpicture}}}
\quad
\overset{\text{(\ref{eq:classicalqmcaxiom})}}{=\joinrel=\joinrel=}
\quad
\vcenter{\hbox{%
\begin{tikzpicture}[font=\small]
\node (s1) at (0,-0.3) {};
\node[star] (s2) at (0,0.1) {};
\node[copier] (c) at (0,0.4) {};
\node[star] (sR) at (0.25,0.7) {};
\node[star] (sL) at (-0.25,0.7) {};
\draw (c) to[out=15,in=-90] (sR);
\draw (sR) to[out=90,in=-90] (-0.25,1.4);
\draw (c) to[out=165,in=-90] (sL);
\draw (sL) to[out=90,in=-90] (0.25,1.4);
\draw (s2) to (c);
\draw (s1) to (s2);
\end{tikzpicture}}}
\quad
\overset{\text{(\ref{eq:markovcatfirstconditions})}}{=\joinrel=\joinrel=\joinrel=}
\quad
\vcenter{\hbox{%
\begin{tikzpicture}[font=\footnotesize]
\node[copier] (c) at (0,0.3) {};
\node[star] (s) at (0,0) {};
\node[star] (sb) at (0,-0.3) {};
\draw (c)
to[out=15,in=-90] (0.25,0.8);
\draw (c)
to[out=165,in=-90] (-0.25,0.8);
\draw (c) to (s);
\draw (s) to (sb);
\draw (sb) to (0,-0.6);
\end{tikzpicture}}}
\quad
\overset{*^2=\id}{=\joinrel=\joinrel=}
\quad
\vcenter{\hbox{%
\begin{tikzpicture}[font=\small]
\node[copier] (c) at (0,0.4) {};
\draw (c)
to[out=15,in=-90] (0.25,0.7);
\draw (c)
to[out=165,in=-90] (-0.25,0.7);
\draw (c) to (0,0.1);
\end{tikzpicture}}}
\quad.
\]
Conversely, (\ref{eq:markovcatfirstconditions}), (\ref{eq:lastsetformarkovcatidentities}), and (\ref{eq:commutativity}) imply (\ref{eq:classicalqmcaxiom}). 
\end{rmk}

\vspace{-1mm}
\begin{rmk}[The grading combines in composition and restricts the tensor product]{a005}
The choice of a functor $\mM_{\text{\Yinyang}}\to\mathbb{B}\Z_{2}$ means that the composite of two morphisms of parities $p_{1}$ and $p_{2}$ is of parity $(p_{1}+p_{2})\mod2.$ 
Pre- or post-composing with $*$ sets up two bijections 
$\mC_{\mathrm{even}}(X,Y)\to\mC_{\mathrm{odd}}(X,Y)$. 
The distinction between even and odd morphisms seems like it might make it a bit awkward for string diagram computations. However, we will see that all string diagram computations will be done in a manner where they pass a ``horizontal line test,'' namely where the morphisms at any height in the string diagram will always have the same degree. Also note that we have to keep track of $*_{I}$ in computations, especially whenever we pull $*_{X}$ through $!_{X}$ as in the last identity in~(\ref{eq:lastsetformarkovcatidentities}). Fortunately, this will rarely show up in string-diagrammatic computations (an exception is Remark~\ref{rmk:CJdisintcontinued}).  
\end{rmk}

The reason to include the odd involution $*$ is to generalize the computations from ordinary Markov categories and classical probability theory~\cite{Fr19,ChJa18} to categories of quantum probability (cf.\ Example~\ref{exa:fdcalgmarkovcat} below). To see this, we first review the classical example. 

\begin{exa}[Stochastic matrices ($\FinStoch$) and Markov kernels ($\Stoch$)]{ex:FinStoch}
%\bx
Our main examples of classical Markov categories are $\FinStoch$ and $\Stoch$. An object of $\FinStoch$ is a finite set. A morphism from $X$ to $Y$ is a Markov kernel/stochastic map/conditional probability from $X$ to $Y$. Such a morphism assigns to each element $x\in X$ a probability measure on $Y$. Composition is defined by the Chapman--Kolmogorov equation (i.e.\ summing over all intermediaries). The tensor product is the cartesian product of sets and the product of Markov kernels for morphisms. The tensor unit is the single element set, often denoted by $\{\bullet\}$. The maps $\D_{X}$ and $!_{X}$ are given by $\D_{X}(x):=(x,x)$ and $!_{X}(x)=\bullet$ for all $x\in X$. 
Notice that axiom~(\ref{eq:commutativity}) in Definition~\ref{defn:qmc} holds.
See Section~\ref{sec:whatisbayes} above, \cite[Example~2.5]{Fr19}, and \cite[Section~2.1]{PaRu19} for more details.
One can also drop the condition that a morphism sends each point to a probability measure and instead associate to each point a signed (finite) measure whose total value is $1$. The resulting category is also a classical Markov category (see Example~11.27 in~\cite{Fr19}). One can also weaken this by dropping the condition that the total measure is $1$, and one ends up with a classical CD category, which we denote by $\FinMeas$. Such morphisms are called \define{transition maps/matrices}. The subcategory where the measures are non-negative, denoted by $\FinMeasP$, is also a classical Markov category. The category $\Stoch$ is a generalization of $\FinStoch$ to measurable spaces (see~\cite[Section~4]{Fr19} for details). 
%\ex
\end{exa}

%\bx
%\label{exa:fdcalgmarkovcat}
\vspace{-1mm}
\begin{exa}[Linear and conjugate-linear maps ($\fdCAlgUY^{\op}$)]{exa:fdcalgmarkovcat}
Our main example of a quantum Markov category is $\fdCAlgUY^{\op}$. The objects here are finite-dimensional unital $C^*$-algebras (see~\cite[Section~2.3]{Pa17} for a review of $C^*$-algebras). Henceforth, all $C^*$-algebras will be assumed unital. Every such finite-dimensional $C^*$-algebra is $^*$-isomorphic to a finite direct sum of (square) matrix algebras~\cite[Theorem~5.5]{Fa01}. A matrix algebra will be written as $\mathcal{M}_{n}(\C)$ indicating the $C^*$-algebra of complex $n\times n$ matrices. 
A morphism from $\mA$ to $\mB$ in $\fdCAlgUY^{\op}$ is \emph{either} a linear or \emph{conjugate-linear} unital map%
%footnote
\footnote{Capital letters for the morphisms will often be used when they describe morphisms between $C^*$-algebras (cf.\ Convention~\ref{conv:translating}).}
%end footnote
 $F:\mB\stoch\mA$ (linear maps are declared even and conjugate-linear maps are declared odd).%
%footnote
\footnote{The subscript \Yinyang\, is used as a reminder that both linear (yang) and conjugate-linear (yin) maps are included. Dropping the subscript will mean taking all even (linear) morphisms.}
%end footnote
 Notice that the function goes backwards because of the superscript ${}^\op$ (in the physics literature, this convention is known as the \emph{Heisenberg picture}). The tensor product (over $\C$) is the tensor product of finite-dimensional $C^*$-algebras. For example, 
\[
\left(\bigoplus_{x\in X}\mathcal{M}_{m_{x}}(\C)\right)\otimes\left(\bigoplus_{y\in Y}\mathcal{M}_{n_{y}}(\C)\right)\cong\bigoplus_{x,y}\Big(\mathcal{M}_{m_{x}}(\C)\otimes\mathcal{M}_{n_{y}}(\C)\Big), 
\]
where $X$ and $Y$ are finite sets labelling the matrix factors. 
The tensor product for morphisms is defined when both are linear or conjugate-linear (cf.\ Example~\ref{ex:linearandantilineartensor}). 
The $*$ operation is the involution on a $C^*$-algebra, which is conjugate-linear (this shows the grading is stable). 
If $\mB\xstoch{F}\mA$ is linear (conjugate-linear), then $F\circ *$ is conjugate-linear (linear) since $(F\circ*)(\l b)=F(\ov\l b^*)=\ov\l F(b^*)=\ov\l(F\circ*)(b)$ (and similarly if $F$ is conjugate-linear). 
We will ignore associators and unitors in what follows. This is permissible thanks to Mac~Lane's coherence theorem~\cite{Ma63}. We define the copy map $\D_{\mA}$ from $\mA$ to $\mA\otimes\mA$ in $\fdCAlgUY^{\op}$ to be the multiplication map determined on elementary tensors by%
%footnote
\footnote{The multiplication map of a $C^*$-algebra is defined as a bilinear map $\mA\times\mA\to\mA$. The fact that it extends to a linear map $\mA\otimes\mA\stoch\mA$ follows from the universal property of the tensor product, which is valid for finite-dimensional $C^*$-algebras. Throughout this article, we will assume this without always explicitly saying so. See Remark~\ref{ex:allCAlg} for the subtleties that occur for possibly infinite-dimensional $C^*$-algebras.}
%end footnote
\[
\vcenter{\hbox{%
\begin{tikzpicture}[font=\small,scale=1.5]
\node[copier] (c) at (0,0.5) {};
\draw (c)
to[out=15,in=-90] (0.44,1.0);
\draw (c)
to[out=165,in=-90] (-0.44,1.0);
\draw (c) to (0,0);
\node at (-0.6,0.9) {$\mA$};
\node at (0.6,0.9) {$\mA$};
\node at (0.2,0.1) {$\mA$};
\end{tikzpicture}}}
\quad\equiv\quad
\xy0;/r.25pc/:
(0,6.5)*+{\mA\otimes\mA}="AA";
(0,-6.5)*+{\mA}="A";
{\ar@{~>}"AA";"A"^{\mu_{\mA}}};
\endxy
\quad
\xy0;/r.25pc/:
(0,6.5)*+{\ni}="AA";
(0,-6.5)*+{\ni}="A";
\endxy
\xy0;/r.25pc/:
(0,6.5)*+{A\otimes A'}="AA";
(0,-6.5)*+{AA'}="A";
{\ar@{|->}"AA";"A"};
\endxy
\]
in $\fdCAlgUY$. The map $\mu_{\mA}$ is linear and unital, but it is not a $^*$-homomorphism unless $\mA$ is commutative. In fact, $\mu_{A}$ is not even positive in general (cf.\ Example~\ref{ex:CPUnotQMC}).
Nevertheless, it is coherent with the involution $*$ (in the sense of the last identity in (\ref{eq:markovcatfirstconditions})) because $(A_{1}A_{2})^*=A_{2}^*A_{1}^*$ for all $A_{1},A_{2}\in\mA$. 
Finally, the discard map $!_{\mA}:\mA\to\C$ in $\fdCAlgUY^{\op}$ is defined to be the unit inclusion map
\[
\vcenter{\hbox{%
\begin{tikzpicture}[font=\small,scale=1.5]
\node[discarder] (d) at (0,0.25) {};
\draw (d) to (0,-0.45);
\node at (0.2,-0.35) {$\mA$};
\node at (0,0.6) {};
\end{tikzpicture}}}
\quad\equiv\quad
\xy0;/r.25pc/:
(0,6.5)*+{\C}="C";
(0,-6.5)*+{\mA}="A";
{\ar"C";"A"^{!_{\mA}}};
\endxy
\quad
\xy0;/r.25pc/:
(0,6.5)*+{\ni}="AA";
(0,-6.5)*+{\ni}="A";
\endxy
\xy0;/r.25pc/:
(0,6.5)*+{\l}="AA";
(0,-6.5)*+{\l1_{\mA}}="A";
{\ar@{|->}"AA";"A"};
\endxy
\]
in $\fdCAlgUY$. 
Here are some of the conditions of a quantum Markov category and their corresponding expressions in terms of these morphisms: 
\[
\vcenter{\hbox{%
\begin{tikzpicture}[font=\small]
\node[copier] (c) at (0,0) {};
\coordinate (x1) at (-0.3,0.3);
\node[discarder] (d) at (x1) {};
\coordinate (x2) at (0.3,0.5);
\draw (c) to[out=165,in=-90] (x1);
\draw (c) to[out=15,in=-90] (x2);
\draw (c) to (0,-0.3);
\end{tikzpicture}}}
\quad=\quad
\vcenter{\hbox{%
\begin{tikzpicture}[font=\small]
\draw (0,0) to (0,0.8);
\end{tikzpicture}}}
\quad=\quad
\vcenter{\hbox{%
\begin{tikzpicture}[font=\small]
\node[copier] (c) at (0,0) {};
\coordinate (x1) at (-0.3,0.5);
\coordinate (x2) at (0.3,0.3);
\node[discarder] (d) at (x2) {};
\draw (c) to[out=165,in=-90] (x1);
\draw (c) to[out=15,in=-90] (x2);
\draw (c) to (0,-0.3);
\end{tikzpicture}}}
\qquad\iff\qquad
1_{\mA}A=A=A1_{\mA}
\qquad\forall\;A\in\mA, 
\]
\[
\vcenter{\hbox{%
\begin{tikzpicture}[font=\small]
\node[copier] (c) at (0,0.4) {};
\draw (c)
to[out=15,in=-90] (0.25,0.7);
\draw (c)
to[out=165,in=-90] (-0.25,0.7);
\draw (c) to (0,0);
\node at (0.6,0.1) {$\mA\otimes\mB$};
\end{tikzpicture}}}
=
\vcenter{\hbox{%
\begin{tikzpicture}[font=\small]
\node[copier] (c) at (0,0) {};
\node[copier] (c2) at (0.4,0) {};
\draw (c) to[out=15,in=-90] +(0.5,0.45);
\draw (c) to[out=165,in=-90] +(-0.4,0.45);
\draw (c) to +(0,-0.4);
\draw (c2) to[out=15,in=-90] +(0.4,0.45);
\draw (c2) to[out=165,in=-90] +(-0.5,0.45);
\draw (c2) to +(0,-0.4);
\node at (-0.2,-0.3) {$\mA$};
\node at (0.6,-0.3) {$\mB$};
\end{tikzpicture}}}
\iff
(A\otimes B)(A'\otimes B')=(AA')\otimes(BB')
\quad\forall\;A,A'\in\mA,\;B,B'\in\mB,
\]
and
\[
\vcenter{\hbox{%
\begin{tikzpicture}[font=\small]
\node[discarder] (d) at (0,0) {};
\node at (0,0.6) {};
\node[star] (s) at (0,-0.3) {};
\draw (d) to (s);
\draw (s) to (0,-0.6);
\node at (0.25,-0.45) {$\mA$};
\end{tikzpicture}}}
=\;
\vcenter{\hbox{%
\begin{tikzpicture}[font=\small]
\node[discarder] (d) at (0,-0.2) {};
\node[star] (s) at (0,0.3) {};
\draw (0,-0.6) to (d);
\draw [gray,dashed] (d) to (s);
\draw [gray,dashed] (s) to (0,0.6);
\node at (0.25,-0.45) {$\mA$};
\end{tikzpicture}}}
\qquad\iff\qquad
(\l1_{\mA})^*=\ov\l1_{\mA}
\qquad\forall\;\l\in\C.
\]
One can check that the rest of the axioms of a quantum Markov category are satisfied for $\fdCAlgUY^{\op}$. In fact, the larger category where we drop the unit-preserving assumption on the morphisms is a quantum CD category. In this paper, we will denote this latter category by $\fdCAlgY^{\op}$. 
We will be lax with our notation and from now on not distinguish between the category $\fdCAlgUY$ and its opposite. When we refer to $\fdCAlgY$ (or $\fdCAlgUY$) as a quantum CD (Markov) category, we will always mean its opposite (though all explicit algebraic calculations will be done in $\fdCAlgY$). In all the string diagrams that appear, the only difference is that we will compose from the top to the bottom of the page (rather than from the bottom to the top).
%\ex
\end{exa}

\vspace{-1mm}
\begin{exa}[Unitality of morphisms]{a006}
It follows from the axioms in Definition~\ref{defn:qmc} that $\id_{X},$ $\D_{X}$, $!_{X}$, and $*_{X}$ are automatically unital for all $X$. 
A morphism in any of the categories of finite sets together with morphisms that associate to each point a signed measure is unital iff the total measure associated to each point is $1$. 
A morphism $F:\mB\stoch\mA$ in any of the categories of $C^*$-algebras we have introduced is unital if and only if it is unital in the usual sense, i.e.\ $F(1_{\mB})=1_{\mA}$.
\end{exa}

We first recall a few definitions. 

\vspace{-1mm}
\begin{defn}[Positive, Schwarz positive, and completely positive maps]{defn:positivemaps}
An element $C$ of a $C^*$-algebra $\mA$ is \define{positive}, written $C\ge0$, iff it equals $A^*A$ for some $A\in\mA.$ A linear map $F:\mB\stoch\mA$ is \define{positive} iff it sends positive elements to positive elements. Let $\fdCAlgPU$ denote the subcategory of $\fdCAlgUY$ consisting of the same objects as $\fdCAlgUY$ but the morphisms are only the positive unital (PU) maps. A linear map $F:\mB\stoch\mA$ is \define{Schwarz positive} (SP) iff it satisfies $F(B^*B)\ge\lVert F(1_{\mB})\rVert\; F(B)^*F(B)$ for all $B\in\mB$. Let $\fdCAlgSPU$ denote the subcategory of $\fdCAlgUY$ consisting of the same objects as $\fdCAlgUY$ but the morphisms are only all the Schwarz positive unital (SPU) maps. A linear map $F:\mB\stoch\mA$ is \define{$n$-positive} iff $\id_{\mathcal{M}_{n}(\C)}\otimes F:\mathcal{M}_{n}(\C)\otimes\mB\stoch\mathcal{M}_{n}(\C)\otimes\mA$ is positive. A linear map $F$ is \define{completely positive} (CP) iff $F$ is $n$-positive for all $n\in\N.$ If $V\in\mA$, let $\Ad_{V}:\mA\stoch\mA$ be the CP map sending $A\in\mA$ to $VAV^{*}$. Let $\fdCAlgCPU$ denote the subcategory of $\fdCAlgUY$ consisting of the same objects as $\fdCAlgUY$ but the morphisms are only the CP unital (CPU) maps. Dropping $\mathbf{U}$ from the notation will be used to denote analogous categories where the morphisms are not necessarily unital (recall, \emph{all} $C^*$-algebras are unital here). 
\end{defn}

\vspace{-1mm}
\begin{rmk}[There is no quantum Markov category of all $C^*$-algebras]{ex:allCAlg}
The previous example \emph{cannot} be generalized to the infinite-dimensional setting. 
Although there are many $C^*$-norms endowing the category of $C^*$-algebras and $*$-homomorphisms with a monoidal structure~\cite{Wa10}, there is no $C^*$-norm endowing the category of $C^*$-algebras and bounded linear maps with a monoidal structure~\cite{Ok70}. However, there is one if one uses completely bounded linear maps~\cite{BlPa91}. Unfortunately, there is no $C^*$-norm for which the bilinear multiplication map $\mB\times\mB\stoch\mB$ on every $C^*$-algebra $\mB$ induces a bounded linear map $\mu_{\mB}:\mB\otimes\mB\stoch\mB$ for which $\mu_{\mB}\left(\sum_{i=1}^{n}B_{i}\otimes B_{i}'\right)=\sum_{i=1}^{n}B_{i}B_{i}'$ on all such finite sums.%
%footnote
\footnote{This problem cannot be fixed by using von~Neumann algebras. In fact, the situation for von~Neumann algebras is worse because the bilinear multiplication map is not even (jointly) weakly continuous~\cite[Lecture~5]{Lu11}.}
%end footnote

To see this explicitly, for $n\in\N$, set $\mA_{n}:=\mM_{n}(\C)$, and let $m_{n}:\mA_{n}\times\mA_{n}\stoch\mA_{n}$ denote the bilinear multiplication map and let $\mu_{n}:\mA_{n}\otimes\mA_{n}\stoch\mA_{n}$ denote the associated linear map. We will first show that $\lVert m_{n}\rVert\le1$ and $\lVert \mu_{n}\rVert\ge n$ for all $n\in\N$. The first fact follows from one of the Banach algebra axioms, which states $\lVert A_{1}A_{2}\rVert\le\lVert A_{1}\rVert\lVert A_{2}\rVert$ for all $A_{1},A_{2}\in\mA_{n}$. Indeed, since the norm on $\mA_{1}\times\mA_{2}$ is given by $\lVert (A_{1},A_{2})\rVert:=\sup\{\lVert A_{1}\rVert,\lVert A_{2}\rVert\}$ and the operator norm is given by (see~\cite[Section~5.1]{Fo99})
\[
\lVert m_{n}\rVert=\sup\big\{\lVert m_{n}(A_{1},A_{2})\rVert\;:\;\lVert (A_{1},A_{2})\rVert=1,\;A_{1},A_{2}\in\mA\big\}, 
\]
one has 
\[
\lVert A_{1}A_{2}\rVert
\le\lVert A_{1}\rVert\lVert A_{2}\rVert
\le1
=\lVert(A_{1},A_{2})\rVert
\]
for all $A_{1},A_{2}\in\mA_{n}$ such that $\lVert (A_{1},A_{2})\rVert=1$. 

Meanwhile, for each $n\in\N$ and $i,j\in\{1,\dots,n\}$, let $E_{ij}^{(n)}$ denote the $ij$-th matrix unit in $\mM_{n}(\C)$. Namely, the $kl$-th entries are $(E_{ij}^{(n)})_{kl}:=\de_{ik}\de_{jl}$. Set
\[
A^{(n)}:=\sum_{i=1}^{n}E^{(n)}_{1i}\otimes E^{(n)}_{i1}
\equiv\begin{bmatrix}E^{(n)}_{11}&\cdots&E^{(n)}_{n1}\\0&\cdots&0\\\vdots&&\vdots\\0&\cdots&0\end{bmatrix},
\]
where we have used the usual isomorphism $\mM_{n}(\C)\otimes\mM_{n}(\C)\cong\mM_{n^2}(\C)$ to relate the tensor product expression to another matrix.
By the definition of the matrix operator norm, $\lVert A^{(n)}\rVert=1$ for all $n\in\N$. This can be obtained by computing 
\[
\begin{split}
(A^{(n)})^{\dag}A^{(n)}
&=\sum_{i,j=1}^{n}(E^{(n)}_{1i}\otimes E^{(n)}_{i1})^{\dag}(E^{(n)}_{1j}\otimes E^{(n)}_{j1})
=\sum_{i,j=1}^{n}E^{(n)}_{ij}\otimes \big(\de_{ij}E^{(n)}_{11}\big)\\
&=\sum_{i=1}^{n}E^{(n)}_{ii}\otimes E^{(n)}_{11}
=\begin{bmatrix}E^{(n)}_{11}&&0\\&\ddots&\\0&&E^{(n)}_{11}\end{bmatrix}
\end{split}
\] 
and taking the square-root of the largest eigenvalue, which is $1$ (we are using the $C^*$-identity here to compute the norm in this way). 

The multiplication map $\mu_{n}:\mA_{n}\otimes\mA_{n}\stoch\mA_{n}$ is determined by how it acts on the basis $E^{(n)}_{ij}\otimes E^{(n)}_{kl}$, which is just $\mu_{n}(E^{(n)}_{ij}\otimes E^{(n)}_{kl})=E^{(n)}_{ij}E^{(n)}_{kl}=\de_{jk}E^{(n)}_{il}$ for all $i,j,k,l\in\{1,\dots,n\}$. Using this, we immediately obtain
\[
\mu_{n}\left(A^{(n)}\right)
=\sum_{i=1}^{n}\mu_{\mA}\left(E^{(n)}_{1i}\otimes E^{(n)}_{i1}\right)
=\sum_{i=1}^{n}E^{(n)}_{11}
=nE^{(n)}_{11}
\equiv\begin{bmatrix}n&0&\cdots&0\\0&0&\cdots&0\\\vdots&\vdots&&\vdots\\0&0&\cdots&0\end{bmatrix},
\]
where the $n$ is actually the number $n$ inside the first \emph{entry} (and the matrix is zero everywhere else). Thus, $\big\lVert \mu_{n}(A^{(n)})\big\rVert=n$ so that $\lVert\mu_{n}\rVert\ge n$. 

One can then extrapolate this construction to $\mA:=\mB(\Hi)$, where $\Hi$ is a separable Hilbert space. One chooses an orthonormal basis $\{e_{i}\}_{i\in\N}$ for $\Hi$ and defines a family of matrix units $\{E_{ij}\in\mB(\Hi)\}_{i,j\in\N}$ via the same formula as in the matrix case, namely $\<e_{k},E_{ij} e_{l}\>=\de_{ik}\de_{jl}$, where $\<\;\cdot\;,\;\cdot\;\>$ denotes the inner product on $\Hi$. From this, one can define a sequence of operators 
\[
\N\ni n\mapsto A^{(n)}:=\sum_{i=1}^{n}E_{1i}\otimes E_{i1}\in\mA\odot\mA,
\]
where $\mA\odot\mA$ is the algebraic tensor product of unital $*$-algebras. 
To compute the norm of this sequence, we first need to choose a $C^*$-norm on $\mA\odot\mA$ (see~\cite[Appendix B]{RaWi98} and~\cite{Wa10} for details and definitions). 
There exists a \emph{spatial/minimal $C^*$-norm} $\lVert\;\cdot\;\rVert_{\s}$ and a \emph{projective/maximal $C^*$-norm} $\lVert\;\cdot\;\rVert_{\mathrm{max}}$ satisfying the condition that given any other $C^*$-norm $\lVert\;\cdot\;\rVert_{\g}$, 
\[
\lVert A\rVert_{\s}\le\lVert A\rVert_{\g}\le\lVert A\rVert_{\mathrm{max}}\qquad\forall\;A\in\mA\odot\mA. 
\]
Since every $C^*$-norm is majorized by the maximal norm $\lVert\;\cdot\;\rVert_{\mathrm{max}}$, proving $\lVert A^{(n)}\rVert_{\mathrm{max}}\le 1$ for all $n\in\N$ will imply $\lVert A^{(n)}\rVert_{\g}\le 1$ for every $C^*$-norm $\lVert\;\cdot\;\rVert_{\g}$. To prove $\lVert A^{(n)}\rVert_{\mathrm{max}}\le 1$, it helps to recall the useful fact that if $\mB$ is a $C^*$-algebra, then 
\[
\lVert B\rVert\le 1 \quad \text{ if and only if }
\begin{bmatrix}1_{\mB}&B\\B^{*}&1_{\mB}\end{bmatrix}\ge0
\]
in $\mM_{2}(\mB)$~\cite[Lemma~3.1]{Pa02}, where $\mM_{2}(\mB)$ is the $C^*$-algebra of $2\times2$ matrices with entries in $\mB$. Thus, it suffices to prove $\left[\begin{smallmatrix}1&A^{(n)}\\(A^{(n)})^{\dag}&1\end{smallmatrix}\right]\ge0$ (in this case, $\mB=\mA\otimes_{\mathrm{max}}\mA$), where $1:=1_{\mA\odot\mA}$. But this is immediate because
\[
\begin{bmatrix}1&A^{(n)}\\(A^{(n)})^{\dag}&1\end{bmatrix}
=
\begin{bmatrix}1&A^{(n)}\\0&1-\big(\sum_{i=1}^{n}E_{ii}\big)\otimes E_{11}\end{bmatrix}^{\dag}\begin{bmatrix}1&A^{(n)}\\0&1-\big(\sum_{j=1}^{n}E_{jj}\big)\otimes E_{11}\end{bmatrix},
\]
as a quick calculation will show. Thus, $\lVert A^{(n)}\rVert_{\mathrm{max}}\le1$ for all $n\in\N$. 

Now, suppose that $\mu_{\mA}:\mA\odot\mA\stoch\mA$ is the linear map for which 
\[
\mu_{\mA}\left(\sum_{i=1}^{n}A_{i}\otimes A_{i}'\right)=\sum_{i=1}^{n}A_{i}A_{i}'
\]
for arbitrary elements in $A_{i},A'_{i}\in\mA$ (the linear map is uniquely determined by these assignments due to the universal property of the \emph{algebraic} tensor product). Then by the previous analysis, the existence of the sequence $\N\ni n\mapsto A^{(n)}$ shows that 
\[
\lVert\mu_{\mA}\rVert:=\sup\big\{\lVert\mu_{\mA}(A)\rVert\;:\;A\in\mA\odot\mA \text{ such that } \lVert A\rVert_{\mathrm{max}}\le1\big\}=\infty,
\]
since $\mu_{\mA}(A^{(n)})=nE_{11}$ has norm $n$ with respect to the usual (operator) norm on $\mA=\mB(\Hi)$. Hence, $\mu_{\mA}$ cannot be continuous. 

Therefore, there is no quantum Markov category whose objects consist of all $C^*$-algebras. These matters are discussed further in Question~\ref{ques:multilinearCAlg}. 
\end{rmk}

The next few examples provide important subcategories of $\fdCAlgY$ that are \emph{neither} quantum \emph{nor} classical Markov categories. Nevertheless, they are the main categories of interest here and the fact that they \emph{embed} into a quantum Markov category allows for a string diagrammatic calculus to be used. 

\begin{exa}[The hierarchy of positivity for maps of $C^*$-algebras]{exa:positivehierarchy}
The categories introduced in Definition~\ref{defn:positivemaps} have the following hierarchical structure
\[
\fdCAlgCPU\subset\fdCAlgSPU\subset\fdCAlgPU\subset\fdCAlgUY, 
\]
where every inclusion is strict (even when restricted to finite-dimensional algebras).
The inclusion $\fdCAlgCPU\subset\fdCAlgSPU$ follows from the fact that every 2-positive unital map $F:\mB\stoch\mA$ (and hence every CPU map) satisfies the \define{Kadison--Schwarz (KS) inequality} 
\[
F(B)^*F(B)\le F(B^*B)\qquad\forall\;B\in\mB, 
\]
which is exactly the Schwarz-positivity condition specialized to the case where $F$ is unital. A proof of this can be found in \cite[Proposition~6]{Ma10} (see~\cite{Ka52} and~\cite{Da57} for the original references).%
%footnote
\footnote{In support of the views expressed in Question~\ref{ques:multilinearCAlg}, even in the infinite-dimensional setting of all $C^*$-algebras, the conditions on the morphisms in the appropriate analogues of the categories here are strong enough to imply that they are automatically continuous maps with respect to the norm topologies. This is because the operator norm of any positive map $F:\mB\stoch\mA$ is given by $\lVert F(1_{\mB})\rVert_{\mA}$, where $\lVert\;\cdot\;\rVert_{\mA}$ denotes the $C^*$-algebra norm on $\mA$.}
%end footnote
Examples instantiating the strict inclusions above are well known, and special instances will be recalled when needed in this work. 
\end{exa}

%\bx
%\label{ex:CPUnotQMC}
\begin{exa}[Quantum channels do not form a quantum Markov category, but...]{ex:CPUnotQMC}
The subcategory $\fdCAlgCPU$ of $\fdCAlgUY$ is not a quantum Markov category because there is no CPU map $\mA\otimes\mA\stoch\mA$ satisfying the conditions of Definition~\ref{defn:qmc} (not even in finite dimensions). In fact, a version of the no-cloning (no-broadcasting) theorem states that a CPU map 
$\mu_{\mA}:\mA\otimes\mA\stoch\mA$ satisfying the first condition in~(\ref{eq:markovcatfirstconditions}), i.e.\ $\mu_{\mA}(1_{\mA}\otimes A)=A=\mu_{\mA}(A\otimes1_{\mA})$ for all $A\in\mA$, exists if and only if $\mA$ is commutative (cf.~\cite[Theorem~6]{Ma10}). 
We will prove a related no-broadcasting theorem in the abstract setting in Theorem~\ref{thm:nocloning}. In fact, this result holds true for SPU maps as well, and the standard proofs (such as \cite[Theorem~6]{Ma10}) actually only use the Kadison--Schwarz inequality. Hence, $\fdCAlgSPU$ is not a quantum Markov category. 
Nevertheless, these categories \emph{embed} into a quantum Markov category, and the morphisms in the ambient quantum Markov category can be used to formulate many notions such as a.e.\ equivalence, disintegrations, Bayesian inversion, conditionals, modularity, determinism, and so on. We will discuss all of these ideas in subsequent sections.
\end{exa}
%\ex

We now introduce a few properties that we wish to distinguish for certain morphisms in quantum Markov categories. The first is the notion of a $*$-preserving morphism, which plays an essential role as a symmetry that allows for a more robust string-diagrammatic calculus in quantum theory. 

%\bd
%\label{defn:selfadjointmorphism}
\begin{defn}[$*$-preserving morphisms]{defn:selfadjointmorphism}
Let $\mM_{\text{\Yinyang}}$ be a quantum Markov category. A morphism $X\xstoch{f}Y$ in $\mM_{\text{\Yinyang}}$ is said to be \define{$*$-preserving} iff $f$ is natural with respect to $*$,%
%footnote
\footnote{This word `natural' is meant in the categorical sense. The assignment $*$ assigns to each object $X$ a morphism $*_{X}$. This assignment is natural (in the sense of natural transformations) precisely for morphisms that are $*$-preserving.}
%end footnote
meaning
\[
\vcenter{\hbox{%
\begin{tikzpicture}[font=\small]
\node[star] (s) at (0,0) {};
\node at (-0.2,-0.25) {\footnotesize$X$};
\node[arrow box] (f) at (0,0.75) {$f$};
\node at (-0.2,1.40) {\footnotesize$Y$};
\draw (0,-0.4) to (s);
\draw (s) to (f);
\draw (f) to (0,1.55);
\end{tikzpicture}}}
\qquad=\qquad
\vcenter{\hbox{%
\begin{tikzpicture}[font=\small]
\node[arrow box] (f) at (0,0.4) {$f$};
\node at (-0.2,-0.25) {\footnotesize$X$};
\node[star] (s) at (0,1.15) {};
\node at (-0.2,1.40) {\footnotesize$Y$};
\draw (0,-0.4) to (f);
\draw (f) to (s);
\draw (s) to (0,1.55);
\end{tikzpicture}}}
\qquad,\qquad\text{i.e.\ }\quad
f\circ*_{X}=*_{Y}\circ f.
\]
%\ed
The monoidal%
%footnote
\footnote{Monoidality follows from the coherence between the tensor and $*$ in Definition~\ref{defn:qmc}.}
%end footnote
 subcategory of all even $*$-preserving morphisms of $\mM_{\text{\Yinyang}}$ will be denoted by $\mM_{*}$. Any subcategory whose morphisms are $*$-preserving (but not necessarily even!) will be called a \define{$*$-preserving subcategory}. 
\end{defn}

\begin{exa}[$*$-preserving morphisms in $\fdCAlgY$]{a007}
A morphism (odd or even!) in $\fdCAlgY$ is $*$-preserving if and only if it takes self-adjoint elements to self-adjoint elements. In particular, every positive map is $*$-preserving~\cite{St13}.%
%footnote
\footnote{In the literature, the terminology ``self-adjoint'' is often used for what we call ``$*$-preserving.''}
%end footnote  
\end{exa}

%\br
\begin{exa}[Copy is not generally $*$-preserving]{a008}
In a quantum Markov category $\mM_{\text{\Yinyang}}$, copy $\D$ is $*$-preserving if and only if $\mM_{\mathrm{even}}$ is a classical Markov category.
%\er
\end{exa}

%\bd
%\label{eq:deterministicmap}
\begin{defn}[Deterministic morphisms]{eq:deterministicmap}
An even%
%footnote
\footnote{A different definition is made for odd morphisms. In this case, there is an additional swapping that needs to be applied. Since we will not be using these morphisms here, we leave this definition out.}
%end footnote
 morphism $X\xstoch{f}Y$ in a quantum Markov category $\mM_{\text{\Yinyang}}$ is called \define{deterministic} iff $f$ is $*$-preserving and natural with respect to $\D$, the latter of which means
\[
%\label{eq:deterministicdefn}
\vcenter{\hbox{%
\begin{tikzpicture}[font=\small]
\coordinate (p) at (0,-1.0);
\node at (-0.2,-0.85) {\footnotesize$X$};
\node[arrow box] (f) at (0,-0.3) {$f$};
\node[copier] (copier) at (0,0.3) {};
\coordinate (X) at (-0.5,0.91);
\node at (-0.7,0.75) {\footnotesize$Y$};
\coordinate (X2) at (0.5,0.91);
\node at (0.7,0.75) {\footnotesize$Y$};
\draw (p) to (f);
\draw (f) to (copier);
\draw (copier) to[out=165,in=-90] (X);
\draw (copier) to[out=15,in=-90] (X2);
\end{tikzpicture}}}
\quad
=
\quad
\vcenter{\hbox{%
\begin{tikzpicture}[font=\small]
\coordinate (p) at (0,-0.3);
\node at (-0.2,-0.15) {\footnotesize$X$};
\node[copier] (copier) at (0,0.3) {};
\node[arrow box] (f) at (-0.5,0.95) {$f$};
\node[arrow box] (e) at (0.5,0.95) {$f$};
\coordinate (X) at (-0.5,1.7);
\node at (-0.7,1.55) {\footnotesize$Y$};
\coordinate (Y) at (0.5,1.7);
\node at (0.7,1.55) {\footnotesize$Y$};
\draw (p) to (copier);
\draw (copier) to[out=165,in=-90] (f);
\draw (f) to (X);
\draw (copier) to[out=15,in=-90] (e);
\draw (e) to (Y);
\end{tikzpicture}}}
\quad, 
\quad\text{i.e.\ }\quad
\D_{Y}\circ f=(f\otimes f)\circ\D_{X}.
\]
The subcategory of $\mM_{\text{\Yinyang}}$ consisting of all (even) deterministic morphisms is denoted by $\mM_{\mathrm{det}}$. 
%\ed
\end{defn}

%\br
%\label{rmk:deterministic}
\begin{rmk}[Tensor and determinism]{rmk:deterministic}
In a quantum Markov category $\mM_{\text{\Yinyang}}$, 
the tensor product of two deterministic maps is deterministic. 
This follows from naturality of the braiding, the definition of determinism, the third identity in~(\ref{eq:markovcatsecondconditions}), and the second identity in~(\ref{eq:lastsetformarkovcatidentities}) in Definition~\ref{defn:qmc}.  
Thus, $\mM_{\mathrm{det}}$ is a symmetric monoidal subcategory of $\mM_{*}$, and hence $\mM_{\text{\Yinyang}}$. 
%\er
\end{rmk}

\begin{exa}[Deterministic maps in $\FinStoch$, $\Stoch$, and $\fdCAlgU^{\op}$]{exa:deterministicmaps}
In $\FinStoch$, deterministic maps correspond to functions, assignments where the measures associated to points are Dirac measures~\cite[Theorems~2.82 and~2.85]{Pa17}.
In $\Stoch$, deterministic maps correspond to $\{0,1\}$-valued measures (cf.\ \cite[Example~10.4]{Fr19}). 
In $\fdCAlgU^{\op}$, deterministic maps correspond to $^*$-homomorphisms. Indeed, if $\mB\xstoch{F}\mA$ is a linear unital map of $C^*$-algebras, then 
the $*$-preserving condition says $F(B^*)=F(B)^*$ for all $B\in\mB$ and naturality with respect to $\D$ says $F(BB')=F(B)F(B')$ for all $B,B'\in\mB$.%
%footnote
\footnote{Although $\fdCAlgCPU^{\op}$ is not a quantum Markov category, since every $*$-homomorphism is automatically CPU, the categories $\fdCAlgCPU^{\op}_{\mathrm{det}}$ and $\fdCAlgU^{\op}_{\mathrm{det}}$ are equal.}
%endfootnote
The relationship to this definition of determinism is closely related to the partial trace from quantum mechanics by using the Hilbert--Schmidt inner product. The reader is referred to Example~2.11, Lemma~2.12, and Remark~2.14 in~\cite{PaEntropy} for further details.
In any case, a concrete example illustrating the type of determinism used here is provided in Example~\ref{exa:QEC} in the context of quantum error correcting codes.
\end{exa}

On occasion, we will make use of special morphisms in a quantum Markov category, special instances of which may be called states. We also make use of a dual concept in a quantum CD category known as an effect. They may also add clarity when interpreting certain results. 

%\bd
%\label{defn:state}
\vspace{-1mm}
\begin{defn}[States and effects]{defn:state}
Let $\mathcal{M}_{\text{\Yinyang}}$ be a quantum CD category and let $\mC$ be an even%
%footnote
\footnote{It is not necessary to restrict to even morphisms. This is done only for simplicity.}
%end footnote
 subcategory of $\mathcal{M}_{\text{\Yinyang}}$. A \define{state} on an object $X$ in $\mC$ is a morphism $I\xstoch{p}X$ in $\mC$. Dually, an \define{effect} on an object $X$ in $\mC$ is a morphism $X\xstoch{\varphi}I$ in $\mC$.%
%footnote
\footnote{If we had restricted ourselves to quantum \emph{Markov} categories, each object would have a unique effect, namely the delete/grounding map. Also note that our usage of states and effects are far more general than the standard definitions. We hope this extension of terminology is not too confusing.}
%end footnote
A state $p$ and an effect $\varphi$ will be drawn in string-diagrammatic notation as 
\[
\vcenter{\hbox{%
\begin{tikzpicture}[font=\small]
\node[state] (p) at (0,0) {$p$};
\node at (0.25,0.7) {\footnotesize$X$};
\node (X) at (0,1.0) {};
\draw (p) to (X);
\end{tikzpicture}}}
\qquad\text{ and }\qquad
\vcenter{\hbox{%
\begin{tikzpicture}[font=\small]
\node[effect] (e) at (0,0.8) {$\varphi$};
\node at (0.25,0.1) {\footnotesize$X$};
\node (X) at (0,-0.2) {};
\draw (X) to (e);
\end{tikzpicture}}}
\qquad,
\]
respectively. 
%\ed
\end{defn}

This definition of states is a bit too general for a proper interpretation in terms of probability theory (both commutative and non-commutative) as the following example illustrates. Nevertheless, we will solve this issue later by working with states in particular subcategories of a quantum CD category that have the appropriate interpretations. 

\vspace{-1mm}
\begin{exa}[States and effects in our main examples]{exa:states}
In $\FinMeas$, a state on $X$ is a (signed) transition matrix $\{\bullet\}\xstoch{p}X$, which defines a signed measure on $X$. An effect is a transition matrix $X\xstoch{\varphi}\{\bullet\}$, which associates to each $x\in X$ a signed measure $\varphi_{x}$ on $\{\bullet\}$. But a measure on $\{\bullet\}$ is uniquely determined by its value on $\bullet$, which is a real number. Hence, $\varphi$ defines a real-valued function on $X$. 

As a special case, in $\FinMeasP$, a state on $X$ defines a (non-negative) measure on $X$ and an effect corresponds to a $(0,\infty)$-valued function on $X$. As an even more special case, in $\FinStoch$, a state defines a probability measure on $X$ and  there is only a single effect because $\varphi_{x}$ must be a probability measure on $\{\bullet\}$, of which there only exists one. In $\FinStoch$, a deterministic state on $X$ is a Dirac delta measure at some point $x\in X$. 

In $\fdCAlg$, a state on $\mA$ is a linear functional $\mA\xstoch{\omega}\C$. An effect is a linear map $\C\xstoch{A}\mA$, which uniquely determines (and is uniquely determined by) an element $A\in\mA$ by taking the image of the element $1\in\C$ (hence the notation). 

As a special case, in $\fdCAlg_{*}$ (the subcategory of $\fdCAlgY$ consisting of even $*$-preserving morphisms), a state is a $*$-preserving linear functional, while an effect $\C\stoch\mA$ corresponds to a \emph{self-adjoint} element of $\mA$. Indeed, let us temporarily denote this map by $\C\xstoch{\chi}\mA$ and set $A:=\chi(1)$. Then $A=\chi(1)=\chi\big(\overline{1}\big)=\chi(1)^{*}=A^*$. As a second special case, in $\fdCAlgP$, a state on $\mA$ is a positive linear functional on $\mA$ and an effect on $\mA$ corresponds to a positive element in $\mA$. Finally, in $\fdCAlgPU$, and equivalently in $\fdCAlgCPU$ by~\cite[Theorem~3]{St55}, a state is a positive unital functional, while there is only a unique effect on every algebra (this latter fact is true even if the positivity assumption is dropped---only unitality is needed). The terminology \emph{state} in the $C^*$-algebra context will always mean PU (and hence CPU). Otherwise, we will use the terminology \emph{linear functional}. 

In general, deterministic states do not exist in $\fdCAlgCPU$. For example, the matrix algebra $\mM_{m}(\C)$ has no deterministic states unless $m=1$. In fact, if $Y$ is a finite set and $\mB:=\bigoplus_{y\in Y}\mM_{n_{y}}(\C)$ is a direct sum of matrix algebras, then $\mB$ has no deterministic states unless at least one of the $n_{y}$ satisfies $n_{y}>1$. This follows from the classification of $*$-homomorphisms between finite-dimensional $C^*$-algebras~\cite{Fa01} and is related to the lack of hard evidence in quantum mechanics~\cite{PaJeffrey}. 
\end{exa}

\begin{conv}[Translating from abstract definitions to $C^*$-algebras]{conv:translating}
The reader may have already noticed some slight notational differences used between the definitions presented and the notation used for $C^*$-algebras. For example, a morphism $X\xstoch{f}Y$ in an arbitrary quantum CD category (or the examples of measurable spaces and transition kernels) is almost always written as $\mB\xstoch{F}\mA$ when instantiated in $\fdCAlgY$. More generally, objects $X,Y,Z$ will be replaced by $\mA,\mB,\mC$. Occasionally, $\Theta$ will be used to represent an object that one interprets as a parameter space in the classical setting. In the quantum setting, $\Theta$ will often be replaced with $\mC$ (since a special case often considered is $\C$). Morphisms in the general setting will often be written with the letters $f,g,h,k$, and are replaced by $F,G,H,K$ when viewed as linear maps between $C^*$-algebras. Note also that the directionality of the arrows changes in $C^*$-algebras since we work exclusively in the Heisenberg picture in all subcategories of $\fdCAlgY$ (cf.\ Convention~\ref{conv:directionality}). States $p,q,r:I\stoch X$ will always be written as $\omega,\xi,\zeta:\mA\stoch\C$, respectively, in $\fdCAlgY$.
\end{conv}

\begin{ques}[Involutive categories]{ques:involutivecategories}
It seems possible that one could have also provided an alternative string-diagrammatic framework for capturing $C^*$-algebras, their (linear) multiplication map, as well as the conjugate-linear involution $*$ by using involutive monoids~\cite{Ja12} (see also~\cite{Ya20}).%
%footnote
\footnote{I thank Aaron Fenyes for suggesting this alternative. I also thank two anonymous referees for additional insight, for bringing my attention to~\cite{Ja12}, and for asking the questions mentioned here.}
%end footnote
Associated to every $C^*$-algebra $\mA$ is a \define{conjugate} $C^*$-algebra $\overline{\mA}$ whose structure is the same as that of $\mA$ (in particular, $\ov{\mA}=\mA$ as a set) except that the scalar multiplication $\C\times\overline{\mA}\to\overline{\mA}$ is defined by $\l\cdot A:=\overline{\l}A$ for all $\l\in\C$ and $A\in\ov{A}$. In this way, the involution $*$ on $\mA$ can be viewed as a \emph{linear} map $\ov{\mA}\to\mA$ (or as $\mA\to\ov{\mA}$). 
Thus, one can view $\mA$ and its associated structure, including the involution, as an involutive monoid in some involutive symmetric monoidal category (see~\cite{Ja12} for definitions and notation). 

More generally, given the category $\mN$ of involutive monoids in an involutive symmetric monoidal category, can one construct an associated quantum Markov category? Conversely, given a quantum Markov category, do the objects form involutive monoids in any sense? 
\end{ques}

\begin{ques}[Are multicategories more suitable for non-commutative probability]{ques:multilinearCAlg}
Remark~\ref{ex:allCAlg} is evidence that the quantum Markov category framework presented here is not robust enough for infinite-dimensional non-commutative probability and statistics. Nevertheless, the reader will notice that almost none of the arguments presented below seem to depend on finite-dimensionality.%
%footnote
\footnote{A large effort has been made to avoid any structure specific to finite dimensions in order that many definitions and theorems described here apply to arbitrary $C^*$-algebras.}
%end footnote
Therefore, one might suspect that the interpretation of the string diagrams in terms of tensor products is the main issue. Another setting where string diagrams appear outside the context of monoidal categories is in \emph{multicategories}. This leads to the following question. Is there a more faithful representation of non-commutative probability and statistics by formulating the notion of a Markov {\emph{multi}category}? The reason for the potential viability of this point of view is because the multiplication map $m_{\mA}:\mA\times\mA\stoch\mA$ is a bounded \emph{multilinear} map, which need not extend to a bounded \emph{linear} map on any tensor product completion (cf.\ Remark~\ref{ex:allCAlg}). 
\end{ques}

%%%%%%%%%%%%%%%%%%%%%%%%%%%%%%%%%%%%%%
\section[S-Positive subcategories]{a009}
\label{sec:pos}
\vspace{-12mm}
\noindent
\begin{tikzpicture}
\coordinate (L) at (-8.75,0);
\coordinate (R) at (8.75,0);
\draw[line width=2pt,orange!20] (L) -- node[anchor=center,rectangle,fill=orange!20]{\strut \Large \textcolor{black}{\textbf{4\;\; S-Positive subcategories}}} (R);
\end{tikzpicture}
%\vspace{1mm}
%%%%%%%%%%%%%%%%%%%%%%%%%%%%%%%%%%%%%%

In this section, we describe an axiom recently introduced by Fritz for certain subcategories of quantum Markov categories~\cite[Definition~11.22]{Fr19}. This axiom is used to prove several theorems about (classical) sufficient statistics, and it is of great interest to adapt this definition to the quantum setting in hopes that abstract graphical reasoning may be used to prove quantum analogues of these statements. Here, we prove that $\fdCAlgSPU$ is a positive subcategory of $\fdCAlgUY$ (Theorem~\ref{thm:fdcalgcpupositive}). We also prove that $\fdCAlgPU$ is \emph{not} a positive subcategory of $\fdCAlgUY$ (Example~\ref{ex:2positiveisFritzpositive}), emphasizing the importance of the Kadison--Schwarz inequality. Although the copy map appears in the definition of positivity, it is not necessary for the copy map to be a morphism in the positive subcategory. In fact, if it is, then this imposes severe restrictions on the subcategory, which we explain in Theorem~\ref{thm:nocloning}. These restrictions are analogous to the no-cloning (technically no-broadcasting) theorem from quantum information theory. 

%\bd
%\label{defn:positivecategory}
\vspace{-1mm}
\begin{defn}[S-positive subcategory]{defn:positivecategory}
Let $\mathcal{M}_{\text{\Yinyang}}$ be a quantum Markov category. A subcategory%
%footnote
\footnote{The subcategory here need not be a classical, nor a quantum, Markov category. It need not even be a monoidal subcategory. 
However, demanding morphisms to be even is essential, since otherwise the diagram would not be well-defined.} 
%end footnote
$\mC\subseteq\mathcal{M}$ (recall, $\mM=\mathcal{M}_{\mathrm{even}}$) is said to be \define{S-positive}%
%footnote
\footnote{The reason we have chosen to call this S-positivity as opposed to just `positivity' will be discussed extensively in detail in this section. The S may stand for `strong' though it may also stand for `Schwarz' (see Question~\ref{ques:Spositivesubcats} for a potentially good reason for the latter terminology).}
%end footnote
iff 
for every pair of composable morphisms $X\xstoch{f}Y\xstoch{g}Z$ in $\mC$ such that $g\circ f$ is deterministic, the equality%
%footnote
\footnote{One could have also written a similar, a-priori inequivalent, equation reflected across the vertical line passing through the copy map. This equation would read $(\id_{Y}\times g)\circ\Delta_{Y}\circ f=(f\times (g\circ f))\circ\Delta_{X}$. The two conditions are equivalent in $*$-preserving subcategories, though we leave the proof of this to Section~\ref{sec:ae}, specifically Lemma~\ref{prop:hfpkgs}. In general, one should be cautious to distinguish between \emph{left} and \emph{right} S-positive subcategories if they are not necessarily $*$-preserving. We avoid this in the present section because we feel it would be a bit too distracting now, and we will come back to these distinctions in later sections. For the rest of this paper, S-positivity will always refer to the diagram condition above exactly as drawn here.}
%end footnote
%\be
%\label{eq:positivitycondition}
\[
\vcenter{\hbox{
\begin{tikzpicture}[font=\footnotesize]
\node[arrow box] (f) at (0,-0.3) {$f$};
\node[copier] (copier) at (0,0.3) {};
\node[arrow box] (g) at (-0.5,0.95) {$g$};
\node at (-0.7,1.45) {$Z$};
\node at (0.7,1.45) {$Y$};
\node at (-0.2,-0.75) {$X$};
\coordinate (X) at (0.5,1.6);
\coordinate (Y) at (-0.5,1.6);
\draw (0,-0.9) to (f);
\draw (f) to (copier);
\draw (copier) to[out=30,in=-90] (X);
\draw (copier) to[out=165,in=-90] (g);
\draw (g) to (Y);
\end{tikzpicture}}}
\quad=\quad
\vcenter{\hbox{%
\begin{tikzpicture}[font=\footnotesize]
\coordinate (q) at (0,-0.2) {};
\node[copier] (copier) at (0,0.3) {};
\node[arrow box] (g2) at (-0.5,0.95) {$f$};
\node[arrow box] (e2) at (-0.5,1.75) {$g$};
\node[arrow box] (g) at (0.5,0.95) {$f$};
\node at (-0.7,2.25) {$Z$};
\node at (0.7,2.25) {$Y$};
\node at (-0.2,-0.05) {$X$};
\coordinate (X) at (-0.5,2.4);
\coordinate (Y) at (0.5,2.4);
\draw (q) to (copier);
\draw (copier) to[out=165,in=-90] (g2);
\draw (g2) to (e2);
\draw (e2) to (X);
\draw (copier) to[out=15,in=-90] (g);
\draw (g) to (Y);
\end{tikzpicture}}}
%\ee
\]
must also hold. 
%\ed
\end{defn}

Any subcategory of even deterministic morphisms is S-positive. The fact that $\FinStoch$ is S-positive was proved in~\cite[Example~11.25]{Fr19} (in fact, this was proved for the larger category of Markov kernels between measurable spaces). Here, we prove a non-commutative version of this result. 

%\bt
%\label{thm:fdcalgcpupositive}
\begin{theo}[$\fdCAlgSPU$ is S-positive]{thm:fdcalgcpupositive}
The category $\fdCAlgCPU$ is an S-positive subcategory of $\fdCAlgUY$. In fact, $\CAlgSPU$ 
is S-positive. 
%\et
\end{theo}

To prove this theorem, we recall an important fact regarding multiplicative properties of SPU maps. This fact and variants of it will be used on multiple occasions throughout this work.

\begin{lem}[The (standard) Multiplication Lemma]{lem:multiplicationtheorem}
%\blem
%[The Multiplication Lemma]
%\label{lem:multiplicationtheorem}
Let $\mB\xstoch{F}\mA$ be an SPU map between $C^*$-algebras. Suppose that 
$F(B^*B)=F(B)^*F(B)$ for some $B\in\mB$. Then
\[
F(B^*C)=F(B)^*F(C)
\quad\text{ and }\quad
F(C^*B)=F(C)^*F(B)
\qquad\forall\;C\in\mB. 
\]
%\elem
\end{lem}

\bprf[Proof of Lemma~\ref{lem:multiplicationtheorem}]
See~\cite[Theorem~4]{Ma10} or the more general result that we will prove later (Lemma~\ref{lem:relmulttheorem}). 
\eprf

\bprf[Proof of Theorem~\ref{thm:fdcalgcpupositive}]
Let $\mC\xstoch{G}\mB\xstoch{F}\mA$ be a pair of composable SPU maps of $C^*$-algebras such that the composite $F\circ G$ is a $^*$-homomorphism. 
Then, 
\be
\begin{split}
F\big(G(C)^*G(C)\big)&\le F\big(G(C^*C)\big)\quad\text{ by Kadison--Schwarz for $G$}\\
&=F\big(G(C)\big)^*F\big(G(C)\big)\quad\text{ since $F\circ G$ is deterministic}\\
&\le F\big(G(C)^*G(C)\big)\quad\text{ by Kadison--Schwarz for $F$}
\end{split}
\ee
holds for all $C\in\mC$. Thus, all inequalities become equalities. In particular, 
\be
F\big(G(C)^*G(C)\big)=F\big(G(C)\big)^*F\big(G(C)\big)\qquad\forall\;C\in\mC.
\ee
By the Multiplicative Lemma (Lemma~\ref{lem:multiplicationtheorem}), this implies
\be
F\big(G(C)^*B\big)=F\big(G(C)\big)^*F(B)\qquad\forall\;C\in\mC,\;B\in\mB.
\ee
Since $F$ and $G$ are $*$-preserving and $*$ is an involution, this reproduces 
\be
F\big(G(C)B\big)=F\big(G(C)\big)F(B)\qquad\forall\;C\in\mC,\;B\in\mB, 
\ee
which is the S-positivity condition in Definition~\ref{defn:positivecategory}.
\eprf

\begin{exa}[$\fdCAlgU$ is not an S-positive subcategory of $\fdCAlgUY$]{a010}
Based on the fact that the category of finite sets together with morphisms assigning signed measures to points embeds fully and faithfully into the (opposite of the) category of finite-dimensional $C^*$-algebras together with linear unital maps, the latter is \emph{not} an S-positive subcategory of $\fdCAlgUY$. This follows immediately from the fact that $\FinStoch_{\pm}$, as defined in \cite[Example~11.27]{Fr19}, is not positive. 
\end{exa}

What is perhaps more surprising is the following result. 

%\bx
%\label{ex:2positiveisFritzpositive}
\begin{exa}[Positive unital maps do not form an S-positive category]{ex:2positiveisFritzpositive} 
The subcategory $\fdCAlgPU$ of finite-dimensional $C^*$-algebras together with positive unital maps is \emph{not} an S-positive subcategory of $\fdCAlgUY$. To see this, take $\mA=\mB=\mathcal{M}_{n}(\C)$ (with $n\ge2$) and set $F:=T=:G$, where $T$ is the map that takes the transpose of matrices. This map is known to be positive and unital but not Schwarz-positive. Furthermore, $G\circ F=T^2=\id$, which is deterministic. Nevertheless, there exist matrices $A$ and $B$ such that 
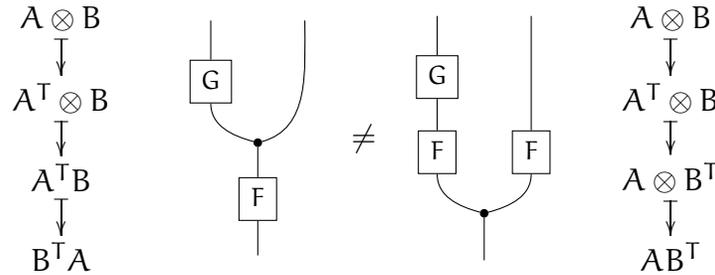
\[
	\xy0;/r.25pc/:
	(0,15)*+{A\otimes B}="aob";
	(0,5)*+{A^{T}\otimes B}="aTob";
	(0,-5)*+{A^{T}B}="ab";
	(0,-15)*+{B^{T}A}="ba";
	{\ar@{|->}"aob";"aTob"};
	{\ar@{|->}"aTob";"ab"};
	{\ar@{|->}"ab";"ba"};
	\endxy
	\qquad
\vcenter{\hbox{
\begin{tikzpicture}[font=\footnotesize,scale=1.25]
\node[arrow box] (f) at (0,-0.3) {$F$};
\node[copier] (copier) at (0,0.3) {};
\node[arrow box] (g) at (-0.5,0.95) {$G$};
\coordinate (X) at (0.5,1.6);
\coordinate (Y) at (-0.5,1.6);
\draw (0,-0.9) to (f);
\draw (f) to (copier);
\draw (copier) to[out=15,in=-90] (X);
\draw (copier) to[out=165,in=-90] (g);
\draw (g) to (Y);
\end{tikzpicture}}}
\quad\ne\quad
\vcenter{\hbox{%
\begin{tikzpicture}[font=\footnotesize,scale=1.25]
\coordinate (q) at (0,-0.2) {};
\node[copier] (copier) at (0,0.3) {};
\node[arrow box] (g2) at (-0.5,0.95) {$F$};
\node[arrow box] (e2) at (-0.5,1.75) {$G$};
\node[arrow box] (g) at (0.5,0.95) {$F$};
\coordinate (X) at (-0.5,2.4);
\coordinate (Y) at (0.5,2.4);
\draw (q) to (copier);
\draw (copier) to[out=165,in=-90] (g2);
\draw (g2) to (e2);
\draw (e2) to (X);
\draw (copier) to[out=15,in=-90] (g);
\draw (g) to (Y);
\end{tikzpicture}}}
\qquad
	\xy0;/r.25pc/:
	(0,15)*+{A\otimes B}="aob";
	(0,5)*+{A^{T}\otimes B}="aTob";
	(0,-5)*+{A\otimes B^{T}}="ab";
	(0,-15)*+{AB^{T}}="ba";
	{\ar@{|->}"aob";"aTob"};
	{\ar@{|->}"aTob";"ab"};
	{\ar@{|->}"ab";"ba"};
	\endxy
\]
so that S-positivity of $\fdCAlgPU$ fails.%
%footnote
\footnote{It is for this reason that we have not chosen to call the axiom in Definition~\ref{defn:positivecategory} positivity, since this might cause confusion amongst the entire field of operator algebras and quantum information theory.}
%end footnote
%\ex
\end{exa}

\begin{ques}[What does a generic S-positive subcategory look like?]{ques:Spositivesubcats}
Theorem~\ref{thm:fdcalgcpupositive} and Example~\ref{ex:2positiveisFritzpositive} prompt the following question. Let $\mC$ be the largest $*$-preserving monoidal S-positive subcategory of $\fdCAlgUY$, i.e.\ for any other $*$-preserving monoidal S-positive subcategory $\mD$ of $\fdCAlgUY$, then $\mD\subseteq\mC$. Then does  $\mC=\fdCAlgCPU$? If so, this would provide an alternative characterization of CPU maps in $\fdCAlgUY$.%
%footnote
\footnote{Current categorical characterizations of CPU maps include one due to Selinger~\cite{Se07} and another one due to Huot and Staton~\cite{HuSt19}.}
%end footnote
More generally, is the largest $*$-preserving S-positive subcategory of $\fdCAlgUY$ equal to $\fdCAlgSPU$? 
\end{ques}

\begin{ques}[A local positivity condition?]{ques:localpositivity}
Currently, the definition of S-positivity is a global condition in the sense that it involves a class of morphisms. 
For example, if for a given morphism $g$ there does not exist a morphism $f$ nor $h$ (with appropriate domains and codomains) such that $g\circ f$ or $h\circ g$ is deterministic, then this seems to suggest that adding the morphism $g$ to the subcategory causes the subcategory to remain S-positive even if the morphism $g$ might have no other good reason to be deemed positive (as in the examples of classical and quantum probability).
Is there a local definition of S-positivity specific to a given morphism rather than an entire subcategory? 
\end{ques}

\begin{ques}[Are S-positive subcategories $*$-preserving?]{ques:Spositivesubcatspreserving}
In this work, we will occasionally assume that our S-positive subcategories will also consist of $*$-preserving morphisms. It is, however, well-known that every positive map between $C^*$-algebras is automatically $*$-preserving. More generally, if $\mC$ is an S-positive subcategory of a quantum Markov category, does $\mC$ necessarily only contain $*$-preserving morphisms? 
\end{ques}

S-positive subcategories of quantum Markov categories have several useful properties. 

%\blem
%\label{lem:positivityinvertibilityimpliesdeterminism}
\begin{lem}[Invertible morphisms are deterministic in S-positive subcategories]{lem:positivityinvertibilityimpliesdeterminism}
Let $\mC$ be an S-positive subcategory of a quantum Markov category. 
Then every morphism in $\mC$ that has an inverse in $\mC$ is deterministic.
%\elem
\end{lem}

This is proved in \cite[Remark~11.28]{Fr19}, but we include the short proof because determinism (and a.e.\ determinism) will play an important role in this work and because this result is a precursor to one of our main results on Bayesian inversion (Theorem~\ref{thm:aemodbayesdisint}). 

\bprf
Let $X\xstoch{f}Y$ be a morphism in $\mC$ for which there exists a morphism $Y\xstoch{g}X$ in $\mC$ such that $f\circ g=\id_{Y}$. Then
\be
\vcenter{\hbox{%
\begin{tikzpicture}[font=\small]
\coordinate (p) at (0,-1.0);
\node[arrow box] (f) at (0,-0.3) {$f$};
\node[copier] (copier) at (0,0.3) {};
\coordinate (X) at (-0.5,0.91);
\coordinate (X2) at (0.5,0.91);
\draw (p) to (f);
\draw (f) to (copier);
\draw (copier) to[out=165,in=-90] (X);
\draw (copier) to[out=15,in=-90] (X2);
\end{tikzpicture}}}
\quad
\overset{f\circ g=\id_{Y}}{=\joinrel=\joinrel=\joinrel=\joinrel=}
\quad
\vcenter{\hbox{%
\begin{tikzpicture}[font=\small]
\coordinate (p) at (0,-1.0);
\node[arrow box] (f) at (0,-0.3) {$f$};
\node[copier] (copier) at (0,0.3) {};
\node[arrow box] (g2) at (-0.5,0.95) {$g$};
\node[arrow box] (e2) at (-0.5,1.75) {$f$};
\coordinate (X) at (-0.5,2.3);
\coordinate (Y) at (0.5,2.3);
\coordinate (Ypre) at (0.5,1.3);
\draw (p) to (f);
\draw (f) to (copier);
\draw (copier) to[out=165,in=-90] (g2);
\draw (g2) to (e2);
\draw (e2) to (X);
\draw (copier) to[out=15,in=-90] (Ypre);
\draw (Ypre) to (Y);
\end{tikzpicture}}}
\quad
\overset{\text{Defn~\ref{defn:positivecategory}}}{=\joinrel=\joinrel=\joinrel=\joinrel=}
\quad
\vcenter{\hbox{%
\begin{tikzpicture}[font=\small]
\coordinate (p) at (0,-0.3);
\node[copier] (copier) at (0,0.3) {};
\node[arrow box] (f) at (-0.5,0.95) {$f$};
\node[arrow box] (e) at (0.5,0.95) {$f$};
\node[arrow box] (g) at (-0.5,1.75) {$g$};
\node[arrow box] (h) at (-0.5,2.55) {$f$};
\coordinate (X) at (-0.5,3.1);
\coordinate (Y) at (0.5,3.1);
\draw (p) to (copier);
\draw (copier) to[out=165,in=-90] (f);
\draw (f) to (g);
\draw (g) to (h);
\draw (h) to (X);
\draw (copier) to[out=15,in=-90] (e);
\draw (e) to (Y);
\end{tikzpicture}}}
\quad
\overset{f\circ g=\id_{Y}}{=\joinrel=\joinrel=\joinrel=\joinrel=}
\quad
\vcenter{\hbox{%
\begin{tikzpicture}[font=\small]
\coordinate (p) at (0,-0.3);
\node[copier] (copier) at (0,0.3) {};
\node[arrow box] (f) at (-0.5,0.95) {$f$};
\node[arrow box] (e) at (0.5,0.95) {$f$};
\coordinate (X) at (-0.5,1.7);
\coordinate (Y) at (0.5,1.7);
\draw (p) to (copier);
\draw (copier) to[out=165,in=-90] (f);
\draw (f) to (X);
\draw (copier) to[out=15,in=-90] (e);
\draw (e) to (Y);
\end{tikzpicture}}}
\quad,
\ee
where S-positivity applies because $g\circ f=\id_{X}$ is deterministic. 
\eprf

The interpretation of Lemma~\ref{lem:positivityinvertibilityimpliesdeterminism} is that if one has an invertible stochastic map (whose inverse is also stochastic), then the stochastic map was in fact deterministic to begin with. Based on our proof that $\fdCAlgSPU$ is S-positive, this shows that the same is true in quantum theory.

%\bc
%\label{cor:positivityinvertibilityimpliesdeterminism}
\begin{cor}[Invertible SPU maps are deterministic]{cor:positivityinvertibilityimpliesdeterminism}
Every invertible morphism in $\fdCAlgSPU$ is deterministic. 
%\ec
\end{cor}

\bprf
Combine Lemma~\ref{lem:positivityinvertibilityimpliesdeterminism} with Theorem~\ref{thm:fdcalgcpupositive}. 
\eprf

Corollary~\ref{cor:positivityinvertibilityimpliesdeterminism} is a fairly well-known result in quantum information theory~\cite{Attal,Ma10} (and is originally due to Choi~\cite{Ch72}), though its proof is quite different from the standard proofs of its classical analogue~\cite{DiRh14}. The quantum Markov category approach shows that these two facts are now consequences of the same axiom and have exactly the same diagrammatic proof. 

\begin{rmk}[The transpose map is invertible but not deterministic]{a011}
Note that the transpose map is a positive unital map with a positive unital inverse (itself). Nevertheless, it is clearly not deterministic. This is consistent with our earlier observations that positive unital maps do not form an S-positive subcategory of $\fdCAlgUY$. 
\end{rmk}

Another interesting corollary for S-positive subcategories is the following general no-broadcasting theorem.

%\bt
%[The no-broadcasting theorem for 2-positive subcategories]
%\label{thm:nocloning}
\begin{theo}[The no-broadcasting theorem for S-positive subcategories]{thm:nocloning}
Let $\mC$ be an S-positive subcategory of a quantum Markov category containing the morphisms 
$\vcenter{\hbox{%
\begin{tikzpicture}[font=\footnotesize]
\node[discarder] (d) at (0,0.15) {};
\draw (d) to (0,-0.25);
\end{tikzpicture}}}$,
$\vcenter{\hbox{%
\begin{tikzpicture}
\node[discarder] (d) at (0,0.05) {};
\draw (d) to (0,-0.25);
\draw (0.5,-0.25) to (0.5,0.4);
\end{tikzpicture}}}$\;,
and 
$\vcenter{\hbox{%
\begin{tikzpicture}
\node[discarder] (d) at (0,0.05) {};
\draw (d) to (0,-0.25);
\draw (-0.5,-0.25) to (-0.5,0.4);
\end{tikzpicture}}}$\;
for each object in $\mC$. 
In addition, suppose that $\mC$ contains a morphism
$\vcenter{\hbox{%
\begin{tikzpicture}[font=\small]
\node[fakecopier] (c) at (0,0.4) {};
\draw (c)
to[out=15,in=-90] (0.25,0.7);
\draw (c)
to[out=165,in=-90] (-0.25,0.7);
\draw (c) to (0,0.1);
\end{tikzpicture}}}$
satisfying 
\be
\label{eq:fakecopygrounding}
\vcenter{\hbox{%
\begin{tikzpicture}[font=\footnotesize]
\node[fakecopier] (c) at (0,0) {};
\coordinate (x1) at (-0.3,0.3);
\node[discarder] (d) at (x1) {};
\coordinate (x2) at (0.3,0.5);
\draw (c) to[out=165,in=-90] (x1);
\draw (c) to[out=15,in=-90] (x2);
\draw (c) to (0,-0.3);
\end{tikzpicture}}}
\quad=\quad
\vcenter{\hbox{%
\begin{tikzpicture}[font=\footnotesize]
\draw (0,0) to (0,0.8);
\end{tikzpicture}}}
\quad=\quad
\vcenter{\hbox{%
\begin{tikzpicture}[font=\footnotesize]
\node[fakecopier] (c) at (0,0) {};
\coordinate (x1) at (-0.3,0.5);
\coordinate (x2) at (0.3,0.3);
\node[discarder] (d) at (x2) {};
\draw (c) to[out=165,in=-90] (x1);
\draw (c) to[out=15,in=-90] (x2);
\draw (c) to (0,-0.3);
\end{tikzpicture}}}
%\quad.
\tag{\Football}
\ee
for every object in $\mC$. 
Then $\vcenter{\hbox{%
\begin{tikzpicture}[font=\small]
\node[fakecopier] (c) at (0,0.4) {};
\draw (c)
to[out=15,in=-90] (0.25,0.7);
\draw (c)
to[out=165,in=-90] (-0.25,0.7);
\draw (c) to (0,0.1);
\end{tikzpicture}}}$ 
is commutative and in fact equals duplication for every object of $\mC$. 
%\et
\end{theo}

\begin{rmk}[Assumptions in the no-broadcasting theorem]{a012}
Before proving the no-broadcasting theorem, we explain the physical meaning of the assumptions. 
The morphism $\vcenter{\hbox{%
\begin{tikzpicture}[font=\footnotesize]
\node[discarder] (d) at (0,0.15) {};
\draw (d) to (0,-0.25);
\end{tikzpicture}}}$ is interpreted as discarding a system. 
The morphisms $\vcenter{\hbox{%
\begin{tikzpicture}
\node[discarder] (d) at (0,0.05) {};
\draw (d) to (0,-0.25);
\draw (0.5,-0.25) to (0.5,0.4);
\end{tikzpicture}}}$
and 
$\vcenter{\hbox{%
\begin{tikzpicture}
\node[discarder] (d) at (0,0.05) {};
\draw (d) to (0,-0.25);
\draw (-0.5,-0.25) to (-0.5,0.4);
\end{tikzpicture}}}$
are interpreted as choosing one of two possible systems in a way that does not alter the other system. 
The morphism $\vcenter{\hbox{%
\begin{tikzpicture}[font=\small]
\node[fakecopier] (c) at (0,0.4) {};
\draw (c)
to[out=15,in=-90] (0.25,0.7);
\draw (c)
to[out=165,in=-90] (-0.25,0.7);
\draw (c) to (0,0.1);
\end{tikzpicture}}}$ 
is an operation that \emph{broadcasts} the information in one system to two copies of that system. A-priori, it is unrelated to the morphism $\vcenter{\hbox{%
\begin{tikzpicture}
\node[copier] (c) at (0,0.4) {};
\draw (c)
to[out=15,in=-90] (0.25,0.7);
\draw (c)
to[out=165,in=-90] (-0.25,0.7);
\draw (c) to (0,0.1);
\end{tikzpicture}}}$. 
The condition~(\ref{eq:fakecopygrounding}) guarantees that once information is transferred to the joint system, each of the two systems has a genuine copy of the system. This means that the marginals of any state (cf.\ Definition~\ref{defn:state}) being broadcast are equal. 
If we interpret the category $\mC$ as one corresponding to admissible operations (say for open system dynamics), then assuming that these morphisms are in $\mC$ means that these are valid physical operations. 
\end{rmk}

The calculation in the following proof is similar to the one in~\cite[Remark~11.29]{Fr19}, though our interpretation of the result is given a more physical meaning. 

\bprf
[Proof of Theorem~\ref{thm:nocloning}]
Set
$g:=\!
\vcenter{\hbox{%
\begin{tikzpicture}
\draw[draw opacity=0, fill=blue, fill opacity=0.1] (-0.35,-0.25) -- (-0.35,0.4) -- (0.6,0.4) -- (0.6,-0.25) -- cycle;
\node[discarder] (d) at (0,0.05) {};
\draw (d) to (0,-0.25);
\draw (0.5,-0.25) to (0.5,0.4);
\end{tikzpicture}}}$\;,
$h:=\!
\vcenter{\hbox{%
\begin{tikzpicture}
\draw[draw opacity=0, fill=yellow, fill opacity=0.2] (-0.65,-0.25) -- (-0.65,0.4) -- (0.3,0.4) -- (0.3,-0.25) -- cycle;
\node[discarder] (d) at (0,0.05) {};
\draw (d) to (0,-0.25);
\draw (-0.5,-0.25) to (-0.5,0.4);
\end{tikzpicture}}}$\;, 
and 
$f:=\!
\vcenter{\hbox{%
\begin{tikzpicture}
\draw[draw opacity=0, fill=red, fill opacity=0.1] (-0.35,0.1) -- (-0.35,0.7) -- (0.35,0.7) -- (0.35,0.1) -- cycle;
\node[fakecopier] (c) at (0,0.4) {};
\draw (c)
to[out=15,in=-90] (0.25,0.7);
\draw (c)
to[out=165,in=-90] (-0.25,0.7);
\draw (c) to (0,0.1);
\end{tikzpicture}}}$.%
%footnote
\footnote{We have colored the background of these diagrams to better illustrate the calculation in~(\ref{eq:provingnocloning}).}
%end footnote
Then $g\circ f=\id=h\circ f$ by~(\ref{eq:fakecopygrounding}) and is therefore deterministic. Hence, 
\be
\label{eq:provingnocloning}
\vcenter{\hbox{%
\begin{tikzpicture}[font=\small]
\node[fakecopier] (c) at (0,0.4) {};
\draw (c)
to[out=15,in=-90] (0.25,0.7)
to[out=90,in=-90] (-0.25,1.4);
\draw (c)
to[out=165,in=-90] (-0.25,0.7)
to[out=90,in=-90] (0.25,1.4);
\draw (c) to (0,0.1);
\end{tikzpicture}}}
%\;%\;
\overset{\text{(\ref{eq:markovcatfirstconditions})}}{=\joinrel=\joinrel=\joinrel=}
%\;%\;
\vcenter{\hbox{%
\begin{tikzpicture}
\draw[draw opacity=0, fill=blue, fill opacity=0.1] (-0.85,0.5) -- (-0.85,1.1) -- (0,1.1) -- (0,0.5) -- cycle;
\draw[draw opacity=0, fill=red, fill opacity=0.1] (-0.35,-0.5) -- (-0.35,0) -- (0.35,0) -- (0.35,-0.5) -- cycle;
\node[discarder] (d) at (-0.5,0.7) {};
\node[copier] (c1) at (-0.25,0.2) {};
\node[copier] (c2) at (0.25,0.2) {};
\node[fakecopier] (c3) at (0,-0.2) {};
\node[discarder] (d2) at (0.5,0.7) {};
\draw (0,-0.5) to (c3);
\draw (c3) to[out=15,in=-90] (c2);
\draw (c3) to[out=165,in=-90] (c1);
\draw (c1) to[out=165,in=-90] (d);
\draw (c1) to[out=15,in=-90] (0.15,1.1);
\draw (c2) to[out=165,in=-90] (-0.15,1.1);
\draw (c2) to[out=15,in=-90] (d2);
\end{tikzpicture}}}
%\;%\;
\underset{\text{\ref{defn:positivecategory}}}{\overset{\text{Defn}}{=\joinrel=\joinrel=}}
%\;%\;
\vcenter{\hbox{%
\begin{tikzpicture}
\draw[draw opacity=0, fill=blue, fill opacity=0.1] (-1,0.5) -- (-1,1.1) -- (-0.025,1.1) -- (-0.025,0.5) -- cycle;
\draw[draw opacity=0, fill=red, fill opacity=0.1] (-1,0) -- (-1,0.5) -- (-0.025,0.5) -- (-0.025,0) -- cycle;
\draw[draw opacity=0, fill=red, fill opacity=0.1] (0.025,0) -- (0.025,0.5) -- (1,0.5) -- (1,0) -- cycle;
\node[discarder] (d) at (-0.7,0.7) {};
\node[discarder] (d2) at (0.7,0.7) {};
\node[fakecopier] (c1) at (-0.4,0.2) {};
\node[fakecopier] (c2) at (0.4,0.2) {};
\node[copier] (c3) at (0,-0.2) {};
\draw (0,-0.5) to (c3);
\draw (c3) to[out=15,in=-90] (c2);
\draw (c3) to[out=165,in=-90] (c1);
\draw (c1) to[out=165,in=-90] (d);
\draw (c1) to[out=15,in=-90] (-0.15,1.1);
\draw (c2) to[out=165,in=-90] (0.15,1.1);
\draw (c2) to[out=15,in=-90] (d2);
\end{tikzpicture}}}
%\;%\;
\overset{\text{(\ref{eq:fakecopygrounding})}}{=\joinrel=}
%\;%\;
\vcenter{\hbox{%
\begin{tikzpicture}[font=\small]
\node[copier] (c) at (0,0.4) {};
\draw (c)
to[out=15,in=-90] (0.25,0.7);
\draw (c)
to[out=165,in=-90] (-0.25,0.7);
\draw (c) to (0,0.1);
\end{tikzpicture}}}
%\;%\;
\overset{\text{(\ref{eq:fakecopygrounding})}}{=\joinrel=}
%\;%\;
\vcenter{\hbox{%
\begin{tikzpicture}
\draw[draw opacity=0, fill=yellow, fill opacity=0.2] (-1,0.5) -- (-1,1.1) -- (-0.025,1.1) -- (-0.025,0.5) -- cycle;
\draw[draw opacity=0, fill=red, fill opacity=0.1] (-1,0) -- (-1,0.5) -- (-0.025,0.5) -- (-0.025,0) -- cycle;
\draw[draw opacity=0, fill=red, fill opacity=0.1] (0.025,0) -- (0.025,0.5) -- (1,0.5) -- (1,0) -- cycle;
\node[discarder] (d) at (-0.25,0.7) {};
\node[discarder] (d2) at (0.25,0.7) {};
\node[fakecopier] (c1) at (-0.55,0.2) {};
\node[fakecopier] (c2) at (0.55,0.2) {};
\node[copier] (c3) at (0,-0.2) {};
\draw (0,-0.5) to (c3);
\draw (c3) to[out=15,in=-90] (c2);
\draw (c3) to[out=165,in=-90] (c1);
\draw (c1) to[out=15,in=-90] (d);
\draw (c1) to[out=165,in=-90] (-0.9,1.1);
\draw (c2) to[out=15,in=-90] (0.9,1.1);
\draw (c2) to[out=165,in=-90] (d2);
\end{tikzpicture}}}
%\;%\;
\underset{\text{\ref{defn:positivecategory}}}{\overset{\text{Defn}}{=\joinrel=\joinrel=}}
%\;%\;
\vcenter{\hbox{%
\begin{tikzpicture}
\draw[draw opacity=0, fill=yellow, fill opacity=0.2] (-0.85,0.5) -- (-0.85,1.1) -- (0,1.1) -- (0,0.5) -- cycle;
\draw[draw opacity=0, fill=red, fill opacity=0.1] (-0.55,-0.5) -- (-0.55,0) -- (0.55,0) -- (0.55,-0.5) -- cycle;
\node[discarder] (d) at (0.25,0.7) {};
\node[copier] (c1) at (-0.45,0.2) {};
\node[copier] (c2) at (0.45,0.2) {};
\node[fakecopier] (c3) at (0,-0.2) {};
\node[discarder] (d2) at (-0.25,0.7) {};
\draw (0,-0.5) to (c3);
\draw (c3) to[out=15,in=-90] (c2);
\draw (c3) to[out=165,in=-90] (c1);
\draw (c1) to[out=165,in=-90] (-0.7,1.1);
\draw (c1) to[out=15,in=-90] (d);
\draw (c2) to[out=15,in=-90] (0.7,1.1);
\draw (c2) to[out=165,in=-90] (d2);
\end{tikzpicture}}}
%\;%\;
\overset{\text{(\ref{eq:markovcatfirstconditions})}}{=\joinrel=\joinrel=\joinrel=}
%\;%\;
\vcenter{\hbox{%
\begin{tikzpicture}[font=\small]
\node[fakecopier] (c) at (0,0.4) {};
\draw (c)
to[out=15,in=-90] (0.25,0.7);
\draw (c)
to[out=165,in=-90] (-0.25,0.7);
\draw (c) to (0,0.1);
\end{tikzpicture}}}
\;,
\ee
which reproduces the identity~(\ref{eq:commutativity}) in Definition~\ref{defn:qmc} since 
$\vcenter{\hbox{%
\begin{tikzpicture}[font=\small]
\node[copier] (c) at (0,0.4) {};
\draw (c)
to[out=15,in=-90] (0.25,0.7);
\draw (c)
to[out=165,in=-90] (-0.25,0.7);
\draw (c) to (0,0.1);
\end{tikzpicture}}}
\;=\;
\vcenter{\hbox{%
\begin{tikzpicture}[font=\small]
\node[fakecopier] (c) at (0,0.4) {};
\draw (c)
to[out=15,in=-90] (0.25,0.7);
\draw (c)
to[out=165,in=-90] (-0.25,0.7);
\draw (c) to (0,0.1);
\end{tikzpicture}}}$. 
\eprf

The following remarks will relate our no-cloning result to others in the literature. A discussion of the standard no-cloning theorems is provided in Remark~\ref{rmk:standardnocloning}. 

\begin{rmk}[A comparison to other categorical no-cloning theorems]{rmk:categoricalnocloning}
This remark will briefly explain the differences between the assumptions in our no-cloning theorem to some other categorical ones in the literature~\cite{Ab09,Fe12}. 
In~\cite[Theorem~11]{Ab09} (see also~\cite[Theorem~4.84]{CoKi17} and~\cite[Theorem~4.27]{HeVi19}), Abramsky provided a categorical formulation of the no-cloning theorem, whose mathematical conclusions are a bit different from ours because the definition of cloning used there is a bit different from ours. Hence, we will focus more on the assumptions used in stating and proving the theorem. Abramsky's definition of cloning is as follows. A symmetric monoidal category $(\mM,\otimes,I)$ has \define{uniform cloning} iff it has a monoidal natural transformation 
\[
\xy0;/r.25pc/:
(-13,0)*+{\mM}="1";
(13,0)*+{\mM}="2";
{\ar@/_1.25pc/"1";"2"_{\id_{\mC}}};
{\ar@/^1.25pc/"1";"2"^{\otimes\circ\Delta_{\mC}}};
{\ar@{=>}(0,-4);(0,4)};
\endxy
\]
that is coassociative and cocommutative (here, $\otimes\circ\Delta_{\mC}$ takes an object $X$ to $X\otimes X$). In terms of string diagrams, this means that every object $X$ has a morphism 
$\vcenter{\hbox{%
\begin{tikzpicture}
\node[fakecopier] (c) at (0,0.4) {};
\draw (c)
to[out=15,in=-90] (0.25,0.7);
\draw (c)
to[out=165,in=-90] (-0.25,0.7);
\draw (c) to (0,0.1);
\end{tikzpicture}}}$
satisfying the following conditions 
\[
\vcenter{\hbox{%
\begin{tikzpicture}[font=\footnotesize]
\node[fakecopier] (c2) at (0,0) {};
\node[fakecopier] (c1) at (-0.3,0.3) {};
\draw (c2) to[out=165,in=-90] (c1);
\draw (c2) to[out=15,in=-90] (0.4,0.6);
\draw (c1) to[out=165,in=-90] (-0.6,0.6);
\draw (c1) to[out=15,in=-90] (0,0.6);
\draw (c2) to (0,-0.3);
\end{tikzpicture}}}
\quad=\quad
\vcenter{\hbox{%
\begin{tikzpicture}[font=\footnotesize]
\node[fakecopier] (c2) at (-0.5,0) {};
\node[fakecopier] (c1) at (-0.2,0.3) {};
\draw (c2) to[out=15,in=-90] (c1);
\draw (c2) to[out=165,in=-90] (-1,0.6);
\draw (c1) to[out=165,in=-90] (-0.5,0.6);
\draw (c1) to[out=15,in=-90] (0.1,0.6);
\draw (c2) to (-0.5,-0.3);
\end{tikzpicture}}}
\qquad
\vcenter{\hbox{%
\begin{tikzpicture}[font=\footnotesize]
\node[fakecopier] (c) at (0,0.4) {};
\draw (c)
to[out=15,in=-90] (0.25,0.7);
\draw (c)
to[out=165,in=-90] (-0.25,0.7);
\draw (c) to (0,0);
\node at (0.5,0.1) {$X\otimes Y$};
\end{tikzpicture}}}
=
\vcenter{\hbox{%
\begin{tikzpicture}[font=\footnotesize]
\node[fakecopier] (c) at (-0.1,0) {};
\node[fakecopier] (c2) at (0.5,0) {};
\draw (c) to[out=15,in=-90] +(0.45,0.45);
\draw (c) to[out=165,in=-90] +(-0.4,0.45);
\draw (c) to +(0,-0.4);
\draw (c2) to[out=15,in=-90] +(0.4,0.45);
\draw (c2) to[out=165,in=-90] +(-0.45,0.45);
\draw (c2) to +(0,-0.4);
\node at (-0.3,-0.3) {$X$};
\node at (0.7,-0.3) {$Y$};
\end{tikzpicture}}}
\qquad
\vcenter{\hbox{%
\begin{tikzpicture}[font=\footnotesize]
\node[fakecopier] (c) at (0,0.4) {};
\draw (c)
to[out=15,in=-90] (0.25,0.7);
\draw (c)
to[out=165,in=-90] (-0.25,0.7);
\draw (c) to (0,0);
\node at (0.2,0.1) {$I$};
\end{tikzpicture}}}
=\;
\vcenter{\hbox{%
\begin{tikzpicture}[font=\footnotesize]
\node at (0.2,-0.05) {};
\draw [gray,dashed] (0,0) rectangle (0.45,0.75);
\end{tikzpicture}}}
\qquad
\vcenter{\hbox{%
\begin{tikzpicture}[font=\small]
\node[fakecopier] (c) at (0,0.4) {};
\draw (c)
to[out=15,in=-90] (0.25,0.65)
to[out=90,in=-90] (-0.25,1.2);
\draw (c)
to[out=165,in=-90] (-0.25,0.65)
to[out=90,in=-90] (0.25,1.2);
\draw (c) to (0,0.1);
\end{tikzpicture}}}
\quad=\quad
\vcenter{\hbox{%
\begin{tikzpicture}[font=\small]
\node[fakecopier] (c) at (0,0.4) {};
\draw (c)
to[out=15,in=-90] (0.25,0.7);
\draw (c)
to[out=165,in=-90] (-0.25,0.7);
\draw (c) to (0,0.1);
\end{tikzpicture}}}
\]
and 
\[
\vcenter{\hbox{%
\begin{tikzpicture}[font=\small]
\coordinate (p) at (0,-1.0);
\node at (-0.2,-0.85) {\footnotesize$X$};
\node[arrow box] (f) at (0,-0.3) {$f$};
\node[fakecopier] (copier) at (0,0.3) {};
\coordinate (X) at (-0.5,0.91);
\node at (-0.7,0.75) {\footnotesize$Y$};
\coordinate (X2) at (0.5,0.91);
\node at (0.7,0.75) {\footnotesize$Y$};
\draw (p) to (f);
\draw (f) to (copier);
\draw (copier) to[out=165,in=-90] (X);
\draw (copier) to[out=15,in=-90] (X2);
\end{tikzpicture}}}
\quad
=
\quad
\vcenter{\hbox{%
\begin{tikzpicture}[font=\small]
\coordinate (p) at (0,-0.3);
\node at (-0.2,-0.15) {\footnotesize$X$};
\node[fakecopier] (copier) at (0,0.3) {};
\node[arrow box] (f) at (-0.5,0.95) {$f$};
\node[arrow box] (e) at (0.5,0.95) {$f$};
\coordinate (X) at (-0.5,1.7);
\node at (-0.7,1.55) {\footnotesize$Y$};
\coordinate (Y) at (0.5,1.7);
\node at (0.7,1.55) {\footnotesize$Y$};
\draw (p) to (copier);
\draw (copier) to[out=165,in=-90] (f);
\draw (f) to (X);
\draw (copier) to[out=15,in=-90] (e);
\draw (e) to (Y);
\end{tikzpicture}}}
\]
for all objects $X,Y$ and all morphisms $X\xstoch{f}Y$. Moreover, the version of the no-cloning theorem proved in~\cite[Theorem~11]{Ab09} is in the setting of a compact closed category (so that every object has a dual satisfying the zig-zag identities) with uniform cloning. Therefore, the setup is quite different from the one we have. There are four crucial differences worth pointing out. First, Abramsky does not assume an explicit discard map. Second, Abramsky demands that $\vcenter{\hbox{%
\begin{tikzpicture}
\node[fakecopier] (c) at (0,0.4) {};
\draw (c)
to[out=15,in=-90] (0.25,0.7);
\draw (c)
to[out=165,in=-90] (-0.25,0.7);
\draw (c) to (0,0.1);
\end{tikzpicture}}}$ is natural, which says that every morphism is deterministic with respect to $\vcenter{\hbox{%
\begin{tikzpicture}
\node[fakecopier] (c) at (0,0.4) {};
\draw (c)
to[out=15,in=-90] (0.25,0.7);
\draw (c)
to[out=165,in=-90] (-0.25,0.7);
\draw (c) to (0,0.1);
\end{tikzpicture}}}$. In particular, this and the other axioms imply  
\[
\vcenter{\hbox{%
\begin{tikzpicture}[font=\small]
\node[state] (p) at (0,-0.2) {$p$};
\node[fakecopier] (c) at (0,0.3) {};
\coordinate (X) at (-0.5,0.85);
\coordinate (X2) at (0.5,0.85);
\draw (p) to (c);
\draw (c) to[out=165,in=-90] (X);
\draw (c) to[out=15,in=-90] (X2);
\end{tikzpicture}}}
\quad=\quad
\vcenter{\hbox{%
\begin{tikzpicture}[font=\small]
\coordinate (p) at (0,-0.3);
\node[state] (f) at (-0.5,0.95) {$p$};
\node[state] (e) at (0.5,0.95) {$p$};
\coordinate (X) at (-0.5,1.7);
\coordinate (Y) at (0.5,1.7);
\draw (f) to (X);
\draw (e) to (Y);
\end{tikzpicture}}}
%%%
\qquad\text{ and }\qquad
%%%
\vcenter{\hbox{%
\begin{tikzpicture}[font=\footnotesize]
\node[state] (s) at (-0.7,-0.2) {\;$s$\;};
\node[state] (s2) at (0.7,-0.2) {\;$s$\;};
\draw (s)++(-0.2,0) to (-0.9,0.7);
\draw (s)++(0.2,0) to (-0.5,0.7);
\draw (s2)++(-0.2,0) to (0.5,0.7);
\draw (s2)++(0.2,0) to (0.9,0.7);
\end{tikzpicture}}}
=
\vcenter{\hbox{%
\begin{tikzpicture}[font=\footnotesize]
\node[fakecopier] (c) at (-0.4,0) {};
\node[fakecopier] (c2) at (0.4,0) {};
\node[state] (s) at (0,-0.6) {\;\;$s$\;\;};
\draw (c) to[out=15,in=-90] +(0.6,0.6);
\draw (c) to[out=165,in=-90] +(-0.5,0.6);
\draw (s)++(-0.4, 0) to (c);
\draw (c2) to[out=15,in=-90] +(0.5,0.6);
\draw (c2) to[out=165,in=-90] +(-0.6,0.6);
\draw (s)++(0.4, 0) to (c2);
\node at (-0.6,-0.3) {$X$};
\node at (0.6,-0.3) {$Y$};
\end{tikzpicture}}}
\]
which give all but one of the axioms assumed in~\cite[Theorem~4.84]{CoKi17}. 
Third, the commutativity condition is assumed. Fourth, by assuming compact closure, the result applies (in particular) to the category of finite-dimensional Hilbert spaces, so that the no-cloning result is a statement about pure states in finite dimensions. Because of the first point, our result does not imply Abramsky's version of the theorem. However, because of the second, third, or fourth points, Abramsky's version does not imply our result either. Nevertheless, the axiom of S-positivity can be viewed as a remnant of the deterministic condition (every deterministic subcategory of a quantum Markov category is automatically S-positive), so that one might suspect that there is some common generalization of Theorem~\ref{thm:nocloning} and Abramsky's result. 

Furthermore, if one assumes that
\[
\vcenter{\hbox{%
\begin{tikzpicture}[font=\small]
\node[state] (p) at (0,-0.9) {$p$};
\node[arrow box] (f) at (0,-0.3) {$f$};
\coordinate (X) at (0,0.3);
\draw (p) to (f);
\draw (f) to (X);
\end{tikzpicture}}}
\quad
=
\quad
\vcenter{\hbox{%
\begin{tikzpicture}[font=\small]
\node[state] (p) at (0,-0.9) {$p$};
\node[arrow box] (f) at (0,-0.3) {$g$};
\coordinate (X) at (0,0.3);
\draw (p) to (f);
\draw (f) to (X);
\end{tikzpicture}}}
\qquad\forall\;
\vcenter{\hbox{%
\begin{tikzpicture}[font=\small]
\node[state] (p) at (0,-0.9) {$p$};
\coordinate (X) at (0,-0.3);
\draw (p) to (X);
\end{tikzpicture}}}
\quad\implies\quad
\vcenter{\hbox{%
\begin{tikzpicture}[font=\small]
\coordinate (p) at (0,-0.9);
\node[arrow box] (f) at (0,-0.3) {$f$};
\coordinate (X) at (0,0.3);
\draw (p) to (f);
\draw (f) to (X);
\end{tikzpicture}}}
\quad
=
\quad
\vcenter{\hbox{%
\begin{tikzpicture}[font=\small]
\coordinate (p) at (0,-0.9);
\node[arrow box] (f) at (0,-0.3) {$g$};
\coordinate (X) at (0,0.3);
\draw (p) to (f);
\draw (f) to (X);
\end{tikzpicture}}}
\]
(which is valid in the category of finite-dimensional Hilbert spaces as well as $\fdCAlg$), then 
\[
\vcenter{\hbox{%
\begin{tikzpicture}[font=\footnotesize]
\node[fakecopier] (c) at (0,0) {};
\coordinate (x1) at (-0.3,0.3);
\node[discarder] (d) at (x1) {};
\coordinate (x2) at (0.3,0.5);
\draw (c) to[out=165,in=-90] (x1);
\draw (c) to[out=15,in=-90] (x2);
\draw (c) to (0,-0.3);
\end{tikzpicture}}}
\quad=\quad
\vcenter{\hbox{%
\begin{tikzpicture}[font=\footnotesize]
\draw (0,0) to (0,0.8);
\end{tikzpicture}}}
\quad=\quad
\vcenter{\hbox{%
\begin{tikzpicture}[font=\footnotesize]
\node[fakecopier] (c) at (0,0) {};
\coordinate (x1) at (-0.3,0.5);
\coordinate (x2) at (0.3,0.3);
\node[discarder] (d) at (x2) {};
\draw (c) to[out=165,in=-90] (x1);
\draw (c) to[out=15,in=-90] (x2);
\draw (c) to (0,-0.3);
\end{tikzpicture}}}
\]
follows from the other axioms. This only formally relates one of our axioms, but it remains unclear how (or if) the other axioms are related to S-positivity. Another no-broadcasting theorem is also presented in~\cite[Section~7.3.2]{HeVi19} in the setting of dagger Frobenius objects in braided monoidal dagger categories, which applies to the setting of finite-dimensional Hilbert spaces. Note that we have avoided an explicit $\dagger$-structure in Theorem~\ref{thm:nocloning}. 

Finally, we briefly mention that Fenyes has also formulated two categorical variants of the no-cloning theorem~\cite{Fe12}, one of which states that no-cloning is not possible even in the classical (symplectic) setting depending on the definition used. However, even in the broader definition including cloning machines (which includes the state wished to be cloned, the raw material needed to construct the clone, and the machine needed to actually put the raw material together to build the clone), one still cannot obtain cloning in the quantum setting. Although Fenyes proved symplectic vector spaces allow cloning in this broader sense, it remains an open question whether all symplectic manifolds admit this cloning. Further investigation would be needed to relate all of these results or to view them from some more general perspective. 
\end{rmk}

\begin{rmk}[A comparison to the standard no-broadcasting/no-cloning theorems]{rmk:standardnocloning}
The above categorical formulations of the no-broadcasting/no-cloning theorems all depend on different definitions of broadcasting/cloning. None of these definitions agree \emph{exactly} with the ones typically encountered in the physics literature~\cite{BBLW,Li99,BCFJS}, though many of them can be viewed as special cases thereof. Typically, one is interested in broadcasting a \emph{subset} of states, rather than \emph{all} the states allowed. From this perspective Theorem~\ref{thm:nocloning} and Theorem~\ref{thm:fdcalgcpupositive} provide a \emph{universal} variant% 
%footnote
\footnote{The adjective ``universal'' here is meant in the sense that the broadcasting operation is valid for \emph{all} input states as opposed to a subset of states (this is the terminology used in~\cite{BBLW} for example). It is \emph{not} meant in the categorical sense as in ``universal property''.}
%end footnote
of the standard no-broadcasting theorem for quantum mechanics~\cite{BCFJS} or algebraic quantum theory~\cite{Lu15,KLL15,KLL17}. 
There is also a no-broadcasting theorem proved in the framework of general probabilistic theories~\cite[Theorem~2]{BBLW}. It is interesting to point out that unlike in the proof from~\cite{BBLW}, we have not used any assumptions regarding convexity.%
%footnote
\footnote{The category of 2-dimensional topological cobordisms~\cite{Ko03} (where copy $S^1\xstoch{\Delta_{S^1}}S^1$ is the pair of pants and discard $S^1\xstoch{!_{S^1}}\varnothing$ is the cap/disc) is a classical CD category that has no obvious notion of a convex structure. Note that this category is not a Markov category since not every morphism is unital (there are non-trivial cobordisms $S^1\stoch\varnothing$).}
%end footnote 
\end{rmk}

There are several other notions of positive subcategories introduced by Fritz, one of which appeared in earlier drafts but did not make it to his final work since it was encompassed by another stronger axiom~\cite{Fr19}. One important consequence of these axioms is that they allow one to prove theorems in statistics abstractly. To state these additional axioms, however, we must take a minor, but exceptionally important, digression and discuss almost everywhere (a.e.) equivalence as well as a.e.\ determinism.

%%%%%%%%%%%%%%%%%%%%%%%%%%%%%%%%%%%%%%
\section[Almost everywhere equivalence]{a013}
\label{sec:ae}
\vspace{-12mm}
\noindent
\begin{tikzpicture}
\coordinate (L) at (-8.75,0);
\coordinate (R) at (8.75,0);
\draw[line width=2pt,orange!20] (L) -- node[anchor=center,rectangle,fill=orange!20]{\strut \Large \textcolor{black}{\textbf{5\;\; Almost everywhere equivalence}}} (R);
\end{tikzpicture}
%\vspace{1mm}
%%%%%%%%%%%%%%%%%%%%%%%%%%%%%%%%%%%%%%

In this section, we define almost everywhere (a.e.) equivalence abstractly in any quantum Markov category. The definition is based on the observation of Cho and Jacobs that measure-theoretic a.e.\ equivalence has a diagrammatic formulation~\cite{ChJa18}. This definition was then later generalized by Fritz to allow for a.e.\ equivalence with respect to morphisms. The latter is useful for applications in statistical models~\cite[Definition~13.1]{Fr19} (see also Example~\ref{exa:Umegakisufficientstatistic}). In our presentation, we distinguish two versions of their definition since they are inequivalent in general quantum Markov categories. 
Lemma~\ref{thm:ncaeequivalence} describes many equivalent definitions of a.e.\ equivalence in $\fdCAlgUY$. It is the main result of this section, as it not only provides a  concrete description of a.e.\ equivalence, but it will also help us prove many of the main theorems for $C^*$-algebras appearing in the rest of this work. The remaining parts of this section describe a.e.\ determinism and prove additional lemmas needed later. 

%\bd
%\label{defn:aeequivalence}
\begin{defn}[Almost everywhere equivalence abstractly]{defn:aeequivalence}
Let $\Theta,$ $X$, and $Y$ be objects in a quantum Markov category, let $\Theta\xstoch{p}X$ and $f,g:X\stoch Y$ be even morphisms. The morphism $f$ is said to be \define{left/right $p$-a.e.\ equivalent to} $g$ iff 
%\be
%\label{eq:aeequivalencediagram}
\[
\left.
\vcenter{\hbox{
\begin{tikzpicture}[font=\small]
\node[arrow box] (p) at (0,-0.3) {$p$};
\node[copier] (copier) at (0,0.3) {};
\node[arrow box] (f) at (-0.5,0.95) {$f$};
\coordinate (X) at (0.5,1.5);
\coordinate (Y) at (-0.5,1.5);
\draw (0,-0.9) to (p);
\draw (p) to (copier);
\draw (copier) to[out=15,in=-90] (X);
\draw (copier) to[out=165,in=-90] (f);
\draw (f) to (Y);
\end{tikzpicture}}}
\quad=\quad
\vcenter{\hbox{
\begin{tikzpicture}[font=\small]
\node[arrow box] (p) at (0,-0.3) {$p$};
\node[copier] (copier) at (0,0.3) {};
\node[arrow box] (g) at (-0.5,0.95) {$g$};
\coordinate (X) at (0.5,1.5);
\coordinate (Y) at (-0.5,1.5);
\draw (0,-0.9) to (p);
\draw (p) to (copier);
\draw (copier) to[out=15,in=-90] (X);
\draw (copier) to[out=165,in=-90] (g);
\draw (g) to (Y);
\end{tikzpicture}}}
\qquad
\middle/
\qquad
\vcenter{\hbox{
\begin{tikzpicture}[font=\small,xscale=-1]
\node[arrow box] (p) at (0,-0.3) {$p$};
\node[copier] (copier) at (0,0.3) {};
\node[arrow box] (f) at (-0.5,0.95) {$f$};
\coordinate (X) at (0.5,1.5);
\coordinate (Y) at (-0.5,1.5);
\draw (0,-0.9) to (p);
\draw (p) to (copier);
\draw (copier) to[out=15,in=-90] (X);
\draw (copier) to[out=165,in=-90] (f);
\draw (f) to (Y);
\end{tikzpicture}}}
\quad=\quad
\vcenter{\hbox{
\begin{tikzpicture}[font=\small,xscale=-1]
\node[arrow box] (p) at (0,-0.3) {$p$};
\node[copier] (copier) at (0,0.3) {};
\node[arrow box] (g) at (-0.5,0.95) {$g$};
\coordinate (X) at (0.5,1.5);
\coordinate (Y) at (-0.5,1.5);
\draw (0,-0.9) to (p);
\draw (p) to (copier);
\draw (copier) to[out=15,in=-90] (X);
\draw (copier) to[out=165,in=-90] (g);
\draw (g) to (Y);
\end{tikzpicture}}}
\right.
\quad.
%\ee
\]
When $f$ is both right and left $p$-a.e.\ equivalent to $g$, we will say $f$ is \define{$p$-a.e.\ equivalent to} $g$, and the notation $f\underset{\raisebox{.6ex}[0pt][0pt]{\scriptsize$p$}}{=}g$ will be used. 
\end{defn}

\vspace{-1mm}
\begin{exa}[Categorical a.e.\ equivalence agrees with the measure-theoretic one]{exa:classicalae}
In a classical Markov category, left and right equivalence are equivalent~\cite{ChJa18}. Furthermore, they agree with a measure-theoretic notion of a.e.\ equivalence in $\Stoch$ (and $\FinStoch$) when $p$ is a probability measure, i.e.\ a state (cf.\ Remark~\ref{rmk:moregeneralae} below). This does depend on your definition of measure-theoretic a.e.\ equivalence, and the two notions agree for a reasonable class of measure spaces (see ~\cite[Appendix~A]{PaRu19}, \cite[Example~13.3]{Fr19}, and \cite[Proposition~5.3]{ChJa18} for a comparison of the two definitions). If $(X,\Sigma,p)$ and $(Y,\Omega,q)$ are two probability spaces, and if $f,g:X\stoch Y$ are two stochastic maps, then $f\aeequals{p}g$ with respect to the diagrammatic definition of a.e.\ equivalence if and only if $\int_{A}f_{x}(B)\,dp(x)=\int_{A}g_{x}(B)\,dp(x)$ for all measurable subsets $A\in\S$ and $B\in\Omega$. The generalized version when $p$ is another morphism (as opposed to a state) is a parametrized version of a.e.\ equivalence and is described in~\cite[Example~13.3]{Fr19}. 
\end{exa}

Although one can directly apply the abstract definition of a.e.\ equivalence to see what it says in $\fdCAlgUY$, the resulting equation does not seem particularly enlightening. We will find a much more explicit and satisfying realization of this definition in Theorem~\ref{thm:ncaeequivalence} after going through several preliminaries. 

%\br
%\label{rmk:moregeneralae}
\begin{rmk}[A.e.\ equivalence relative to a state]{rmk:moregeneralae}
As in Example~\ref{exa:classicalae}, $p$ in Definition~\ref{defn:aeequivalence} could be a state. 
In this case, the two notions look like
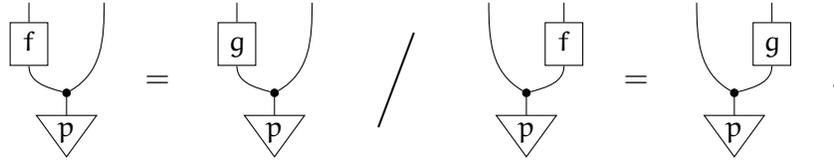
\[
\left.
\vcenter{\hbox{
\begin{tikzpicture}[font=\small]
\node[state] (p) at (0,0) {$p$};
\node[copier] (copier) at (0,0.3) {};
\node[arrow box] (f) at (-0.5,0.95) {$f$};
\coordinate (X) at (0.5,1.5);
\coordinate (Y) at (-0.5,1.5);
\draw (p) to (copier);
\draw (copier) to[out=30,in=-90] (X);
\draw (copier) to[out=165,in=-90] (f);
\draw (f) to (Y);
\end{tikzpicture}}}
\quad=\quad
\vcenter{\hbox{
\begin{tikzpicture}[font=\small]
\node[state] (p) at (0,0) {$p$};
\node[copier] (copier) at (0,0.3) {};
\node[arrow box] (g) at (-0.5,0.95) {$g$};
\coordinate (X) at (0.5,1.5);
\coordinate (Y) at (-0.5,1.5);
\draw (p) to (copier);
\draw (copier) to[out=30,in=-90] (X);
\draw (copier) to[out=165,in=-90] (g);
\draw (g) to (Y);
\end{tikzpicture}}}
\qquad
\middle/
\qquad
\vcenter{\hbox{
\begin{tikzpicture}[font=\small]
\node[state] (p) at (0,0) {$p$};
\node[copier] (copier) at (0,0.3) {};
\node[arrow box] (f) at (0.5,0.95) {$f$};
\coordinate (X) at (-0.5,1.5);
\coordinate (Y) at (0.5,1.5);
\draw (p) to (copier);
\draw (copier) to[out=150,in=-90] (X);
\draw (copier) to[out=15,in=-90] (f);
\draw (f) to (Y);
\end{tikzpicture}}}
\quad=\quad
\vcenter{\hbox{
\begin{tikzpicture}[font=\small]
\node[state] (p) at (0,0) {$p$};
\node[copier] (copier) at (0,0.3) {};
\node[arrow box] (g) at (0.5,0.95) {$g$};
\coordinate (X) at (-0.5,1.5);
\coordinate (Y) at (0.5,1.5);
\draw (p) to (copier);
\draw (copier) to[out=150,in=-90] (X);
\draw (copier) to[out=15,in=-90] (g);
\draw (g) to (Y);
\end{tikzpicture}}}
\right.
\quad.
\]
In fact, we will often assume that $p$ is in an S-positive subcategory, although many of our \emph{diagrammatic} results will not rely on such an assumption. 
Hence, the reader may wish to assume $p$ is a state (in the usual sense) for physical interpretations whenever convenient. In any case, we will make it clear when these assumptions are made. 
\end{rmk}

The fact that we have two notions of a.e.\ equivalence (left and right) is solely due to the fact that the morphisms need not be $*$-preserving. 
For \emph{$*$-preserving} morphisms, the two notions of a.e.\ equivalence are themselves equivalent.  

%\bn
%\label{prop:hfpkgs}
\begin{lem}[Symmetry of a.e.\ equivalence for $*$-preserving morphisms]{prop:hfpkgs}
Let 
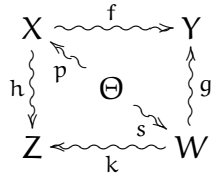
\[
\xy0;/r.25pc/:
(0,0)*+{\Theta}="0";
(-10,7.5)*+{X}="X";
(10,7.5)*+{Y}="Y";
(10,-7.5)*+{W}="W";
(-10,-7.5)*+{Z}="Z";
{\ar@{~>}"X";"Y"^{f}};
{\ar@{~>}"X";"Z"_{h}};
{\ar@{~>}"W";"Z"^{k}};
{\ar@{~>}"W";"Y"_{g}};
{\ar@{~>}"0";"X"^{p}};
{\ar@{~>}"0";"W"_{s}};
\endxy
\]
be a (not necessarily commuting) diagram of $*$-preserving morphisms in a quantum Markov category. Then 
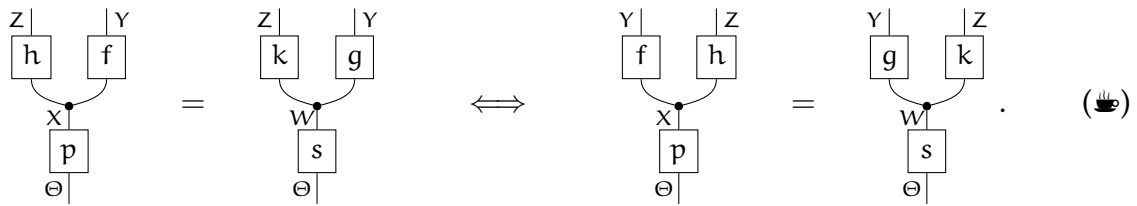
\be
\label{eq:hfpkhs}
\vcenter{\hbox{%
\begin{tikzpicture}[font=\small]
\node[copier] (copier) at (0,0.3) {};
\node at (-0.20,0.15) {\scriptsize$X$};
\node[arrow box] (h) at (-0.5,0.95) {$h$};
\node at (-0.70,1.45) {\scriptsize$Z$};
\node[arrow box] (f) at (0.5,0.95) {$f$};
\node at (0.70,1.45) {\scriptsize$Y$};
\node[arrow box] (p) at (0,-0.3) {$p$};
\node at (-0.20,-0.80) {\scriptsize$\Theta$};
\coordinate (X) at (-0.5,1.6);
\coordinate (Y) at (0.5,1.6);
\draw (0,-1.0) to (p);
\draw (p) to (copier);
\draw (copier) to[out=165,in=-90] (h);
\draw (h) to (X);
\draw (copier) to[out=15,in=-90] (f);
\draw (f) to (Y);
\end{tikzpicture}}}
\quad
=
\quad
\vcenter{\hbox{%
\begin{tikzpicture}[font=\small]
\node[copier] (copier) at (0,0.3) {};
\node at (-0.20,0.15) {\scriptsize$W$};
\node[arrow box] (h) at (-0.5,0.95) {$k$};
\node at (-0.70,1.45) {\scriptsize$Z$};
\node[arrow box] (f) at (0.5,0.95) {$g$};
\node at (0.70,1.45) {\scriptsize$Y$};
\node[arrow box] (p) at (0,-0.3) {$s$};
\node at (-0.20,-0.80) {\scriptsize$\Theta$};
\coordinate (X) at (-0.5,1.6);
\coordinate (Y) at (0.5,1.6);
\draw (0,-1.0) to (p);
\draw (p) to (copier);
\draw (copier) to[out=165,in=-90] (h);
\draw (h) to (X);
\draw (copier) to[out=15,in=-90] (f);
\draw (f) to (Y);
\end{tikzpicture}}}
\qquad
\iff
\qquad
\vcenter{\hbox{%
\begin{tikzpicture}[font=\small]
\node[copier] (copier) at (0,0.3) {};
\node at (-0.20,0.15) {\scriptsize$X$};
\node[arrow box] (h) at (-0.5,0.95) {$f$};
\node at (-0.70,1.45) {\scriptsize$Y$};
\node[arrow box] (f) at (0.5,0.95) {$h$};
\node at (0.70,1.45) {\scriptsize$Z$};
\node[arrow box] (p) at (0,-0.3) {$p$};
\node at (-0.20,-0.80) {\scriptsize$\Theta$};
\coordinate (X) at (-0.5,1.6);
\coordinate (Y) at (0.5,1.6);
\draw (0,-1.0) to (p);
\draw (p) to (copier);
\draw (copier) to[out=165,in=-90] (h);
\draw (h) to (X);
\draw (copier) to[out=15,in=-90] (f);
\draw (f) to (Y);
\end{tikzpicture}}}
\quad
=
\quad
\vcenter{\hbox{%
\begin{tikzpicture}[font=\small]
\node[copier] (copier) at (0,0.3) {};
\node at (-0.20,0.15) {\scriptsize$W$};
\node[arrow box] (h) at (-0.5,0.95) {$g$};
\node at (-0.70,1.45) {\scriptsize$Y$};
\node[arrow box] (f) at (0.5,0.95) {$k$};
\node at (0.70,1.45) {\scriptsize$Z$};
\node[arrow box] (p) at (0,-0.3) {$s$};
\node at (-0.20,-0.80) {\scriptsize$\Theta$};
\coordinate (X) at (-0.5,1.6);
\coordinate (Y) at (0.5,1.6);
\draw (0,-1.0) to (p);
\draw (p) to (copier);
\draw (copier) to[out=165,in=-90] (h);
\draw (h) to (X);
\draw (copier) to[out=15,in=-90] (f);
\draw (f) to (Y);
\end{tikzpicture}}}
.
\tag{\Coffeecup}
\ee
\end{lem}

\bprf
Assume the left-hand-side of~(\ref{eq:hfpkhs}) holds. Then 
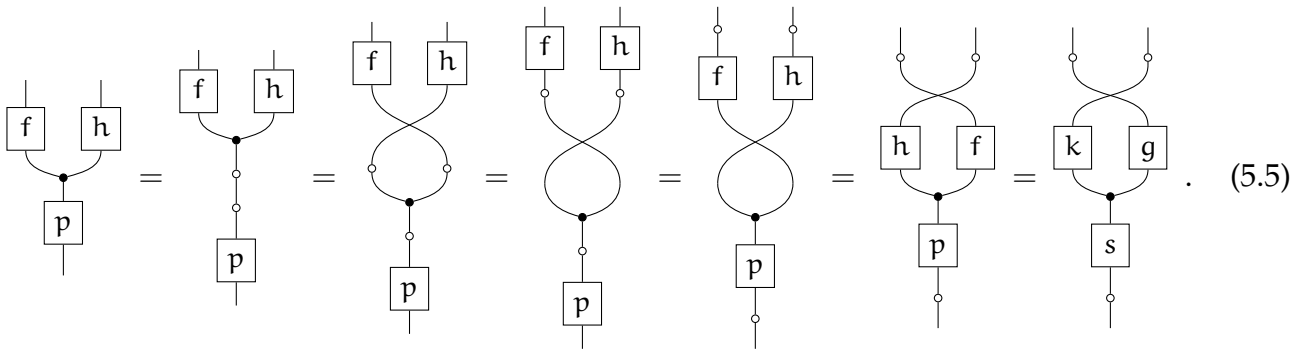
\be
\vcenter{\hbox{%
\begin{tikzpicture}[font=\small]
\node[copier] (copier) at (0,0.3) {};
\node[arrow box] (h) at (-0.5,0.95) {$f$};
\node[arrow box] (f) at (0.5,0.95) {$h$};
\node[arrow box] (p) at (0,-0.3) {$p$};
\coordinate (X) at (-0.5,1.6);
\coordinate (Y) at (0.5,1.6);
\draw (0,-1.0) to (p);
\draw (p) to (copier);
\draw (copier) to[out=165,in=-90] (h);
\draw (h) to (X);
\draw (copier) to[out=15,in=-90] (f);
\draw (f) to (Y);
\end{tikzpicture}}}
%end of first one
\;=
\vcenter{\hbox{
\begin{tikzpicture}[font=\small]
\node[arrow box] (p) at (0,-1.3) {$p$};
\node[star] (s1) at (0,-0.6) {};
\node[star] (s2) at (0,-0.15) {};
\node[copier] (copier) at (0,0.3) {};
\node[arrow box] (f) at (-0.5,0.95) {$f$};
\node[arrow box] (h) at (0.5,0.95) {$h$};
\coordinate (X) at (0.5,1.5);
\coordinate (Y) at (-0.5,1.5);
\draw (0,-1.9) to (p);
\draw (p) to (s1);
\draw (s1) to (s2);
\draw (s2) to (copier);
\draw (copier) to[out=15,in=-90] (h);
\draw (h) to (X);
\draw (copier) to[out=165,in=-90] (f);
\draw (f) to (Y);
\end{tikzpicture}}}
%end of second one
\;=
\vcenter{\hbox{
\begin{tikzpicture}[font=\small]
\node[arrow box] (p) at (0,-1.3) {$p$};
\node[star] (s1) at (0,-0.6) {};
\node[copier] (copier) at (0,-0.15) {};
\node[star] (R) at (0.5,0.3) {};
\node[star] (L) at (-0.5,0.3) {};
\coordinate (Ls) at (-0.5,1.4) {};
\coordinate (Rs) at (0.5,1.4) {};
\node[arrow box] (f) at (-0.5,1.7) {$f$};
\node[arrow box] (h) at (0.5,1.7) {$h$};
\coordinate (X) at (0.5,2.25);
\coordinate (Y) at (-0.5,2.25);
\draw (0,-1.9) to (p);
\draw (p) to (s1);
\draw (s1) to (copier);
\draw (copier) to [out=15,in=-90] (R);
\draw (R) to [out=90,in=-90] (Ls);
\draw (L) to [out=90,in=-90] (Rs);
\draw (Ls) to (f);
\draw (Rs) to (h);
\draw (h) to (X);
\draw (copier) to[out=165,in=-90] (L);
\draw (f) to (Y);
\end{tikzpicture}}}
%end of third one
\;=
\vcenter{\hbox{
\begin{tikzpicture}[font=\small]
\node[arrow box] (p) at (0,-1.3) {$p$};
\node[star] (s1) at (0,-0.6) {};
\node[copier] (copier) at (0,-0.15) {};
\coordinate (R) at (0.5,0.3) {};
\coordinate (L) at (-0.5,0.3) {};
\coordinate (Ls) at (-0.5,1.4) {};
\coordinate (Rs) at (0.5,1.4) {};
\node[star] (s2) at (-0.5,1.5) {};
\node[star] (s3) at (0.5,1.5) {};
\node[arrow box] (f) at (-0.5,2.1) {$f$};
\node[arrow box] (h) at (0.5,2.1) {$h$};
\coordinate (X) at (0.5,2.65);
\coordinate (Y) at (-0.5,2.65);
\draw (0,-1.9) to (p);
\draw (p) to (s1);
\draw (s1) to (copier);
\draw (copier) to [out=15,in=-90] (R);
\draw (R) to [out=90,in=-90] (Ls);
\draw (L) to [out=90,in=-90] (Rs);
\draw (Ls) to (s2);
\draw (s2) to (f);
\draw (Rs) to (s3);
\draw (s3) to (h);
\draw (h) to (X);
\draw (copier) to[out=165,in=-90] (L);
\draw (f) to (Y);
\end{tikzpicture}}}
%end of fourth one
\;=
\vcenter{\hbox{
\begin{tikzpicture}[font=\small]
\node[arrow box] (p) at (0,-0.8) {$p$};
\node[star] (s1) at (0,-1.5) {};
\node[copier] (copier) at (0,-0.15) {};
\coordinate (R) at (0.5,0.3) {};
\coordinate (L) at (-0.5,0.3) {};
\coordinate (Ls) at (-0.5,1.4) {};
\coordinate (Rs) at (0.5,1.4) {};
\node[star] (s2) at (-0.5,2.35) {};
\node[star] (s3) at (0.5,2.35) {};
\node[arrow box] (f) at (-0.5,1.7) {$f$};
\node[arrow box] (h) at (0.5,1.7) {$h$};
\coordinate (X) at (0.5,2.65);
\coordinate (Y) at (-0.5,2.65);
\draw (0,-1.9) to (s1);
\draw (s1) to (p);
\draw (p) to (copier);
\draw (copier) to [out=15,in=-90] (R);
\draw (R) to [out=90,in=-90] (Ls);
\draw (L) to [out=90,in=-90] (Rs);
\draw (Ls) to (f);
\draw (f) to (s2);
\draw (Rs) to (h);
\draw (h) to (s3);
\draw (s3) to (X);
\draw (copier) to[out=165,in=-90] (L);
\draw (s2) to (Y);
\end{tikzpicture}}}
%end of fifth one
\;=\;
\vcenter{\hbox{
\begin{tikzpicture}[font=\small]
\node[arrow box] (p) at (0,-0.8) {$p$};
\node[star] (s1) at (0,-1.5) {};
\node[copier] (copier) at (0,-0.15) {};
\coordinate (R) at (0.5,0.3) {};
\coordinate (Ls) at (-0.5,1.6) {};
\coordinate (Rs) at (0.5,1.6) {};
\node[star] (s2) at (-0.5,1.7) {};
\node[star] (s3) at (0.5,1.7) {};
\node[arrow box] (f) at (0.5,0.5) {$f$};
\node[arrow box] (h) at (-0.5,0.5) {$h$};
\coordinate (X) at (0.5,2.1);
\coordinate (Y) at (-0.5,2.1);
\draw (0,-1.9) to (s1);
\draw (s1) to (p);
\draw (p) to (copier);
\draw (copier) to [out=15,in=-90] (f);
\draw (f) to [out=90,in=-90] (Ls);
\draw (h) to [out=90,in=-90] (Rs);
\draw (Ls) to (s2);
\draw (Rs) to (s3);
\draw (s3) to (X);
\draw (copier) to[out=165,in=-90] (h);
\draw (s2) to (Y);
\end{tikzpicture}}}
%end of sixth one
\;=
\vcenter{\hbox{
\begin{tikzpicture}[font=\small]
\node[arrow box] (p) at (0,-0.8) {$s$};
\node[star] (s1) at (0,-1.5) {};
\node[copier] (copier) at (0,-0.15) {};
\coordinate (R) at (0.5,0.3) {};
\coordinate (Ls) at (-0.5,1.6) {};
\coordinate (Rs) at (0.5,1.6) {};
\node[star] (s2) at (-0.5,1.7) {};
\node[star] (s3) at (0.5,1.7) {};
\node[arrow box] (g) at (0.5,0.5) {$g$};
\node[arrow box] (h) at (-0.5,0.5) {$k$};
\coordinate (X) at (0.5,2.1);
\coordinate (Y) at (-0.5,2.1);
\draw (0,-1.9) to (s1);
\draw (s1) to (p);
\draw (p) to (copier);
\draw (copier) to [out=15,in=-90] (g);
\draw (g) to [out=90,in=-90] (Ls);
\draw (h) to [out=90,in=-90] (Rs);
\draw (Ls) to (s2);
\draw (Rs) to (s3);
\draw (s3) to (X);
\draw (copier) to[out=165,in=-90] (h);
\draw (s2) to (Y);
\end{tikzpicture}}}
\;\;.
\ee 
One then rewinds the steps with $f$ replaced by $g$, $h$ replaced by $k$,  and $p$ replaced by $s$. A completely analogous argument holds if the right-hand-side of~(\ref{eq:hfpkhs}) is assumed. 
\eprf

%\bc
%\label{cor:leftequalsrightaeequivalence}
\begin{cor}[When left and right a.e.\ equivalence coincide]{cor:leftequalsrightaeequivalence}
In a quantum Markov category (using the same notation as in Definition~\ref{defn:aeequivalence}),
$f$ is right $p$-a.e.\ equivalent to $g$ if and only if $f$ is left $p$-a.e.\ equivalent to $g$ provided $f,p,$ and $g$ are $*$-preserving.  
In particular, if $f$ is left (or right) $p$-a.e.\ equivalent to $g$, then $f$ is $p$-a.e.\ equivalent to $g$. 
%\ec
\end{cor}

\bprf 
This follows from Lemma~\ref{prop:hfpkgs} in the special case described by the diagram 
\be
\xy0;/r.25pc/:
(0,0)*+{\Theta}="0";
(-10,7.5)*+{X}="X";
(10,7.5)*+{Y}="Y";
(10,-7.5)*+{X}="W";
(-10,-7.5)*+{X}="Z";
{\ar@{~>}"X";"Y"^{f}};
{\ar"X";"Z"_{\id_{X}}};
{\ar"W";"Z"^{\id_{X}}};
{\ar@{~>}"W";"Y"_{g}};
{\ar@{~>}"0";"X"^{p}};
{\ar@{~>}"0";"W"_{p}};
\endxy
.
\ee
\eprf

If the morphisms in Definition~\ref{defn:aeequivalence} are \emph{not} $*$-preserving, which can happen in the quantum setting, there are instances where left and right a.e.\ equivalence are inequivalent (besides the following remark, see also Example~\ref{rmk:aedoesnotimplydoubleae} and Remark~\ref{rmk:nonstarpresbayesnotae}).

%\br
%\label{rmk:hfpkgs}
\begin{rmk}[$*$-preservation is crucial for left-right a.e.\ symmetry]{rmk:hfpkgs}
The $*$-preserving condition in Lemma~\ref{prop:hfpkgs}, and hence Corollary~\ref{cor:leftequalsrightaeequivalence}, is crucial. Here is a counterexample in the category of finite-dimensional $C^*$-algebras and unital linear maps. 
Let $\mA$ and $\mB$ both be $\mathcal{M}_{2}(\C)$ and let $\w=\tr(\r\;\cdot\;)$, where $\rho=\left[\begin{smallmatrix}1&0\\0&0\end{smallmatrix}\right]$. Let $F=\id_{\mathcal{M}_{2}(\C)}$ and set 
\[
\begin{split}
\mathcal{M}_{2}(\C)&\xstoch{F'}\mathcal{M}_{2}(\C)\\
\begin{bmatrix}b_{11}&b_{12}\\b_{21}&b_{22}\end{bmatrix}&\xmapsto{\qquad}
\begin{bmatrix}b_{11}&b_{12}\\0&b_{22}\end{bmatrix}.
\end{split}
\] 
Then, one can easily check that $F$ is unital, $\w=\w\circ F=\w\circ F'$, and
\[
\w\big(F(B)A\big)=\w\big(F'(B)A\big)\quad\forall\;A,B,
\quad\text{while}\quad
\w\big(AF(B)\big)\ne\w\big(AF'(B)\big)\quad\forall\;A,B
.
\]
This is because $F'$ is not $*$-preserving. However, if $\rho$ happens to commute with the images of $F$ and $F'$, then the two notions agree and all the expressions above are equal (due to the cyclicity of trace). Therefore, one can view the difference of these two notions of a.e.\ equivalence as being related to the non-commutativity present in the quantum setting. 
\end{rmk}
%\er

One of the key ingredients to the Gelfand--Naimark--Segal construction is the concept of the nullspace of a state on a $C^*$-algebra. A closely related idea is the right nullspace of a PU map, which was studied by Choi, among others, under the name ``left kernel''~\cite{Ch74}. 

%\bd
%\label{defn:nullspace}
\begin{defn}[The left and right nullspace of a PU map]{defn:nullspace}
Let $\mA\xstoch{\omega}\mC$ be a PU 
map of $C^*$-algebras, the \define{left/right nullspace} of $\omega$ is the subset 
of $\mA$ given by%
%footnote
\footnote{When $\omega$ is SPU, the nullspaces ${}_{\w}\mathcal{N}$ and $\mathcal{N}_{\omega}$ are the largest right and left ideals, respectively (note the order), contained in the kernel of $\omega$~\cite[Remark~3.4]{Ch74}.}
%end footnote
\[
{}_{\w}\mathcal{N}:=\big\{A\in\mA\;:\;\omega(AA^*)=0\big\}
\quad/\quad
\mathcal{N}_{\omega}:=\big\{A\in\mA\;:\;\omega(A^*A)=0\big\}
.
\]
The PU map $\omega$ is said to be \define{faithful} iff $\mathcal{N}_{\omega}=0$.%
%footnote
\footnote{This is equivalent to ${}_{\omega}\mathcal{N}=0$ since $*$ is an involution. Indeed, suppose $\omega(AA^*)=0$. Then $\omega(\uline{A}^*\uline{A})=0$, where $\uline{A}:=A^*$. Hence, $\uline{A}=0$, which implies $A=0$. This shows $\mathcal{N}_{\omega}=0$ implies ${}_{\omega}\mathcal{N}=0$. A similar calculation shows the converse.}
%end footnote
%\ed
\end{defn}

We find the nullspace of a PU map between $C^*$-algebras to be a convenient mathematical object that serves to facilitate certain computations and relate two (a-prior different) notions of a.e.\ equivalence in the category $\fdCAlgY$. States on finite-dimensional algebras (or von~Neumann algebras) admit a particularly convenient form for the nullspace. 

%\blem
%\label{lem:support}
\begin{lem}[The support of a state]{lem:support}
Let $\mA$ be a finite-dimensional $C^*$-algebra (or more generally a von~Neumann algebra) and let $\mA\xstoch{\w}\C$ be a (PU) state. Then there exists a unique smallest projection $P_{\w}$ satisfying
%\be
%\label{eq:supportstateidentities}
\[
\w(A)=\w(P_{\w}A)=\w(AP_{\w})=\w(P_{\w}AP_{\w})\qquad\forall\;A\in\mA. 
\]
%\ee
The projection $P_{\w}$ is called the \define{support} of $\omega$.
In particular, if $P_{\w}^{\perp}:=1_{\mA}-P_{\w}$, then $\w(P_{\w}^{\perp}A)=0$ and $\w(AP_{\w}^{\perp})=0$ for all $A\in\mA$. Furthermore, 
\[
{}_{\w}\mathcal{N}=P_{\omega}^{\perp}\mA
\quad\text{ and }\quad
\mathcal{N}_{\omega}=\mA P_{\omega}^{\perp}.
\]
%\elem
\end{lem}

\bprf
See Section~1.14 of Sakai~\cite{Sa71}. 
\eprf

\begin{rmk}[Techniques on positivity]{rmk:supportconditions}
Any one of the three equalities in Lemma~\ref{lem:support} implies the other two. The proof illustrates some simple, but useful, techniques, which we will use very often.%
%footnote
\footnote{It is in fact these kinds of manipulations that motivated us to introduce the involution into the string-diagrammatic language.}
%end footnote
 For example, suppose $\w(P_{\w}AP_{\w})=\w(AP_{\w})$ for all $A\in\mA$. Then 
\[
\w(P_{\w}A)=\w\big((A^*P_{\w})^*\big)=\overline{\w(A^*P_{\w})}=\overline{\w(P_{\w}A^*P_{\w})}=\w(P_{\w}AP_{\w})
\]
since $\w(A^*)=\overline{\w(A)}$ for all $A\in\mA$ (this is a consequence of positivity of $\w$). This proves one of the two remaining equalities. 
For the last equality, $\omega(A)=\omega(AP_{\omega})$, first note that if $A\ge0$, then 
\[
\begin{split}
0&\le\w\big(P_{\w}^{\perp}AP_{\w}^{\perp}\big)\quad\text{ since $\w$ is a positive functional and $P_{\w}^{\perp}AP_{\w}^{\perp}\ge0$}\\
&\le\lVert A\rVert\w\big(P_{\w}^{\perp}\big)\quad\text{ since $A\le\lVert A\rVert1_{\mA}$ and $\w$ is linear ($\lVert\;\cdot\;\rVert$ is the norm on $\mA$)}\\
&=0\quad\text{ by the previous equality}.
\end{split}
\]
This proves $\w\big(P_{\w}^{\perp}AP_{\w}^{\perp}\big)=0$ for $A\ge0$. This equality also holds for arbitrary $A$ since every $A$ can always be expressed as a linear combination of at most four positive elements.%
%footnote
\footnote{Break $A$ up into a self-adjoint and skew-adjoint element, and then spectrally decompose the two results and split the negative and positive terms apart.}
%end footnote
Therefore, 
\[
\w(A)=\w\big(AP_{\w}+P_{\w}AP_{\w}^{\perp}+P_{\w}^{\perp}AP_{\w}^{\perp}\big)=\w(AP_{\w})+\w\big(P_{\w}AP_{\w}^{\perp}\big)+\w\big(P_{\w}^{\perp}AP_{\w}^{\perp}\big)=\w(AP_{\w})
\]
since the second term equals $\w\big(P_{\w}AP_{\w}^{\perp}\big)=\w\big(P_{\w}AP_{\w}^{\perp}P_{\w}\big)=0$ by our assumed identity. 
\end{rmk}

The following theorem is one of the main results in this paper. It offers a useful technique translating between the string-diagrammatic definition of a.e.\ equivalence and the operator-algebraic definition from~\cite{PaRu19}. This result will be utilized in proving many of the main theorems referring to $C^*$-algebras from now on. 

%\bt
%\label{thm:ncaeequivalence}
\begin{theo}[The many equivalent definitions of a.e.\ equivalence in $\fdCAlgY$]{thm:ncaeequivalence}
Let $\mA$, $\mB$, and $\mC$ be $C^*$-algebras, let $\mA\xstoch{\w}\mC$ be an SPU map, and let $F,G:\mB\stoch\mA$ be linear maps.
Consider the following four conditions. 
\begin{enumerate}[(a)]
\itemsep0pt
\item
\label{item:leftae}
$F$ is left $\w$-a.e.\ equivalent to $G$ in the sense of Definition~\ref{defn:aeequivalence}.
\item
\label{item:rightae}
$F$ is right $\w$-a.e.\ equivalent to $G$ in the sense of Definition~\ref{defn:aeequivalence}.
\item
\label{item:nullspaceae}
$F(B)-G(B)$ is in the right nullspace $\mathcal{N}_{\w}$ of $\w$ for all $B\in\mB$. 
\item
\label{item:aeP}
If $\mC=\C$ (so that $\omega$ is a state), then 
$F(B)P_{\w}=G(B)P_{\w}$ for all $B\in\mB$. 
\end{enumerate}
Then the following facts hold.
\begin{enumerate}[i.]
\itemsep0pt
\item
Conditions (\ref{item:rightae}) and (\ref{item:nullspaceae}) are equivalent.%
%footnote
\footnote{A similar equivalence holds between (\ref{item:leftae}) and $F(B)-G(B)\in{}_{\omega}\mathcal{N}$ for all $B\in\mB$.}
%end footnote
\item
If $F$ and $G$ are $*$-preserving, then conditions (\ref{item:leftae}), (\ref{item:rightae}), and (\ref{item:nullspaceae}) are all equivalent. 
\item
If $\mC=\C$, then conditions (\ref{item:rightae}), (\ref{item:nullspaceae}), 
and (\ref{item:aeP}) are all equivalent. 
\item
If $\mC=\C$ and if $F$ and $G$ are $*$-preserving, then all conditions are equivalent. 
\end{enumerate}
%\et
\end{theo}

\pagebreak
\bprf
{\color{white}{you found me!}}

\begin{enumerate}[i.]
\itemsep0pt
\item
To see that (\ref{item:rightae}) is equivalent to (\ref{item:nullspaceae}), first suppose (\ref{item:rightae}) holds. This means
\be
\w\big(AF(B)\big)=\w\big(AG(B)\big)
\qquad\forall\;A\in\mA,\;B\in\mB.
\ee
By linearity of $\w$, this is equivalent to
\be
\w\Big(A\big(F(B)-G(B)\big)\Big)=0\qquad\forall\;A\in\mA,\;B\in\mB.
\ee
In particular, one can set $A:=\big(F(B)-G(B)\big)^*$. This immediately gives condition (\ref{item:nullspaceae}).
Now, suppose (\ref{item:nullspaceae}) holds. Then 
\be
\label{eq:provingae0}
\begin{split}
0&=\omega\Big(\big(F(B)-G(B)\big)^*\big(F(B)-G(B)\big)\Big)\quad\text{ by assumption}\\
&\ge\omega\Big(\big(F(B)-G(B)\big)\Big)^*\omega\Big(\big(F(B)-G(B)\big)\Big)\quad\text{ since $\omega$ is SPU}\\
&\ge0\quad\text{ since $C^*C\ge0$.}
\end{split}
\ee
Hence, all inequalities become equalities. Therefore, 
\be
\begin{split}
\omega\Big(A^*\big(F(B)-G(B)\big)\Big)&=\omega(A)^*\omega\Big(\big(F(B)-G(B)\big)\Big)\quad\text{by the Multiplication Lemma}\\
&=0\quad\text{by (\ref{eq:provingae0})}, 
\end{split}
\ee
where the last step applies since $C^*C=0$ implies $C=0$. 
Rearranging this gives $\omega\big(A^*F(B)\big)=\omega\big(A^*G(B)\big)$. Since $*$ is an involution, this proves (\ref{item:rightae}). 

\item
This follows from the previous steps and Corollary~\ref{cor:leftequalsrightaeequivalence} ($\w$ is $*$-preserving because it is positive), which proves (\ref{item:leftae}) is equivalent to (\ref{item:rightae}). 

\item
The equivalence between conditions (\ref{item:nullspaceae}) and (\ref{item:aeP}) follows from the identity $\mathcal{N}_{\w}=\mathcal{A}P_{\w}^{\perp}$ (cf.\ Lemma~3.42 in~\cite{PaRu19}). 

\item
This follows from all the previous statements. \qedhere
\end{enumerate}
\eprf

The definition of a.e.\ equivalence for morphisms of $C^*$-algebras in terms of the nullspace of a state (item~(\ref{item:nullspaceae}) in Theorem~\ref{thm:ncaeequivalence}) was introduced in~\cite{PaRu19}. It was motivated by the GNS construction and had little to do with diagrammatic reasoning, so it is quite satisfying that our definition coincides with the categorical (Definition~\ref{defn:aeequivalence}) due to Cho and Jacobs~\cite{ChJa18} (and extended by Fritz in~\cite[Definition~13.1]{Fr19}) when $\omega$ is SPU and when $F$ and $G$ are $*$-preserving. However, when $F$ and $G$ are merely linear, there the distinction between left and right a.e.\ equivalence becomes important. This difference will manifest itself when we distinguish certain properties that one can demand on subcategories of quantum Markov categories. These include a.e.\ determinism, a.e.\ modularity, causality, strict positivity, and Bayesian inversion, all of which will be discussed in later sections. 

\begin{rmk}[A simpler proof for a.e.\ equivalence with respect to a state]{a014}
When $\mC=\C$, there is an even simpler proof that (\ref{item:aeP}) implies (\ref{item:rightae}) in Theorem~\ref{thm:ncaeequivalence} using projections, namely
\[
\w\big(AF(B)\big)
\overset{\text{Lem~\ref{lem:support}}}{=\joinrel=\joinrel=\joinrel=\joinrel=}
\w\big(AF(B)P_{\w}\big)
=
\w\big(AG(B)P_{\w}\big)
\overset{\text{Lem~\ref{lem:support}}}{=\joinrel=\joinrel=\joinrel=\joinrel=}
\w\big(AG(B)\big)
\]
for all $A\in\mA$ and $B\in\mB$. 
One of the convenient properties of condition~(\ref{item:aeP}) in Theorem~\ref{thm:ncaeequivalence} is that it is linear and involves only a single variable, as opposed to the definition of right a.e.\ equivalence from Definition~\ref{defn:aeequivalence}, which involves two variable inputs. 
\end{rmk}

The following remark describes a feature of a.e.\ equivalence in the non-commutative setting and is meant to clarify any potential misunderstandings. 

\begin{rmk}[A.e.\ equivalence sees more than the supported corner]{rmk:supportedcornernotenough}
If $\mM_{m}(\C)\xstoch{F}\mM_{m}(\C)$ is a CPU map, $\mM_{m}(\C)\xstoch{\w}\C$ is a state, and $F\aeequals{\omega}\id_{\mM_{m}(\C)}$, i.e.\ $F(A)P_{\omega}=AP_{\omega}$, then $F=\id_{\mM_{m}(\C)}$~\cite[Theorem~3.67]{PaRu19}.%
%footnote
\footnote{This is false if $F$ is merely $*$-preserving and unital.}
%end footnote 
However, if $\mathrm{Ad}_{P_{\omega}}\circ F=\mathrm{Ad}_{P_{\omega}}$, i.e.\ $P_{\omega}F(A)P_{\omega}=P_{\omega}AP_{\omega}$, then it is not necessarily the case that $F$ equals the identity map. 
Indeed, consider the $m=2$ case and take the density matrix $\rho=e_{1}e_{1}^*$ ($e_{1}$ is the first standard unit vector of $\C^{m}$) with associated state $\w=\tr(\rho\;\cdot\;)$. Then the CPU map (which is even a $^*$-isomorphism)
\[
F\left(
\begin{bmatrix}
a&b\\c&d
\end{bmatrix}
\right):=
\mathrm{Ad}_{\left[\begin{smallmatrix}1&0\\0&-1\end{smallmatrix}\right]}\left(
\begin{bmatrix}
a&b\\c&d
\end{bmatrix}
\right)
=
\begin{bmatrix}
a&-c\\-b&d
\end{bmatrix}
\]
satisfies $\w=\w\circ F$, $\mathrm{Ad}_{P_{\w}}\circ F=\mathrm{Ad}_{P_{\w}}$, and $F\circ \mathrm{Ad}_{P_{\w}}=\mathrm{Ad}_{P_{\w}}$, but $F\ne\id_{\mathcal{M}_{n}}$. 
Thus, one should keep in mind that there is a good deal of information about $F(A)$ in $F(A)P_{\w}$, which would be lost if one worked with only $P_{\w}F(A)P_{\w}$. The notion of a.e.\ equivalence we are using keeps track of this additional information. 
\end{rmk}

%%%%%%%%%%%%%%%%%%%%%%%%%%%%%%%%%%%%%%
\section[Almost everywhere determinism]{a016}
\label{sec:aedet}
\vspace{-12mm}
\noindent
\begin{tikzpicture}
\coordinate (L) at (-8.75,0);
\coordinate (R) at (8.75,0);
\draw[line width=2pt,orange!20] (L) -- node[anchor=center,rectangle,fill=orange!20]{\strut \Large \textcolor{black}{\textbf{6\;\; Almost everywhere determinism}}} (R);
\end{tikzpicture}
%\vspace{1mm}
%%%%%%%%%%%%%%%%%%%%%%%%%%%%%%%%%%%%%%

We have two reasonable notions of being deterministic almost everywhere (besides just the distinction between left/right notions as in a.e.\ equivalence).
The notion of a morphism being a.e.\ deterministic was introduced recently by Fritz~\cite[Definition~13.10]{Fr19}. In the setting of operator algebras while studying the relationship between disintegrations and Bayesian inversion, we have also found this notion to be of great importance (though our initial discovery was obtained in terms of projections). It is again reassuring that the categorical definition agrees with ours. In addition, another natural candidate we found important to distinguish is that of being a.e.\ equivalent to a deterministic morphism. This would allow one to replace a morphism with a deterministic one in certain computations. 

%\bd
%\label{defn:aeequivalentdeterministic}
\begin{defn}[A.e.\ deterministic morphisms]{defn:aeequivalentdeterministic}
Let $\Theta\xstoch{p}X$ be any morphism and let $X\xstoch{f}Y$ be an even morphism in a quantum Markov category. The morphism $f$ is \define{left/right $p$-a.e.\ equivalent to a deterministic morphism} iff there exists a deterministic morphism $X\xstoch{g}Y$ such that 
$f$ is left/right $p$-a.e.\ equal to $g$. 
The morphism $f$ is \define{left/right $p$-a.e.\ deterministic} iff 
%\be
%\label{eq:rightaedet}
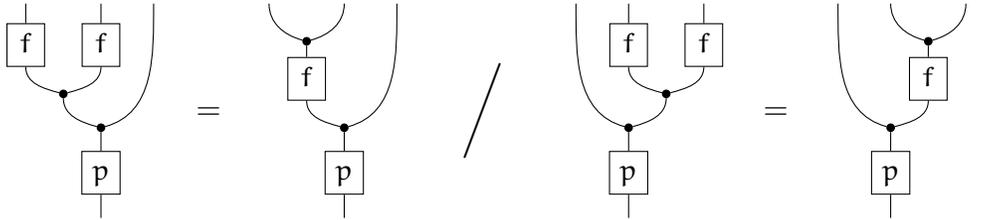
\[
\left.
\vcenter{\hbox{
\begin{tikzpicture}[font=\small]
\node[arrow box] (p) at (0,-0.3) {$p$};
\node[copier] (c) at (0,0.3) {};
\node[copier] (c2) at (-0.5,0.75) {};
\node[arrow box] (f) at (-1,1.4) {$f$};
\node[arrow box] (e) at (0,1.4) {$f$};
\coordinate (X) at (0.7,1.95);
\coordinate (Y1) at (-1,1.95);
\coordinate (Y2) at (0,1.95);
\draw (0,-0.9) to (p);
\draw (p) to (c);
\draw (c) to[out=15,in=-90] (X);
\draw (c) to[out=165,in=-90] (c2);
\draw (c2) to[out=165,in=-90] (f);
\draw (c2) to[out=15,in=-90] (e);
\draw (f) to (Y1);
\draw (e) to (Y2);
\end{tikzpicture}}}
\quad=\quad
\vcenter{\hbox{
\begin{tikzpicture}[font=\small]
\node[arrow box] (p) at (0,-0.3) {$p$};
\node[copier] (c) at (0,0.3) {};
\node[copier] (c2) at (-0.5,1.45) {};
\node[arrow box] (f) at (-0.5,0.95) {$f$};
\coordinate (X) at (0.7,1.95);
\coordinate (Y1) at (-1,1.95);
\coordinate (Y2) at (0,1.95);
\draw (0,-0.9) to (p);
\draw (p) to (c);
\draw (c) to[out=15,in=-90] (X);
\draw (c) to[out=165,in=-90] (f);
\draw (c2) to[out=165,in=-90] (-1,1.95);
\draw (c2) to[out=15,in=-90] (0,1.95);
\draw (f) to (c2);
\end{tikzpicture}}}
\qquad
\middle/
\qquad
\vcenter{\hbox{
\begin{tikzpicture}[font=\small]
\node[arrow box] (p) at (0,-0.3) {$p$};
\node[copier] (c) at (0,0.3) {};
\node[copier] (c2) at (0.5,0.75) {};
\node[arrow box] (f) at (1,1.4) {$f$};
\node[arrow box] (e) at (0,1.4) {$f$};
\coordinate (X) at (-0.7,1.95);
\coordinate (Y1) at (1,1.95);
\coordinate (Y2) at (0,1.95);
\draw (0,-0.9) to (p);
\draw (p) to (c);
\draw (c) to[out=165,in=-90] (X);
\draw (c) to[out=15,in=-90] (c2);
\draw (c2) to[out=15,in=-90] (f);
\draw (c2) to[out=165,in=-90] (e);
\draw (f) to (Y1);
\draw (e) to (Y2);
\end{tikzpicture}}}
\quad=\quad
\vcenter{\hbox{
\begin{tikzpicture}[font=\small]
\node[arrow box] (p) at (0,-0.3) {$p$};
\node[copier] (c) at (0,0.3) {};
\node[copier] (c2) at (0.5,1.45) {};
\node[arrow box] (f) at (0.5,0.95) {$f$};
\coordinate (X) at (-0.7,1.95);
\coordinate (Y1) at (1,1.95);
\coordinate (Y2) at (0,1.95);
\draw (0,-0.9) to (p);
\draw (p) to (c);
\draw (c) to[out=165,in=-90] (X);
\draw (c) to[out=15,in=-90] (f);
\draw (c2) to[out=15,in=-90] (Y1);
\draw (c2) to[out=165,in=-90] (0,1.95);
\draw (f) to (c2);
\end{tikzpicture}}}
\right.
\quad.
%\ee
\]
%\ed
\end{defn}

A-priori these notions are all different, and only in certain subcategories of certain quantum Markov categories do they agree. For example, left and right notions are equivalent in classical Markov categories and $*$-preserving subcategories of quantum Markov categories (cf.\ Lemma~\ref{lem:symmetryaedeterminism}). Rather than proving this now, we first provide some examples and non-examples.

%\bx
%\label{ex:paedetforfinstoch}
\begin{exa}[A.e.\ determinism in $\FinStoch$]{ex:paedetforfinstoch}
Using the same notation as in Definition~\ref{defn:aeequivalentdeterministic} but in the category $\FinStoch$, 
a stochastic map $f$ is $p$-a.e.\ deterministic if and only if 
\[
f_{yx}f_{y'x}p_{x\theta}=\de_{y'y}f_{yx}p_{x\theta}\qquad\forall\;\theta\in\Theta,\;x\in X,\;y,y'\in Y. 
\]
Therefore, $f_{yx}f_{y'x}=\de_{yy'}f_{yx}$ for some $x\in X$ and all $y,y'\in Y$ if there exists a $\q$ such that $p_{x\theta}>0$. In this case, $f_{yx}f_{y'x}=0$ when $y\ne y'$ and $(f_{yx})^2=f_{yx}$. This means $f_{yx}\in\{0,1\}$. Since $f_{x}$ is a probability measure, this implies there exists a unique $y$ such that $f_{yx}=1$. Set
\[
N_{p}:=\big\{x\in X\;:\;p_{x\theta}=0\;\forall\;\theta\in\Theta\big\}=\bigcap_{\theta\in\Theta}N_{p_{\theta}},
\]
where $N_{p_{\theta}}\subset X$ is the usual nullspace of the probability measure $p_{\theta}$. Then this analysis says that $f$, when restricted to $X\setminus N_{p}$, corresponds to a function from $X\setminus N_{p}$ to $Y$. Hence, for $x\in X\setminus N_{p}$, set $g_{yx}:=f_{yx}$ for all $y\in Y$. No information about the form of $f_{x}$ is obtained for $x\in N_{p}$ from the condition of a.e.\ determinism. Nevertheless, if $x\in N_{p}$, then set $g_{x}$ to be \emph{any} (unit) point measure on $Y$. Then $g$ is deterministic and $f\underset{\raisebox{.6ex}[0pt][0pt]{\scriptsize$p$}}{=}g$. This shows that in $\FinStoch$, every morphism $f$ that is $p$-a.e.\ deterministic is $p$-a.e.\ equivalent to some deterministic map. The converse turns out to be true as well, and this will be proved more generally in an arbitrary classical Markov category in Proposition~\ref{prop:aetodetisaedet} (see also~\cite[Lemma~13.12]{Fr19}).
%\ex
\end{exa}

\begin{exa}[A.e.\ determinism in $\Stoch$]{ex:paedetforstoch}
Generalizing the previous example, in $\Stoch$, the category of measurable spaces and Markov kernels, the condition of a.e.\ determinism reads 
\[
\int_{A}f_{x}(B)f_{x}(C)\,dp_{\theta}(x)=\int_{A}f_{x}(B\cap C)\,dp_{\theta}(x)
\]
for all measurable subsets $A\subseteq X$ and $B,C\subseteq Y$ and for all $\theta\in\Theta$. In the special case where $B=C$, this equation becomes 
\[
\int_{A}\Big(f_{x}(B)-f_{x}(B)^2\Big)\,dp_{\theta}(x)=0
\]
for all measurable subsets $A\subseteq X$ and $B\subseteq Y$ and for all $\theta\in\Theta$. The integrand is non-negative because $f_{x}$ is a probability measure. Hence, taking $A=X$, there exists a $p_{\theta}$-nullset $N_{\theta,B}\subset X$ such that $f_{x}(B)^2=f_{x}(B)$ for all $x\in X\setminus N_{\theta,B}$. Thus, $f_{x}(B)\in\{0,1\}$ for all $x\in X\setminus N_{\theta,B}$. In other words, given any $\theta\in\Theta$ and any measurable $B\subseteq Y$, the stochastic map $f$ restricts to a $\{0,1\}$-valued measure on a set of full $p_{\theta}$ measure. This is not necessarily the same as a function on arbitrary measurable spaces, though it is under additional, often considered reasonable, assumptions on the measurable spaces (see \cite[Example~10.4]{Fr19} for further details). 
\end{exa}

\begin{exa}[A.e.\ deterministic maps in $\fdCAlgU$]{exa:aedetinCAlg}
In the quantum Markov category $\fdCAlgUY$, given a linear unital map
$\mB\xstoch{F}\mA$ and an SPU map $\mA\xstoch{\omega}\mC$, then $F$ is right $\w$-a.e.\ deterministic if and only if 
\[
F(B_{1}B_{2})-F(B_{1})F(B_{2})\in\mN_{\omega}\qquad\forall\;B_{1},B_{2}\in\mB
\]
by part~(\ref{item:nullspaceae}) of Theorem~\ref{thm:ncaeequivalence}. 
As a special case, if $\omega$ is a (PU) state $\mA\xstoch{\omega}\C$, then $F$ is right $\w$-a.e.\ deterministic if and only if 
\[
F(B_{1}B_{2})P_{\w}=F(B_{1})F(B_{2})P_{\w}\qquad\forall\;B_{1},B_{2}\in\mB.
\] 
by part~(\ref{item:aeP}) of Theorem~\ref{thm:ncaeequivalence}. 
Note that a.e.\ determinism does \emph{not} say that the composite 
\[
\mB\xstoch{F}\mA\xstoch{S_{\omega}}P_{\omega}\mA P_{\omega},
\]
where $S_{\omega}(A):=P_{\omega}AP_{\omega}$, is deterministic. Indeed, if $\mB:=\mM_{n}(\C),$ $\mA:=\mM_{m}(\C)$, and $\mathrm{rank}(P_{\omega})<n$, then there are no $*$-homomorphisms (not even non-unital ones!) $\mB\to P_{\omega}\mA P_{\omega}$. 
We will find a simpler conditions describing a.e.\ determinism involving only a single variable (and without the restriction to states) in Example~\ref{exa:aedeterminismforSPU}. 
\end{exa}

\begin{exa}[A.e.\ deterministic does not imply a.e.\ equivalent to a deterministic map]{exa:aedniaeequivdet}
In Example~\ref{ex:paedetforfinstoch}, we showed that every $p$-a.e.\ deterministic morphism is $p$-a.e.\ equivalent to a deterministic morphism in $\FinStoch$. This turns out to be false in the category $\fdCAlgCPU$ of finite-dimensional $C^*$-algebras and CPU maps. Namely, if a CPU map is a.e.\ deterministic, it is not necessary equivalent to a deterministic map. 
To see this, take 
\[
\begin{split}
\mB:=\mathcal{M}_{n}(\C)&\xstoch{F}\mathcal{M}_{m}(\C)=:\mA\\
A&\xmapsto{\quad}\begin{bmatrix}A&0\\0&\frac{1}{n}\tr(A)\mathds{1}_{m-n}\end{bmatrix}
\end{split}
\]
supposing $n<m$ and $m$ is \emph{not} an integral multiple of $n$. Let $\sigma\in\mM_{n}(\C)$ be a density matrix and let 
\[
\w:=\tr(\rho\;\cdot\;),\quad\text{ where }\quad\rho:=\begin{bmatrix}\sigma&0\\0&0\end{bmatrix}\in\mM_{m}(\C).
\]
Then $F$ is $\omega$-a.e.\ deterministic, but it is not $\omega$-a.e.\ equivalent to a deterministic morphism because a $*$-homomorphism from $\mM_{n}(\C)$ to $\mM_{m}(\C)$ does not exist unless $m$ is an integral multiple of $n$. 
\end{exa}

To describe a.e.\ determinism more explicitly and conveniently in $\fdCAlgUY$, we state two (independent) lemmas that will be used several times. 

%\blem
%\label{lem:minilemmaleftideal}
\begin{lem}[Positive elements in the nullspace]{lem:minilemmaleftideal}
Let $\mA$ and $\mC$ be $C^*$-algebras, let $\mA\xstoch{\w}\mC$ be SPU, let $\mathcal{N}_{\w}\subseteq\mA$ be its associated right nullspace, and fix $A\ge0$. Then $A\in\ker\w$ if and only if $A\in\mathcal{N}_{\w}$. An analogous condition holds for the left nullspace. 
%\elem
\end{lem}

\bprf
%[Proof of Lemma~\ref{lem:minilemmaleftideal}]
One direction follows from $\mathcal{N}_{\omega}\subseteq\ker\omega$ because $0\le\omega(A)^*\omega(A)\le\omega(A^*A)\le0$. For the other direction, write $A$ as $A=D^*D$. Then $D\in\mathcal{N}_{\w}$ since $\omega(A)=0$. Since $\mathcal{N}_{\w}$ is a left ideal in $\mA$ (see footnote in Definition~\ref{defn:nullspace}), $D^*D\in\mathcal{N}_{\w}$. 
\eprf

%\blem
%\label{lem:Attalslemmaaegeneral}
\begin{lem}[The weak a.e.\ Multiplication Lemma]{lem:Attalslemmaaegeneral}
Let $\mA$, $\mB$, and $\mC$ be $C^*$-algebras, 
let $\mA\xstoch{\w}\mC$ be an SPU map, and let $\mB\xstoch{F}\mA$ be a linear (not necessarily $*$-preserving) map.
Then 
%\be
%\label{eq:FBsBNomega}
\[
F(B^*B)-F(B)^*F(B)\in\mathcal{N}_{\omega}\qquad\forall\;B\in\mB
\]
%\ee
if and only if $F$ is right $\w$-a.e.\ deterministic, i.e.\
\[
F(B^*C)-F(B)^*F(C)\in\mathcal{N}_{\omega}\qquad\forall\;B,C\in\mB. 
\]
Diagrammatically, 
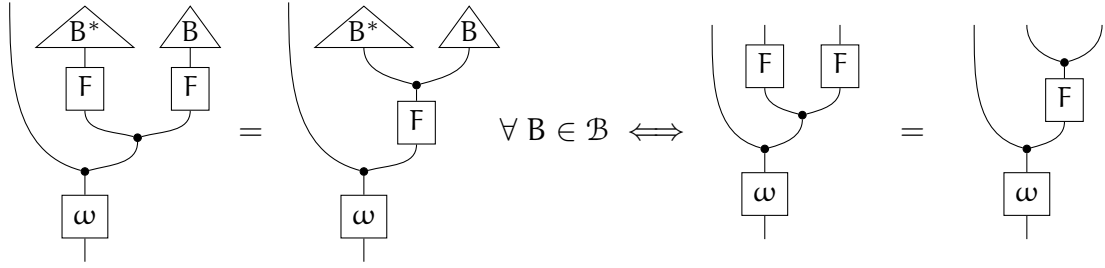
\[
\vcenter{\hbox{
\begin{tikzpicture}[font=\small]
\node[arrow box] (p) at (0,-0.3) {$\omega$};
\node[copier] (c) at (0,0.3) {};
\node[copier] (c2) at (0.7,0.75) {};
\node[arrow box] (f) at (1.4,1.4) {$F$};
\node[arrow box] (e) at (0,1.4) {$F$};
\coordinate (X) at (-1.0,2.55);
\node[effect] (Y1) at (1.4,1.95) {$B$};
\node[effect] (Y2) at (0,1.95) {$B^*$};
\draw (0,-0.9) to (p);
\draw (p) to (c);
\draw (c) to[out=165,in=-90] (X);
\draw (c) to[out=15,in=-90] (c2);
\draw (c2) to[out=15,in=-90] (f);
\draw (c2) to[out=165,in=-90] (e);
\draw (f) to (Y1);
\draw (e) to (Y2);
\end{tikzpicture}}}
%%%%%%%%%%%%%%%%%%%
\;=\;
%%%%%%%%%%%%%%%%%%%
\vcenter{\hbox{
\begin{tikzpicture}[font=\small]
\node[arrow box] (p) at (0,-0.3) {$\omega$};
\node[copier] (c) at (0,0.3) {};
\node[copier] (c2) at (0.7,1.45) {};
\node[arrow box] (f) at (0.7,0.95) {$F$};
\coordinate (X) at (-1.0,2.55);
\node[effect] (Y1) at (1.4,1.95) {$B$};
\node[effect] (Y2) at (0,1.95) {$B^*$};
\draw (0,-0.9) to (p);
\draw (p) to (c);
\draw (c) to[out=165,in=-90] (X);
\draw (c) to[out=15,in=-90] (f);
\draw (c2) to[out=15,in=-90] (Y1);
\draw (c2) to[out=165,in=-90] (0,1.95);
\draw (f) to (c2);
\end{tikzpicture}}}
%%%%%%%%%%%%%%%%%%%
\forall\;B\in\mB\iff
%%%%%%%%%%%%%%%%%%%
\vcenter{\hbox{
\begin{tikzpicture}[font=\small]
\node[arrow box] (p) at (0,-0.3) {$\omega$};
\node[copier] (c) at (0,0.3) {};
\node[copier] (c2) at (0.5,0.75) {};
\node[arrow box] (f) at (1,1.4) {$F$};
\node[arrow box] (e) at (0,1.4) {$F$};
\coordinate (X) at (-0.7,1.95);
\coordinate (Y1) at (1,1.95);
\coordinate (Y2) at (0,1.95);
\draw (0,-0.9) to (p);
\draw (p) to (c);
\draw (c) to[out=165,in=-90] (X);
\draw (c) to[out=15,in=-90] (c2);
\draw (c2) to[out=15,in=-90] (f);
\draw (c2) to[out=165,in=-90] (e);
\draw (f) to (Y1);
\draw (e) to (Y2);
\end{tikzpicture}}}
\quad=\quad
\vcenter{\hbox{
\begin{tikzpicture}[font=\small]
\node[arrow box] (p) at (0,-0.3) {$\omega$};
\node[copier] (c) at (0,0.3) {};
\node[copier] (c2) at (0.5,1.45) {};
\node[arrow box] (f) at (0.5,0.95) {$F$};
\coordinate (X) at (-0.7,1.95);
\coordinate (Y1) at (1,1.95);
\coordinate (Y2) at (0,1.95);
\draw (0,-0.9) to (p);
\draw (p) to (c);
\draw (c) to[out=165,in=-90] (X);
\draw (c) to[out=15,in=-90] (f);
\draw (c2) to[out=15,in=-90] (Y1);
\draw (c2) to[out=165,in=-90] (0,1.95);
\draw (f) to (c2);
\end{tikzpicture}}}
\]
%\elem
An analogous condition holds for left a.e.\ determinism with appropriate morphisms appearing on the left instead of the right. 
\end{lem}

\bprf
%[Proof of Lemma~\ref{lem:Attalslemmaaegeneral}]
The only non-trivial direction is the forward one. 
Fix $B,C\in\mB$ and set $D:=B+C$. 
Then 
\be
\label{eq:FBCi}
\begin{split}
F(D^*D)-F(D)^*F(D)&=F(B^*B)+F(B^*C)+F(C^*B)+F(C^*C)\\
&-F(B)^*F(B)-F(B)^*F(C)-F(C)^*F(B)-F(C)^*F(C).
\end{split}
\ee
Since $F(D^*D)-F(D)^*F(D)\in\mathcal{N}_{\omega}$ by assumption, this means
\be
\label{eq:omegaFD}
\omega\Big(\big(F(D^*D)-F(D)^*F(D)\big)^*\big(F(D^*D)-F(D)^*F(D)\big)\Big)=0.
\ee
Expanding the expression inside $\omega$ in terms of $B$ and $C$ results in 64 terms, which we will not write. However, we will describe the types of terms that arise. First, 8 of the terms come from 
\be
\big(F(V^*V)-F(V)^*F(V)\big)^*\big(F(V^*V)-F(V)^*F(V)\big)
\ee
with $V\in\{B,C\}$, which vanish by assumption. Second, 32 terms have an \emph{odd} number of $B$'s and $C$'s in them. By taking $B\mapsto B$ and $C\mapsto -C$ and adding the two resulting expressions cancels out all of these terms. 
Thirdly, 8 additional terms with exactly two $C$'s and two $B$'s satisfy the condition that taking $B\mapsto B$ and $C\mapsto iC$ negates all these 8 terms (for example, $F(C^*B)^*F(B)F(C)$ is one such term). 
For these terms, taking $B\mapsto B$ and $C\mapsto iC$ and adding the resulting equations causes these terms to cancel as well. Thus, one is left with 64-8-32-8=16 terms, which are equal to 
\be
\label{eq:fourtermtypesleft}
\begin{split}
\uline{\big(F(B^*B)-F(B)^*F(B)\big)^*\big(F(C^*C)-F(C)^*F(C)\big)}&+\big(F(B^*C)-F(B)^*F(C)\big)^*\big(F(B^*C)-F(B)^*F(C)\big)\\
\big(F(C^*B)-F(C)^*F(B)\big)^*\big(F(C^*B)-F(C)^*F(B)\big)&+\uline{\big(F(C^*C)-F(C)^*F(C)\big)^*\big(F(B^*B)-F(B)^*F(B)\big)}.
\end{split}
\ee
The underlined terms vanish inside $\omega$ by the (standard) Multiplication Lemma (Lemma~\ref{lem:multiplicationtheorem}). In more detail, since
\be
\begin{split}
0&=\omega\Big(\big(F(B^*B)-F(B)^*F(B)\big)^*\big(F(B^*B)-F(B)^*F(B)\big)\Big)\quad\text{ by assumption}\\
&\ge\omega\big(F(B^*B)-F(B)^*F(B)\big)^*\omega\big(F(B^*B)-F(B)^*F(B)\big)\quad\text{ since $\omega$ is SPU}\\
&=0\quad\text{ since $F(B^*B)-F(B)^*F(B)\in\mathcal{N}_{\omega}\subseteq\ker\omega$}
\end{split}
\ee
so that all intermediate inequalities become equalities. Hence, $F(B^*B)-F(B)^*F(B)$ is in the multiplicative domain of $\omega$ and the (standard) Multiplication Lemma applies. Therefore, applying $\omega$ to the remaining terms in (\ref{eq:fourtermtypesleft}) gives
\be
0\overset{\text{(\ref{eq:omegaFD})}}{=\joinrel=\joinrel=}\omega(W^*W)+\omega(W'^*W'),
\ee
where
\be
W:=F(B^*C)-F(B)^*F(C)
\quad\text{ and }\quad
W':=F(C^*B)-F(C)^*F(B).
\ee
Since $\omega$ is positive, this is the sum of two positive terms, which can only be zero if both terms are zero. In particular, $\omega(W^*W)=0$, which proves the lemma. 
\eprf

%\bc
%\label{lem:Attalslemmaae}
\begin{cor}[A.e.\ determinism in $\fdCAlgY$ with respect to a state]{cor:Attalslemmaae}
Let $\mA$ and $\mB$ be finite-dimensional $C^*$-algebras, 
let $\mA\xstoch{\w}\C$ be a state on $\mA$, and let $\mB\xstoch{F}\mA$ be a linear map. Then 
%\be
%\label{eq:FBsBP}
\[
F(B^*B)P_{\w}=F(B)^*F(B)P_{\w}\qquad\forall\;B\in\mB
\]
%\ee
if and only if $F$ is $\w$-a.e.\ deterministic, i.e.\
\[
F(B^*C)P_{\w}=F(B)^*F(C)P_{\w}\qquad\forall\;B,C\in\mB. 
\]
%\ec
\end{cor}

\bprf
%[Proof of Corollary~\ref{lem:Attalslemmaae}]
This follows from the weak a.e.\ Multiplication Lemma (Lemma~\ref{lem:Attalslemmaaegeneral}) and Theorem~\ref{thm:ncaeequivalence}. One can also prove this statement much more directly by following the proof of the weak a.e.\ Multiplication Lemma but using projections instead---one obtains 8 terms instead of 64. 
\eprf

These results provide some indication that we have a reasonable definition for a.e.\ determinism in the $C^*$-algebra setting (see~Theorem~\ref{thm:aemodbayesdisint} for additional justification). One can simplify this even further if one assumes the map $F$ is SPU as well. 

\begin{exa}[A.e.\ determinism for SPU maps]{exa:aedeterminismforSPU}
In the quantum Markov category $\fdCAlgY$, given a linear map
$\mB\xstoch{F}\mA$ and an SPU map $\mA\xstoch{\omega}\mC$, $F$ is right $\w$-a.e.\ deterministic if and only if 
\[
F(B^*B)-F(B)^*F(B)\in\mN_{\omega}\qquad\forall\;B\in\mB
\]
by the weak a.e.\ Multiplication Lemma (Lemma~\ref{lem:Attalslemmaaegeneral}). 
When $F$ is $*$-preserving, this is equivalent to 
\[
F(B^*B)-F(B)^*F(B)\in\mathcal{N}_{\omega}\cap{}_{\omega}\mathcal{N}\qquad\forall\;B\in\mB
\]
since this expression is self-adjoint. Finally, when $F$ is SPU, this simplifies even further to 
\[
F(B^*B)-F(B)^*F(B)\in\ker\omega\qquad\forall\;B\in\mB
\]
by Lemma~\ref{lem:minilemmaleftideal}, which applies because $F(B^*B)-F(B)^*F(B)\ge0$ by the Kadison--Schwarz inequality for $F$. 
\end{exa}

The following lemma shows that left and right a.e.\ determinism are equivalent for $*$-preserving morphisms. In fact, a more general result is proved in order to motivate some of the definitions and results that follow. 

%\br
%\label{rmk:aedeterminismnotsymmetric}
\begin{lem}[Symmetry of a.e.\ determinism for $*$-preserving morphisms]{lem:symmetryaedeterminism}
In terms of the notation from Definition~\ref{defn:aeequivalentdeterministic} and assuming that the morphisms $f$ and $p$ are $*$-preserving, then $f$ is left p-a.e.\ deterministic if and only if $f$ is right $p$-a.e.\ deterministic. More generally,%
%footnote
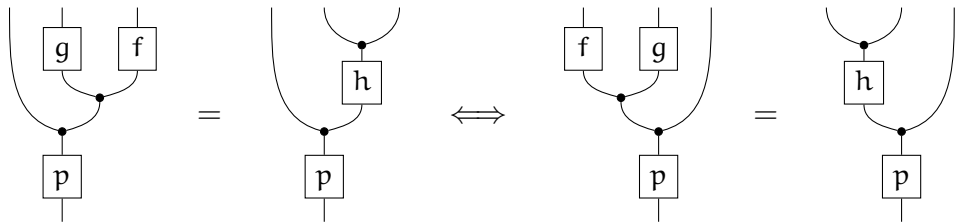
\footnote{Notice the order of the morphisms here. In a classical Markov category, $f$ and $g$ can be swapped on either side without changing this equivalence. The order matters in a general quantum Markov category, even when using even $*$-preserving morphisms.}
%end footnote
\[
\vcenter{\hbox{
\begin{tikzpicture}[font=\small]
\node[arrow box] (p) at (0,-0.3) {$p$};
\node[copier] (c) at (0,0.3) {};
\node[copier] (c2) at (0.5,0.75) {};
\node[arrow box] (f) at (1,1.4) {$f$};
\node[arrow box] (e) at (0,1.4) {$g$};
\coordinate (X) at (-0.7,1.95);
\coordinate (Y1) at (1,1.95);
\coordinate (Y2) at (0,1.95);
\draw (0,-0.9) to (p);
\draw (p) to (c);
\draw (c) to[out=165,in=-90] (X);
\draw (c) to[out=15,in=-90] (c2);
\draw (c2) to[out=15,in=-90] (f);
\draw (c2) to[out=165,in=-90] (e);
\draw (f) to (Y1);
\draw (e) to (Y2);
\end{tikzpicture}}}
\quad=\quad
\vcenter{\hbox{
\begin{tikzpicture}[font=\small]
\node[arrow box] (p) at (0,-0.3) {$p$};
\node[copier] (c) at (0,0.3) {};
\node[copier] (c2) at (0.5,1.45) {};
\node[arrow box] (f) at (0.5,0.95) {$h$};
\coordinate (X) at (-0.7,1.95);
\coordinate (Y1) at (1,1.95);
\coordinate (Y2) at (0,1.95);
\draw (0,-0.9) to (p);
\draw (p) to (c);
\draw (c) to[out=165,in=-90] (X);
\draw (c) to[out=15,in=-90] (f);
\draw (c2) to[out=15,in=-90] (Y1);
\draw (c2) to[out=165,in=-90] (0,1.95);
\draw (f) to (c2);
\end{tikzpicture}}}
\quad\iff\quad
\vcenter{\hbox{
\begin{tikzpicture}[font=\small]
\node[arrow box] (p) at (0,-0.3) {$p$};
\node[copier] (c) at (0,0.3) {};
\node[copier] (c2) at (-0.5,0.75) {};
\node[arrow box] (f) at (-1,1.4) {$f$};
\node[arrow box] (e) at (0,1.4) {$g$};
\coordinate (X) at (0.7,1.95);
\coordinate (Y1) at (-1,1.95);
\coordinate (Y2) at (0,1.95);
\draw (0,-0.9) to (p);
\draw (p) to (c);
\draw (c) to[out=15,in=-90] (X);
\draw (c) to[out=165,in=-90] (c2);
\draw (c2) to[out=165,in=-90] (f);
\draw (c2) to[out=15,in=-90] (e);
\draw (f) to (Y1);
\draw (e) to (Y2);
\end{tikzpicture}}}
\quad=\quad
\vcenter{\hbox{
\begin{tikzpicture}[font=\small]
\node[arrow box] (p) at (0,-0.3) {$p$};
\node[copier] (c) at (0,0.3) {};
\node[copier] (c2) at (-0.5,1.45) {};
\node[arrow box] (f) at (-0.5,0.95) {$h$};
\coordinate (X) at (0.7,1.95);
\coordinate (Y1) at (-1,1.95);
\coordinate (Y2) at (0,1.95);
\draw (0,-0.9) to (p);
\draw (p) to (c);
\draw (c) to[out=15,in=-90] (X);
\draw (c) to[out=165,in=-90] (f);
\draw (c2) to[out=165,in=-90] (-1,1.95);
\draw (c2) to[out=15,in=-90] (0,1.95);
\draw (f) to (c2);
\end{tikzpicture}}}
\]
whenever $f,g,h,$ and $p$ are even $*$-preserving morphisms. 
%\er
\end{lem}

\begin{proof}
It suffices to show the right equality implies the left equality. Beginning with the left and using the axioms and assumption gives
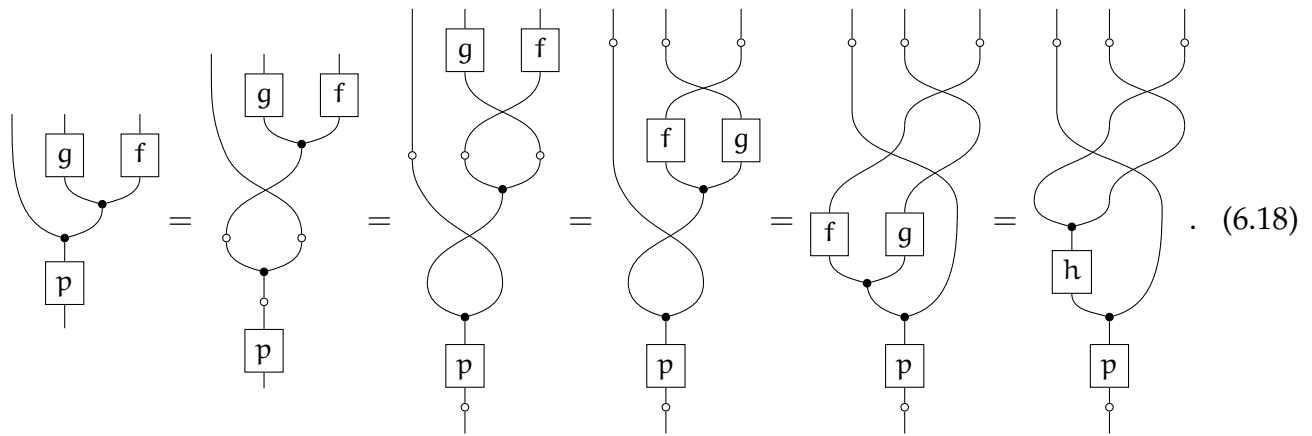
\be
\vcenter{\hbox{
\begin{tikzpicture}[font=\small]
\node[arrow box] (p) at (0,-0.3) {$p$};
\node[copier] (c) at (0,0.3) {};
\node[copier] (c2) at (0.5,0.75) {};
\node[arrow box] (f) at (1,1.4) {$f$};
\node[arrow box] (e) at (0,1.4) {$g$};
\coordinate (X) at (-0.7,1.95);
\coordinate (Y1) at (1,1.95);
\coordinate (Y2) at (0,1.95);
\draw (0,-0.9) to (p);
\draw (p) to (c);
\draw (c) to[out=165,in=-90] (X);
\draw (c) to[out=15,in=-90] (c2);
\draw (c2) to[out=15,in=-90] (f);
\draw (c2) to[out=165,in=-90] (e);
\draw (f) to (Y1);
\draw (e) to (Y2);
\end{tikzpicture}}}
%%%%%%%%%%%
=
%%%%%%%%%%%second below
\vcenter{\hbox{
\begin{tikzpicture}[font=\small]
\node[arrow box] (p) at (0,-1.2) {$p$};
\node[star] (s1) at (0,-0.55) {};
\node[copier] (c) at (0,-0.15) {};
\node[star] (R) at (0.5,0.3) {};
\node[star] (L) at (-0.5,0.3) {};
\coordinate (Ls) at (-0.5,1.4) {};
\coordinate (Rs) at (0.5,1.4) {};
\node[copier] (c2) at (0.5,1.55) {};
\node[arrow box] (f) at (1,2.2) {$f$};
\node[arrow box] (e) at (0,2.2) {$g$};
\coordinate (X) at (-0.7,2.75);
\coordinate (Y1) at (1,2.75);
\coordinate (Y2) at (0,2.75);
\draw (0,-1.7) to (p);
\draw (p) to (s1);
\draw (s1) to (c);
\draw (c) to [out=15,in=-90] (R);
\draw (c) to[out=165,in=-90] (L);
\draw (R) to [out=90,in=-65] (Ls);
%\draw (R) to [out=90,in=-90] (X);
\draw (Ls) to[out=115,in=-90] (X);
%\draw (L) to [out=90,in=-90] (Rs);
\draw (L) to [out=90,in=-90] (c2);
%\draw (Rs) to (c2);
\draw (c2) to[out=15,in=-90] (f);
\draw (c2) to[out=165,in=-90] (e);
\draw (f) to (Y1);
\draw (e) to (Y2);
\end{tikzpicture}}}
%%%%%%%%%%%
=
%%%%%%%%%%%third below
\vcenter{\hbox{
\begin{tikzpicture}[font=\small]
\node[arrow box] (p) at (0,-0.80) {$p$};
\node[star] (s1) at (0,-1.35) {};
\node[copier] (c) at (0,-0.15) {};
\coordinate (R) at (0.5,0.3) {};
\coordinate (L) at (-0.5,0.3) {};
\coordinate (Ls) at (-0.5,1.4) {};
\coordinate (Rs) at (0.5,1.4) {};
\node[copier] (c2) at (0.5,1.55) {};
\node[star] (sL) at (-0.7,2.0) {};
\node[star] (sTL) at (0,2.0) {};
\node[star] (sTR) at (1,2.0) {};
\coordinate (TLs) at (0,3.1) {};
\coordinate (TRs) at (1,3.1) {};
\node[arrow box] (f) at (1,3.4) {$f$};
\node[arrow box] (g) at (0,3.4) {$g$};
\coordinate (X) at (-0.7,3.95);
\coordinate (Y1) at (1,3.95);
\coordinate (Y2) at (0,3.95);
\draw (0,-1.7) to (s1);
\draw (s1) to (p);
\draw (p) to (c);
\draw (c) to [out=15,in=-90] (R);
\draw (c) to[out=165,in=-90] (L);
\draw (R) to [out=90,in=-55] (Ls);
\draw (Ls) to[out=125,in=-90] (sL);
\draw (sL) to (X);
\draw (L) to [out=90,in=-90] (c2);
\draw (c2) to[out=15,in=-90] (sTR);
\draw (c2) to[out=165,in=-90] (sTL);
\draw (sTL) to [out=90,in=-90] (TRs);
\draw (sTR) to [out=90,in=-90] (TLs);
\draw (TRs) to (f);
\draw (TLs) to (g);
\draw (f) to (Y1);
\draw (g) to (Y2);
\end{tikzpicture}}}
%%%%%%%%%%%
=
%%%%%%%%%%%fourth below
\vcenter{\hbox{
\begin{tikzpicture}[font=\small]
\node[arrow box] (p) at (0,-0.80) {$p$};
\node[star] (s1) at (0,-1.35) {};
\node[copier] (c) at (0,-0.15) {};
\coordinate (R) at (0.5,0.3) {};
\coordinate (L) at (-0.5,0.3) {};
\coordinate (Ls) at (-0.5,1.4) {};
\coordinate (Rs) at (0.5,1.4) {};
\node[copier] (c2) at (0.5,1.55) {};
\coordinate (sL0) at (-0.7,2.0) {};
\node[star] (sL) at (-0.7,3.5) {};
\node[arrow box] (sTL) at (0,2.2) {$f$};
\node[arrow box] (sTR) at (1,2.2) {$g$};
\coordinate (TLs) at (0,3.3) {};
\coordinate (TRs) at (1,3.3) {};
\node[star] (f) at (1,3.5) {};
\node[star] (g) at (0,3.5) {};
\coordinate (X) at (-0.7,3.95);
\coordinate (Y1) at (1,3.95);
\coordinate (Y2) at (0,3.95);
\draw (0,-1.7) to (s1);
\draw (s1) to (p);
\draw (p) to (c);
\draw (c) to [out=15,in=-90] (R);
\draw (c) to[out=165,in=-90] (L);
\draw (R) to [out=90,in=-55] (Ls);
\draw (Ls) to[out=125,in=-90] (sL0);
\draw (sL0) to (sL);
\draw (sL) to (X);
\draw (L) to [out=90,in=-90] (c2);
\draw (c2) to[out=15,in=-90] (sTR);
\draw (c2) to[out=165,in=-90] (sTL);
\draw (sTL) to [out=90,in=-90] (TRs);
\draw (sTR) to [out=90,in=-90] (TLs);
\draw (TRs) to (f);
\draw (TLs) to (g);
\draw (f) to (Y1);
\draw (g) to (Y2);
\end{tikzpicture}}}
%%%%%%%%%%%
=
%%%%%%%%%%%fifth below
\vcenter{\hbox{
\begin{tikzpicture}[font=\small]
\node[arrow box] (p) at (0,-0.80) {$p$};
\node[star] (s1) at (0,-1.35) {};
\node[copier] (c) at (0,-0.15) {};
\coordinate (R) at (0.7,1.3) {};
\coordinate (L) at (-0.5,0.3) {};
\coordinate (Ls) at (-0.7,2.6) {};
\coordinate (Rs) at (0.7,1.4) {};
\node[copier] (c2) at (-0.5,0.3) {};
\node[star] (sL) at (-0.7,3.5) {};
\node[arrow box] (sTL) at (-1,0.95) {$f$};
\node[arrow box] (sTR) at (0,0.95) {$g$};
\coordinate (oldf) at (0,2.4) {};
\coordinate (oldg) at (1,2.4) {};
\coordinate (TLs) at (0,3.3) {};
\coordinate (TRs) at (1,3.3) {};
\node[star] (f) at (1,3.5) {};
\node[star] (g) at (0,3.5) {};
\coordinate (X) at (-0.7,3.95);
\coordinate (Y1) at (1,3.95);
\coordinate (Y2) at (0,3.95);
\draw (0,-1.7) to (s1);
\draw (s1) to (p);
\draw (p) to (c);
\draw (c) to [out=15,in=-90] (R);
\draw (c) to[out=165,in=-90] (L);
\draw (R) to [out=90,in=-90] (Ls);
\draw (Ls) to (sL);
\draw (sL) to (X);
\draw (L) to [out=90,in=-90] (c2);
\draw (c2) to[out=15,in=-90] (sTR);
\draw (c2) to[out=165,in=-90] (sTL);
\draw (sTL) to [out=90,in=-90] (oldf);
\draw (sTR) to [out=90,in=-90] (oldg);
\draw (oldf) to [out=90,in=-90] (TRs);
\draw (oldg) to [out=90,in=-90] (TLs);
\draw (TRs) to (f);
\draw (TLs) to (g);
\draw (f) to (Y1);
\draw (g) to (Y2);
\end{tikzpicture}}}
%%%%%%%%%%%
=
%%%%%%%%%%%sixth below
\vcenter{\hbox{
\begin{tikzpicture}[font=\small]
\node[arrow box] (p) at (0,-0.80) {$p$};
\node[star] (s1) at (0,-1.35) {};
\node[copier] (c) at (0,-0.15) {};
\coordinate (R) at (0.7,1.3) {};
\node[arrow box] (L) at (-0.5,0.45) {$h$};
\coordinate (Ls) at (-0.7,2.6) {};
\coordinate (Rs) at (0.7,1.4) {};
\node[copier] (c2) at (-0.5,1.05) {};
\node[star] (sL) at (-0.7,3.5) {};
\coordinate (sTL) at (-1,1.4) {};%
\coordinate (sTR) at (0,1.4) {};%
\coordinate (oldf) at (0,2.4) {};
\coordinate (oldg) at (1,2.4) {};
\coordinate (TLs) at (0,3.3) {};
\coordinate (TRs) at (1,3.3) {};
\node[star] (f) at (1,3.5) {};
\node[star] (g) at (0,3.5) {};
\coordinate (X) at (-0.7,3.95);
\coordinate (Y1) at (1,3.95);
\coordinate (Y2) at (0,3.95);
\draw (0,-1.7) to (s1);
\draw (s1) to (p);
\draw (p) to (c);
\draw (c) to [out=15,in=-90] (R);
\draw (c) to[out=165,in=-90] (L);
\draw (R) to [out=90,in=-90] (Ls);
\draw (Ls) to (sL);
\draw (sL) to (X);
\draw (L) to [out=90,in=-90] (c2);
\draw (c2) to[out=15,in=-90] (sTR);
\draw (c2) to[out=165,in=-90] (sTL);
\draw (sTL) to [out=90,in=-90] (oldf);
\draw (sTR) to [out=90,in=-90] (oldg);
\draw (oldf) to [out=90,in=-90] (TRs);
\draw (oldg) to [out=90,in=-90] (TLs);
\draw (TRs) to (f);
\draw (TLs) to (g);
\draw (f) to (Y1);
\draw (g) to (Y2);
\end{tikzpicture}}}
.
%%%%%%%%%%%
\ee
Then by undoing the process of passing the involutions through, the claim is obtained. 
\end{proof}

Note that Lemma~\ref{lem:symmetryaedeterminism} is not a special case of Lemma~\ref{prop:hfpkgs}. 
We now mention a result that holds in classical Markov categories but fails in quantum Markov categories. This lemma will be important in proving that a morphism that is a.e.\ equivalent to a deterministic morphism is itself a.e.\ deterministic in a \emph{classical} Markov category. We will come back to what happens in a quantum Markov category afterwards since the situation is a bit more subtle.

%\blem
%\label{lem:fpaegimpliesffpaegg}
\begin{lem}[A.e.\ equality implies doubling a.e.\ equality in the classical setting]{lem:fpaegimpliesffpaegg}
In a classical Markov category,
\[
%\be
%\label{eq:fpaegimpliesffpaegg}
\vcenter{\hbox{
\begin{tikzpicture}[font=\small,xscale=-1]
\node[arrow box] (p) at (0,-0.3) {$p$};
\node[copier] (copier) at (0,0.3) {};
\node[arrow box] (f) at (-0.5,0.95) {$f$};
\coordinate (X) at (0.5,1.5);
\coordinate (Y) at (-0.5,1.5);
\draw (0,-0.9) to (p);
\draw (p) to (copier);
\draw (copier) to[out=15,in=-90] (X);
\draw (copier) to[out=165,in=-90] (f);
\draw (f) to (Y);
\end{tikzpicture}}}
\quad=\quad
\vcenter{\hbox{
\begin{tikzpicture}[font=\small,xscale=-1]
\node[arrow box] (p) at (0,-0.3) {$p$};
\node[copier] (copier) at (0,0.3) {};
\node[arrow box] (g) at (-0.5,0.95) {$g$};
\coordinate (X) at (0.5,1.5);
\coordinate (Y) at (-0.5,1.5);
\draw (0,-0.9) to (p);
\draw (p) to (copier);
\draw (copier) to[out=15,in=-90] (X);
\draw (copier) to[out=165,in=-90] (g);
\draw (g) to (Y);
\end{tikzpicture}}}
\quad
\implies
\quad
\vcenter{\hbox{
\begin{tikzpicture}[font=\small,xscale=-1]
\node[arrow box] (p) at (0,-0.3) {$p$};
\node[copier] (c) at (0,0.3) {};
\node[copier] (c2) at (-0.5,0.75) {};
\node[arrow box] (f) at (-1,1.4) {$f$};
\node[arrow box] (e) at (0,1.4) {$f$};
\coordinate (X) at (0.7,1.95);
\coordinate (Y1) at (-1,1.95);
\coordinate (Y2) at (0,1.95);
\draw (0,-0.9) to (p);
\draw (p) to (c);
\draw (c) to[out=15,in=-90] (X);
\draw (c) to[out=165,in=-90] (c2);
\draw (c2) to[out=165,in=-90] (f);
\draw (c2) to[out=15,in=-90] (e);
\draw (f) to (Y1);
\draw (e) to (Y2);
\end{tikzpicture}}}
\quad=\quad
\vcenter{\hbox{
\begin{tikzpicture}[font=\small,xscale=-1]
\node[arrow box] (p) at (0,-0.3) {$p$};
\node[copier] (c) at (0,0.3) {};
\node[copier] (c2) at (-0.5,0.75) {};
\node[arrow box] (f) at (-1,1.4) {$g$};
\node[arrow box] (e) at (0,1.4) {$g$};
\coordinate (X) at (0.7,1.95);
\coordinate (Y1) at (-1,1.95);
\coordinate (Y2) at (0,1.95);
\draw (0,-0.9) to (p);
\draw (p) to (c);
\draw (c) to[out=15,in=-90] (X);
\draw (c) to[out=165,in=-90] (c2);
\draw (c2) to[out=165,in=-90] (f);
\draw (c2) to[out=15,in=-90] (e);
\draw (f) to (Y1);
\draw (e) to (Y2);
\end{tikzpicture}}}
%\ee
\]
%\elem
\end{lem}

\bprf
Although a proof is given in~\cite[Lemma~13.12]{Fr19}, we illustrate it here for comparison. It follows from 
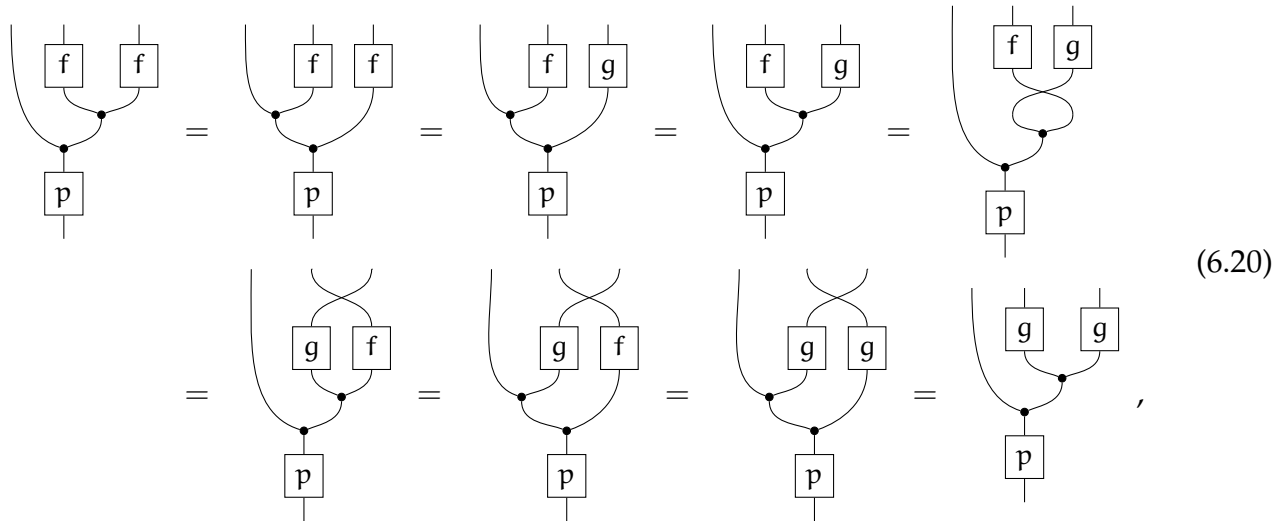
\be
\begin{split}
\vcenter{\hbox{
\begin{tikzpicture}[font=\small,xscale=-1]
\node[arrow box] (p) at (0,-0.3) {$p$};
\node[copier] (c) at (0,0.3) {};
\node[copier] (c2) at (-0.5,0.75) {};
\node[arrow box] (f) at (-1,1.4) {$f$};
\node[arrow box] (e) at (0,1.4) {$f$};
\coordinate (X) at (0.7,1.95);
\coordinate (Y1) at (-1,1.95);
\coordinate (Y2) at (0,1.95);
\draw (0,-0.9) to (p);
\draw (p) to (c);
\draw (c) to[out=15,in=-90] (X);
\draw (c) to[out=165,in=-90] (c2);
\draw (c2) to[out=165,in=-90] (f);
\draw (c2) to[out=15,in=-90] (e);
\draw (f) to (Y1);
\draw (e) to (Y2);
\end{tikzpicture}}}
%%%%%end of first picture
\;\;&=\;\;
\vcenter{\hbox{
\begin{tikzpicture}[font=\small,xscale=-1]
\node[arrow box] (p) at (0,-0.3) {$p$};
\node[copier] (c) at (0,0.3) {};
\node[copier] (c2) at (0.5,0.75) {};
\node[arrow box] (f) at (-0.8,1.4) {$f$};
\node[arrow box] (e) at (0,1.4) {$f$};
\coordinate (X) at (0.9,1.95);
\coordinate (Y1) at (-0.8,1.95);
\coordinate (Y2) at (0,1.95);
\draw (0,-0.9) to (p);
\draw (p) to (c);
\draw (c) to[out=15,in=-90] (c2);
\draw (c) to[out=165,in=-90] (f);
\draw (c2) to[out=165,in=-90] (e);
\draw (c2) to[out=15,in=-90] (X);
\draw (f) to (Y1);
\draw (e) to (Y2);
\end{tikzpicture}}}
%%%%%end of second picture
\;\;=\;\;
\vcenter{\hbox{
\begin{tikzpicture}[font=\small,xscale=-1]
\node[arrow box] (p) at (0,-0.3) {$p$};
\node[copier] (c) at (0,0.3) {};
\node[copier] (c2) at (0.5,0.75) {};
\node[arrow box] (f) at (-0.8,1.4) {$g$};
\node[arrow box] (e) at (0,1.4) {$f$};
\coordinate (X) at (0.9,1.95);
\coordinate (Y1) at (-0.8,1.95);
\coordinate (Y2) at (0,1.95);
\draw (0,-0.9) to (p);
\draw (p) to (c);
\draw (c) to[out=15,in=-90] (c2);
\draw (c) to[out=165,in=-90] (f);
\draw (c2) to[out=165,in=-90] (e);
\draw (c2) to[out=15,in=-90] (X);
\draw (f) to (Y1);
\draw (e) to (Y2);
\end{tikzpicture}}}
%%%%%end of third picture
\;\;=\;\;
%%%%%end of fourth picture
\vcenter{\hbox{
\begin{tikzpicture}[font=\small,xscale=-1]
\node[arrow box] (p) at (0,-0.3) {$p$};
\node[copier] (c) at (0,0.3) {};
\node[copier] (c2) at (-0.5,0.75) {};
\node[arrow box] (g) at (-1,1.4) {$g$};
\node[arrow box] (f) at (0,1.4) {$f$};
\coordinate (X) at (0.7,1.95);
\coordinate (Y1) at (-1,1.95);
\coordinate (Y2) at (0,1.95);
\draw (0,-0.9) to (p);
\draw (p) to (c);
\draw (c) to[out=15,in=-90] (X);
\draw (c) to[out=165,in=-90] (c2);
\draw (c2) to[out=165,in=-90] (g);
\draw (c2) to[out=15,in=-90] (f);
\draw (g) to (Y1);
\draw (f) to (Y2);
\end{tikzpicture}}}
%%%%%end of fourth (correct) picture
\;\;=\;\;
\vcenter{\hbox{
\begin{tikzpicture}[font=\small,xscale=-1]
\node[arrow box] (p) at (0,-0.3) {$p$};
\node[copier] (c) at (0,0.3) {};
\node[copier] (c2) at (-0.5,0.75) {};
\coordinate (L) at (-0.9,1.0);
\coordinate (R) at (-0.1,1.0);
\node[arrow box] (g) at (-0.9,1.9) {$g$};
\node[arrow box] (f) at (-0.1,1.9) {$f$};
\coordinate (X) at (0.7,2.45);
\coordinate (Y1) at (-0.9,2.45);
\coordinate (Y2) at (-0.1,2.45);
\draw (0,-0.9) to (p);
\draw (p) to (c);
\draw (c) to[out=15,in=-90] (X);
\draw (c) to[out=165,in=-90] (c2);
\draw (c2) to[out=165,in=-90] (L);
\draw (c2) to[out=15,in=-90] (R);
\draw (L) to[out=90,in=-90] (f);
\draw (R) to[out=90,in=-90] (g);
\draw (g) to (Y1);
\draw (f) to (Y2);
\end{tikzpicture}}}
%%%%%end of fifth picture
\\
\;\;&=\;\;
\vcenter{\hbox{
\begin{tikzpicture}[font=\small,xscale=-1]
\node[arrow box] (p) at (0,-0.3) {$p$};
\node[copier] (c) at (0,0.3) {};
\node[copier] (c2) at (-0.5,0.75) {};
\coordinate (L) at (-0.9,2.45);
\coordinate (R) at (-0.1,2.45);
\node[arrow box] (f) at (-0.9,1.4) {$f$};
\node[arrow box] (g) at (-0.1,1.4) {$g$};
\coordinate (X) at (0.7,2.45);
\coordinate (Y1) at (-0.9,2.45);
\coordinate (Y2) at (-0.1,2.45);
\draw (0,-0.9) to (p);
\draw (p) to (c);
\draw (c) to[out=15,in=-90] (X);
\draw (c) to[out=165,in=-90] (c2);
\draw (c2) to[out=165,in=-90] (f);
\draw (c2) to[out=15,in=-90] (g);
\draw (f) to[out=90,in=-90] (R);
\draw (g) to[out=90,in=-90] (L);
\draw (L) to (Y1);
\draw (R) to (Y2);
\end{tikzpicture}}}
%%%%%end of sixth picture
\;\;=\;\;
\vcenter{\hbox{
\begin{tikzpicture}[font=\small,xscale=-1]
\node[arrow box] (p) at (-0.2,-0.3) {$p$};
\node[copier] (c) at (-0.2,0.3) {};
\node[copier] (c2) at (0.4,0.75) {};
\coordinate (L) at (-0.9,2.45);
\coordinate (R) at (-0.1,2.45);
\node[arrow box] (f) at (-0.9,1.4) {$f$};
\node[arrow box] (g) at (-0.1,1.4) {$g$};
\coordinate (X) at (0.8,2.45);
\coordinate (Y1) at (-0.9,2.45);
\coordinate (Y2) at (-0.1,2.45);
\draw (-0.2,-0.9) to (p);
\draw (p) to (c);
\draw (c) to[out=15,in=-90] (c2);
\draw (c) to[out=165,in=-90] (f);
\draw (c2) to[out=165,in=-90] (g);
\draw (c2) to[out=15,in=-90] (X);
\draw (f) to[out=90,in=-90] (R);
\draw (g) to[out=90,in=-90] (L);
\draw (L) to (Y1);
\draw (R) to (Y2);
\end{tikzpicture}}}
%%%%%end of seventh picture
\;\;=\;\;
\vcenter{\hbox{
\begin{tikzpicture}[font=\small,xscale=-1]
\node[arrow box] (p) at (-0.2,-0.3) {$p$};
\node[copier] (c) at (-0.2,0.3) {};
\node[copier] (c2) at (0.4,0.75) {};
\coordinate (L) at (-0.9,2.45);
\coordinate (R) at (-0.1,2.45);
\node[arrow box] (f) at (-0.9,1.4) {$g$};
\node[arrow box] (g) at (-0.1,1.4) {$g$};
\coordinate (X) at (0.8,2.45);
\coordinate (Y1) at (-0.9,2.45);
\coordinate (Y2) at (-0.1,2.45);
\draw (-0.2,-0.9) to (p);
\draw (p) to (c);
\draw (c) to[out=15,in=-90] (c2);
\draw (c) to[out=165,in=-90] (f);
\draw (c2) to[out=165,in=-90] (g);
\draw (c2) to[out=15,in=-90] (X);
\draw (f) to[out=90,in=-90] (R);
\draw (g) to[out=90,in=-90] (L);
\draw (L) to (Y1);
\draw (R) to (Y2);
\end{tikzpicture}}}
%%%%%end of eigth picture
\;\;=\;\;
\vcenter{\hbox{
\begin{tikzpicture}[font=\small,xscale=-1]
\node[arrow box] (p) at (0,-0.3) {$p$};
\node[copier] (c) at (0,0.3) {};
\node[copier] (c2) at (-0.5,0.75) {};
\node[arrow box] (f) at (-1,1.4) {$g$};
\node[arrow box] (e) at (0,1.4) {$g$};
\coordinate (X) at (0.7,1.95);
\coordinate (Y1) at (-1,1.95);
\coordinate (Y2) at (0,1.95);
\draw (0,-0.9) to (p);
\draw (p) to (c);
\draw (c) to[out=15,in=-90] (X);
\draw (c) to[out=165,in=-90] (c2);
\draw (c2) to[out=165,in=-90] (f);
\draw (c2) to[out=15,in=-90] (e);
\draw (f) to (Y1);
\draw (e) to (Y2);
\end{tikzpicture}}}
\;\;,
\end{split}
\ee
where the symmetry (commutativity) axiom of a classical Markov category was used in the fourth and eighth equalities. 
\eprf

Note that unlike in the proof of Lemma~\ref{lem:symmetryaedeterminism} in a $*$-preserving subcategory of a general quantum Markov category, in the proof of Lemma~\ref{lem:fpaegimpliesffpaegg}, we are able to swap the order of $f$ and $g$ in the fourth equality, which allows us to use the fact that $f\aeequals{p}g$. It is not possible to do this using the involution $*$ in the more general setting. More precisely, we have the following fact.  

\vspace{-1mm}
%\bn
%\label{prop:aedoesnotimplydoubleae}
\begin{exa}[Doubling a.e.\ equality fails in $\fdCAlgCPU$]{rmk:aedoesnotimplydoubleae}
Lemma~\ref{lem:fpaegimpliesffpaegg} does not generally hold in a quantum Markov category, even if all morphisms are even and $*$-preserving. 
We will supply a simple counterexample in the subcategory $\fdCAlgCPU$ of the quantum Markov category $\fdCAlgUY$. 
Set $\mB\equiv\mA:=\mathcal{M}_{2}(\C)$ and set 
$\rho:=\left[\begin{smallmatrix}1&0\\0&0\end{smallmatrix}\right]$.
Let $\w:=\tr(\rho\;\cdot\;)$ be its associated state, and let $P_{\w}$ denote its support (in this case, $\rho=P_{\w}$). 
For any $\l\in(0,1),$ set
\[
\begin{split}
F&:=\l\id+(1-\l)\mathrm{Ad}_{P_{\w}}+(1-\l)\mathrm{Ad}_{\left[\begin{smallmatrix}0&1\\1&0\end{smallmatrix}\right]}\circ\mathrm{Ad}_{P_{\w}}
\quad\text{ and }\\
G&:=\l\id+(1-\l)\mathrm{Ad}_{P_{\w}}+(1-\l)\mathrm{Ad}_{P_{\w}^{\perp}},
\end{split}
\]
which explicitly shows that $F$ and $G$ are CP~\cite{Ch75}. 
In terms of their action on matrices, these maps are given by 
\[
F\left(\begin{bmatrix}a&b\\c&d\end{bmatrix}\right):=\begin{bmatrix}a&\l b\\\l c&(1-\l)a+\l d\end{bmatrix}
\quad\text{ and }\quad
G\left(\begin{bmatrix}a&b\\c&d\end{bmatrix}\right):=\begin{bmatrix}a&\l b\\ \l c&d\end{bmatrix}, 
\]
from which it easily follows $F$ and $G$ are unital, and therefore CPU. These maps are $\w$-a.e.\ equivalent because multiplying both expressions by $P_{\w}$ on the right gives the same result (we are freely using the equivalent notions of a.e.\ equivalence from Theorem~\ref{thm:ncaeequivalence} because $\omega$ is SPU). Using these formulas, we find
\[
F\left(\begin{bmatrix}a&b\\c&d\end{bmatrix}\right)^2P_{\w}-G\left(\begin{bmatrix}a&b\\c&d\end{bmatrix}\right)^2P_{\w}=\begin{bmatrix}0&0\\\l(1-\l)c(a-d)&0\end{bmatrix},
\]
which is non-zero in general. 
Therefore, the equality on the right-hand-side in Lemma~\ref{lem:fpaegimpliesffpaegg} does not hold in $\fdCAlgCPU$. 
%\en
\end{exa}

A consequence of Lemma~\ref{lem:fpaegimpliesffpaegg} in the classical Markov category setting is the following innocent sounding statement. 

%\bn
%\label{prop:aetodetisaedet}
\begin{prop}[A.e.\ equivalence to a deterministic map implies a.e.\ determinism]{prop:aetodetisaedet}
Let $\Theta\xstoch{p}X$ be a morphism in a classical Markov category. If $X\xstoch{f}Y$ is $p$-a.e.\ equivalent to a deterministic map $X\xstoch{g}Y$, then $f$ is $p$-a.e.\ deterministic. 
%\en
\end{prop}

Although this fact is proved in~\cite[Lemma~13.11]{Fr19}, it is worth explicitly showing where the commutativity assumption is made and to motivate an upcoming definition.

\bprf
This follows from
\be
\vcenter{\hbox{
\begin{tikzpicture}[font=\small,xscale=-1]
\node[arrow box] (p) at (0,-0.3) {$p$};
\node[copier] (c) at (0,0.3) {};
\node[copier] (c2) at (-0.5,0.75) {};
\node[arrow box] (f) at (-1,1.4) {$f$};
\node[arrow box] (e) at (0,1.4) {$f$};
\coordinate (X) at (0.7,1.95);
\coordinate (Y1) at (-1,1.95);
\coordinate (Y2) at (0,1.95);
\draw (0,-0.9) to (p);
\draw (p) to (c);
\draw (c) to[out=15,in=-90] (X);
\draw (c) to[out=165,in=-90] (c2);
\draw (c2) to[out=165,in=-90] (f);
\draw (c2) to[out=15,in=-90] (e);
\draw (f) to (Y1);
\draw (e) to (Y2);
\end{tikzpicture}}}
\quad\overset{\text{Lem~\ref{lem:fpaegimpliesffpaegg}}}{=\joinrel=\joinrel=\joinrel=\joinrel=}\quad
\vcenter{\hbox{
\begin{tikzpicture}[font=\small,xscale=-1]
\node[arrow box] (p) at (0,-0.3) {$p$};
\node[copier] (c) at (0,0.3) {};
\node[copier] (c2) at (-0.5,0.75) {};
\node[arrow box] (f) at (-1,1.4) {$g$};
\node[arrow box] (e) at (0,1.4) {$g$};
\coordinate (X) at (0.7,1.95);
\coordinate (Y1) at (-1,1.95);
\coordinate (Y2) at (0,1.95);
\draw (0,-0.9) to (p);
\draw (p) to (c);
\draw (c) to[out=15,in=-90] (X);
\draw (c) to[out=165,in=-90] (c2);
\draw (c2) to[out=165,in=-90] (f);
\draw (c2) to[out=15,in=-90] (e);
\draw (f) to (Y1);
\draw (e) to (Y2);
\end{tikzpicture}}}
\quad=\quad
\vcenter{\hbox{
\begin{tikzpicture}[font=\small,xscale=-1]
\node[arrow box] (p) at (0,-0.3) {$p$};
\node[copier] (c) at (0,0.3) {};
\node[copier] (c2) at (-0.5,1.45) {};
\node[arrow box] (g) at (-0.5,0.95) {$g$};
\coordinate (X) at (0.7,1.95);
\coordinate (Y1) at (-1,1.95);
\coordinate (Y2) at (0,1.95);
\draw (0,-0.9) to (p);
\draw (p) to (c);
\draw (c) to[out=15,in=-90] (X);
\draw (c) to[out=165,in=-90] (g);
\draw (c2) to[out=165,in=-90] (-1,1.95);
\draw (c2) to[out=15,in=-90] (0,1.95);
\draw (g) to (c2);
\end{tikzpicture}}}
\quad=\quad
\vcenter{\hbox{
\begin{tikzpicture}[font=\small,xscale=-1]
\node[arrow box] (p) at (0,-0.3) {$p$};
\node[copier] (c) at (0,0.3) {};
\node[copier] (c2) at (-0.5,1.45) {};
\node[arrow box] (f) at (-0.5,0.95) {$f$};
\coordinate (X) at (0.7,1.95);
\coordinate (Y1) at (-1,1.95);
\coordinate (Y2) at (0,1.95);
\draw (0,-0.9) to (p);
\draw (p) to (c);
\draw (c) to[out=15,in=-90] (X);
\draw (c) to[out=165,in=-90] (f);
\draw (c2) to[out=165,in=-90] (-1,1.95);
\draw (c2) to[out=15,in=-90] (0,1.95);
\draw (f) to (c2);
\end{tikzpicture}}}
.
%\qedhere
\ee
\eprf

Since we know Lemma~\ref{lem:fpaegimpliesffpaegg} is false in a quantum Markov category (by Example~\ref{rmk:aedoesnotimplydoubleae}), one wonders if the reasonable-sounding Proposition~\ref{prop:aetodetisaedet} still holds in quantum Markov categories, perhaps due to Lemma~\ref{lem:symmetryaedeterminism}. Somewhat surprisingly, Proposition~\ref{prop:aetodetisaedet} fails in a general quantum Markov category (cf.\ Example~\ref{exa:Unotdetreason}). 
This motivates the following definition. 

\begin{defn}[Deterministically reasonable subcategory]{defn:detreason}
A subcategory $\mC$ of a quantum Markov category is \define{deterministically reasonable} iff for any morphisms $f$ and $p$ in $\mC$ such that $f$ is right/left $p$-a.e.\ equivalent to a deterministic morphism, then $f$ is right/left $p$-a.e.\ deterministic. 
\end{defn}

Although every classical Markov category is deterministically reasonable 
by Proposition~\ref{prop:aetodetisaedet}, 
this is not true in an arbitrary quantum Markov category, even if $\mC$ is $*$-preserving. Fortunately, it holds for our main subcategories of interest. 

\vspace{-1mm}
\begin{exa}[$\fdCAlgSPU$ is deterministically reasonable]{exa:SPUdetreason}
Let $\mB\xstoch{F}\mA$ and $\mA\xstoch{\omega}\mC$ be SPU maps. Suppose that $F\underset{\raisebox{.6ex}[0pt][0pt]{\tiny$\omega$}}{=}G$, where $\mB\xstoch{G}\mA$ is deterministic. Fixing $B\in\mB$ immediately gives
\[
\omega\big(F(B^*B)-F(B)^*F(B)\big)=\omega\big(G(B^*B)-G(B)^*G(B)\big)=0,  
\]
since $F\aeequals{\omega}G$ and since $G$ is deterministic. 
But $F(B^*B)-F(B)^*F(B)\ge0$ by KS for $F$. Hence, Lemma~\ref{lem:minilemmaleftideal} implies $F(B^*B)-F(B)^*F(B)\in\mathcal{N}_{\omega}$. Since $B$ was arbitrary, $F$ is $\omega$-a.e.\ deterministic by the weak a.e.\ Multiplication Lemma (Lemma~\ref{lem:Attalslemmaaegeneral}). 
\end{exa}

\vspace{-1mm}
\begin{exa}[$\fdCAlgU_{*}$ is not deterministically reasonable]{exa:Unotdetreason}
Set $\mA:=\mM_{4}(\C)$ and $\mB:=\mM_{2}(\C)$. Set $F,G:\mB\stoch\mA$ and $\mA\xstoch{\omega}\C$ to be the unital linear $*$-preserving maps 
\[
\mM_{2}(\C)\ni\begin{bmatrix}b_{11}&b_{12}\\b_{21}&b_{22}\end{bmatrix}\mapsto F(B):=\begin{bmatrix}b_{11}&b_{12}&0&0\\b_{21}&b_{22}&b_{21}&b_{21}\\0&b_{12}&b_{11}&b_{12}\\0&b_{12}&b_{21}&b_{22}\end{bmatrix},\quad
G(B):=\begin{bmatrix}b_{11}&b_{12}&0&0\\b_{21}&b_{22}&0&0\\0&0&b_{11}&b_{12}\\0&0&b_{21}&b_{22}\end{bmatrix}
\]
and
\[
\omega=\tr(\rho\;\cdot\;),\quad\text{ where }\quad
\rho=\begin{bmatrix}1&0&0&0\\0&0&0&0\\0&0&0&0\\0&0&0&0\end{bmatrix}. 
\]
Then $G$ is deterministic, $F\aeequals{\omega}G$ Theorem~\ref{thm:ncaeequivalence}, and yet $F$ is not right $\omega$-a.e.\ deterministic by Corollary~\ref{cor:Attalslemmaae}. This shows that  $\fdCAlgU_{*}$ is not deterministically reasonable. 
\end{exa}

\vspace{-1mm}
\begin{ques}[Is $\fdCAlgPU$ deterministically reasonable?]{a016}
The proof provided in Example~\ref{exa:SPUdetreason} used the Kadison--Schwarz inequality for $F$. It is unclear to us if this is necessary. Example~\ref{exa:Unotdetreason} showed us that $\fdCAlgU_{*}$ is not deterministically reasonable. What about the subcategory $\fdCAlgPU$?
\end{ques}

We end this section with a concept that is an a.e.\ weakening of unitality. The introduction of this concept is motivated by theorems regarding disintegrations from measure theory~\cite[Section~452]{FrV4}. We will actually prove these statements abstractly in the Markov category language in Remark~\ref{rmk:aeunitaldisint}. 
The reader uninterested in these results can safely skip this discussion and move to the next section. 

\begin{defn}[A.e.\ unitality]{eq:aecausal}
Let $\Theta\xstoch{p}X\xstoch{f}Y$ be a composable pair of morphisms in a quantum CD category, where %$p$ is 2-positive and 
$f$ is even. The morphism $f$ is said to be \define{left/right $p$-a.e.\ unital} iff 
\[
\left. 
\vcenter{\hbox{
\begin{tikzpicture}[font=\small]
\node[arrow box] (p) at (0,-0.3) {$p$};
\node[copier] (copier) at (0,0.3) {};
\node[arrow box] (f) at (-0.5,0.95) {$f$};
\coordinate (X) at (0.5,1.7);
\node[discarder] (Y) at (-0.5,1.5) {};
\draw (0,-0.9) to (p);
\draw (p) to (copier);
\draw (copier) to[out=15,in=-90] (X);
\draw (copier) to[out=165,in=-90] (f);
\draw (f) to (Y);
\end{tikzpicture}}}
\quad=\quad
\vcenter{\hbox{
\begin{tikzpicture}[font=\small,xscale=-1]
\node[arrow box] (p) at (0,0.3) {$p$};
\coordinate (X) at (0,1.6);
\draw (0,-1.0) to (p);
\draw (p) to (X);
\end{tikzpicture}}}
\qquad\middle/\qquad
\vcenter{\hbox{
\begin{tikzpicture}[font=\small,xscale=-1]
\node[arrow box] (p) at (0,-0.3) {$p$};
\node[copier] (copier) at (0,0.3) {};
\node[arrow box] (f) at (-0.5,0.95) {$f$};
\coordinate (X) at (0.5,1.7);
\node[discarder] (Y) at (-0.5,1.5) {};
\draw (0,-0.9) to (p);
\draw (p) to (copier);
\draw (copier) to[out=15,in=-90] (X);
\draw (copier) to[out=165,in=-90] (f);
\draw (f) to (Y);
\end{tikzpicture}}}
\quad=\quad
\vcenter{\hbox{
\begin{tikzpicture}[font=\small,xscale=-1]
\node[arrow box] (p) at (0,0.3) {$p$};
\coordinate (X) at (0,1.6);
\draw (0,-1.0) to (p);
\draw (p) to (X);
\end{tikzpicture}}}
\right.
\quad.
\]
\end{defn}

\begin{exa}[A.e.\ unitality in $\fdCAlgP$ and $\FinMeas$]{a017}
Given a state $\mA\xstoch{\w}\C$ in $\fdCAlgPU$, a positive (not necessarily unital) map $\mB\xstoch{F}\mA$ is right $\w$-a.e.\ unital if and only if 
$F(1_{\mB})P_{\w}=P_{\w}$ by Theorem~\ref{thm:ncaeequivalence} and the first axiom in~(\ref{eq:markovcatfirstconditions}) in Definition~\ref{defn:qmc}. 
Since 
$P_{\w}^{\perp}F(1_{\mB})P_{\w}=P_{\w}^{\perp}P_{\w}=0$ and $P_{\w}F(1_{\mB})P_{\w}^{\perp}=(P_{\w}^{\perp}F(1_{\mB})P_{\w})^*=0$
(because $F$ is $*$-preserving), we conclude  
\[
F(1_{\mB})=P_{\w}F(1_{\mB})P_{\w}+P_{\w}^{\perp}F(1_{\mB})P_{\w}^{\perp}=P_{\w}+\mathrm{Ad}_{P_{\w}^{\perp}}\big(F(1_{\mB})\big).
\]
This guarantees $F(1_{\mB})\ge P_{\w}$ since $F$ is a positive map. 

For $\FinMeas$, in the special case $\mA=\C^{X}$ and $\mB=\C^{Y}$, the state $\w$ corresponds to a probability measure $I\xstoch{p}X$, and this condition means that the corresponding map $X\xstoch{f}Y$ associates to each $x\in X\setminus N_{p}$ a probability measure on $Y$. However, it can assign \emph{any} (possibly signed) measure 
on $Y$ to each element of $N_{p}$. Indeed, a.e.\ unitality provides us with the equation 
\[
p_{x}=\sum_{y\in Y}f_{yx}p_{x}=p_{x}\sum_{y\in Y}f_{yx}\qquad\forall\; x\in X.
\]
When $x\in X\setminus N_{p}$, this gives the constraint $\sum_{y\in Y}f_{yx}=1$, but when $x\in N_{p}$, this gives no condition. 
\end{exa}

%%%%%%%%%%%%%%%%%%%%%%%%%%%%%%%%%%%%%%
\section[Inversion, disintegration, and Bayesian inversion]{a018}
\label{sec:disintbayes}
\vspace{-12mm}
\noindent
\begin{tikzpicture}
\coordinate (L) at (-8.75,0);
\coordinate (R) at (8.75,0);
\draw[line width=2pt,orange!20] (L) -- node[anchor=center,rectangle,fill=orange!20]{\strut \Large \textcolor{black}{\textbf{7\;\; Inversion, disintegration, and Bayesian inversion}}} (R);
\end{tikzpicture}
%\vspace{1mm}
%%%%%%%%%%%%%%%%%%%%%%%%%%%%%%%%%%%%%%

Now that we have described a.e.\ equivalence, a.e.\ determinism, and a.e\ unitality, we focus on disintegrations and Bayesian inversion in this section. We recall the definitions, we state several properties, and we review the current status on their existence in the non-commutative setting~\cite{PaRu19,PaRuBayes} illustrating that disintegrations, Bayesian inverses, and conditionals (in the sense of~\cite[Definition~11.5]{Fr19}) do not always exist in $\fdCAlgSPU$ (nor do they exist in $\fdCAlgCPU$). Nevertheless, necessary and sufficient conditions for when disintegrations and Bayesian inverses exist are provided in~\cite{PaRu19,PaRuBayes}. The purpose of this and the next section is to elucidate the relationships between them in the cases when they do exist. Examples in classical and quantum error correcting codes are also provided. 

%\bd
%\label{defn:disintegration}
\begin{defn}[State-preserving morphisms and disintegrations]{defn:disintegration}
Let $\mathcal{M}_{\text{\Yinyang}}$ be a quantum Markov category and let $\mC$ be a $*$-preserving subcategory of $\mathcal{M}_{\text{\Yinyang}}$. 
Given states $I\xstoch{p}X$ and $I\xstoch{q}Y$ in $\mM_{\text{\Yinyang}}$, 
a morphism $X\xstoch{f}Y$ in $\mathcal{M}_{\text{\Yinyang}}$ is said to be 
\define{state-preserving} iff 
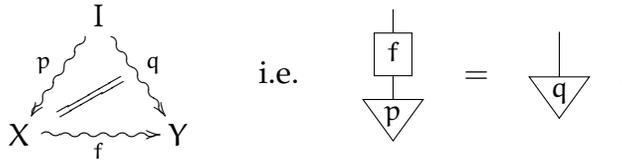
\[
\xy0;/r.25pc/:
(0,7.5)*+{I}="o";
(-10,-7.5)*+{X}="X";
(10,-7.5)*+{Y}="Y";
{\ar@{~>}"o";"X"_{p}};
{\ar@{~>}"o";"Y"^{q}};
{\ar@{~>}"X";"Y"_{f}};
{\ar@{=}(-5,-5);(3,-0.5)};
\endxy
\qquad\text{i.e.}\qquad
\vcenter{\hbox{%
\begin{tikzpicture}[font=\small]
\node[state] (p) at (0,-0.9) {$p$};
\node[arrow box] (f) at (0,-0.3) {$f$};
\coordinate (X) at (0,0.3);
\draw (p) to (f);
\draw (f) to (X);
\end{tikzpicture}}}
\quad
=
\quad
\vcenter{\hbox{%
\begin{tikzpicture}[font=\small]
\node[state] (q) at (0,-0.9) {$q$};
\coordinate (X) at (0,-0.3);
\draw (q) to (X);
\end{tikzpicture}}}
\quad.
\]
Such data will be denoted by $(f,p,q)$. 
If $\mC$ is even and if $f$ is also in $\mC$, a \define{(right) disintegration of $(f,p,q)$ in $\mC$ (or in $\mM$)} is 
a morphism%
%footnote
\footnote{An explanation of this terminology is given in Example~\ref{exa:disintsinstoch}. Technically, it would be more mathematically appropriate and consistent to say that ``A morphism $Y\xstoch{g}X$ in $\mC$ (or in $\mM$) is a \define{(right) disintegration of $(f,p,q)$} iff the given equations hold.'' Example~\ref{exa:disintsinstoch} explains that not every disintegration in the measure-theoretic sense is necessarily represented by a transition kernel (cf.\ \cite[Definition~A.1]{PaRu19}). However, for a clean categorical account of disintegrations, we demand that they are.}
%end footnote
 $Y\xstoch{g}X$ in $\mC$ (or in $\mM$) such that%
%footnote
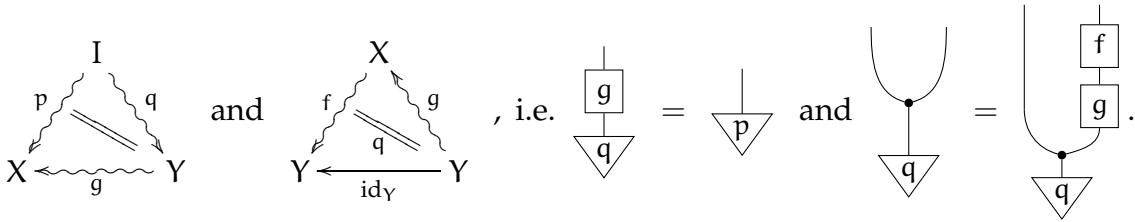
\footnote{Note that if $g$ is only in $\mM$, it need not be $*$-preserving~\cite{PaRuBayes}. The equality $f\circ g\aeequals{q}\id_{Y}$ should be read as right a.e.\ equivalence to be consistent with the string diagram. Analogously, there is a similar notion of \define{left disintegration}, where the second diagram has the $f$ and $g$ on the left string instead of the right string. If $g$ is in $\mC$, then these two notions agree.}
%end footnote
\[
\xy0;/r.25pc/:
(0,7.5)*+{I}="o";
(-10,-7.5)*+{X}="X";
(10,-7.5)*+{Y}="Y";
{\ar@{~>}"o";"X"_{p}};
{\ar@{~>}"o";"Y"^{q}};
{\ar@{~>}"Y";"X"^{g}};
{\ar@{=}(-3,0);(5,-4.5)};
\endxy
\;\;\text{and}\;\;
\xy0;/r.25pc/:
(0,7.5)*+{X}="X";
(10,-7.5)*+{Y}="Y1";
(-10,-7.5)*+{Y}="Y2";
{\ar@{~>}"Y1";"X"_{g}};
{\ar@{~>}"X";"Y2"_{f}};
{\ar"Y1";"Y2"^{\mathrm{id}_{Y}}};
{\ar@{=}(-3,0);(5,-4.5)_{q}};
\endxy
\;\;,\;\text{i.e.}\;\;
\vcenter{\hbox{%
\begin{tikzpicture}[font=\small]
\node[state] (q) at (0,-0.9) {$q$};
\node[arrow box] (g) at (0,-0.3) {$g$};
\coordinate (Y) at (0,0.3);
\draw (q) to (g);
\draw (g) to (Y);
\end{tikzpicture}}}
\;\;
=
\;\;
\vcenter{\hbox{%
\begin{tikzpicture}[font=\small]
\node[state] (p) at (0,-0.9) {$p$};
\coordinate (X) at (0,-0.3);
\draw (p) to (X);
\end{tikzpicture}}}
\;\;\;\text{and}\;\;
\vcenter{\hbox{%
\begin{tikzpicture}[font=\small]
\node[state] (q) at (0,-0.4) {$q$};
\node[copier] (copier) at (0,0.3) {};
\coordinate (X) at (-0.5,1.31);
\coordinate (X2) at (0.5,1.31);
\draw (q) to (copier);
\draw (copier) to[out=165,in=-90] (X);
\draw (copier) to[out=15,in=-90] (X2);
\end{tikzpicture}}}
\;\;
=
\;\;
\vcenter{\hbox{%
\begin{tikzpicture}[font=\small]
\node[state] (q) at (0,0) {$q$};
\node[copier] (copier) at (0,0.3) {};
\coordinate (f) at (-0.5,0.91) {};
\node[arrow box] (g) at (0.5,0.95) {$g$};
\node[arrow box] (e) at (0.5,1.75) {$f$};
\coordinate (X) at (-0.5,2.3);
\coordinate (Y) at (0.5,2.3);
\draw (q) to (copier);
\draw (copier) to[out=165,in=-90] (f);
\draw (f) to (X);
\draw (copier) to[out=15,in=-90] (g);
\draw (g) to (e);
\draw (e) to (Y);
\end{tikzpicture}}}
\;.
\]
As a shorthand, `$g$ is a disintegration of $(f,p,q)$' will refer to the case when $g$ is in $\mC$.%
%footnote
\footnote{One can slightly generalize the definition of a disintegration by not requiring $p$ nor $q$ to be states but arbitrary morphisms. We will do this in future sections and still refer to them as disintegrations (see Remark~\ref{rmk:statestomapsfordisint} for more details). This will be particularly relevant when relating this notion to those of sufficient statistics in the non-commutative setting.}
%end footnote
\end{defn}

\vspace{-1mm}
\begin{rmk}[Allowing $f$ to be stochastic for disintegrations]{a019}
The definition of a disintegration in Definition~\ref{defn:disintegration} differs from the one introduced in~\cite{PaRu19} in that we are no longer assuming $X\xstoch{f}Y$ is deterministic. The reason for this is to apriori allow the possibility for more morphisms to have disintegrations. 
However, it turns out that in suitable subcategories (of which $\FinStoch$ and $\fdCAlgSPU$ are examples), if a morphism $X\xstoch{f}Y$ together with a state $I\xstoch{p}X$ has a disintegration, then $f$ is $p$-a.e.\ deterministic (see Theorem~\ref{thm:aemodbayesdisint}). 
\end{rmk}

\begin{exa}[Disintegrations for deterministic morphisms exist in $\FinStoch$]{thm:disintsexistfordeterministicmaps}
Let $(X,p)$ and $(Y,q)$ be finite probability spaces.
Let $f:X\to Y$ be a measure-preserving function (cf.\ Example~\ref{exa:probpreserving}).
Then there exists a disintegration $Y\xstoch{g}X$ of 
$(f,p,q)$. Moreover, $g$ is unique $q$-a.e.\ and a formula for such a (representative of a) disintegration is given by 
%\be
%\label{eq:disintformula}
\[
g_{xy}:=
\begin{cases}
p_{x}\de_{yf(x)}/q_{y}&\mbox{ if } y\in Y\setminus N_{q}\\
1/|X|&\mbox{ if } y\in N_{q}
\end{cases}
\]
%\ee
(see~\cite[Section~2.2]{PaRu19} for details). However, disintegrations of \emph{stochastic maps} (as opposed to functions) do not always exist. We will show why in Corollary~\ref{cor:binversiondisint}.  
\end{exa}

%\pagebreak
%\phantom{why do I need this here?}
%\vspace{-6mm}
\vspace{-1mm}
\begin{exa}[Hamming codes as disintegrations in $\FinStoch$]{a020}
Hamming's error correcting codes~\cite{Ha50} can be described in terms of disintegrations. We will illustrate this in the concrete case of Hamming's $(7,4)$ code (though the analysis works for any of them).%
%footnote
\footnote{The presentation we follow here is based on an Exercises \#53 and \#54 in~\cite[Section 3.1]{Br97}. I have also benefitted from discussions with Christian Carmellini and Philip Parzygnat.}
%end footnote
 Set $X:=\Z_{2}^{4}$ and $Y:=\Z_{2}^{7}$. We interpret $X$ as the set of possible messages a sender wishes to transmit and $Y$ as the set of possible messages that the receiver acquires after transmission. Since $X$ and $Y$ are composed of vectors, we will write an element of $X$ and $Y$ as $\vec{x}$ and $\vec{y}$, respectively. Here, $\Z_{2}=\{0,1\}$ with modulo 2 arithmetic. Set
\[
Q:=\begin{bmatrix}1&0&1&1\\1&1&0&1\\1&1&1&0\end{bmatrix}, 
\qquad
H:=\begin{bmatrix}\mathds{1}_{3}&Q\end{bmatrix},
\qquad
M:=\begin{bmatrix}Q\\\mathds{1}_{4}\end{bmatrix}.
\]
Then, one can show that 
\[
0\to\Z_{2}^{4}\xrightarrow{M}\Z_{2}^{7}\xrightarrow{H}\Z_{2}^{3}\to0
\]
is an exact sequence%
%footnote
\footnote{A sequence of abelian groups $\cdots \rightarrow V_{i-1}\xrightarrow{f_{i}}V_{i}\xrightarrow{f_{i+1}}V_{i+1}\rightarrow\cdots$ is \define{exact} iff $\mathrm{Im}(f_{i})=\ker(f_{i+1})$ for all $i$.}
%end footnote
~\cite{PaLinear}. Exactness of this sequence says that $M$ embeds $\Z_{2}^{4}$ into a subspace of $\Z_{2}^{7}$, henceforth referred to as the \emph{code subspace}. Furthermore, $H$ takes \emph{only} this subspace to $0$. The rest of $\Z_{2}^7$ is generated from the code subspace by the addition of any one of the seven unit vectors in $\Z_{2}^{7}$ in the sense that 
\[
\Z_{2}^{7}=M(\Z_{2}^{4})\cup\bigcup_{i=1}^{7}\Big(M(\Z_{2}^{4})+\vec{e}_{i}\Big)
\]
and all terms in this union are mutually disjoint. 

Now, the sender wants to send a message, an element of $X$, across a communication channel to a receiver. During the process, the sender first embeds the initial code using $M$. A stochastic map $X\xstoch{f}Y$, called the \emph{error}, is used to represent the evolution of any initial message across the channel. The specific conditional probabilities $f_{\vec{y}\vec{x}}$ associated to the channel (namely, the probability that the messages $\vec{y}$ is received given that $\vec{x}$ was sent) will not be relevant to our discussion, but one crucial assumption is made. This assumption is 
\[
f_{\vec{y}\vec{x}}=0 \quad\text{ if }\quad \left\lVert\vec{y}-M\vec{x}\right\rVert>1,
\] 
where $\lVert\;\cdot\;\rVert$ here denotes the \emph{Hamming distance}, defined via%
%footnote
\footnote{The modulo $2$ arithmetic is used only for each term! Once $z_{i}-y_{i}$ is computed in $\Z_{2}$, the sum $\sum_{i=1}^{7}$ is a standard sum of non-negative integers. For example, $\lVert(1,0,1,0,0,1,0)-(1,0,1,1,1,1,1)\rVert=3$.}
%footnote 
\[
\lVert\vec{y}-\vec{z}\rVert:=\sum_{i=1}^{n}\Big((z_{i}-y_{i})\!\!\mod2\Big)\qquad\forall\;\vec{y},\vec{z}\in\Z_{2}^{n}. 
\]
This is the assumption that at most one error occurs during the transmission. 

Once a message $\vec{y}$ is received, the receiver must apply the transformation $H$ to obtain a vector in $\Z_{2}^{3}$. The number of vectors in $\Z_{2}^{3}$ is $2^3=8$, corresponding to the total number of possible errors plus no error. The outcome $H\vec{y}$ is therefore either one of the columns of $H$ or is the zero vector. If it is the zero vector, no error has occurred. If it is the $i$-th column of $H$, then an error occurred in the $i$-th entry of the vector $\vec{y}$. In the latter case, the receiver must correct for this error to obtain $\vec{y}+\vec{e}_{i}$. Afterwards, the receiver applies the projection map onto the first four entries of this vector to obtain a vector in $\Z_{2}^{4}$. This describes a deterministic map $Y\xrightarrow{g}X$, called the \emph{recovery}, whose definition we summarize as follows. Let $Z:=\Z_{2}^{3}$, define $Z\xrightarrow{k}\{0,1,\dots,7\}$ by 
\[
Z\ni\vec{z}\mapsto k(\vec{z}):=
\begin{cases}
0&\mbox{ if $\vec{z}=\vec{0}$}\\
i&\mbox{ if $\vec{z}=H\vec{e}_{i}$,}
\end{cases}
\]
and let $\Z_{2}^{7}\xrightarrow{\pi}\Z_{2}^{4}$ be the projection onto the first four entries. 
Finally, define $Y\xrightarrow{g}X$ by 
\[
Y\ni\vec{y}\mapsto g(\vec{y}):=
\begin{cases}
\pi(H\vec{y})&\mbox{ if $k(H\vec{y})=0$}\\
\pi\big(H\vec{y}+\vec{e}_{k(H\vec{y})}\big)&\mbox{ if $k(H\vec{y})\ne0$.}
\end{cases}
\]
One can easily check that $g\circ f=\id_{X}$. In particular, given \emph{any} probability measure $\{\bullet\}\xstoch{p}X$ (describing a probability measure on the collection of possible messages the sender may send), and defining $q:=f\circ p$ (which is the probability measure describing the possible messages the receiver may see), it follows that the error map $f$ is a disintegration of the recovery map $(g,q,p)$.%
%footnote
\footnote{Normally, disintegrations can be viewed as stochastic processes that recover the information loss associated to a deterministic process~\cite{BFL,PaEntropy}. Therefore, it may sound strange that error is a disintegration of recovery. Nevertheless, this is the case because one is meant to perform a specific (i.e.\ deterministic) task to recover the message.}
%end footnote
\end{exa}

%\pagebreak
\phantom{why do I need this here?}
\vspace{-11mm}
%\vspace{-2mm}
\begin{exa}[Disintegrations in $\Stoch$]{exa:disintsinstoch}
The terminology of \emph{disintegration} comes from the relationship of this concept to one from measure theory, which is normally defined as follows~\cite[Section~452]{FrV4} (see also~\cite{Pa79,We94}).%
%footnote
\footnote{For the relationship to the terminology used by Cho--Jacobs from~\cite{ChJa18}, see Remark~\ref{rmk:ChoJacobsDisintegrations}.}
%end footnote
 Given a measure-preserving function $(X,\Sigma,p)\xrightarrow{f}(Y,\Omega,q)$, where $\Sigma$ and $\Omega$ are $\sigma$-algebras on $X$ and $Y$, respectively, a \define{disintegration of $p$ over $q$ consistent with $f$} is a family of pairs $(\Sigma_{y},g_{y})$, indexed by $y\in Y$, where $\Sigma_{y}$ is a $\s$-algebra on $X$ and $g_{y}$ is a measure on $(X,\Sigma_{y})$ satisfying the following three conditions. 
\begin{enumerate}
\itemsep0pt
\item
For each $A\in\S$, there exists a $N_{A}\in\Omega$ such that $q(N_{A})=0$, $A\in\S_{y}$ for all $y\in Y\setminus N_{A}$, and 
\[
Y\setminus N_{A}\ni y \mapsto g_{y}(A)\in[0,\infty]
\]
is $\Omega_{\upharpoonright Y\setminus N_{A}}$-measurable. Here, $\Omega_{\upharpoonright Y\setminus N_{A}}$ is the natural $\s$-algebra on $Y\setminus N_{A}$ viewed as a submeasurable space of $(Y,\Omega)$. 
\item
For each $A\in\S$, 
\[
\int_{Y}g_{y}(A)\,dq(y)=p(A), 
\]
where the integral on the left-hand-side is technically defined as $\int_{Y\setminus N_{A}}g_{y}(A)\,dq(y)$, where $N_{A}$ is as in the previous item (and the measure $q$ is the unique one determined by $q\big(B\cap(Y\setminus N_{A})\big):=q(B)$ for all $B\in\Omega$). 
\item
For each $B\in\Omega$, there exists a $q$-null set $N_{B}\in\Omega$ such that $f^{-1}(B)\in\S_{y}$ for all $y\in Y\setminus N_{B}$ and 
\[
g_{y}\big(f^{-1}(B)\big)=1\qquad\forall\;y\in(Y\setminus N_{B})\cap B.
\]
\end{enumerate}
It is not necessary that the family of $\s$-algebras $\S_{y}$ satisfies $\Sigma\subseteq\Sigma_{y}$ for this notion of disintegration (cf.\ \cite[Example~1.2]{Pa79}). %for an example involving unit intervals and the Lebesgue measure). 
If this \emph{were} the case, then one could restrict all the $\s$-algebras to be $\Sigma$, and then one could define a transition kernel $(Y,\Omega)\xstoch{g}(X,\Sigma)$ representing the family of measures. 

This motivates the following closely related definition of a (measure-theoretic) disintegration by demanding that it be represented by a transition kernel. A transition kernel $(Y,\Omega)\xstoch{g}(X,\Sigma)$ is a \define{disintegration of $p$ over $q$ consistent with $f$} iff 
\[
\int_{Y}g_{y}(A)\,dq(y)=p(A)\qquad\forall\;A\in\Sigma
\]
and for each $B\in\Omega$, there exists a $q$-null set $N_{B}\in\Omega$ such that 
\[
g_{y}\big(f^{-1}(B)\big)=1\qquad\forall\;y\in(Y\setminus N_{B})\cap B.
\]
The equivalence between \emph{this} definition and the one from Definition~\ref{defn:disintegration} above is explained in~\cite[Appendix~A]{PaRu19} (it is immediate from the definitions that the first equality is just the statement $p=g\circ q$). This definition of a transition kernel being a disintegration is equivalent to the first definition above (where the family of measures need not form a transition kernel) when $(X,\Sigma,p)$ is a non-empty countably compact measure space (cf.\ \cite[Exercise 452X (l)]{FrV4}). Therefore, we do not lose too much generality by considering our disintegrations (as in Definition~\ref{defn:disintegration}) to be represented by transition kernels. 

Note, however, that our usage of the concept of a disintegration from Definition~\ref{defn:disintegration} is also (a-priori) slightly more general than the (latter) measure-theoretic one because we allow $f$ to be a \emph{Markov kernel}, as opposed to just a \emph{measurable map}. 
 Furthermore, although $g$ is a-priori only a transition kernel (as opposed to a Markov kernel), it is a fact~\cite[Proposition~452G (a)]{FrV4} that if $g$ is a transition kernel that is a disintegration of $p$ over $q$ consistent with $f$, then $g$ is a \emph{Markov kernel} $q$-a.e. (see Remark~\ref{rmk:aeunitaldisint} for a diagrammatic proof). Hence, we lose no generality by assuming that we work with unital morphisms (Markov kernels) in our definition of disintegration. 
\emph{Regular conditional probabilities} can also be described in terms of disintegrations (though their categorical relationship will be made more explicit in Remark~\ref{exa:RCP}). Finally, disintegrations are sometimes also called \emph{optimal hypotheses}.% 
%footnote
\footnote{The definition of a regular conditional probability is a-priori different from the definition of a disintegration. Also, the definition of an optimal hypothesis comes from~\cite{BaFr14}. Technically, that definition requires a stochastic section on the nose as opposed to a.e.\ equality to the identity, but it seems reasonable to apply the same terminology to this setting.} 
%end footnote
 These relationships are also explained in~\cite[Appendix~A]{PaRu19}.
\end{exa}

%\phantom{why do I need this here?}
%\vspace{-7mm}
\vspace{-1mm}
\begin{rmk}[Relationship to Cho--Jacobs' disintegration]{rmk:ChoJacobsDisintegrations}
In~\cite{ChJa18}, Cho and Jacobs defined a disintegration (called a \emph{CJ disintegration} in the remainder of this remark) as one of two assignments that take a joint state $I\xstoch{s}X\times Y$ and produce probability-preserving morphism $(X,p)\xstoch{f}(Y,q)$ and $(Y,q)\xstoch{g}(X,p)$ such that 
\[
\vcenter{\hbox{%
\begin{tikzpicture}[font=\small]
\node[state] (q) at (0,0) {$q$};
%\node[arrow box] (f) at (0,-0.2) {$f$};
%\node[state] (p) at (0,-0.8) {$p$};
\node[copier] (copier) at (0,0.3) {};
\node[arrow box] (g) at (-0.5,0.95) {$g$};
\coordinate (X) at (-0.5,1.5);
\coordinate (Y) at (0.5,1.5);
\draw (q) to (copier);
%\draw (f) to (copier);
%\draw (p) to (f);
\draw (copier) to[out=165,in=-90] (g);
\draw (copier) to[out=15,in=-90] (Y);
\draw (g) to (X);
\path[scriptstyle]
node at (-0.7,1.45) {$X$}
node at (0.7,1.45) {$Y$};
\end{tikzpicture}}}
%\;
\quad
=
\quad
\vcenter{\hbox{%
\begin{tikzpicture}[font=\small]
\node[state] (omega) at (0,0) {\;$s$\;};
\coordinate (X) at (-0.25,0.55) {};
\coordinate (Y) at (0.25,0.55) {};
\draw (omega) ++(-0.25, 0) to (X);
\draw (omega) ++(0.25, 0) to (Y);
\path[scriptstyle]
node at (-0.45,0.4) {$X$}
node at (0.45,0.4) {$Y$};
\end{tikzpicture}}}
\quad
=
\quad
%\;
\vcenter{\hbox{%
\begin{tikzpicture}[font=\small]
\node[state] (p) at (0,0) {$p$};
\node[copier] (copier) at (0,0.3) {};
\node[arrow box] (f) at (0.5,0.95) {$f$};
\coordinate (X) at (-0.5,1.5);
\coordinate (Y) at (0.5,1.5);
\draw (p) to (copier);
\draw (copier) to[out=165,in=-90] (X);
\draw (copier) to[out=15,in=-90] (f);
\draw (f) to (Y);
\path[scriptstyle]
node at (-0.7,1.45) {$X$}
node at (0.7,1.45) {$Y$};
\end{tikzpicture}}}
\;\;,
\]
where $p$ and $q$ are the marginals of $s$, namely the composites $I\xstoch{s}X\times Y\xrightarrow{\pi_{X}}X$ and $I\xstoch{s}X\times Y\xrightarrow{\pi_{Y}}Y$, respectively. 
The assignments taking a joint state to conditionals are special kinds of disintegrations in our language where the maps one disintegrates are \emph{projections}. In more detail, they arise as the horizontal squiggly arrows in the diagram 
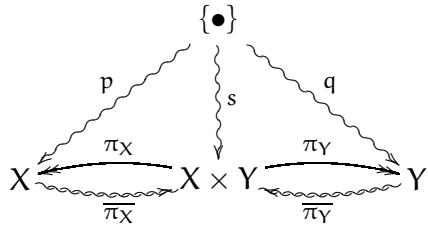
\[
\xy0;/r.25pc/:
(-25,-10)*+{X}="X";
(25,-10)*+{Y}="Y";
(0,-10)*+{X\times Y}="XY";
(0,10)*+{\{\bullet\}}="1";
{\ar@{~>}"1";"X"_{p}};
{\ar@{~>}"1";"Y"^{q}};
{\ar@{~>}"1";"XY"^{s}};
{\ar@{->>}@/_0.5pc/"XY";"X"_{\pi_{X}}};
{\ar@{->>}@/^0.5pc/"XY";"Y"^{\pi_{Y}}};
%{\ar@{=}(-5,-4);(-8,-2)};
%{\ar@{=}(5,-4);(8,-2)};
{\ar@{~>}@/_0.5pc/"X";"XY"_{\overline{\pi_{X}}}};
{\ar@{~>}@/^0.5pc/"Y";"XY"^{\overline{\pi_{Y}}}};
\endxy
\]
where $\overline{\pi_{X}}$ and $\overline{\pi_{Y}}$ are disintegrations of $(\pi_{X},s,p)$ and $(\pi_{Y},s,q)$, respectively. 
The morphisms $f$ and $g$ are then obtained by post-composing these disintegrations with the other projection (following the horizontal morphisms in succession). More explicitly $f=\pi_{Y}\circ\overline{\pi_{X}}$ and $g=\pi_{X}\circ\overline{\pi_{Y}}$. The proof that $f$ and $g$ as defined here satisfy the CJ disintegration condition will follow from additional assumptions on the specific subcategory that these morphisms live in, namely those covered by Theorem~\ref{thm:aemodbayesdisint} (cf.\ Remark~\ref{rmk:CJdisintcontinued}). 

We point out this remark because the two ideas of disintegration are semantically different in quantum (and classical!) Markov categories since the input data are different. The CJ disintegration associates conditionals from joint states (see also Remark~\ref{rmk:EPRconditionals}). In contrast, the disintegration as presented in Definition~\ref{defn:disintegration} assumes a conditional and describes a conditional in the opposite direction. 
Our reason for distinguishing them is because the associated joint state need not be a positive linear functional in the setting of $C^*$-algebras, rendering it uninterpretable as a genuine physical state. Indeed, if $f$ is the identity, the associated joint functional is the copy map followed by the state one begins with. Since the copy map is not positive unital, the associated joint functional is not positive~\cite[Remark~5.96]{PaRuBayes}. 
\end{rmk}

\vspace{-1mm}
\begin{exa}[Knill--Laflamme quantum error correction as a disintegration]{exa:QEC}
Just as Hamming error correcting codes can be described in terms of disintegrations, some quantum error correcting codes can be as well.%
%footnote
\footnote{I thank Benjamin Russo for discussions and for bringing my attention to~\cite{KLP05}.}
%end footnote
 Though much can be said about general quantum error correcting codes~\cite{KLP05}, we content ourselves here with a concrete example~\cite{KnLa97}, leaving the more general study to forthcoming work. In what follows, we implement Dirac brac-ket notation. In particular, let $|0\>$ and $|1\>$ denote the unit vectors $\vec{e}_{1}$ and $\vec{e}_{2}$ in $\C^{2}=:\mathcal{Q}$, respectively. Let $\g\in[0,\infty)$ and set $a_{\pm}:=\sqrt{\frac{1\pm e^{-\gamma}}{2}}$. Set
\[
|0_{L}\>:=\frac{1}{2^{3/2}}\big(|0\>+|1\>\big)^{\otimes 3}
\quad
\text{ and }
\quad
|1_{L}\>:=\frac{1}{2^{3/2}}\big(|0\>-|1\>\big)^{\otimes 3}
\]
to be the \emph{logical} $0$ and $1$ states in $\Hi:=(\C^{2})^{\otimes 3}$. Set $\mC:=\mathrm{span}\big\{|0_{L}\>,|1_{L}\>\big\}$ to be the \emph{code subspace} in $\Hi$. 
Let $V:\mathcal{Q}\hookrightarrow\Hi$ be the isometry, called the \emph{encoding}, sending $\mathcal{Q}$ onto $\mC$, namely $V|0\>:=|0_{L}\>$ and $V|1\>:=|1_{L}\>$. 
Define the single qubit error generators
\[
\Lambda_{+}:=a_{+}\mathds{1}_{2}
\quad\text{ and }\quad
\Lambda_{-}:=a_{-}\sigma_{z},
\quad\text{ where }\quad\s_{z}:=\begin{bmatrix}1&0\\0&-1\end{bmatrix}, 
\]
and let 
\[
\Gamma:=\frac{1}{4}\left(2-e^{-3\gamma}+3e^{-\g}\right). 
\]
Define the \emph{error operators} by 
\begin{align*}
E_{0}&:=\frac{1}{\sqrt{\Gamma}}\Lambda_{+}\otimes \Lambda_{+}\otimes \Lambda_{+},&
E_{1}&:=\frac{1}{\sqrt{\Gamma}}\Lambda_{-}\otimes \Lambda_{+}\otimes \Lambda_{+},\\
E_{2}&:=\frac{1}{\sqrt{\Gamma}}\Lambda_{+}\otimes \Lambda_{-}\otimes \Lambda_{+},&
E_{3}&:=\frac{1}{\sqrt{\Gamma}}\Lambda_{+}\otimes \Lambda_{+}\otimes \Lambda_{-}. 
\end{align*}
These error operators are the Kraus operators of a CPU map%
%footnote
\footnote{Since we are in the Heisenberg picture, our $E$ is more closely related to the Hilbert--Schmidt dual of the error map $\mathcal{E}$ in~\cite{KnLa97}.}
%endfootnote
$\mB\xstoch{E}\mB$ defined by 
\[
\mB\ni B\xmapsto{E} \sum_{i=0}^{3}E_{i}^{\dag}AE_{i}, 
\]
where $\mB:=\mB(\Hi)$. The composite 
\[
F:=\left(\mA\xstoch{E}\mB\xstoch{\Ad_{V^{\dag}}}\mA\right),
\]
where $\mA:=\mB(\mathcal{Q})$, is called the \emph{error}, and is also a CPU map. Just as in the case of the Hamming code, it includes the encoding done by the sender as well as the possible errors that can occur along the transmission channel. In this case, the allowed errors include flipping a qubit in a \emph{single} component of the tensor product (3 such possibilities) or no error. The receiver will apply a recovery map to obtain the message being sent. To define this map, let $P_{\mC}:=|0_{L}\>\<0_{L}|+|1_{L}\>\<1_{L}|$ denote the projection onto $\mC$ (as an element of $\mB(\Hi)$) and define the \emph{recovery operators} by
\begin{align*}
R_{0}&:=P_{\mC}(\mathds{1}_{2}\otimes\mathds{1}_{2}\otimes\mathds{1}_{2}),&
R_{1}&:=P_{\mC}(\s_{z}\otimes\mathds{1}_{2}\otimes\mathds{1}_{2}),\\
R_{2}&:=P_{\mC}(\mathds{1}_{2}\otimes\s_{z}\otimes\mathds{1}_{2}),&
R_{3}&:=P_{\mC}(\mathds{1}_{2}\otimes\mathds{1}_{2}\otimes\s_{z}).\\
\end{align*}
From these, define the CPU map $\mB\xstoch{R}\mB$ by 
\[
\mB\ni B\xmapsto{R}\sum_{i=0}^{3}R_{i}^{\dag}AR_{i}.
\]
Unitality of $R$ follows from the fact that the $R_{i}^{\dag}R_{i}$ are orthogonal projections onto mutually orthogonal 2-dimensional subspaces of $\Hi$. For example, $R_{0}^{\dag}R_{0}=P_{\mC}$. Similarly, $R_{i}^{\dag}R_{i}$ is the projection onto the subspace $R_{i}^{\dag}(\mC)$, which is where an element of $\mathcal{Q}$ gets sent under a bit flip in the $i$-th tensor component. Thus, by applying these orthogonal projections, the receiver can \emph{detect} which subspace the message is in. Furthermore, the recovery operators automatically correct the error as well because $\s_{z}^{2}=\mathds{1}_{2}$. Hence, the receiver is able to recover what the original message was. Mathematically, 
the composite
\[
G:=\left(\mA\xstoch{\Ad_{V}}\mB\xstoch{R}\mB\right)
\]
is called the \emph{recovery}, and it describes this process. It is also a CPU map (which is not immediately obvious because $\Ad_{V}$ is not unital). In fact, $G$ is a $*$-homomorphism.%
%footnote
\footnote{This is the determinism we mentioned in Example~\ref{exa:deterministicmaps}. In the Schr\"odinger picture, this corresponds to a partial trace.}
%end footnote

Putting all our data together, we obtain the following (not necessarily commutative) diagram (recall that this is described in the Heisenberg picture) 
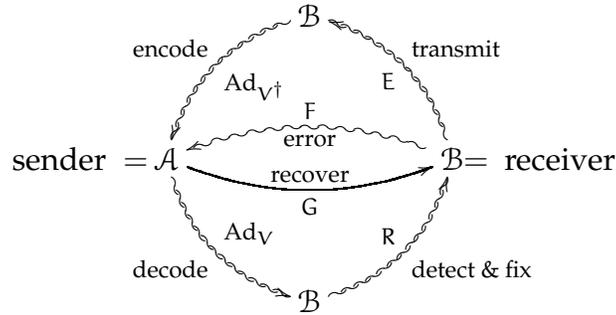
\[
\xy0;/r.25pc/:
(-29,0)*{\text{sender }=};
(-18,0)*+{\mA}="B";
(18,0)*+{\mB}="A";
(29,0)*{=\text{ receiver}};
(0,18)*+{\mB}="At";
(0,-18)*+{\mB}="Ab";
{\ar@/_1.0pc/"B";"A"_{G}^{\text{recover}}};
{\ar@{~>}@/_1.0pc/"A";"B"_{F}^{\text{error}}};
{\ar@{~>}@/_1.0pc/"B";"Ab"^{\Ad_{V}}_{\text{decode}}};
{\ar@{~>}@/_1.0pc/"Ab";"A"^{R}_{\text{detect \& fix}}};
{\ar@{~>}@/_1.0pc/"A";"At"^{E}_{\text{transmit}}};
{\ar@{~>}@/_1.0pc/"At";"B"^{\Ad_{V^{\dag}}}_{\text{encode}}};
\endxy
\]
One can also show that $F\circ G=\id_{\mA}$. Hence, the error map $F$ is a disintegration of the recovery map $(G,\omega\circ F,\omega)$ for every state $\mA\xstoch{\omega}\C$ that the sender wishes to communicate to the receiver. 
\end{exa}

\begin{rmk}[A.e.\ unitality for disintegrations]{rmk:aeunitaldisint}
In Definition~\ref{defn:disintegration}, suppose $\mM_{\text{\Yinyang}}$ is a quantum CD category instead, and suppose $f$ is unital. If we had only demanded that a right Bayes map $g$ be a morphism in $\mM$, so that it is even but not necessarily unital, then it is automatically right $q$-a.e.\ unital. This follows from 
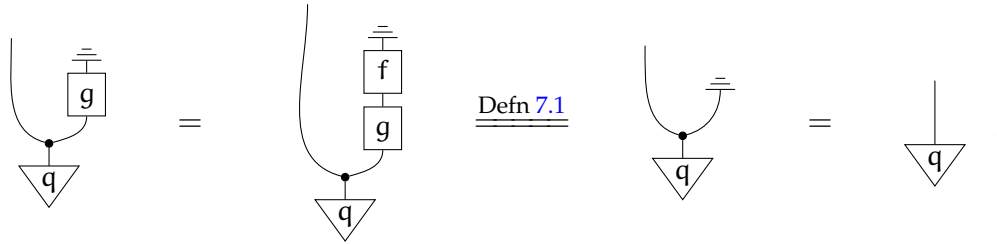
\[
\vcenter{\hbox{%
\begin{tikzpicture}[font=\small]
\node[state] (q) at (0,0) {$q$};
\node[copier] (copier) at (0,0.3) {};
\node[arrow box] (g) at (0.5,0.95) {$g$};
\coordinate (X) at (-0.5,1.7);
\node[discarder] (Y) at (0.5,1.4) {};
\draw (q) to (copier);
\draw (copier) to[out=165,in=-90] (X);
\draw (copier) to[out=15,in=-90] (g);
\draw (g) to (Y);
\end{tikzpicture}}}
\qquad
=
\qquad
\vcenter{\hbox{%
\begin{tikzpicture}[font=\small]
\node[state] (q) at (0,0) {$q$};
\node[copier] (copier) at (0,0.3) {};
\node[arrow box] (g) at (0.5,0.95) {$g$};
\coordinate (X) at (-0.5,2.6);
\node[discarder] (Y) at (0.5,2.15) {};
\node[arrow box] (f) at (0.5,1.7) {$f$};
\draw (q) to (copier);
\draw (copier) to[out=165,in=-90] (X);
\draw (copier) to[out=15,in=-90] (g);
\draw (g) to (f);
\draw (f) to (Y);
\end{tikzpicture}}}
\qquad
\overset{\text{Defn~\ref{defn:disintegration}}}{=\joinrel=\joinrel=\joinrel=\joinrel=}
\qquad
\vcenter{\hbox{%
\begin{tikzpicture}[font=\small]
\node[state] (q) at (0,0) {$q$};
\node[copier] (copier) at (0,0.3) {};
\coordinate (X) at (-0.5,1.5);
\coordinate (c) at (0.5,0.91);
\node[discarder] (Y) at (c) {};
\draw (q) to (copier);
\draw (copier) to[out=165,in=-90] (X);
\draw (copier) to[out=15,in=-90] (c);
\end{tikzpicture}}}
\qquad
=
\qquad
\vcenter{\hbox{%
\begin{tikzpicture}[font=\small]
\node[state] (q) at (0,0) {$q$};
\node (X) at (0,1.0) {};
\draw (q) to (X);
\end{tikzpicture}}}
\quad,
\]
where we have used unitality of $f$ in the first equality. This provides an abstract proof of a standard result~\cite[Proposition~452G (a)]{FrV4}, the latter of which is a special case by Proposition~\ref{thm:Bayesianinverseofdeterministicisadisint} (modulo the subtlety mentioned in Example~\ref{exa:disintsinstoch}). In fact, it is a partial generalization since our $f$ is not assumed to be a measurable function but is only assumed to be stochastic (cf.\ Example~\ref{exa:disintsinstoch}). 
\end{rmk}

\begin{defn}[Bayesian inverses and Bayes maps]{defn:bayesianinverse}
Let $\mathcal{M}_{\text{\Yinyang}}$ be a quantum Markov category and let $\mC$ be an even subcategory of $\mathcal{M}_{\text{\Yinyang}}$. 
Given states $I\xstoch{p}X$ and $I\xstoch{q}Y$ in $\mC$ and a state-preserving morphism $X\xstoch{f}Y$ in $\mC$, a \define{left Bayesian inverse of $(f,p,q)$} is 
a morphism $Y\xstoch{g}X$ in $\mC$ such that%
%footnote
\footnote{In other words, one can `slide' either morphism over $\D$ but this swaps the states and morphisms.}
%end footnote
%\be
%\label{eq:Bayesdiagramsabstract}
\[
\xy0;/r.25pc/:
(0,7.5)*+{I}="1";
(-25,7.5)*+{Y}="Y";
(25,7.5)*+{X}="X";
(-25,-7.5)*+{Y\times Y}="YY";
(25,-7.5)*+{X\times X}="XX";
(0,-7.5)*+{X\times Y}="XY";
{\ar@{~>}"1";"X"^{p}};
{\ar@{~>}"1";"Y"_{q}};
{\ar"Y";"YY"_{\Delta_{Y}}};
{\ar"X";"XX"^{\Delta_{X}}};
{\ar@{~>}"YY";"XY"_{g\times\id_{Y}}};
{\ar@{~>}"XX";"XY"^{\id_{X}\times f}};
(0,0)*{=\joinrel=\joinrel=};
\endxy
\quad,\;\;\text{i.e.}\qquad
\vcenter{\hbox{%
\begin{tikzpicture}[font=\small]
\node[state] (q) at (0,0) {$q$};
%\node[arrow box] (f) at (0,-0.2) {$f$};
%\node[state] (p) at (0,-0.8) {$p$};
\node[copier] (copier) at (0,0.3) {};
\node[arrow box] (g) at (-0.5,0.95) {$g$};
\coordinate (X) at (-0.5,1.5);
\coordinate (Y) at (0.5,1.5);
\draw (q) to (copier);
%\draw (f) to (copier);
%\draw (p) to (f);
\draw (copier) to[out=165,in=-90] (g);
\draw (copier) to[out=15,in=-90] (Y);
\draw (g) to (X);
\path[scriptstyle]
node at (-0.7,1.45) {$X$}
node at (0.7,1.45) {$Y$};
\end{tikzpicture}}}
%\;
\quad
=
\quad
%\;
\vcenter{\hbox{%
\begin{tikzpicture}[font=\small]
\node[state] (p) at (0,0) {$p$};
\node[copier] (copier) at (0,0.3) {};
\node[arrow box] (f) at (0.5,0.95) {$f$};
\coordinate (X) at (-0.5,1.5);
\coordinate (Y) at (0.5,1.5);
\draw (p) to (copier);
\draw (copier) to[out=165,in=-90] (X);
\draw (copier) to[out=15,in=-90] (f);
\draw (f) to (Y);
\path[scriptstyle]
node at (-0.7,1.45) {$X$}
node at (0.7,1.45) {$Y$};
\end{tikzpicture}}}
\;.
\]
%\ee
This condition will be referred to as \define{the Bayes condition}.
More generally, a \define{left Bayes map for $(f,p,q)$} is a morphism $f$ in $\mM$ such that the Bayes condition holds.
%footnote
\footnote{Just as we did for disintegrations, one can slightly generalize these definitions by not requiring $p$ nor $q$ to be states but arbitrary morphisms (cf.\ Remark~\ref{rmk:statestomapsfordisint}). We will do this in future sections and still refer to them as Bayesian inverses or Bayes maps.}
%end footnote
%\ed
A completely analogous definition is made for \define{right Bayesian inverses} and \define{right Bayes maps}, and the adjectives `right' and `left' are dropped when the two agree, when it does not matter which is used, or when it is clear which one of the two is being used.%
%footnote
\footnote{As usual, the left and right notions are equivalent when $\mC$ is $*$-preserving and all morphisms are in $\mC$.}
%end footnote
\end{defn}

Thus, a Bayesian inverse is a particular kind of Bayes map. The definition of a Bayesian inverse was already motivated in Section~\ref{sec:whatisbayes} along with the theorem regarding its existence for $\FinStoch$. It also exists for more general probability spaces (standard Borel spaces)~\cite{CuSt14,CDDG17,ChJa18}. Example~\ref{exa:RCP} will describe how regular conditional probabilities are special kinds of Bayesian inverses. 

\begin{exa}[Bayes maps and Bayesian inverses for matrix algebras]{exa:Bayesmapsformatrixalgebras}
The general theory of Bayesian inverses for CPU maps (in the sense presented here) between finite-dimensional $C^*$-algebras is the main subject of~\cite{PaRuBayes}. A special case is simple enough to describe here. If $\mA:=\mM_{m}(\C)\xstoch{\omega=\tr(\rho\;\cdot\;)}\C$ and $\mB:=\mM_{n}(\C)\xstoch{\xi=\tr(\s\;\cdot\;)}\C$ are states (with corresponding density matrices), and if $\mB\xstoch{F}\mA$ is a state-preserving CPU map, then a Bayes map $\mA\xstoch{G}\mB$ for $(F,\omega,\xi)$ must necessarily satisfy $P_{\xi}G(A)=\hat{\sigma}F^*(\rho A)$ for all $A\in\mA$. Here, $\hat{\sigma}$ is the pseudo-inverse of $\sigma$ (cf.\ \cite{Mo1920,Pe55}) and $F^*$ is the Hilbert--Schmidt adjoint of $F$. Thus, a Bayes map always exists. However, demanding that the Bayes map be $*$-preserving is a non-trivial condition. In general, therefore, a Bayesian inverse need not exist as a CPU map. In the special case that $P_{\xi}=\mathds{1}_{n}$, a CPU Bayesian inverse of $(F,\omega,\xi)$ exists if and only if $F(\sigma B)\rho=\rho F(B\sigma)$ for all $B\in\mB$. When this condition is satisfied, the Bayesian inverse takes the form $G=\mathrm{Ad}_{\sqrt{\hat{\sigma}}}\circ F^*\circ\mathrm{Ad}_{\sqrt{\rho}}$. The expression here is often known in the quantum information literature as the \emph{Petz recovery map}~\cite{Pe84,Wi15}. It was originally introduced by Accardi and Cecchini as a generalized conditional expectation~\cite{AcCe82}. It has also been rediscovered by Barnum--Knill as well as Leifer in more recent years~\cite{BaKn02,Le06}. 
However, when $\sigma$ and/or $\rho$ do not have full support, the notion of Bayesian inverse as presented here is generally different from the Petz recovery map, and it generalizes another notion studied by Accardi and Cecchini in the case of faithful states~\cite{AcCe82}. 

Physically, demanding that a Bayesian inverse is CPU suggests that it can be interpreted, and potentially implemented, as a physical process. In the case that no such CPU Bayesian inverse exists, one \emph{always} obtains left and right Bayes maps. It is generally believed that linear maps that are not CP are not implementable as physical operations.%
%footnote
\footnote{In recent years, there has been some dispute about this in the literature~\cite{Pe94,Al95,ShSu05,SRS19}.}
%end footnote
 Nevertheless (and perhaps surprisingly), this map has interesting consequences in terms of inference in certain cases. These details are explored in~\cite{PaJeffrey} (a preview is given in Remark~\ref{rmk:EPRconditionals}). 
\end{exa}

%\br
%\label{rmk:Bayesianinversespreservestates}
\begin{lem}[A.e.\ unitality and state-preservation for Bayesian inverses]{lem:Bayesianinversespreservestates}
If $\mM_{\text{\Yinyang}}$ in Definition~\ref{defn:bayesianinverse} is replaced with a quantum CD category, one can define Bayesian inverses in essentially the same manner. Under these more general assumptions the following facts hold. 
\begin{enumerate}
\item
A left Bayes map $g$ for $(f,p,q)$ is necessarily left $q$-a.e.\ unital. 
\item
If $f$ is right $p$-a.e.\ unital, then $g$ is state-preserving, i.e.\ $q\circ g=p$. 
\end{enumerate}
%\er
\end{lem}

\bprf
{\color{white}{you found me!}}

\begin{enumerate}
\item
A left Bayes map $g$ is necessarily left $q$-a.e.\ unital since
\be
\vcenter{\hbox{%
\begin{tikzpicture}[font=\small,xscale=-1]
\node[state] (q) at (0,0) {$q$};
\node[copier] (copier) at (0,0.3) {};
\node[arrow box] (g) at (0.5,0.95) {$g$};
\coordinate (X) at (-0.5,1.7);
\node[discarder] (Y) at (0.5,1.4) {};
\draw (q) to (copier);
\draw (copier) to[out=165,in=-90] (X);
\draw (copier) to[out=15,in=-90] (g);
\draw (g) to (Y);
\end{tikzpicture}}}
\qquad
\overset{\text{Defn~\ref{defn:bayesianinverse}}}{=\joinrel=\joinrel=\joinrel=\joinrel=}
\qquad
\vcenter{\hbox{%
\begin{tikzpicture}[font=\small,xscale=-1]
\node[state] (p) at (0,0) {$p$};
\node[copier] (copier) at (0,0.3) {};
\coordinate (c) at (0.5,0.91);
\node[discarder] (Y) at (c) {};
\node[arrow box] (f) at (-0.5,0.95) {$f$};
\coordinate (X) at (-0.5,1.7);
\node[discarder] (Y) at (c) {};
\draw (p) to (copier);
\draw (copier) to[out=165,in=-90] (f);
\draw (f) to (X);
\draw (copier) to[out=15,in=-90] (c);
\end{tikzpicture}}}
\qquad
=
\qquad
\vcenter{\hbox{%
\begin{tikzpicture}[font=\small]
\node[state] (p) at (0,0) {$p$};
\node[arrow box] (f) at (0,0.75) {$f$};
\draw (p) to (f);
\draw (f) to (0,1.55);
\end{tikzpicture}}}
\qquad
=
\qquad
\vcenter{\hbox{%
\begin{tikzpicture}[font=\small]
\node[state] (q) at (0,0) {$q$};
\node (X) at (0,1.0) {};
\draw (q) to (X);
\end{tikzpicture}}}
\quad.
\ee
\item
If $f$ is right $p$-a.e.\ unital, then 
\be
\vcenter{\hbox{%
\begin{tikzpicture}[font=\small]
\node[state] (p) at (0,0) {$p$};
\node (X) at (0,1.0) {};
\draw (p) to (X);
\end{tikzpicture}}}
\quad
=
\quad
\vcenter{\hbox{%
\begin{tikzpicture}[font=\small]
\node[state] (p) at (0,0) {$p$};
\node[copier] (copier) at (0,0.3) {};
\coordinate (X) at (-0.5,1.5);
\coordinate (c) at (0.5,0.91);
\node[discarder] (Y) at (c) {};
\draw (p) to (copier);
\draw (copier) to[out=165,in=-90] (X);
\draw (copier) to[out=15,in=-90] (c);
\end{tikzpicture}}}
\quad
=
\quad
\vcenter{\hbox{%
\begin{tikzpicture}[font=\small]
\node[state] (p) at (0,0) {$p$};
\node[copier] (copier) at (0,0.3) {};
\node[arrow box] (f) at (0.5,0.95) {$f$};
\coordinate (X) at (-0.5,1.7);
\node[discarder] (Y) at (0.5,1.4) {};
\draw (p) to (copier);
\draw (copier) to[out=165,in=-90] (X);
\draw (copier) to[out=15,in=-90] (f);
\draw (f) to (Y);
\end{tikzpicture}}}
\quad
\overset{\text{Defn~\ref{defn:bayesianinverse}}}{=\joinrel=\joinrel=\joinrel=\joinrel=}
\quad
\vcenter{\hbox{%
\begin{tikzpicture}[font=\small]
\node[state] (q) at (0,0) {$q$};
\node[copier] (copier) at (0,0.3) {};
\coordinate (c) at (0.5,0.91);
\node[discarder] (Y) at (c) {};
\node[arrow box] (g) at (-0.5,0.95) {$g$};
\coordinate (X) at (-0.5,1.7);
\node[discarder] (Y) at (c) {};
\draw (q) to (copier);
\draw (copier) to[out=165,in=-90] (g);
\draw (g) to (X);
\draw (copier) to[out=15,in=-90] (c);
\end{tikzpicture}}}
\quad
=
\quad
\vcenter{\hbox{%
\begin{tikzpicture}[font=\small]
\node[state] (q) at (0,0) {$q$};
\node[arrow box] (g) at (0,0.75) {$g$};
\draw (q) to (g);
\draw (g) to (0,1.55);
\end{tikzpicture}}}
\quad.
\ee
\end{enumerate}
\eprf

The following lemma shows that Bayesian inverses are a.e.\ unique whenever they are $*$-preserving. 

%\blem
%\label{lem:bayesaeunique}
\begin{lem}[Bayesian inverses are left a.e.\ unique]{lem:bayesaeunique}
Let $I\xstoch{p}X\xstoch{f}Y$ be a composable pair of even morphisms in a quantum Markov category and set $q:=f\circ p$. 
\begin{enumerate}
\item
If $g,g':Y\stoch X$ are two left (right) Bayes maps for $(f,p,q)$, then $g$ and $g'$ are left (right) $q$-a.e.\ equal. In particular, if $f,p,q,g,$ and $g'$ are $*$-preserving, then $g\underset{\raisebox{.6ex}[0pt][0pt]{\scriptsize$q$}}{=}g'$. 
\item
If $X\xstoch{f'} Y$ is right (left) $p$-a.e. equal to $f$, then $Y\xstoch{g}X$ is a left (right) Bayes map for $(f,p,q)$ if and only if $g$ is a left (right) Bayes map for $(f',p,q)$. 
\end{enumerate}
%\elem
\end{lem}

\bprf
Although this is explained in~\cite[Section~5]{ChJa18}, we illustrate the concise proof. 

\begin{enumerate}
\item
By assumption,
\be 
\vcenter{\hbox{
\begin{tikzpicture}[font=\small]
\node[state] (q) at (0,0) {$q$};
\node[copier] (copier) at (0,0.3) {};
\node[arrow box] (g) at (-0.5,0.95) {$g$};
\coordinate (X) at (0.5,1.5);
\coordinate (Y) at (-0.5,1.5);
\draw (q) to (copier);
\draw (copier) to[out=15,in=-90] (X);
\draw (copier) to[out=165,in=-90] (g);
\draw (g) to (Y);
\end{tikzpicture}}}
\quad=\quad
\vcenter{\hbox{
\begin{tikzpicture}[font=\small]
\node[state] (p) at (0,0) {$p$};
\node[copier] (copier) at (0,0.3) {};
\node[arrow box] (f) at (0.5,0.95) {$f$};
\coordinate (X) at (-0.5,1.5);
\coordinate (Y) at (0.5,1.5);
\draw (p) to (copier);
\draw (copier) to[out=165,in=-90] (X);
\draw (copier) to[out=15,in=-90] (f);
\draw (f) to (Y);
\end{tikzpicture}}}
\quad=\quad
\vcenter{\hbox{
\begin{tikzpicture}[font=\small]
\node[state] (q) at (0,0) {$q$};
\node[copier] (copier) at (0,0.3) {};
\node[arrow box] (g) at (-0.5,0.95) {$g'$};
\coordinate (X) at (0.5,1.5);
\coordinate (Y) at (-0.5,1.5);
\draw (q) to (copier);
\draw (copier) to[out=15,in=-90] (X);
\draw (copier) to[out=165,in=-90] (g);
\draw (g) to (Y);
\end{tikzpicture}}}
\quad.
\ee
Applying Corollary~\ref{cor:leftequalsrightaeequivalence} gives the required result when the morphisms are $*$-preserving. 
\item
A similar calculation proves this. \qedhere
\end{enumerate}
\eprf

Lemma~\ref{lem:bayesaeunique} says that Bayesian inverses are a.e.\ unique in a certain sense. A similar result also holds for disintegrations, though the proof is given much later in Corollary~\ref{cor:binversiondisint} in a.e.\ modular subcategories. 

%\br
%\label{rmk:nonstarpresbayesnotae}
\begin{rmk}[On the a.e.\ equivalence of Bayesian inverses]{rmk:nonstarpresbayesnotae}
Without assuming $g$ and $g'$ are $*$-preserving in Lemma~\ref{lem:bayesaeunique}, we can only conclude that $g$ is \emph{left} $q$-a.e.\ equivalent to $g$', and not necessarily \emph{right} a.e.\ equivalent. 
Indeed, in the category $\fdCAlgU$, if two Bayesian inverses are not $*$-preserving, then they need not be a.e.\ equivalent. Explicit examples are provided in~\cite{PaRuBayes}. In any case, we will be clear about these subtle points in this work whenever they arise. 
%\er
\end{rmk}

We now describe some of the compositional/functorial properties of disintegrations and Bayesian inverses in appropriate subcategories of quantum Markov categories. We first define the associated category of states and then return to these facts. 

\begin{defn}[The category of state-preserving morphisms]{defn:statepreservingcat}
Let $\mC$ be an even subcategory of a quantum Markov category.
Let $I_{/\mC}$ be the category whose objects are pairs $(X,I\xstoch{p}X)$ (with $X$ in $\mC$ and $p$ a state in $\mC$) and a morphism from $(X,I\xstoch{p}X)$ to $(Y,I\xstoch{q}Y)$ is a state-preserving morphism $X\xstoch{f}Y$ in $\mC$, i.e.\ $f\circ p=q$. 
The category $I_{/\mC}$ is called the \define{category of state-preserving morphisms in $\mC$}. The category $\mC$ is said to be \define{left/right a.e.\ well-defined with respect to states} iff for every quadruple of morphisms 
\[
\xy0;/r.25pc/:
(-20,0)*+{(X,p)}="X";
(0,0)*+{(Y,q)}="Y";
(20,0)*+{(Z,p)}="Z";
{\ar@{~>}@<0.6ex>"X";"Y"^{f}};
{\ar@{~>}@<-0.6ex>"X";"Y"_{f'}};
{\ar@{~>}@<0.6ex>"Y";"Z"^{g}};
{\ar@{~>}@<-0.6ex>"Y";"Z"_{g'}};
\endxy
\]
for which the pairs $(f,f')$ and $(g,g')$ are left/right $p$- and $q$-a.e.\ equivalent, respectively, then the composites $g\circ f$ and $g'\circ f'$ are left/right $p$-a.e.\ equivalent. If $\mC$ is also $*$-preserving, 
then the left/right adjectives will be dropped.
\end{defn}

%\bn
%\label{prop:functorialityofdisints}
\begin{prop}[Compositionality of disintegrations]{prop:functorialityofdisints}
Let $\mC$ be a subcategory of a quantum Markov category that is %left/right 
a.e.\ well-defined with respect to states (cf.\ Definition~\ref{defn:statepreservingcat}).%
%footnote
\footnote{Example~\ref{exa:aeclassescompose} explains that this holds for our two main examples of classical and quantum probability, while~\cite[Proposition~13.9]{Fr19} (whose statement is repeated in Corollary~\ref{cor:aeclassesofstatepreserving}) shows that this holds more generally when $\mC$ is causal (cf.\ Definition~\ref{defn:causal}).}
%endfootnote
If $\ov{g}:Z\stoch Y$ and $\ov{f}:Y\stoch X$ are disintegrations of 
$(X\xstoch{f}Y,I\xstoch{p}X,q:=f\circ p)$ and $(Y\xstoch{g}Z,I\xstoch{q}Y,r:=g\circ q)$, 
respectively, then $\ov{f}\circ\ov{g}$ is a disintegration of $(g\circ f,p,r)$ (all morphisms here are in $\mC$). 
%\en
\end{prop}

\bprf
The probability-preserving condition is immediate. 
The second condition for a disintegration follows from 
\be
\xy0;/r.25pc/:
(-6.5,0)*+{Y}="Y2";
(6.5,0)*+{Y}="Y1";
(0,12.5)*+{X}="X";
(22.5,-12.5)*+{Z}="Z1";
(-22.5,-12.5)*+{Z}="Z2";
{\ar@{~>}@/_1.5pc/"Z1";"X"_{\ov{f}\circ\ov{g}}};
{\ar@{~>}@/_1.5pc/"X";"Z2"_{g\circ f}};
{\ar@{~>}"Z1";"Y1"^{\ov{g}}};
{\ar@{~>}"Y1";"X"_{\ov{f}}};
{\ar@{~>}"X";"Y2"_{f}};
{\ar@{~>}"Y2";"Z2"^{g}};
{\ar"Y1";"Y2"^{\id_{Y}}};
{\ar"Z1";"Z2"^{\id_{Z}}};
{\ar@{=}(-3,-6.5);(3,-6.5)_{r}};
{\ar@{=}(-1.5,5);(1.5,5)_{q}};
\endxy
%\quad,
\ee
since composing a.e.\ equivalence classes of morphisms is well-defined.
\eprf

\begin{exa}[A.e.\ classes of positive unital maps compose]{exa:aeclassescompose}
The categories $\FinStoch$ and $\fdCAlgSPU$ are a.e.\ well-defined with respect to states (cf.\ Definition~\ref{defn:statepreservingcat}). The latter fact was first proved in~\cite[Proposition~3.106]{PaRu19}. In fact, it also holds for $\fdCAlgPU$~\cite[Theorem~3.113]{PaRu19}. As a consequence, Proposition~\ref{prop:functorialityofdisints} applies to these categories. 
\end{exa}

Bayesian inversion, on the other hand, is compositional in \emph{any} quantum Markov category. 

%\bn
%\label{prop:bayesfunctorial}
\begin{prop}[Bayesian inversion is compositional and symmetric]{prop:bayesfunctorial}
Let $\mM_{\text{\Yinyang}}$ be a quantum Markov category and let $\mC$ be an even subcategory of $\mM_{\text{\Yinyang}}$. Let $I\xstoch{p}X$, $I\xstoch{q}Y$, and $I\xstoch{r}Z$ be three states in $\mC$, let $X\xstoch{f}Y$ and $Y\xstoch{g}Z$ be state-preserving morphisms in $\mC$, so that $q=f\circ p$ and $r=g\circ q$.
\begin{enumerate}
\item
If $Z\xstoch{\ov{g}}Y$ and $Y\xstoch{\ov{f}}X$ are Bayes maps for 
$(f,p,q)$ and $(g,q,r)$ (so that $\ov{g}$ and $\ov{f}$ are assumed to only be in $\mM$),
respectively, then $\ov{f}\circ\ov{g}$ is a Bayes map for $(g\circ f,p)$.
\item
If $\mC$ is also $*$-preserving and $\ov{f}$ is a Bayesian inverse of $(f,p,q)$ (so that $\ov{f}$ is in $\mC$ now), then $f$ is a Bayesian inverse of $(\ov{f},q,p)$.%
%footnote
\footnote{The $*$-preserving assumption is crucial here because of Remark~\ref{rmk:hfpkgs}.}
%end footnote
\end{enumerate}
%\en
\end{prop}

\bprf
{\color{white}{you found me!}}

\begin{enumerate}
\item
The calculation 
\be
\vcenter{\hbox{%
\begin{tikzpicture}[font=\small]
\node[state] (q) at (0,0) {$r$};
\node[copier] (copier) at (0,0.3) {};
\coordinate (f) at (0.5,0.91) {};
\node[arrow box] (g) at (-0.5,0.95) {$\ov{g}$};
\node[arrow box] (e) at (-0.5,1.75) {$\ov{f}$};
\coordinate (X) at (0.5,2.3);
\coordinate (Y) at (-0.5,2.3);
\draw (q) to (copier);
\draw (copier) to[out=15,in=-90] (f);
\draw (f) to (X);
\draw (copier) to[out=165,in=-90] (g);
\draw (g) to (e);
\draw (e) to (Y);
\end{tikzpicture}}}
\quad
=
\quad
\vcenter{\hbox{%
\begin{tikzpicture}[font=\small]
\node[state] (p) at (0,0) {$q$};
\node[copier] (copier) at (0,0.3) {};
\node[arrow box] (f) at (0.5,0.95) {$g$};
\node[arrow box] (e) at (-0.5,1.75) {$\ov{f}$};
\coordinate (X) at (0.5,2.3);
\coordinate (Y) at (-0.5,2.3);
\draw (p) to (copier);
\draw (copier) to[out=15,in=-90] (f);
\draw (f) to (X);
\draw (copier) to[out=165,in=-90] (e);
\draw (e) to (Y);
\end{tikzpicture}}}
\quad
=
\quad
\vcenter{\hbox{%
\begin{tikzpicture}[font=\small]
\node[state] (p) at (0,0) {$q$};
\node[copier] (copier) at (0,0.3) {};
\node[arrow box] (f) at (-0.5,0.95) {$\ov{f}$};
\node[arrow box] (e) at (0.5,1.75) {$g$};
\coordinate (X) at (-0.5,2.3);
\coordinate (Y) at (0.5,2.3);
\draw (p) to (copier);
\draw (copier) to[out=165,in=-90] (f);
\draw (f) to (X);
\draw (copier) to[out=15,in=-90] (e);
\draw (e) to (Y);
\end{tikzpicture}}}
\quad
=
\quad
\vcenter{\hbox{%
\begin{tikzpicture}[font=\small]
\node[state] (q) at (0,0) {$p$};
\node[copier] (copier) at (0,0.3) {};
\coordinate (f) at (-0.5,0.91) {};
\node[arrow box] (g) at (0.5,0.95) {$f$};
\node[arrow box] (e) at (0.5,1.75) {$g$};
\coordinate (X) at (-0.5,2.3);
\coordinate (Y) at (0.5,2.3);
\draw (q) to (copier);
\draw (copier) to[out=165,in=-90] (f);
\draw (f) to (X);
\draw (copier) to[out=15,in=-90] (g);
\draw (g) to (e);
\draw (e) to (Y);
\end{tikzpicture}}}
\;\;\;.
\ee
proves the Bayes condition.  
\item
This is an immediate consequence of Lemma~\ref{prop:hfpkgs} applied to the diagram 
\be
\xy0;/r.25pc/:
(0,0)*+{I}="0";
(-10,7.5)*+{X}="X";
(10,7.5)*+{Y}="Y";
(10,-7.5)*+{Y}="W";
(-10,-7.5)*+{X}="Z";
{\ar@{~>}"X";"Y"^{f}};
{\ar"X";"Z"_{\id_{X}}};
{\ar@{~>}"W";"Z"^{\ov{f}}};
{\ar"W";"Y"_{\id_{Y}}};
{\ar@{~>}"0";"X"^{p}};
{\ar@{~>}"0";"W"_{q}};
\endxy
,
\ee
which entails 
\be
\label{eq:bayessymmetric}
\vcenter{\hbox{
\begin{tikzpicture}[font=\small]
\node[state] (q) at (0,0.0) {$q$};
\node[copier] (copier) at (0,0.3) {};
\node[arrow box] (g) at (-0.5,0.95) {$\ov{f}$};
\node at (-0.7,1.5) {$X$};
\node at (0.7,1.5) {$Y$};
\coordinate (X) at (0.5,1.6);
\coordinate (Y) at (-0.5,1.6);
\draw (q) to (copier);
\draw (copier) to[out=15,in=-90] (X);
\draw (copier) to[out=165,in=-90] (g);
\draw (g) to (Y);
\end{tikzpicture}}}
\quad=\quad
\vcenter{\hbox{
\begin{tikzpicture}[font=\small]
\node[state] (p) at (0,0) {$p$};
\node[copier] (copier) at (0,0.3) {};
\node[arrow box] (f) at (0.5,0.95) {$f$};
\node at (-0.7,1.5) {$X$};
\node at (0.7,1.5) {$Y$};
\coordinate (X) at (-0.5,1.6);
\coordinate (Y) at (0.5,1.6);
\draw (p) to (copier);
\draw (copier) to[out=165,in=-90] (X);
\draw (copier) to[out=15,in=-90] (f);
\draw (f) to (Y);
\end{tikzpicture}}}
\qquad
\iff
\qquad
\vcenter{\hbox{
\begin{tikzpicture}[font=\small]
\node[state] (p) at (0,0) {$p$};
\node[copier] (copier) at (0,0.3) {};
\node[arrow box] (f) at (-0.5,0.95) {$f$};
\node at (-0.7,1.5) {$Y$};
\node at (0.7,1.5) {$X$};
\coordinate (X) at (0.5,1.6);
\coordinate (Y) at (-0.5,1.6);
\draw (p) to (copier);
\draw (copier) to[out=15,in=-90] (X);
\draw (copier) to[out=165,in=-90] (f);
\draw (f) to (Y);
\end{tikzpicture}}}
\quad
=
\quad
\vcenter{\hbox{
\begin{tikzpicture}[font=\small]
\node[state] (q) at (0,0) {$q$};
\node[copier] (copier) at (0,0.3) {};
\node[arrow box] (g) at (0.5,0.95) {$\ov{f}$};
\node at (-0.7,1.5) {$Y$};
\node at (0.7,1.5) {$X$};
\coordinate (X) at (-0.5,1.6);
\coordinate (Y) at (0.5,1.6);
\draw (q) to (copier);
\draw (copier) to[out=165,in=-90] (X);
\draw (copier) to[out=15,in=-90] (g);
\draw (g) to (Y);
\end{tikzpicture}}}
.
\ee
The implication towards the right proves the proposition. \qedhere
\end{enumerate}
\eprf

The following remark explains that Bayesian inversion defines a dagger functor on a.e.\ equivalence classes. This is essentially Remark~13.9 in~\cite{Fr19}, but now it includes the non-commutative setting. 

\begin{rmk}[Bayesian inversion as a dagger functor]{rmk:bayesdagger}
Let $\mC\subseteq\mathcal{M}_{\text{\Yinyang}}$ be an even $*$-preserving subcategory of a quantum Markov category.
Let $\mathfrak{B}I_{/\mC}$ be the subcategory of $I_{/\mC}$ (recall Definition~\ref{defn:statepreservingcat}) consisting of the same objects as $I_{/\mC}$ but whose morphisms consist of all Bayesian invertible morphisms (whose Bayesian inverses are also in $\mC$).

Now, consider two a.e.\ equivalent pairs of composable morphisms $f,f':(X,I\xstoch{p}X)\to(Y,I\xstoch{q}Y)$ and $g,g':(Y,I\xstoch{q}Y)\to(Z,I\xstoch{r}Z)$, i.e.\ $f\underset{\raisebox{.6ex}[0pt][0pt]{\scriptsize$p$}}{=}f'$ and $g\underset{\raisebox{.6ex}[0pt][0pt]{\scriptsize$q$}}{=}g'$, in $\mathfrak{B}I_{/\mC}$. Then $g\circ f\underset{\raisebox{.6ex}[0pt][0pt]{\scriptsize$p$}}{=}g'\circ f'$ follows from 
\[
\vcenter{\hbox{%
\begin{tikzpicture}[font=\small]
\node[state] (q) at (0,0) {$p$};
\node[copier] (copier) at (0,0.3) {};
\coordinate (f) at (-0.5,0.91) {};
\node[arrow box] (g) at (0.5,0.95) {$f$};
\node[arrow box] (e) at (0.5,1.75) {$g$};
\coordinate (X) at (-0.5,2.3);
\coordinate (Y) at (0.5,2.3);
\draw (q) to (copier);
\draw (copier) to[out=165,in=-90] (f);
\draw (f) to (X);
\draw (copier) to[out=15,in=-90] (g);
\draw (g) to (e);
\draw (e) to (Y);
\end{tikzpicture}}}
\quad
=
\quad
\vcenter{\hbox{%
\begin{tikzpicture}[font=\small]
\node[copier] (copier) at (0,0.3) {};
\node[arrow box] (h) at (-0.5,0.95) {$\ov{f}$};
\node[arrow box] (f) at (0.5,0.95) {$g$};
\node[state] (p) at (0,-0.1) {$q$};
\coordinate (X) at (-0.5,1.6);
\coordinate (Y) at (0.5,1.6);
\draw (p) to (copier);
\draw (copier) to[out=165,in=-90] (h);
\draw (h) to (X);
\draw (copier) to[out=15,in=-90] (f);
\draw (f) to (Y);
\end{tikzpicture}}}
\quad
\overset{g\underset{\raisebox{.6ex}[0pt][0pt]{\tiny$q$}}{=}g'}{=\joinrel=\joinrel=}
\quad
\vcenter{\hbox{%
\begin{tikzpicture}[font=\small]
\node[copier] (copier) at (0,0.3) {};
\node[arrow box] (h) at (-0.5,0.95) {$\ov{f}$};
\node[arrow box] (f) at (0.5,0.95) {$g'$};
\node[state] (p) at (0,-0.1) {$q$};
\coordinate (X) at (-0.5,1.6);
\coordinate (Y) at (0.5,1.6);
\draw (p) to (copier);
\draw (copier) to[out=165,in=-90] (h);
\draw (h) to (X);
\draw (copier) to[out=15,in=-90] (f);
\draw (f) to (Y);
\end{tikzpicture}}}
\quad
\overset{\text{Prop~\ref{prop:bayesfunctorial}}}{=\joinrel=\joinrel=\joinrel=\joinrel=\joinrel=}
\quad
\vcenter{\hbox{%
\begin{tikzpicture}[font=\small]
\node[state] (q) at (0,0) {$p$};
\node[copier] (copier) at (0,0.3) {};
\coordinate (f) at (-0.5,0.91) {};
\node[arrow box] (g) at (0.5,0.95) {$f$};
\node[arrow box] (e) at (0.5,1.75) {$g'$};
\coordinate (X) at (-0.5,2.3);
\coordinate (Y) at (0.5,2.3);
\draw (q) to (copier);
\draw (copier) to[out=165,in=-90] (f);
\draw (f) to (X);
\draw (copier) to[out=15,in=-90] (g);
\draw (g) to (e);
\draw (e) to (Y);
\end{tikzpicture}}}
\quad
\overset{f\underset{\raisebox{.6ex}[0pt][0pt]{\tiny$p$}}{=}f'}{=\joinrel=\joinrel=}
\quad
\vcenter{\hbox{%
\begin{tikzpicture}[font=\small]
\node[state] (q) at (0,0) {$p$};
\node[copier] (copier) at (0,0.3) {};
\coordinate (f) at (-0.5,0.91) {};
\node[arrow box] (g) at (0.5,0.95) {$f'$};
\node[arrow box] (e) at (0.5,1.75) {$g'$};
\coordinate (X) at (-0.5,2.3);
\coordinate (Y) at (0.5,2.3);
\draw (q) to (copier);
\draw (copier) to[out=165,in=-90] (f);
\draw (f) to (X);
\draw (copier) to[out=15,in=-90] (g);
\draw (g) to (e);
\draw (e) to (Y);
\end{tikzpicture}}}
\quad,
\]
where we have used a Bayesian inverse $\ov{f}$ for $(f,p,q)$ in the first equality. 
By Lemma~\ref{lem:bayesaeunique}, if $f$ is $p$-a.e.\ equivalent to $f'$ and $f$ has $\ov{f}$ as a Bayesian inverse, then $\ov{f}$ is also a Bayesian inverse of $f'$.
It is even easier to check that the identity is a Bayesian inverse of the identity for any states.
Hence, taking a.e.\ equivalence classes of morphisms in $\mathfrak{B}I_{/\mC}$ defines a category, which will be denoted by $\mathfrak{B}^{\mathrm{ae}}I_{/\mC}$. It consists of a.e.\ equivalence classes of even $*$-preserving Bayesian invertible morphisms. 
These facts together with Proposition~\ref{prop:bayesfunctorial} say that 
Bayesian inversion defines a dagger functor on $\mathfrak{B}^{\mathrm{ae}}I_{/\mC}$.%
%footnote
\footnote{Our original proof of this claim was initially valid for our examples of commutative and non-commutative probability, where we had explicitly shown that a.e.\ equivalence classes of PU maps compose~\cite{PaRu19}. This is to be contrasted with Fritz proof, which was valid for all causal (classical) Markov categories (cf.\ Definition~\ref{defn:causal}). Since $\fdCAlgSPU$ is causal in an appropriate sense (cf. Proposition~\ref{prop:SPUcausal}), Fritz' proof can be adapted to this case (cf.\ Corollary~\ref{cor:aeclassesofstatepreserving}). However, we have actually found that $\fdCAlgPU$ is not causal in this sense, even though it has the property that a.e.\ equivalence classes of morphisms compose. Hence, our result simultaneously extends Fritz' result in two respects (cf.\ Question~\ref{ques:diagramaePUcompose}).}
%end footnote

This is to be contrasted with the notion of having a disintegration. Even in categories where composition of state-preserving a.e.\ classes is well-defined so that disintegrations are compositional when they exist, if $(f,p,q)$ has a disintegration $\ov{f}$, it is almost never the case that $f$ is a disintegration of $(\ov{f},q,p)$. More on this will be explained in Section~\ref{sec:modposcaus} when we discover that having a disintegration of $(f,p,q)$ imposes constraints on $f$ (namely, $p$-a.e.\ determinism).
\end{rmk}

\vspace{-1mm}
\begin{prop}[Invertible implies disintegrable and Bayesian invertible]{prop:invertimpliesdisintandbayes}
Let $\mC$ be an S-positive subcategory of a quantum CD category $\mM_{\text{\Yinyang}}$. 
Let $I\xstoch{p}X$ and $X\xstoch{f}Y$ be unital morphisms in $\mM$ and set $q:=f\circ p$. 
\begin{enumerate}[i.]
\itemsep0pt
\item
If $f$ is invertible, with inverse in $\mM$, then $f^{-1}$ is a disintegration of $(f,p,q)$ in $\mM$. 
\item
\label{item:inverseSpositiveimpliesBayes}
Assuming that $f$ and $p$ are now in $\mC$, if $f$ is invertible, with inverse also in $\mC$, then $f^{-1}$ is a Bayesian inverse of $(f,p,q)$. 
\end{enumerate}
\end{prop}

\bprf
Note that even if $\mM_{\text{\Yinyang}}$ is a quantum CD category, if $f$ is unital, then $f^{-1}$ is unital as well because
\be
\vcenter{\hbox{%
\begin{tikzpicture}[font=\small]
\node[arrow box] (c) at (0,0) {$f^{-1}$};
\node[discarder] (d) at (0,0.5) {};
\draw  (c) to (d);
\draw (c) to (0,-0.5);
\end{tikzpicture}}}
\quad
=
\quad
\vcenter{\hbox{%
\begin{tikzpicture}[font=\small]
\node[arrow box] (b) at (0,-0.8) {$f^{-1}$};
\node[arrow box] (c) at (0,0) {$f$};
\node[discarder] (d) at (0,0.5) {};
\draw  (b) to (c);
\draw  (c) to (d);
\draw (b) to (0,-1.4);
\end{tikzpicture}}}
\quad
=
\quad
\vcenter{\hbox{%
\begin{tikzpicture}[font=\small]
\node[discarder] (d) at (0,0) {};
\node at (0,0.2) {};
\draw (d) to (0,-0.5);
\end{tikzpicture}}}
\quad
.
\ee

\begin{enumerate}[i.]
\itemsep0pt
\item
The disintegration conditions for $f^{-1}$
follow from
\be
\vcenter{\hbox{%
\begin{tikzpicture}[font=\small]
\node[state] (q) at (0,-0.9) {$q$};
\node[arrow box] (g) at (0,-0.3) {$f^{-1}$};
\coordinate (Y) at (0,0.3);
\draw (q) to (g);
\draw (g) to (Y);
\end{tikzpicture}}}
\quad
=
\quad
\vcenter{\hbox{%
\begin{tikzpicture}[font=\small]
\node[state] (p) at (0,-1.7) {$p$};
\node[arrow box] (q) at (0,-1.1) {$f$};
\node[arrow box] (g) at (0,-0.3) {$f^{-1}$};
\coordinate (Y) at (0,0.3);
\draw (p) to (q);
\draw (q) to (g);
\draw (g) to (Y);
\end{tikzpicture}}}
\quad
=
\quad
\vcenter{\hbox{%
\begin{tikzpicture}[font=\small]
\node[state] (p) at (0,-0.9) {$p$};
\coordinate (X) at (0,-0.3);
\draw (p) to (X);
\end{tikzpicture}}}
\qquad
\text{ and }
\qquad
\vcenter{\hbox{%
\begin{tikzpicture}[font=\small]
\node[state] (q) at (0,0) {$q$};
\node[copier] (copier) at (0,0.3) {};
\coordinate (f) at (-0.5,0.91) {};
\node[arrow box] (g) at (0.5,0.95) {$f^{-1}$};
\node[arrow box] (e) at (0.5,1.75) {$f$};
\coordinate (X) at (-0.5,2.3);
\coordinate (Y) at (0.5,2.3);
\draw (q) to (copier);
\draw (copier) to[out=165,in=-90] (f);
\draw (f) to (X);
\draw (copier) to[out=15,in=-90] (g);
\draw (g) to (e);
\draw (e) to (Y);
\end{tikzpicture}}}
\quad
=
\quad
\vcenter{\hbox{%
\begin{tikzpicture}[font=\small]
\node[state] (q) at (0,-0.2) {$q$};
\node[copier] (copier) at (0,0.3) {};
\coordinate (X) at (-0.5,1.31);
\coordinate (X2) at (0.5,1.31);
\draw (q) to (copier);
\draw (copier) to[out=165,in=-90] (X);
\draw (copier) to[out=15,in=-90] (X2);
\end{tikzpicture}}}
\quad.
\ee

\item
The Bayes condition for $f^{-1}$ follows from 
\be
\vcenter{\hbox{%
\begin{tikzpicture}[font=\small]
\node[state] (q) at (0,0) {$q$};
\node[copier] (copier) at (0,0.3) {};
\node[arrow box] (g) at (-0.5,0.95) {$f^{-1}$};
\coordinate (X) at (-0.5,1.5);
\coordinate (Y) at (0.5,1.5);
\draw (q) to (copier);
\draw (copier) to[out=165,in=-90] (g);
\draw (copier) to[out=15,in=-90] (Y);
\draw (g) to (X);
\end{tikzpicture}}}
\quad
=
\quad
\vcenter{\hbox{%
\begin{tikzpicture}[font=\small,xscale=-1]
\node[state] (q) at (0,0) {p};
\node[copier] (copier) at (0,1.2) {};
\node[arrow box] (e2) at (0.5,1.75) {$f^{-1}$};
\node[arrow box] (g) at (0,0.6) {$f$};
\coordinate (X) at (-0.5,2.3);
\coordinate (Y) at (0.5,2.3);
\draw (q) to (g);
\draw (g) to (copier);
\draw (copier) to[out=165,in=-90] (X);
\draw (e2) to (Y);
\draw (copier) to[out=15,in=-90] (e2);
\end{tikzpicture}}}
\quad
\overset{\text{Lem~\ref{lem:positivityinvertibilityimpliesdeterminism}}}{=\joinrel=\joinrel=\joinrel=\joinrel=}
\quad
\vcenter{\hbox{%
\begin{tikzpicture}[font=\small,xscale=-1]
\node[state] (q) at (0,0) {p};
\node[copier] (copier) at (0,0.3) {};
\node[arrow box] (g2) at (-0.5,0.95) {$f$};
\node[arrow box] (e2) at (0.5,1.75) {$f^{-1}$};
\node[arrow box] (g) at (0.5,0.95) {$f$};
\coordinate (X) at (-0.5,2.3);
\coordinate (Y) at (0.5,2.3);
\draw (q) to (copier);
\draw (copier) to[out=165,in=-90] (g2);
\draw (g2) to (X);
\draw (e2) to (Y);
\draw (copier) to[out=15,in=-90] (g);
\draw (g) to (e2);
\end{tikzpicture}}}
\quad
=
\quad
\vcenter{\hbox{%
\begin{tikzpicture}[font=\small]
\node[state] (p) at (0,0) {$p$};
\node[copier] (copier) at (0,0.3) {};
\node[arrow box] (f) at (0.5,0.95) {$f$};
\coordinate (X) at (-0.5,1.5);
\coordinate (Y) at (0.5,1.5);
\draw (p) to (copier);
\draw (copier) to[out=165,in=-90] (X);
\draw (copier) to[out=15,in=-90] (f);
\draw (f) to (Y);
\end{tikzpicture}}}
\quad.
\ee
since $f$ is deterministic by S-positivity. \qedhere
\end{enumerate}
\eprf

\begin{rmk}[S-positivity is crucial in part~\ref{item:inverseSpositiveimpliesBayes}.\ of Proposition~\ref{prop:invertimpliesdisintandbayes}]{rmk:disintdnibayesforpu}
To see that S-positivity is a crucial assumption in part~\ref{item:inverseSpositiveimpliesBayes}.\ of Proposition~\ref{prop:invertimpliesdisintandbayes}, consider $\fdCAlgPU$. Take $\mA:=\mM_{m}(\C)=:\mB$, let $\mA\xstoch{\omega:=\tr(\rho\;\cdot\;)}\C$ be any state, and take $\mB\xstoch{F:=T}\mA$ to be the transpose map. Then (by arguments similar to those in Remark~\ref{rmk:supportconditions}), taking $B\in\mB$ to compute the Hilbert--Schmidt adjoint of $F$ gives
\[
\tr\big(F^*(\rho)B\big)
=\tr\big(\rho F(B)\big)
=\tr(\rho B^{T})
=\tr\big((B\rho^{T})^T\big)
=\tr(B\rho^{T})
=\tr(\rho^{T}B),
\]
where the invariance of trace under the transpose was used in the fourth equality. Hence, $F^*=F$ and $\xi:=\tr(\rho^{T}\;\cdot\;)$ is the pullback state. Furthermore, since $F$ is invertible ($F^{-1}=F$), $G:=F$ is a disintegration of $(F,\omega,\xi)$. However, the Bayes condition fails because given $A,B\in\mM_{m}(\C)$, one has
\[
\xi\big(G(A)B\big)
=\tr(\rho^{T}A^{T}B)
=\tr\big((B^TA\rho)^T\big)
=\tr(B^{T}A\rho)
=\tr(\rho B^{T}A)
\ne
\tr(\rho AB^{T})
=\omega\big(AF(B)\big).
\]
On the other hand, we will later see that every SPU disintegration of an SPU map is indeed a Bayesian inverse in Corollary~\ref{cor:summary}. 
\end{rmk}

We now state another fact that provides an indication of how disintegrations are related to Bayesian inverses. 

%\bt
%\label{thm:Bayesianinverseofdeterministicisadisint}
\begin{prop}[Bayes maps for a.e.\ deterministic morphisms are disintegrations]{thm:Bayesianinverseofdeterministicisadisint}
Let $\mM_{\text{\Yinyang}}$ be a quantum Markov category, 
let $I\xstoch{p}X$ be a state in $\mM$, let $X\xstoch{f}Y$ be \emph{either} a left or a right $p$-a.e.\ deterministic map in $\mM$, and set $q:=f\circ p$. If $f$ has a  \emph{left} (right) Bayes map $Y\xstoch{g}X$, then $g$ is a \emph{left} (right) disintegration of $(f,p,q)$ in $\mM$.  
In particular, if all morphisms are $*$-preserving, then all adjectives `left' and `right' may be dropped. 
\end{prop}
%\et

\bprf
The state-preserving condition of a disintegration follows from Lemma~\ref{lem:Bayesianinversespreservestates}. 
For concreteness, suppose that $f$ is right $p$-a.e.\ deterministic (we will later explain how the calculation is modified if $f$ is left $p$-a.e.\ deterministic). 
The second condition of a left disintegration follows from the string diagram calculation
\be
\label{eq:binvofdetimpdisint}
\vcenter{\hbox{%
\begin{tikzpicture}[font=\small]
\node[state] (q) at (0,-0.1) {$q$};
\node[copier] (copier) at (0,0.3) {};
\coordinate (X) at (-0.5,1.21);
\coordinate (X2) at (0.5,1.21);
\draw (q) to (copier);
\draw (copier) to[out=165,in=-90] (X);
\draw (copier) to[out=15,in=-90] (X2);
\end{tikzpicture}}}
\;\;
=
\;\;
\vcenter{\hbox{%
\begin{tikzpicture}[font=\small]
\node[state] (p) at (0,-0.9) {$p$};
\node[arrow box] (f) at (0,-0.3) {$f$};
\node[copier] (copier) at (0,0.3) {};
\coordinate (X) at (-0.5,0.91);
\coordinate (X2) at (0.5,0.91);
\draw (p) to (f);
\draw (f) to (copier);
\draw (copier) to[out=165,in=-90] (X);
\draw (copier) to[out=15,in=-90] (X2);
\end{tikzpicture}}}
\;\;
=
\;\;
\vcenter{\hbox{
\begin{tikzpicture}[font=\small]
\node[state] (p) at (0,0) {$p$};
\node[copier] (c) at (0,0.3) {};
\node[copier] (c2) at (0.5,1.45) {};
\node[arrow box] (f) at (0.5,0.95) {$f$};
\node[discarder] (X) at (-0.7,1.65) {};
\coordinate (Y1) at (1,1.95);
\coordinate (Y2) at (0,1.95);
\draw (p) to (c);
\draw (c) to[out=165,in=-90] (X);
\draw (c) to[out=15,in=-90] (f);
\draw (c2) to[out=15,in=-90] (Y1);
\draw (c2) to[out=165,in=-90] (0,1.95);
\draw (f) to (c2);
\end{tikzpicture}}}
\;\;
=
\;\;
\vcenter{\hbox{
\begin{tikzpicture}[font=\small]
\node[state] (p) at (0,0) {$p$};
\node[copier] (c) at (0,0.3) {};
\node[copier] (c2) at (0.5,0.75) {};
\node[arrow box] (f) at (1,1.4) {$f$};
\node[arrow box] (e) at (0,1.4) {$f$};
\node[discarder] (X) at (-0.7,1.65) {};
\coordinate (Y1) at (1,1.95);
\coordinate (Y2) at (0,1.95);
\draw (p) to (c);
\draw (c) to[out=165,in=-90] (X);
\draw (c) to[out=15,in=-90] (c2);
\draw (c2) to[out=15,in=-90] (f);
\draw (c2) to[out=165,in=-90] (e);
\draw (f) to (Y1);
\draw (e) to (Y2);
\end{tikzpicture}}}
\;\;
=
\;\;
\vcenter{\hbox{%
\begin{tikzpicture}[font=\small]
\node[state] (p) at (0,0) {$p$};
\node[copier] (copier) at (0,0.3) {};
\node[arrow box] (f) at (-0.5,0.95) {$f$};
\node[arrow box] (e) at (0.5,0.95) {$f$};
\coordinate (X) at (-0.5,2.2);
\coordinate (Y) at (0.5,2.2);
\draw (p) to (copier);
\draw (copier) to[out=165,in=-90] (f);
\draw (f) to (X);
\draw (copier) to[out=15,in=-90] (e);
\draw (e) to (Y);
\end{tikzpicture}}}
\;\;
%\overset{\text{Prop~\ref{prop:bayesfunctorial}}}{=\joinrel=\joinrel=\joinrel=\joinrel=\joinrel=}
=
\;\;
\vcenter{\hbox{%
\begin{tikzpicture}[font=\small,xscale=-1]
\node[state] (q) at (0,0) {$q$};
\node[copier] (copier) at (0,0.3) {};
\coordinate (f) at (-0.5,0.91) {};
\node[arrow box] (g) at (0.5,0.95) {$g$};
\node[arrow box] (e) at (0.5,1.75) {$f$};
\coordinate (X) at (-0.5,2.3);
\coordinate (Y) at (0.5,2.3);
\draw (q) to (copier);
\draw (copier) to[out=165,in=-90] (f);
\draw (f) to (X);
\draw (copier) to[out=15,in=-90] (g);
\draw (g) to (e);
\draw (e) to (Y);
\end{tikzpicture}}}
\;\;, 
\ee
where right $p$-a.e.\ determinism of $f$ was used in the third equality and the \emph{left} Bayes condition was used in the last equality. Notice that if $f$ was left $p$-a.e.\ deterministic instead, one could extend the grounding operation to the right, instead of the left, as was done in the second equality. 
\eprf

Proposition~\ref{thm:Bayesianinverseofdeterministicisadisint} is useful because it allows one to use the simple formula for Bayes maps from Example~\ref{exa:Bayesmapsformatrixalgebras} (and more generally~\cite{PaRuBayes}) to construct \emph{linear} disintegrations. Although not quantum operations, such maps can still be useful in inference and Bayesian updating~\cite{PaJeffrey}.

\begin{exa}[Regular conditional probabilities]{exa:RCP}
Let $(X,\S,\mu)$ and $(Y,\W,\nu)$ be measure spaces and let 
$X\xrightarrow{f}Y$ be a measure-preserving map. 
A \define{regular conditional probability} (cf.\ \cite[Definition~2.1]{LFR04})
is a transition kernel $Y\xstoch{g}X$ for which 
there exists a $\nu$-null set $N\in\W$ such that 
$g_{y}$ is a probability measure for all $y\in Y\setminus N$ and 
\[
\mu\big(A\cap f^{-1}(B)\big)=\int_{B}g_{y}(A)\,d\nu(y)
\qquad\forall\;A\in\S\text{ and }\forall\;B\in\W.
\]
Writing this equation string-diagrammatically immediately shows that this is precisely the Bayes condition, namely $g$ is a Bayesian inverse of $(f,\mu,\nu)$ in $\Stoch$. 
Indeed, viewing $f$ as a Markov kernel, $\int_{A}f_{x}(B)\,d\mu(x)=\int_{A}\chi_{f^{-1}(B)}(x)\,d\mu(x)=\int_{A\cap f^{-1}(B)}d\mu=\mu\big(A\cap f^{-1}(B)\big)$, where $\chi$ is used to denote the characteristic function. 
In~\cite[Proposition~A.29]{PaRu19}, it was shown (using the explicit language of measure theory) that the notion of a regular conditional probability is equivalent to that of a disintegration. By Proposition~\ref{thm:Bayesianinverseofdeterministicisadisint}, we get a simple string-diagrammatic proof that a regular conditional probability is a disintegration%
%footnote
\footnote{Technically, the regular conditional probability is only a.e.\ unital.}
%end footnote
 since $f$ is deterministic. A string-diagrammatic proof that a disintegration is a regular conditional probability will follow from Theorem~\ref{thm:aemodbayesdisint} and the fact that $\Stoch$ is a.e.\ modular, as defined in Definition~\ref{defn:aemodularity} (see \cite[Example~13.19]{Fr19} for a proof that $\Stoch$ satisfies this condition). 

Therefore, since the Bayes condition makes sense even if $f$ is a Markov kernel, a Bayesian inverse can be viewed as a generalization of the notion of a regular conditional probability. There are also other generalizations of regular conditional probabilities in the literature. For example, a \emph{conditional distribution of a Markov kernel given another} in~\cite[Definition 4]{No13} is a slight generalization of Bayesian inversion where one includes an extra morphism (to reproduce our Bayes condition, set $M_{1}=\id$ in~\cite[Definition 4]{No13}). 
\end{exa}

\vspace{-2mm}
\begin{cor}[Bayesian inverses of a.e.\ deterministic maps in $\fdCAlgSPU$]{cor:bayesdetCAlg}
In the category $\fdCAlgSPU$, 
let $\mA\xstoch{\w}\C$ be a state on $\mA$, let $\mB\xstoch{F}\mA$ be a $\w$-a.e.\ deterministic map, and set $\xi:=\w\circ F$. If $\mA\xstoch{G}\mB$ is a Bayesian inverse of $(F,\omega,\xi)$, then $G$ is a disintegration of $(F,\w,\xi)$.
\end{cor}

A partial converse to Proposition~\ref{thm:Bayesianinverseofdeterministicisadisint} and Corollary~\ref{cor:bayesdetCAlg} will be provided in the next section. 

%\vspace{-1mm}
\begin{rmk}[Replacing states with arbitrary morphisms]{rmk:statestomapsfordisint}
As briefly mentioned in Definitions~\ref{defn:disintegration} and~\ref{defn:bayesianinverse}, one can remove the restriction that $p$, $q$, and $r$ are states. Indeed, one can replace $p$ with an arbitrary morphism $\Theta\xstoch{p}X$ (and similarly for $q$ and $r$) provided that it is still in the same subcategory as needed for the definitions. One then obtains analogous results to those obtained specifically for states. However, our understanding of Bayesian inversion and disintegration in the category $\fdCAlgCPU$ in this more general case is not yet complete. See Question~\ref{ques:moregeneraldisintsandbayes} for details. 
\end{rmk}

\vspace{-1mm}
\begin{ques}[Characterizing generalized disintegrations and Bayesian inverses]{ques:moregeneraldisintsandbayes}
Our earlier work provides necessary and sufficient conditions for disintegrations and Bayesian inverses to exist in the category $\fdCAlgCPU$ with respect to (CPU) \emph{states}~\cite{PaRu19,PaRuBayes}. However, if the states are replaced with arbitrary CPU maps, we do not yet have more concrete necessary and sufficient conditions for their existence (although some preliminary results are known). For example, we have two open problems. 
\begin{enumerate}
\itemsep0pt
\item
Given CPU maps $\mA\xstoch{\omega}\mC$ and $\mB\xstoch{\xi}\mC$ together with a deterministic map $\mB\xrightarrow{F}\mA$ (or more generally, an $\omega$-a.e.\ deterministic map) such that $\xi=\w\circ F$, find convenient necessary and sufficient conditions for the existence of a CPU map $\mA\xstoch{G}\mB$ such that $\w=\xi\circ G$ and $G\circ F\aeequals{\xi}\id_{\mC}$. 
\item
Given CPU maps $\mA\xstoch{\omega}\mC$ and $\mB\xstoch{\xi}\mC$ together with a CPU map $\mB\xstoch{F}\mA$ such that $\xi=\w\circ F$, find convenient necessary and sufficient conditions for the existence of a CPU map $\mA\xstoch{G}\mB$ such that $\xi\big(G(A)B\big)=\omega\big(AF(B)\big)$ for all $A\in\mA$ and $B\in\mB$. 
\end{enumerate}
Such generalizations appear in classical probability in the form of parametrizing probability measures~\cite[Section~14]{Fr19}. 
In the quantum setting, special cases of these concepts (where $\mC$ is a commutative $C^*$-algebra, for example) are known in the non-commutative statistics literature, some of which will be explored in Example~\ref{exa:Umegakisufficientstatistic}. 
What is the appropriate interpretation of these more general notions in non-commutative statistics? 
\end{ques}

\vspace{-1mm}
\begin{rmk}[$\fdCAlgCPU$ does not have conditionals]{rmk:EPRconditionals}
Although we do not discuss conditionals in this work, we briefly mention that the subcategory $\fdCAlgCPU$ does not have conditional distributions in the sense of~\cite[Definition~11.1]{Fr19} (and therefore does not have conditionals in the sense of~\cite[Definition~11.5]{Fr19}).%
%footnote
\footnote{This is also consistent with our observation that $\fdCAlgCPU$ is not strictly positive (Example~\ref{exa:strictpositive}) and Fritz' observation~\cite[Lemma~13.18]{Fr19}, which states: ``if $\mC$ is a classical Markov category with conditionals, then $\mC$ is strictly positive.''}
%end footnote
In the setting of $C^*$-algebras, a ``conditional distribution'' is replaced by a PU map~\cite{Pa17}. One might suspect the lack of conditional distributions since the copy map is not positive and therefore it is rare that using it to construct a joint functional will give a positive functional~\cite[Remark~5.96]{PaRuBayes}. However, if one already has a joint state, can a conditional be found that gives rise to that joint state? We include a simple counterexample here to illustrate that $\fdCAlgCPU$ does not have conditional distributions in this sense. Let $\rho$ be the EPR density matrix on $\mA\otimes\mB:=\mM_{2}(\C)\otimes\mM_{2}(\C)\cong\mM_{4}(\C)$ given by
\[
\rho=\frac{1}{2}\begin{bmatrix}0&0&0&0\\0&1&-1&0\\0&-1&1&0\\0&0&0&0\end{bmatrix}.
\]
Set $\omega:=\tr(\rho\;\cdot\;)$ and $\omega_{\mA}:=\omega\circ!_{\mA}$ (the state obtained by computing the partial trace of $\rho$ over $\mB$). In this case, $\omega_{\mA}(A)=\frac{1}{2}\tr(A)$ for all $A\in\mA$. For $\omega$ to have a conditional distribution, there must exist a CPU map $\mB\xstoch{F}\mA$ such that $\omega(A\otimes B)=\omega_{\mA}\big(AF(B)\big)$ for all $A\in\mA$ and $B\in\mB$. One can show that the only such $F$ is the map 
\[
\mB\ni\begin{bmatrix}b_{11}&b_{12}\\b_{21}&b_{22}\end{bmatrix}\xmapsto{F}
\begin{bmatrix}b_{22}&-b_{12}\\-b_{21}&b_{11}\end{bmatrix}. 
\]
This map is unital and in fact positive. However, it is not completely positive since its Choi matrix (cf.\ \cite[Theorem~2]{Ch75}) has eigenvalues $-1$ and $+1$. In fact, the map $F$ can be expressed as 
\[
\mB\ni B\mapsto F(B)=\begin{bmatrix}0&1\\-1&0\end{bmatrix}B^{T}\begin{bmatrix}0&-1\\1&0\end{bmatrix}
\]
showing explicitly that it is PU but not CPU. 
It is interesting that one actually obtains a positive map $F$ in this example.%
%footnote
\footnote{However, this is not always the case for all states $\omega$ on $\mA\otimes\mB$.}
%end footnote
It is also interesting to note that this positive map provides the correct inference update rule reproducing the EPR correlations. For example, if Alice measures spin up in some direction and Bob applies the map $F$, then the induced state on Bob's system corresponds to the spin down state. 
Further details and examples are deferred to forthcoming work~\cite{PaJeffrey}. 
\end{rmk}

%\pagebreak
%%%%%%%%%%%%%%%%%%%%%%%%%%%%%%%%%%%%%%
\section[A.e.\ modularity, strict positivity, and causality]{a021}
\label{sec:modposcaus}
\vspace{-12mm}
\noindent
\begin{tikzpicture}
\coordinate (L) at (-8.75,0);
\coordinate (R) at (8.75,0);
\draw[line width=2pt,orange!20] (L) -- node[anchor=center,rectangle,fill=orange!20]{\strut \Large \textcolor{black}{\textbf{8\;\; A.e.\ modularity, strict positivity, and causality}}} (R);
\end{tikzpicture}
%\vspace{1mm}
%%%%%%%%%%%%%%%%%%%%%%%%%%%%%%%%%%%%%%

Here we prove our main results relating Bayesian inversion to disintegration. 
We also review a variety of additional axioms for Markov categories (besides S-positivity) introduced by Fritz~\cite{Fr19}, and we prove which ones hold for $\CAlgSPU$. 
In particular, we bring back Fritz' initial axiom of positivity, which appeared in an earlier version of his work, but which was subsequently dropped (since it was subsumed by a stronger axiom that was valid for many of the examples he considered). To avoid conflating the terminology, we call this dropped axiom a.e.\ modularity (since it describes an almost everywhere version of the modularity properties of conditional expectations from the theory of operator algebras~\cite{Um54,To57}). Briefly, our main results are that $\fdCAlgSPU$ is a.e.\ modular and causal, but it is \emph{not} strictly positive (and neither is $\fdCAlgCPU$). We will also see that a.e.\ modularity is the diagrammatic axiom that relates inverses, deterministic morphisms, disintegrations, and Bayesian inverses in Theorem~\ref{thm:aemodbayesdisint}, which is valid in the non-commutative setting as well (Theorem~\ref{thm:SPUmodular}). The results in this section lay the foundations for instantiating theorems proved synthetically to the non-commutative setting, though we do not explore all these theorems in detail here (we only briefly mention a quantum Fisher--Neyman factorization theorem).  

%\vspace{-1mm}
\begin{defn}[A.e.\ modularity]{defn:aemodularity}
An even 
subcategory $\mC$ of a quantum Markov category is \define{left a.e.\ modular} iff for all $f,g,q$ in $\mC$ such that%
%footnote
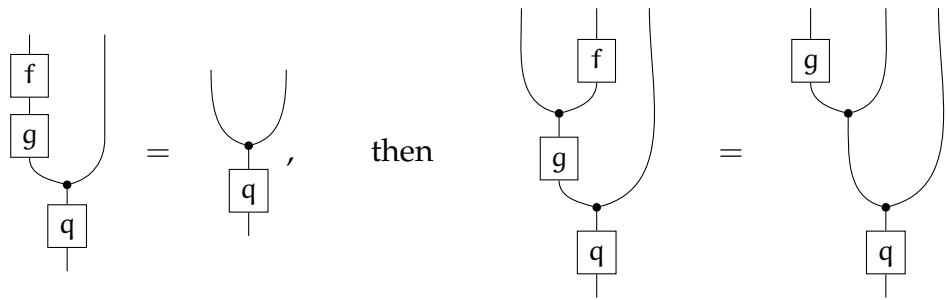
\footnote{Notice that the diagram on the left says that $g$ is a left disintegration of $(f,p:=g\circ q,q)$. This definition is a slight generalization of the previous definition of disintegration to the case where $q$ need not be a state (see Remark~\ref{rmk:statestomapsfordisint}).}
%end footnote
\vspace{-1mm}
\[
\vcenter{\hbox{%
\begin{tikzpicture}[font=\small]
\node[arrow box] (q) at (0,-0.25) {$q$};
\node[copier] (copier) at (0,0.3) {};
\coordinate (f) at (0.5,0.91) {};
\node[arrow box] (g) at (-0.5,0.95) {$g$};
\node[arrow box] (e) at (-0.5,1.75) {$f$};
\coordinate (X) at (0.5,2.3);
\coordinate (Y) at (-0.5,2.3);
\draw (0,-0.85) to (q);
\draw (q) to (copier);
\draw (copier) to[out=15,in=-90] (f);
\draw (f) to (X);
\draw (copier) to[out=165,in=-90] (g);
\draw (g) to (e);
\draw (e) to (Y);
\end{tikzpicture}}}
%%%%%%%%%%%%%%%%%%%
\quad=\quad
%%%%%%%%%%%%%%%%%%%
\vcenter{\hbox{%
\begin{tikzpicture}[font=\small]
\node[arrow box] (q) at (0,-0.3) {$q$};
\node[copier] (copier) at (0,0.3) {};
\coordinate (X) at (-0.5,1.31);
\coordinate (X2) at (0.5,1.31);
\draw (0,-0.9) to (q);
\draw (q) to (copier);
\draw (copier) to[out=165,in=-90] (X);
\draw (copier) to[out=15,in=-90] (X2);
\end{tikzpicture}}}
%%%%%%%%%%%%%%%%%%%
,\qquad\text{ then }\qquad
%%%%%%%%%%%%%%%%%%%
\vcenter{\hbox{
\begin{tikzpicture}[font=\small]
\node[arrow box] (q) at (0,-0.3) {$q$};
\node[copier] (c) at (0,0.3) {};
\node[copier] (c2) at (-0.5,1.55) {};
\node[arrow box] (g) at (-0.5,0.95) {$g$};
\node[arrow box] (f) at (0,2.25) {$f$};
\coordinate (X) at (0.7,2.95);
\coordinate (Y1) at (0,2.95);
\coordinate (Y2) at (-1,2.95);
\draw (0,-0.9) to (q);
\draw (q) to (c);
\draw (c) to[out=15,in=-90] (X);
\draw (c) to[out=165,in=-90] (g);
\draw (c2) to[out=15,in=-90] (f);
\draw (f) to (Y1);
\draw (c2) to[out=165,in=-90] (Y2);
\draw (g) to (c2);
\end{tikzpicture}}}
%%%%%%%%%%%%%%%%%%%
\quad=\quad
%%%%%%%%%%%%%%%%%%%
\vcenter{\hbox{
\begin{tikzpicture}[font=\small]
\node[arrow box] (q) at (0,-0.3) {$q$};
\node[copier] (c) at (0,0.3) {};
\node[copier] (c2) at (-0.5,1.55) {};
\node[arrow box] (g) at (-1,2.25) {$g$};
\coordinate (X) at (0.7,2.95);
\coordinate (Y1) at (0,2.95);
\coordinate (Y2) at (-1,2.95);
\draw (0,-0.9) to (q);
\draw (q) to (c);
\draw (c) to[out=15,in=-90] (X);
\draw (c) to[out=165,in=-90] (c2);
\draw (c2) to[out=15,in=-90] (Y1);
\draw (g) to (Y2);
\draw (c2) to[out=165,in=-90] (g);
\end{tikzpicture}}}
\]
A morphism $g$ is said to be \define{left $q$-a.e.\ modular with respect to $f$} iff the right equality holds. 
A similar definition is made for a \define{right a.e.\ modular} subcategory as well as \define{right $q$-a.e.\ modular with respect to $f$}.%
%footnote
\footnote{This definition is made in such a way so that the diagrams are reflected across the vertical line passing through $q$. It is defined in such a way so that an even $*$-preserving subcategory of a quantum Markov category is left a.e.\ modular if and only if it is right a.e.\ modular (the proof is similar to arguments we have already seen---see Lemma~\ref{prop:hfpkgs} for example).} 
%end footnote
\end{defn}

\vspace{-2mm}

In other words, an even subcategory $\mC$ is left (right) a.e.\ modular iff every disintegration as above is left (right) $q$-a.e.\ modular. 

%\vspace{1mm}
\vspace{-1mm}
\begin{exa}[The reason for choosing the terminology ``a.e.\ modular'']{exa:whyaemodular}
In the category $\fdCAlgU$, 
let $\mB\xstoch{F}\mA$ and $\mA\xstoch{G}\mB$ be linear unital maps. Let $\mB\xstoch{\xi}\C$ be a state (CPU). Theorem~\ref{thm:ncaeequivalence} shows that $G$ is a left $\xi$-a.e.\ modular with respect to $F$ if and only if
\[
P_{\xi}G\big(AF(B)\big)=P_{\xi}G(A)B\qquad\forall\;A\in\mA,\;B\in\mB. 
\]
If one assumes $F$ is an injective $*$-homomorphism (so that $\mB$ is viewed as a $C^*$-subalgebra of $\mA$) and $P_{\xi}=1_{\mB}$, then this condition is the modularity condition for conditional expectations of inclusions (see~\cite[Definition~9.5.3]{Fi96} for example). Thus, the definition of a.e.\ modularity generalizes this definition to the case where $\xi$ need not be a state, $F$ need not be injective (it need not even be a $*$-homomorphism), and the equality condition is weakened to an a.e.\ equality. 
\end{exa}

Both $\FinStoch$ and its measurable infinite-dimensional analogue $\Stoch$ are a.e.\ modular Markov categories. This was explicitly proved in an earlier version of~\cite{Fr19} (it also follows from results in the published version). Before checking whether $\fdCAlgSPU$ is a.e.\ modular, we present one of our main results relating a.e.\ modularity to two other conditions. 
This theorem also provides a partial converse to Proposition~\ref{thm:Bayesianinverseofdeterministicisadisint}, the latter of which provided a simple condition under which a Bayesian inverse is necessarily a disintegration. 
Furthermore, the theorem provides an alternative criterion for determining whether a subcategory is a.e.\ modular, and it is this fact that will help us prove $\fdCAlgSPU$ is a.e.\ modular. 

\begin{theo}[A.e.\ modularity, Bayesian inversion, and disintegrations abstractly]{thm:aemodbayesdisint}
Let $f,g,q$ be even morphisms in a quantum Markov category such that $g$ is a left disintegration of $(f,p:=g\circ q,q)$ (the left equality Definition~\ref{defn:aemodularity}). Then $g$ is left $q$-a.e.\ modular with respect to $f$ (the right equality in Definition~\ref{defn:aemodularity}) if and only if both the equalities 
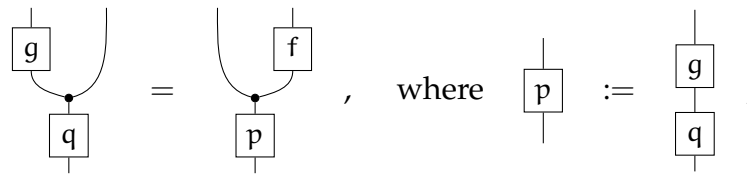
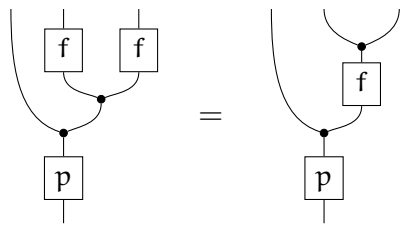
\begin{enumerate}[(a)]
\item
\[
\vcenter{\hbox{%
\begin{tikzpicture}[font=\small]
\node[arrow box] (q) at (0,-0.2) {$q$};
\node[copier] (copier) at (0,0.3) {};
\node[arrow box] (g) at (-0.5,0.95) {$g$};
\coordinate (X) at (-0.5,1.5);
\coordinate (Y) at (0.5,1.5);
\draw (0,-0.7) to (q); 
\draw (q) to (copier);
\draw (copier) to[out=165,in=-90] (g);
\draw (copier) to[out=15,in=-90] (Y);
\draw (g) to (X);
\end{tikzpicture}}}
%%%%%%%%%%%%%%%%%%%%
\quad=\quad
%%%%%%%%%%%%%%%%%%%%
\vcenter{\hbox{%
\begin{tikzpicture}[font=\small]
\node[arrow box] (p) at (0,-0.2) {$p$};
\node[copier] (copier) at (0,0.3) {};
\node[arrow box] (f) at (0.5,0.95) {$f$};
\coordinate (X) at (-0.5,1.5);
\coordinate (Y) at (0.5,1.5);
\draw (0,-0.7) to (p); 
\draw (p) to (copier);
\draw (copier) to[out=165,in=-90] (X);
\draw (copier) to[out=15,in=-90] (f);
\draw (f) to (Y);
\end{tikzpicture}}}
%%%%%%%%%%%%%%%%%%%%%%
\quad,\quad\text{ where }\quad
%%%%%%%%%%%%%%%%%%%%%%
\vcenter{\hbox{%
\begin{tikzpicture}[font=\small]
\node[arrow box] (p) at (0,0) {$p$};
\draw (0,-0.7) to (p);
\draw (p) to (0,0.7);
\end{tikzpicture}}}
%%%%%%%%%%%%%%%%%%%%%%
\quad:=\quad
%%%%%%%%%%%%%%%%%%%%%%
\vcenter{\hbox{%
\begin{tikzpicture}[font=\small]
\node[arrow box] (q) at (0,-0.1) {$q$};
\node[arrow box] (g) at (0,0.8) {$g$};
\draw (0,-0.6) to (q);
\draw (q) to (g);
\draw (g) to (0,1.55);
\end{tikzpicture}}}
\quad, 
\]
\item
and
\[
\vcenter{\hbox{
\begin{tikzpicture}[font=\small,xscale=-1]
\node[arrow box] (p) at (0,-0.3) {$p$};
\node[copier] (c) at (0,0.3) {};
\node[copier] (c2) at (-0.5,0.75) {};
\node[arrow box] (f) at (-1,1.4) {$f$};
\node[arrow box] (e) at (0,1.4) {$f$};
\coordinate (X) at (0.7,1.95);
\coordinate (Y1) at (-1,1.95);
\coordinate (Y2) at (0,1.95);
\draw (0,-0.9) to (p);
\draw (p) to (c);
\draw (c) to[out=15,in=-90] (X);
\draw (c) to[out=165,in=-90] (c2);
\draw (c2) to[out=165,in=-90] (f);
\draw (c2) to[out=15,in=-90] (e);
\draw (f) to (Y1);
\draw (e) to (Y2);
\end{tikzpicture}}}
%%%%%%%%%%%%%%%%%%%%%%
\quad=\quad
%%%%%%%%%%%%%%%%%%%%%%
\vcenter{\hbox{
\begin{tikzpicture}[font=\small,xscale=-1]
\node[arrow box] (p) at (0,-0.3) {$p$};
\node[copier] (c) at (0,0.3) {};
\node[copier] (c2) at (-0.5,1.45) {};
\node[arrow box] (f) at (-0.5,0.95) {$f$};
\coordinate (X) at (0.7,1.95);
\coordinate (Y1) at (-1,1.95);
\coordinate (Y2) at (0,1.95);
\draw (0,-0.9) to (p);
\draw (p) to (c);
\draw (c) to[out=15,in=-90] (X);
\draw (c) to[out=165,in=-90] (f);
\draw (c2) to[out=165,in=-90] (-1,1.95);
\draw (c2) to[out=15,in=-90] (0,1.95);
\draw (f) to (c2);
\end{tikzpicture}}}
\]
\end{enumerate}
hold. In words, an even subcategory $\mC$ of a quantum Markov category is left a.e.\ modular if and only if for every left disintegration $g$ of $(f,p,q)$ in $\mC$, $g$ is a left Bayesian inverse for $(f,p,q)$ and $f$ is right $p$-a.e.\ deterministic. 
\end{theo}

\bprf
{\color{white}{you found me!}}

\noindent($\Rightarrow$)
Suppose that $f,g,q$ satisfy the right identity in Definition~\ref{defn:aemodularity}. Then 
\be
\label{eq:bayesfromdisint}
\vcenter{\hbox{%
\begin{tikzpicture}[font=\small]
\node[arrow box] (q) at (0,-0.2) {$q$};
\node[copier] (copier) at (0,0.3) {};
\node[arrow box] (g) at (-0.5,0.95) {$g$};
\coordinate (X) at (-0.5,1.5);
\coordinate (Y) at (0.5,1.5);
\draw (0,-0.7) to (q); 
\draw (q) to (copier);
\draw (copier) to[out=165,in=-90] (g);
\draw (copier) to[out=15,in=-90] (Y);
\draw (g) to (X);
\end{tikzpicture}}}
%%%%%%%%%%%%%%%%%%%%
\quad=\quad
%%%%%%%%%%%%%%%%%%%%
\vcenter{\hbox{
\begin{tikzpicture}[font=\small]
\node[arrow box] (q) at (0,-0.3) {$q$};
\node[copier] (c) at (0,0.3) {};
\node[copier] (c2) at (-0.5,1.55) {};
\node[arrow box] (g) at (-1,2.25) {$g$};
\node[discarder] (X) at (0.7,2.95) {};
\coordinate (Y1) at (0,2.95);
\coordinate (Y2) at (-1,2.95);
\draw (0,-0.9) to (q);
\draw (q) to (c);
\draw (c) to[out=15,in=-90] (X);
\draw (c) to[out=165,in=-90] (c2);
\draw (c2) to[out=15,in=-90] (Y1);
\draw (g) to (Y2);
\draw (c2) to[out=165,in=-90] (g);
\end{tikzpicture}}}
%%%%%%%%%%%%%%%%%%%%
\quad\overset{\text{Defn~\ref{defn:aemodularity}}}{=\joinrel=\joinrel=\joinrel=\joinrel=\joinrel=}\quad
%%%%%%%%%%%%%%%%%%%%
\vcenter{\hbox{
\begin{tikzpicture}[font=\small]
\node[arrow box] (q) at (0,-0.3) {$q$};
\node[copier] (c) at (0,0.3) {};
\node[copier] (c2) at (-0.5,1.55) {};
\node[arrow box] (g) at (-0.5,0.95) {$g$};
\node[arrow box] (f) at (0,2.25) {$f$};
\node[discarder] (X) at (0.7,2.95) {};
\coordinate (Y1) at (0,2.95);
\coordinate (Y2) at (-1,2.95);
\draw (0,-0.9) to (q);
\draw (q) to (c);
\draw (c) to[out=15,in=-90] (X);
\draw (c) to[out=165,in=-90] (g);
\draw (c2) to[out=15,in=-90] (f);
\draw (f) to (Y1);
\draw (c2) to[out=165,in=-90] (Y2);
\draw (g) to (c2);
\end{tikzpicture}}}
%%%%%%%%%%%%%%%%%%%%
\quad=\quad
%%%%%%%%%%%%%%%%%%%%
\vcenter{\hbox{%
\begin{tikzpicture}[font=\small]
\node[arrow box] (q) at (0,-1.0) {$q$};
\node[arrow box] (g) at (0,-0.2) {$g$};
\node[copier] (copier) at (0,0.3) {};
\node[arrow box] (f) at (0.5,0.95) {$f$};
\coordinate (X) at (-0.5,1.5);
\coordinate (Y) at (0.5,1.5);
\draw (0,-1.6) to (q); 
\draw (q) to (g);
\draw (g) to (copier);
\draw (copier) to[out=165,in=-90] (X);
\draw (copier) to[out=15,in=-90] (f);
\draw (f) to (Y);
\end{tikzpicture}}}
%%%%%%%%%%%%%%%%%%%
\quad=\quad
%%%%%%%%%%%%%%%%%%%
\vcenter{\hbox{%
\begin{tikzpicture}[font=\small]
\node[arrow box] (p) at (0,-0.2) {$p$};
\node[copier] (copier) at (0,0.3) {};
\node[arrow box] (f) at (0.5,0.95) {$f$};
\coordinate (X) at (-0.5,1.5);
\coordinate (Y) at (0.5,1.5);
\draw (0,-0.7) to (p); 
\draw (p) to (copier);
\draw (copier) to[out=165,in=-90] (X);
\draw (copier) to[out=15,in=-90] (f);
\draw (f) to (Y);
\end{tikzpicture}}}
\ee
and 
\be
\label{eq:aemodularityaedeterministic}
\vcenter{\hbox{
\begin{tikzpicture}[font=\small,xscale=-1]
\node[arrow box] (p) at (0,-0.3) {$p$};
\node[copier] (c) at (0,0.3) {};
\node[copier] (c2) at (-0.5,0.75) {};
\node[arrow box] (f) at (-1,1.4) {$f$};
\node[arrow box] (e) at (0,1.4) {$f$};
\coordinate (X) at (0.7,1.95);
\coordinate (Y1) at (-1,1.95);
\coordinate (Y2) at (0,1.95);
\draw (0,-0.9) to (p);
\draw (p) to (c);
\draw (c) to[out=15,in=-90] (X);
\draw (c) to[out=165,in=-90] (c2);
\draw (c2) to[out=165,in=-90] (f);
\draw (c2) to[out=15,in=-90] (e);
\draw (f) to (Y1);
\draw (e) to (Y2);
\end{tikzpicture}}}
%%%%%%%%%%%%%%%%%%%%%%end of first
=
\!
%%%%%%%%%%%%%%%%%%%%%%
\vcenter{\hbox{
\begin{tikzpicture}[font=\small,xscale=-1]
\node[arrow box] (p) at (0,-0.3) {$p$};
\node[copier] (c) at (0,0.3) {};
\node[copier] (c2) at (0.5,0.75) {};
\node[arrow box] (f) at (-0.8,1.4) {$f$};
\node[arrow box] (e) at (0,1.4) {$f$};
\coordinate (X) at (0.9,1.95);
\coordinate (Y1) at (-0.8,1.95);
\coordinate (Y2) at (0,1.95);
\draw (0,-0.9) to (p);
\draw (p) to (c);
\draw (c) to[out=15,in=-90] (c2);
\draw (c) to[out=165,in=-90] (f);
\draw (c2) to[out=165,in=-90] (e);
\draw (c2) to[out=15,in=-90] (X);
\draw (f) to (Y1);
\draw (e) to (Y2);
\end{tikzpicture}}}
%%%%%%%%%%%%%%%%%%%%%%end of second
\overset{\text{(\ref{eq:bayesfromdisint})}}{=\joinrel=\joinrel=}
\!
%%%%%%%%%%%%%%%%%%%%%%
\vcenter{\hbox{
\begin{tikzpicture}[font=\small]
\node[arrow box] (q) at (0,-0.3) {$q$};
\node[copier] (c) at (0,0.3) {};
\node[copier] (c2) at (-0.5,1.55) {};
\node[arrow box] (g) at (-0.5,0.95) {$g$};
\node[arrow box] (f) at (0,2.25) {$f$};
\coordinate (X) at (0.7,2.95);
\coordinate (Y1) at (0,2.95);
\coordinate (Y2) at (-1,2.95);
\draw (0,-0.9) to (q);
\draw (q) to (c);
\draw (c) to[out=15,in=-90] (X);
\draw (c) to[out=165,in=-90] (g);
\draw (c2) to[out=15,in=-90] (f);
\draw (f) to (Y1);
\draw (c2) to[out=165,in=-90] (Y2);
\draw (g) to (c2);
\end{tikzpicture}}}
%%%%%%%%%%%%%%%%%%%%%%
%%%%%%%%%%%%%%%%%%%%
\!\!
\underset{\text{\ref{defn:aemodularity}}}{\overset{\text{Defn}}{=\joinrel=\joinrel=}}
%%%%%%%%%%%%%%%%%%%% 
\!\!
%%%%%%%%%%%%%%%%%%%%%%
\vcenter{\hbox{
\begin{tikzpicture}[font=\small]
\node[arrow box] (q) at (0,-0.3) {$q$};
\node[copier] (c) at (0,0.3) {};
\node[copier] (c2) at (-0.5,1.55) {};
\node[arrow box] (g) at (-1,2.25) {$g$};
\coordinate (X) at (0.7,2.95);
\coordinate (Y1) at (0,2.95);
\coordinate (Y2) at (-1,2.95);
\draw (0,-0.9) to (q);
\draw (q) to (c);
\draw (c) to[out=15,in=-90] (X);
\draw (c) to[out=165,in=-90] (c2);
\draw (c2) to[out=15,in=-90] (Y1);
\draw (g) to (Y2);
\draw (c2) to[out=165,in=-90] (g);
\end{tikzpicture}}}
%%%%%%%%%%%%%%%%%%%%%%
\!
=
%%%%%%%%%%%%%%%%%%%%%%
\vcenter{\hbox{
\begin{tikzpicture}[font=\small,xscale=-1]
\node[arrow box] (q) at (0,-0.3) {$q$};
\node[copier] (c) at (0,0.3) {};
\node[copier] (c2) at (-0.5,1.45) {};
\node[arrow box] (g) at (0.6,1.05) {$g$};
\coordinate (X) at (0.6,1.95);
\coordinate (Y1) at (-1,1.95);
\coordinate (Y2) at (0,1.95);
\draw (0,-0.9) to (q);
\draw (q) to (c);
\draw (c) to[out=165,in=-90] (c2);
\draw (c) to[out=15,in=-90] (g);
\draw (c2) to[out=165,in=-90] (-1,1.95);
\draw (c2) to[out=15,in=-90] (0,1.95);
\draw (g) to (X);
\end{tikzpicture}}}
%%%%%%%%%%%%%%%%%%%%%%
\!
\overset{\text{(\ref{eq:bayesfromdisint})}}{=\joinrel=\joinrel=}
\!
%%%%%%%%%%%%%%%%%%%%%%
\vcenter{\hbox{
\begin{tikzpicture}[font=\small,xscale=-1]
\node[arrow box] (p) at (0,-0.3) {$p$};
\node[copier] (c) at (0,0.3) {};
\node[copier] (c2) at (-0.5,1.45) {};
\node[arrow box] (f) at (-0.5,0.95) {$f$};
\coordinate (X) at (0.7,1.95);
\coordinate (Y1) at (-1,1.95);
\coordinate (Y2) at (0,1.95);
\draw (0,-0.9) to (p);
\draw (p) to (c);
\draw (c) to[out=15,in=-90] (X);
\draw (c) to[out=165,in=-90] (f);
\draw (c2) to[out=165,in=-90] (-1,1.95);
\draw (c2) to[out=15,in=-90] (0,1.95);
\draw (f) to (c2);
\end{tikzpicture}}}
\ee

\noindent($\Leftarrow$)
The proof that $f,g,q$ satisfy conditions (a) and (b) implies they satisfy the right identity in Definition~\ref{defn:aemodularity} follows similar manipulations to the ones just witnessed---it is obtained by cyclically permuting the list of equalities in~(\ref{eq:aemodularityaedeterministic}). 
\eprf

\begin{cor}[Bayesian inversion and disintegration in an a.e.\ modular category]{cor:binversiondisint}
Let $\Theta\xstoch{p}X$ and $X\xstoch{f}Y$ be morphisms in a $*$-preserving a.e.\ modular subcategory $\mC$ of a quantum Markov category. Set $q:=f\circ p$. Then the following conditions are equivalent. 
\begin{enumerate}[(a)]
\item
\label{item:gdisint}
A disintegration of $(f,p,q)$ exists. 
\item
\label{item:gbayesandfdet}
$f$ is $p$-a.e.\ deterministic and a Bayesian inverse of $(f,p,q)$ exists.  
\end{enumerate}
Furthermore, in either of these cases, any disintegration of $(f,p,q)$ is $q$-a.e.\ equivalent to any Bayesian inverse of $(f,p,q)$. In particular, any two disintegrations are $q$-a.e.\ equivalent. 
\end{cor}

\bprf
Proposition~\ref{thm:Bayesianinverseofdeterministicisadisint} provides the implication (\ref{item:gbayesandfdet})$\Rightarrow$(\ref{item:gdisint}). Theorem~\ref{thm:aemodbayesdisint} provides the implication (\ref{item:gdisint})$\Rightarrow$(\ref{item:gbayesandfdet}). 
\eprf

%%\pagebreak
%\phantom{why do I need this here?}
%\vspace{-5mm}
\vspace{-1mm}
\begin{rmk}[Cho--Jacobs' disintegration continued]{rmk:CJdisintcontinued}
We now restate and prove the claim made in Remark~\ref{rmk:ChoJacobsDisintegrations} regarding the relationship to Cho and Jacob's version of disintegration in the context of a $*$-preserving a.e.\ modular category and supposing that certain disintegrations exist (the ones needed to make sense of what follows). Let $I\xstoch{s}X\times Y$ be a state with $X$ marginal $I\xstoch{p}X$ and $Y$ marginal $I\xstoch{q}Y.$ Let $X\xstoch{\overline{\pi_{X}}}X\times Y$ and $Y\xstoch{\overline{\pi_{Y}}}X\times Y$ be disintegrations of $(\pi_{X},s,p)$ and $(\pi_{Y},s,q)$, respectively. Then by Theorem~\ref{thm:aemodbayesdisint}, $\overline{\pi_{X}}$ and $\overline{\pi_{Y}}$ are Bayesian inverses of $(\pi_{X},s,p)$ and $(\pi_{Y},s,q)$, respectively, i.e.\ 
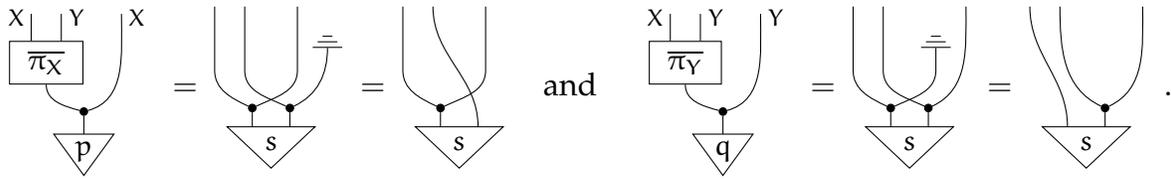
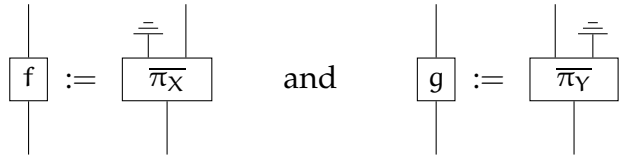
\[
\vcenter{\hbox{%
\begin{tikzpicture}[font=\small]
\node[state] (q) at (0,0) {$p$};
\node[copier] (copier) at (0,0.3) {};
\node[arrow box] (g) at (-0.5,0.95) {\;$\overline{\pi_{X}}$\;};
\coordinate (X) at (-0.7,1.6);
\coordinate (Y) at (-0.3,1.6);
\coordinate (X2) at (0.5,1.6);
\draw (q) to (copier);
\draw (copier) to[out=165,in=-90] (g);
\draw (copier) to[out=15,in=-90] (X2);
\draw (g.north)++(-0.2,0) to (X);
\draw (g.north)++(0.2,0) to (Y);
\path[scriptstyle]
node at (-0.9,1.55) {$X$}
node at (-0.1,1.55) {$Y$}
node at (0.7,1.55) {$X$};
\end{tikzpicture}}}
%\quad
\;
=
\;
%\quad
%
\vcenter{\hbox{%
\begin{tikzpicture}[font=\small]
\coordinate (di1) at (-0.75,1.6);
\coordinate (di2) at (0.75,1.05);
\node[discarder] at (di2) {};
\node[state] (s) at (0,0) {\;$s$\;};
\node[copier] (c1) at (-0.25,0.25) {};
\node[copier] (c2) at (0.25,0.25) {};
\draw (s) ++(-0.25,0) to (c1);
\draw (s) ++(0.25,0) to (c2);
\draw (c1) to[out=155,in=-90] (-0.75,0.75)
to (di1);
\draw (c2) to[out=25,in=-90] (di2);
\draw (c1) to[out=25,in=-90] (0.35,0.75) to (0.35,1.6);
\draw (c2) to[out=155,in=-90] (-0.35,0.75) to (-0.35,1.6);
\end{tikzpicture}}}
%
%\quad
\;
=
\;
%\quad
%
\vcenter{\hbox{%
\begin{tikzpicture}[font=\small]
\coordinate (di1) at (-0.75,1.6);
\coordinate (di2) at (0.75,1.05);
\node[state] (s) at (0,0) {\;$s$\;};
\node[copier] (c1) at (-0.25,0.25) {};
%\coordinate (c2) at (0.25,0.25);
%
\draw (s) ++(-0.25,0) to (c1);
\draw (s) ++(0.25,0) to[out=90,in=-90] (-0.35,1.6);
\draw (c1) to[out=155,in=-90] (-0.75,0.75)
to (di1);
\draw (c1) to[out=25,in=-90] (0.35,0.75) to (0.35,1.6);
%\draw (c2) to[out=155,in=-90] (-0.45,0.75) to (-0.45,1.6);
\end{tikzpicture}}}
\quad\text{ and }\quad
\vcenter{\hbox{%
\begin{tikzpicture}[font=\small]
\node[state] (q) at (0,0) {$q$};
\node[copier] (copier) at (0,0.3) {};
\node[arrow box] (g) at (-0.5,0.95) {\;$\overline{\pi_{Y}}$\;};
\coordinate (X) at (-0.7,1.6);
\coordinate (Y) at (-0.3,1.6);
\coordinate (X2) at (0.5,1.6);
\draw (q) to (copier);
\draw (copier) to[out=165,in=-90] (g);
\draw (copier) to[out=15,in=-90] (X2);
\draw (g.north)++(-0.2,0) to (X);
\draw (g.north)++(0.2,0) to (Y);
\path[scriptstyle]
node at (-0.9,1.55) {$X$}
node at (-0.1,1.55) {$Y$}
node at (0.7,1.55) {$Y$};
\end{tikzpicture}}}
%
%\quad
\;
=
\;
%\quad
%
\vcenter{\hbox{%
\begin{tikzpicture}[font=\small]
\coordinate (di1) at (-0.75,1.6);
\coordinate (di2) at (0.75,1.6);
\coordinate (d) at (0.35,1.05);
\node[discarder] at (d) {};
\node[state] (s) at (0,0) {\;$s$\;};
\node[copier] (c1) at (-0.25,0.25) {};
\node[copier] (c2) at (0.25,0.25) {};
\draw (s) ++(-0.25,0) to (c1);
\draw (s) ++(0.25,0) to (c2);
\draw (c1) to[out=155,in=-90] (-0.75,0.75)
to (di1);
\draw (c2) to[out=25,in=-90] (di2);
\draw (c1) to[out=25,in=-90] (0.35,0.75) to (d);
\draw (c2) to[out=155,in=-90] (-0.35,0.75) to (-0.35,1.6);
\end{tikzpicture}}}
%
%\quad
\;
=
\;
%\quad
\vcenter{\hbox{%
\begin{tikzpicture}[font=\small]
\coordinate (di1) at (-0.75,1.6);
\coordinate (di2) at (0.75,1.6);
\coordinate (d) at (0.35,1.05);
\node[state] (s) at (0,0) {\;$s$\;};
\node[copier] (c2) at (0.25,0.25) {};
\draw (s) ++(-0.25,0) to[out=90,in=-90] (di1);
\draw (s) ++(0.25,0) to (c2);
to (di1);
\draw (c2) to[out=25,in=-90] (di2);
\draw (c2) to[out=155,in=-90] (-0.35,1.6);
\end{tikzpicture}}}
\;\;.
\]
Therefore, defining 
\[
\vcenter{\hbox{
\begin{tikzpicture}[font=\small]
\node[arrow box] (p) at (0,0) {$f$};
\coordinate (X) at (0,1.0);
\draw (0,-1.0) to (p);
\draw (p) to (X);
\end{tikzpicture}}}
\;
:=
\;
\vcenter{\hbox{
\begin{tikzpicture}[font=\small]
\node[arrow box] (p) at (0,0) {\;\;$\overline{\pi_{X}}$\;\;};
\coordinate (X) at (-0.25,0.6);
\coordinate (Y) at (0.25,1.0);
\node[discarder] at (X) {};
\draw (0,-1.0) to (p);
\draw (p.north)++(-0.25,0) to (X);
\draw (p.north)++(0.25,0) to (Y);
\end{tikzpicture}}}
\qquad\text{ and }\qquad
\vcenter{\hbox{
\begin{tikzpicture}[font=\small]
\node[arrow box] (p) at (0,0) {$g$};
\coordinate (X) at (0,1.0);
\draw (0,-1.0) to (p);
\draw (p) to (X);
\end{tikzpicture}}}
\;
:=
\;
\vcenter{\hbox{
\begin{tikzpicture}[font=\small]
\node[arrow box] (p) at (0,0) {\;\;$\overline{\pi_{Y}}$\;\;};
\coordinate (X) at (-0.25,1.0);
\coordinate (Y) at (0.25,0.6);
\node[discarder] at (Y) {};
\draw (0,-1.0) to (p);
\draw (p.north)++(-0.25,0) to (X);
\draw (p.north)++(0.25,0) to (Y);
\end{tikzpicture}}}
\]
as in Remark~\ref{rmk:ChoJacobsDisintegrations} shows that these maps satisfy
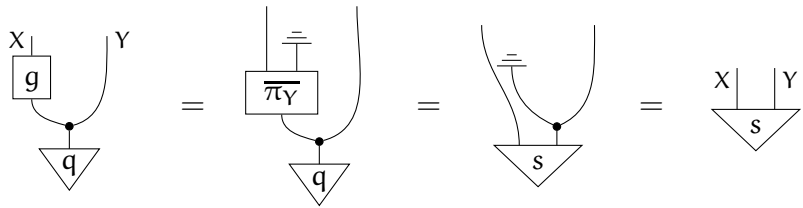
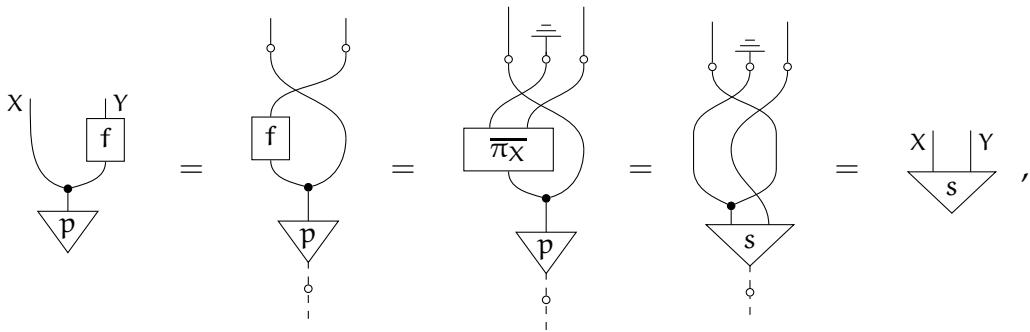
\[
\vcenter{\hbox{%
\begin{tikzpicture}[font=\small]
\node[state] (q) at (0,0) {$q$};
%\node[arrow box] (f) at (0,-0.2) {$f$};
%\node[state] (p) at (0,-0.8) {$p$};
\node[copier] (copier) at (0,0.3) {};
\node[arrow box] (g) at (-0.5,0.95) {$g$};
\coordinate (X) at (-0.5,1.5);
\coordinate (Y) at (0.5,1.5);
\draw (q) to (copier);
%\draw (f) to (copier);
%\draw (p) to (f);
\draw (copier) to[out=165,in=-90] (g);
\draw (copier) to[out=15,in=-90] (Y);
\draw (g) to (X);
\path[scriptstyle]
node at (-0.7,1.45) {$X$}
node at (0.7,1.45) {$Y$};
\end{tikzpicture}}}
\quad
%\;
=
%\;
\quad
\vcenter{\hbox{%
\begin{tikzpicture}[font=\small]
\node[state] (q) at (0,0) {$q$};
\node[copier] (copier) at (0,0.3) {};
\node[arrow box] (g) at (-0.5,0.95) {\;$\overline{\pi_{Y}}$\;};
\coordinate (X) at (-0.7,2.1);
\coordinate (Y) at (-0.3,1.6);
\coordinate (X2) at (0.5,2.1);
\node[discarder] at (Y) {};
\draw (q) to (copier);
\draw (copier) to[out=165,in=-90] (g);
\draw (copier) to[out=15,in=-90] (X2);
\draw (g.north)++(-0.2,0) to (X);
\draw (g.north)++(0.2,0) to (Y);
\end{tikzpicture}}}
\quad
%\;
=
%\;
\quad
\vcenter{\hbox{%
\begin{tikzpicture}[font=\small]
\coordinate (di1) at (-0.75,1.6);
\coordinate (di2) at (0.75,1.6);
\coordinate (d) at (-0.35,1.05);
\node[state] (s) at (0,0) {\;$s$\;};
\node[copier] (c2) at (0.25,0.25) {};
\node[discarder] at (d) {};
\draw (s) ++(-0.25,0) to[out=90,in=-90] (di1);
\draw (s) ++(0.25,0) to (c2);
to (di1);
\draw (c2) to[out=25,in=-90] (di2);
\draw (c2) to[out=155,in=-90] (d);
\end{tikzpicture}}}
\quad
%\;
=
%\;
\quad
\vcenter{\hbox{%
\begin{tikzpicture}[font=\small]
\node[state] (omega) at (0,0) {\;$s$\;};
\coordinate (X) at (-0.25,0.55) {};
\coordinate (Y) at (0.25,0.55) {};
\draw (omega) ++(-0.25, 0) to (X);
\draw (omega) ++(0.25, 0) to (Y);
\path[scriptstyle]
node at (-0.45,0.4) {$X$}
node at (0.45,0.4) {$Y$};
\end{tikzpicture}}}
\]
and
\[
\vcenter{\hbox{%
\begin{tikzpicture}[font=\small]
\node[state] (p) at (0,0) {$p$};
\node[copier] (copier) at (0,0.3) {};
\node[arrow box] (f) at (0.5,0.95) {$f$};
\coordinate (X) at (-0.5,1.5);
\coordinate (Y) at (0.5,1.5);
\draw (p) to (copier);
\draw (copier) to[out=165,in=-90] (X);
\draw (copier) to[out=15,in=-90] (f);
\draw (f) to (Y);
\path[scriptstyle]
node at (-0.7,1.45) {$X$}
node at (0.7,1.45) {$Y$};
\end{tikzpicture}}}
\quad
%\;
=
%\;
\quad
\vcenter{\hbox{
\begin{tikzpicture}[font=\small]
\node[state] (p) at (0,-0.6) {$p$};
\node[star] (s1) at (0,-1.5) {};
\node[copier] (copier) at (0,-0.15) {};
\coordinate (R) at (0.5,0.3) {};
\coordinate (Ls) at (-0.5,1.6) {};
\coordinate (Rs) at (0.5,1.6) {};
\node[star] (s2) at (-0.5,1.7) {};
\node[star] (s3) at (0.5,1.7) {};
\coordinate (r) at (0.5,0.5);
\node[arrow box] (h) at (-0.5,0.5) {$f$};
\coordinate (X) at (0.5,2.1);
\coordinate (Y) at (-0.5,2.1);
\draw[dashed] (0,-1.9) to (s1);
%\draw[dashed] (s1) to (p); %this isn't working...
\draw[dashed] (s1) to (0,-1.15);
\draw (p) to (copier);
\draw (copier) to [out=15,in=-90] (r);
\draw (r) to [out=90,in=-90] (Ls);
\draw (h) to [out=90,in=-90] (Rs);
\draw (Ls) to (s2);
\draw (Rs) to (s3);
\draw (s3) to (X);
\draw (copier) to[out=165,in=-90] (h);
\draw (s2) to (Y);
\end{tikzpicture}}}
\quad
%\;
=
%\;
\quad
\vcenter{\hbox{
\begin{tikzpicture}[font=\small]
\node[state] (p) at (0,-0.6) {$p$};
\node[star] (s1) at (0,-1.5) {};
\node[copier] (copier) at (0,-0.15) {};
\coordinate (R) at (0.5,0.3) {};
\coordinate (Ls) at (-0.5,1.6) {};
\coordinate (Rs) at (0.5,1.6) {};
\coordinate (d) at (0,2.0);
\node[discarder] at (d) {};
\node[star] (s4) at (0,1.7) {};
\node[star] (s2) at (-0.5,1.7) {};
\node[star] (s3) at (0.5,1.7) {};
\coordinate (r) at (0.5,0.5);
\node[arrow box] (h) at (-0.5,0.5) {\;\;$\overline{\pi_{X}}$\;\;};
\coordinate (X) at (0.5,2.4);
\coordinate (X2) at (0,2.4);
\coordinate (Y) at (-0.5,2.4);
\draw[dashed] (0,-1.9) to (s1);
%\draw[dashed] (s1) to (p); %this isn't working...
\draw[dashed] (s1) to (0,-1.15);
\draw (p) to (copier);
\draw (copier) to [out=15,in=-90] (r);
\draw (r) to [out=90,in=-90] (Ls);
\draw (h.north)++(0.25,0) to [out=90,in=-90] (Rs);
\draw (h.north)++(-0.25,0) to [out=90,in=-90] (s4);
%\draw[dashed] (d) to (s4); %this doesn't work either...
\draw[dashed] (0,1.8) to (s4);
\draw (s4) to (d);
\draw (Ls) to (s2);
\draw (Rs) to (s3);
\draw (s3) to (X);
\draw (copier) to[out=165,in=-90] (h);
\draw (s2) to (Y);
\end{tikzpicture}}}
\quad
%\;
=
%\;
\quad
\vcenter{\hbox{%
\begin{tikzpicture}[font=\small]
\coordinate (d) at (0,2.3);
\coordinate (M1) at (-0.75,1.2);
\coordinate (M2) at (-0.25,1.0);
\coordinate (M3) at (0.35,1.2);
\node[star] (s4) at (0,2.1) {};
\node[star] (s2) at (-0.5,2.1) {};
\node[star] (s3) at (0.5,2.1) {};
\node[star] (s1) at (0,-0.9) {};
\node[discarder] at (d) {};
\node[state] (s) at (0,0) {\;$s$\;};
\node[copier] (c1) at (-0.25,0.25) {};
%\coordinate (c2) at (0.25,0.25);
%
\draw[dashed] (0,-1.2) to (s1);
\draw[dashed] (s1) to (0,-0.55);
\draw (s) ++(-0.25,0) to (c1);
\draw (s) ++(0.25,0) to[out=90,in=-90] (M2) to[out=90,in=-90] (s3);
\draw (s3) to (0.5,2.7);
\draw (c1) to[out=155,in=-90] (-0.75,0.75)
to (M1) to[out=90,in=-90] (s4);
\draw (s4) to (d);
\draw (c1) to[out=25,in=-90] (0.35,0.75) to (M3) to[out=90,in=-90] (s2);
\draw (s2) to (-0.5,2.7);
%\draw (c2) to[out=155,in=-90] (-0.45,0.75) to (-0.45,1.6);
\end{tikzpicture}}}
\quad
%\;
=
%\;
\quad
\vcenter{\hbox{%
\begin{tikzpicture}[font=\small]
\node[state] (omega) at (0,0) {\;$s$\;};
\coordinate (X) at (-0.25,0.55) {};
\coordinate (Y) at (0.25,0.55) {};
\draw (omega) ++(-0.25, 0) to (X);
\draw (omega) ++(0.25, 0) to (Y);
\path[scriptstyle]
node at (-0.45,0.4) {$X$}
node at (0.45,0.4) {$Y$};
\end{tikzpicture}}}
\;\;,
\]
which are the CJ disintegration conditions. Note that the procedure used here also constructs Bayesian inverses from disintegrating the projection maps and composing, just as is in the measure-theoretic case (cf.\ \cite[Section 3.2]{CuSt14}). 
\end{rmk}

\vspace{-2mm}
\begin{rmk}[Relation to the abstract Fisher--Neyman factorization theorem]{rmk:fisherneyman}
Fritz' formulation of the Fisher--Neyman factorization theorem (Theorem~14.5 in~\cite{Fr19}) is an immediate consequence of our Theorem~\ref{thm:aemodbayesdisint} in the special case where one assumes $f$ is deterministic. Hence, our result can be viewed as a stronger version of this factorization theorem. Note that although Fritz assumed strict positivity (cf.\ Definition~\ref{defn:strictpositive} below), the proof follows through using a.e.\ modularity, which was in fact the axiom used in an earlier version of his paper. 
\end{rmk}

We now use this theorem to prove one of our main theorems showing that SPU maps between $C^*$-algebras form an a.e.\ modular category. 

\begin{theo}[$\fdCAlgSPU$ is a.e.\ modular]{thm:SPUmodular}
$\fdCAlgSPU$ and $\fdCAlgCPU$ are a.e.\ modular subcategories of $\fdCAlgUY$. 
\end{theo}

\vspace{-3mm}

To provide a proof of this, we need two lemmas, which are interesting in their own right. 

\vspace{-1mm}
\begin{lem}[The relative Multiplication Lemma]{lem:relmulttheorem}
Let $\mA\xstoch{G}\mB$ and $\mB\xstoch{\xi}\mC$ be morphisms in $\fdCAlgSPU$.%
%footnote
\footnote{Technically, we only need $\xi$ to be a positive (not necessarily unital) map. It need not be Schwarz-positive.} 
%end footnote
Suppose that%
%footnote
\footnote{This is equivalent to $G(A^*A)-G(A)^*G(A)\in\mathcal{N}_{\xi}$ since $G$ is SPU by Lemma~\ref{lem:minilemmaleftideal}.}
%end footnote
$G(A^*A)-G(A)^*G(A)\in\ker\xi$ for some $A\in\mA$. Then%
%footnote
\footnote{Notice the quantifiers here and how this claim differs from the weak a.e.\ Multiplication Lemma (Lemma~\ref{lem:Attalslemmaaegeneral}).}
%end footnote
\[
G(A^*D)-G(A)^*G(D)\;,\; G(D^*A)-G(D)^*G(A)\in\ker\xi\qquad\forall\; D\in\mA.
\]
In terms of string diagrams (in the larger quantum CD category $\fdCAlgY$), if 
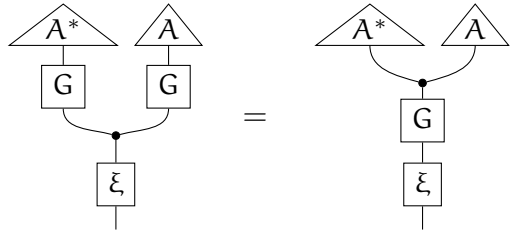
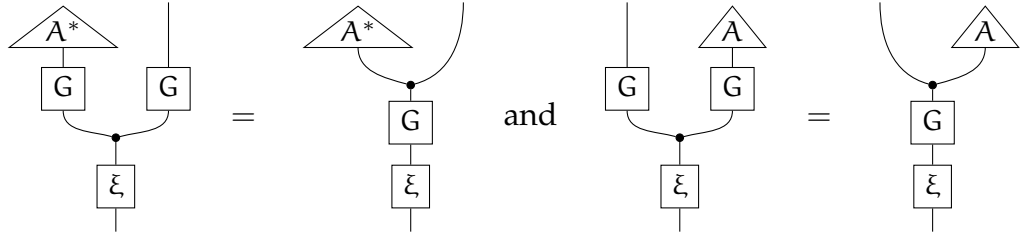
\[
\vcenter{\hbox{
\begin{tikzpicture}[font=\small]
\node[arrow box] (p) at (0,0.1) {$\xi$};
\node[copier] (c2) at (0,0.75) {};
\node[arrow box] (f) at (0.7,1.4) {$G$};
\node[arrow box] (e) at (-0.7,1.4) {$G$};
\node[effect] (Y1) at (0.7,1.95) {$A$};
\node[effect] (Y2) at (-0.7,1.95) {$A^*$};
\draw (0,-0.5) to (p);
\draw (p) to (c2);
\draw (c2) to[out=15,in=-90] (f);
\draw (c2) to[out=165,in=-90] (e);
\draw (f) to (Y1);
\draw (e) to (Y2);
\end{tikzpicture}}}
%%%%%%%%%%%%%%%%%%%
\quad=\quad
%%%%%%%%%%%%%%%%%%%
\vcenter{\hbox{
\begin{tikzpicture}[font=\small]
\node[arrow box] (p) at (0,0.1) {$\xi$};
\node[copier] (c2) at (0,1.45) {};
\node[arrow box] (f) at (0,0.95) {$G$};
\node[effect] (Y1) at (0.7,1.95) {$A$};
\node[effect] (Y2) at (-0.7,1.95) {$A^*$};
\draw (0,-0.5) to (p);
\draw (p) to (f);
\draw (c2) to[out=15,in=-90] (Y1);
\draw (c2) to[out=165,in=-90] (Y2);
\draw (f) to (c2);
\end{tikzpicture}}}
\]
for some $A\in\mA$, then
\[
\vcenter{\hbox{
\begin{tikzpicture}[font=\small]
\node[arrow box] (p) at (0,0.1) {$\xi$};
\node[copier] (c2) at (0,0.75) {};
\node[arrow box] (f) at (0.7,1.4) {$G$};
\node[arrow box] (e) at (-0.7,1.4) {$G$};
\coordinate (Y1) at (0.7,2.55);
\node[effect] (Y2) at (-0.7,1.95) {$A^*$};
\draw (0,-0.5) to (p);
\draw (p) to (c2);
\draw (c2) to[out=15,in=-90] (f);
\draw (c2) to[out=165,in=-90] (e);
\draw (f) to (Y1);
\draw (e) to (Y2);
\end{tikzpicture}}}
%%%%%%%%%%%%%%%%%%%
\quad=\quad
%%%%%%%%%%%%%%%%%%%
\vcenter{\hbox{
\begin{tikzpicture}[font=\small]
\node[arrow box] (p) at (0,0.1) {$\xi$};
\node[copier] (c2) at (0,1.45) {};
\node[arrow box] (f) at (0,0.95) {$G$};
\coordinate (Y1) at (0.7,2.55);
\node[effect] (Y2) at (-0.7,1.95) {$A^*$};
\draw (0,-0.5) to (p);
\draw (p) to (f);
\draw (c2) to[out=15,in=-90] (Y1);
\draw (c2) to[out=165,in=-90] (Y2);
\draw (f) to (c2);
\end{tikzpicture}}}
%%%%%%%%%%%%%%%%%%%
\quad\text{ and }\quad
%%%%%%%%%%%%%%%%%%%
\vcenter{\hbox{
\begin{tikzpicture}[font=\small]
\node[arrow box] (p) at (0,0.1) {$\xi$};
\node[copier] (c2) at (0,0.75) {};
\node[arrow box] (f) at (0.7,1.4) {$G$};
\node[arrow box] (e) at (-0.7,1.4) {$G$};
\node[effect] (Y1) at (0.7,1.95) {$A$};
\coordinate (Y2) at (-0.7,2.55) {};
\draw (0,-0.5) to (p);
\draw (p) to (c2);
\draw (c2) to[out=15,in=-90] (f);
\draw (c2) to[out=165,in=-90] (e);
\draw (f) to (Y1);
\draw (e) to (Y2);
\end{tikzpicture}}}
%%%%%%%%%%%%%%%%%%%
\quad=\quad
%%%%%%%%%%%%%%%%%%%
\vcenter{\hbox{
\begin{tikzpicture}[font=\small]
\node[arrow box] (p) at (0,0.1) {$\xi$};
\node[copier] (c2) at (0,1.45) {};
\node[arrow box] (f) at (0,0.95) {$G$};
\node[effect] (Y1) at (0.7,1.95) {$A$};
\coordinate (Y2) at (-0.7,2.55) {};
\draw (0,-0.5) to (p);
\draw (p) to (f);
\draw (c2) to[out=15,in=-90] (Y1);
\draw (c2) to[out=165,in=-90] (Y2);
\draw (f) to (c2);
\end{tikzpicture}}}
\]
\end{lem}

\bprf
[Proof of Lemma~\ref{lem:relmulttheorem}]
Fix $D\in\mA$ and $\l>0$. Then 
\be
G\big((A+\l D)^*(A+\l D)\big)-G(A+\l D)^*G(A+\l D)\ge0
\ee
by Kadison--Schwarz for $G$. Hence 
\be
\begin{split}
0&\le\xi\Big(G\big((A+\l D)^*(A+\l D)\big)-G(A+\l D)^*G(A+\l D)\Big)\quad\text{ since $\xi$ is positive}\\
&=\xi\Big(G(A^*A)-G(A)^*G(A)+\\
&\qquad+\l\big(G(A^*D)+G(D^*A)-G(A)^*G(D)-G(D)^*G(A)\big)\\
&\qquad+\l^2\big(G(D^*D)-G(D)^*G(D)\big)\Big)
\end{split}
\ee
The order $\l^0$ term vanishes by assumption. Hence, dividing both sides by $\l$ and taking the $\l\to0$ limit causes the initially $\l^2$ term to vanish and we are left with 
\be
0\le\xi\big(G(A^*D)+G(D^*A)-G(A)^*G(D)-G(D)^*G(A)\big).
\ee
Taking $A\to iA$ and $D\to-iD$ gives the reverse inequality. Hence, 
\be
0=\xi\big(G(A^*D)+G(D^*A)-G(A)^*G(D)-G(D)^*G(A)\big).
\ee
Now, taking $A\to A$ and $D\to iD$ gives 
\be
0=\xi\big(G(A^*D)-G(D^*A)-G(A)^*G(D)+G(D)^*G(A)\big).
\ee
Adding the two resulting expressions together gives
\be
0=\xi\big(G(A^*D)-G(A)^*G(D)\big).
\ee
The other condition is obtained similarly. 
\eprf

The reason for the terminology is because a special case (when $\xi$ is the identity) reproduces the standard Multiplication Lemma (Lemma~\ref{lem:multiplicationtheorem}). 

%\blem
%\label{lem:disintaehomproperty}
\begin{lem}[The relative conditional expectation property]{lem:disintaehomproperty}
Let $\mA$ and $\mB$ be $C^*$-algebras, 
let $\mB\xstoch{\xi}\mC$ be a positive (not necessarily unital nor SP) map, let $\mB\xstoch{F}\mA$ be an SPU map, and suppose there exists an SPU map $\mA\xstoch{G}\mB$ such that $G\circ F\underset{\raisebox{.6ex}[0pt][0pt]{\scriptsize$\xi$}}{=}\id_{\mB}$. Then  
\[
\xi(B^*B)
=\xi\Big(G\big(F(B)\big)^*G\big(F(B)\big)\Big)
=\xi\Big(G\big(F(B)^*F(B)\big)\Big)
=\xi\Big(G\big(F(B^*B)\big)\Big)
\quad\forall\;B\in\mB.
\]
%\elem
String-diagrammatically, 
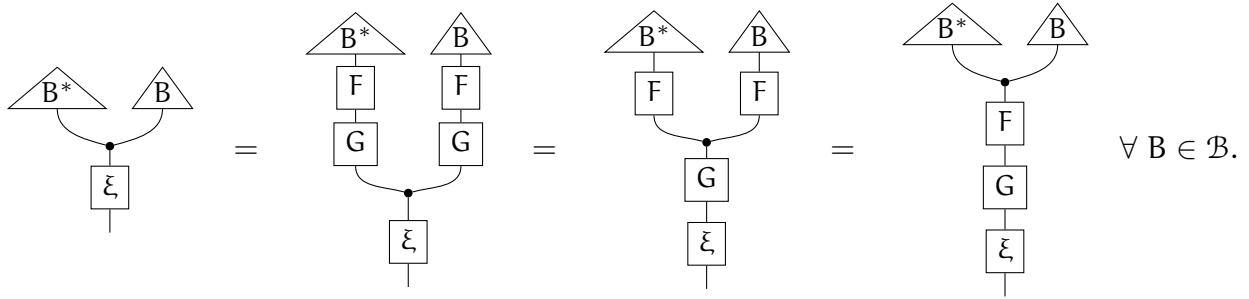
\[
\vcenter{\hbox{
\begin{tikzpicture}[font=\small]
\node[arrow box] (p) at (0,0.9) {$\xi$};
\node[copier] (c2) at (0,1.45) {};
\node[effect] (Y1) at (0.7,1.95) {$B$};
\node[effect] (Y2) at (-0.7,1.95) {$B^*$};
\draw (0,0.3) to (p);
\draw (p) to (c2);
\draw (c2) to[out=15,in=-90] (Y1);
\draw (c2) to[out=165,in=-90] (Y2);
\end{tikzpicture}}}
%%%%%%%%%%%%%%%%%%%
\quad=\quad
%%%%%%%%%%%%%%%%%%%
\vcenter{\hbox{
\begin{tikzpicture}[font=\small]
\node[arrow box] (p) at (0,0.1) {$\xi$};
\node[copier] (c2) at (0,0.75) {};
\node[arrow box] (f) at (0.7,1.4) {$G$};
\node[arrow box] (e) at (-0.7,1.4) {$G$};
\node[arrow box] (fL) at (0.7,2.15) {$F$};
\node[arrow box] (fR) at (-0.7,2.15) {$F$};
\node[effect] (Y1) at (0.7,2.6) {$B$};
\node[effect] (Y2) at (-0.7,2.6) {$B^*$};
\draw (0,-0.5) to (p);
\draw (p) to (c2);
\draw (c2) to[out=15,in=-90] (f);
\draw (c2) to[out=165,in=-90] (e);
\draw (f) to (fL);
\draw (e) to (fR);
\draw (fL) to (Y1);
\draw (fR) to (Y2);
\end{tikzpicture}}}
%%%%%%%%%%%%%%%%%%%
\quad=\quad
%%%%%%%%%%%%%%%%%%%
\vcenter{\hbox{
\begin{tikzpicture}[font=\small]
\node[arrow box] (p) at (0,0.1) {$\xi$};
\node[copier] (c2) at (0,1.45) {};
\node[arrow box] (f) at (0,0.95) {$G$};
\node[arrow box] (fL) at (-0.7,2.1) {$F$};
\node[arrow box] (fR) at (0.7,2.1) {$F$};
\node[effect] (Y1) at (0.7,2.65) {$B$};
\node[effect] (Y2) at (-0.7,2.65) {$B^*$};
\draw (0,-0.5) to (p);
\draw (p) to (f);
\draw (c2) to[out=15,in=-90] (fR);
\draw (c2) to[out=165,in=-90] (fL);
\draw (f) to (c2);
\draw (fL) to (Y2);
\draw (fR) to (Y1);
\end{tikzpicture}}}
%%%%%%%%%%%%%%%%%%%
\quad=\quad
%%%%%%%%%%%%%%%%%%%
\vcenter{\hbox{
\begin{tikzpicture}[font=\small]
\node[arrow box] (p) at (0,-0.8) {$\xi$};
\node[arrow box] (F) at (0,0.9) {$F$};
\node[arrow box] (G) at (0,0.05) {$G$};
\node[copier] (c2) at (0,1.45) {};
\node[effect] (Y1) at (0.7,1.95) {$B$};
\node[effect] (Y2) at (-0.7,1.95) {$B^*$};
\draw (0,-1.4) to (p);
\draw (p) to (G);
\draw (G) to (F);
\draw (F) to (c2);
\draw (c2) to[out=15,in=-90] (Y1);
\draw (c2) to[out=165,in=-90] (Y2);
\end{tikzpicture}}}
\quad\forall\;B\in\mB.
\]
\end{lem}

\bprf
[Proof of Lemma~\ref{lem:disintaehomproperty}]
Let $B\in\mB$. Then 
\be
\label{eq:squeezing}
\begin{split}
\xi(B^*B)&=\xi\Big(G\big(F(B)\big)^*G\big(F(B)\big)\Big)\quad\text{since $G\circ F\aeequals{\xi}\id_{\mB}$}\\
&\le\xi\Big(G\big(F(B)^*F(B)\big)\Big)\quad\text{ by KS for $G$ and positivity of $\xi$}\\
&\le\xi\Big(G\big(F(B^*B)\big)\Big)\quad\text{ by KS for $F$ and positivity of $\xi$}\\
&=\xi(B^*B)\quad\text{since $G\circ F\aeequals{\xi}\id_{\mB}$.}
\end{split}
\ee
Hence, all intermediate inequalities become equalities. 
\eprf

We finally proceed to prove that $\fdCAlgSPU$ is a.e.\ modular (and hence $\fdCAlgCPU$ as well). 

\bprf
[Proof of Theorem~\ref{thm:SPUmodular}]
By Theorem~\ref{thm:aemodbayesdisint}, it suffices to prove that if $\mB\xstoch{F}\mA,\mA\xstoch{G}\mB,$ and $\mB\xstoch{\xi}\mC$ satisfy the condition that $G\circ F\aeequals{\xi}\id_{\mB}$, then conditions (a) and (b) of Theorem~\ref{thm:aemodbayesdisint} hold. In what follows, set $\omega:=\xi\circ G$. 
\begin{enumerate}[(a)]
\item
Let $B\in\mB$. Then the conditions of Lemma~\ref{lem:disintaehomproperty} hold. In particular, \emph{temporarily} setting $A:=F(B)$ provides
\be
\xi\big(G(A^*A)\big)=\xi\big(G(A)^*G(A)\big). 
\ee
Thus, the relative Multiplication Lemma (Lemma~\ref{lem:relmulttheorem}) applies giving
\be
\label{eq:RMTimpliesbayes}
\xi\Big(G\big(D^*F(B)\big)\Big)
\overset{\text{Lem~\ref{lem:relmulttheorem}}}{=\joinrel=\joinrel=\joinrel=\joinrel=}\xi\Big(G(D)^*G\big(F(B)\big)\Big)
\overset{G\circ F\underset{\raisebox{.6ex}[0pt][0pt]{\tiny$\xi$}}{=}\id_{\mB}}{=\joinrel=\joinrel=\joinrel=\joinrel=}\xi\big(G(D)^*B\big)
\ee
for all $D\in\mA$. Since $B$ was arbitrary and all morphisms are $*$-preserving (and because $*$ is an involution), 
\be
\omega\big(AF(B)\big)
\overset{\omega:=\xi\circ G}{=\joinrel=\joinrel=\joinrel=\joinrel=}
\xi\Big(G\big(AF(B)\big)\Big)
\overset{\text{(\ref{eq:RMTimpliesbayes})}}{=\joinrel=\joinrel=}
\xi\big(G(A)B\big)
\qquad\forall\;A\in\mA,\;B\in\mB.
\ee
Thus, $G$ is a Bayesian inverse of $(F,\omega,\xi)$. 
\item
By the weak a.e.\ Multiplication Lemma (Lemma~\ref{lem:Attalslemmaaegeneral}), it suffices to prove $F(B^*B)-F(B)^*F(B)\in\mathcal{N}_{\omega}$ for all $B\in\mB$. By the Kadison--Schwarz inequality for $F$, we have 
\be
\label{eq:FBBKadisonS}
F(B^*B)-F(B)^*F(B)\ge0. 
\ee
Hence, $F(B^*B)-F(B)^*F(B)\in\mathcal{N}_{\omega}$ if and only if $F(B^*B)-F(B)^*F(B)\in\ker\omega$ by Lemma~\ref{lem:minilemmaleftideal}. 
Therefore, to see that the latter holds, note that 
\be
\begin{split}
\w\big(F(B^*B)-F(B)^*F(B)\big)
&=\xi\Big(G\big(F(B^*B)\big)-G\big(F(B)^*F(B)\big)\Big)\quad\text{ since $\w=\xi\circ F$}\\
&=0\quad\text{ by Lemma~\ref{lem:disintaehomproperty}}.
\end{split}
\ee
Thus, $F$ is $\omega$-a.e.\ deterministic. \qedhere
\end{enumerate}
\eprf

\vspace{-2mm}
\begin{rmk}[On partial invertibility and the $*$-homomorphism property]{rmk:leftinversednidet}
Our proof of Theorem~\ref{thm:SPUmodular} (specifically condition~(b) in Theorem~\ref{thm:aemodbayesdisint} in the setting of $C^*$-algebras) was inspired by Attal's Theorem~6.38 in~\cite{Attal}, although the latter has slightly different assumptions. It states that if $\Hi$ is a Hilbert space and if a CPU map $F:\mathcal{B}(\Hi)\stoch\mathcal{B}(\Hi)$ has a CPU left inverse $G,$ then $F$ is a $^*$-homomorphism. If one allows a different codomain for $F$, then this claim is false. Indeed, a simple example, even in finite dimensions, is provided in Example~\ref{exa:aedniaeequivdet}. 
It is therefore interesting that merely adding a state-preserving assumption to Attal's theorem guarantees the $^*$-homomorphism claim almost surely regardless of the domain and codomain. 
\end{rmk}

\begin{rmk}[A.e.\ determinism for disintegrations on $C^*$-algebras]{rmk:aedeterministicDNIequivtodeterministic}
If $(\mB,\xi)\xstoch{F}(\mA,\omega)$ is a state-preserving SPU map for which there exists a CPU disintegration, it is generally \emph{not} true that $F$ is $\w$-a.e.\ equivalent to a $*$-homomorphism. This is because there are CPU disintegrations of CPU maps of the form $\mathcal{M}_{n}(\C)\xstoch{F}\mathcal{M}_{m}(\C)$, where $m$ is not a multiple of $n$. Indeed, the example in Remark~\ref{exa:aedniaeequivdet} provides such an instance if one equips $\mathcal{M}_{m}(\C)$ with a density matrix of the form 
$\left[\begin{smallmatrix}\sigma&0\\0&0\end{smallmatrix}\right]$, 
where $\sigma$ is of size $n\times n$. 
\end{rmk}

Another immediate corollary of Theorem~\ref{thm:SPUmodular} is a \emph{non-commutative} version of the Fisher--Neyman factorization theorem. The reader is referred to~\cite[Theorem~14.5]{Fr19} for details (and noting that a.e.\ modularity suffices for the theorem to hold). An explicit form of the factorization will be provided in future work (Theorem~5.108 in~\cite{PaRu19} describes a special case). Rather than restating the factorization theorem here, we instead mention an example relating sufficiency in the sense of~\cite[Definition~14.3]{Fr19} to Umegaki's notion of sufficiency. 

\begin{exa}[Umegaki sufficiency in quantum statistical decision theory]{exa:Umegakisufficientstatistic}
In 1958, Umegaki extended the notion of sufficient statistics from classical statistics to operator algebras~\cite{Um59} (see also Holevo's work~\cite[Section~2]{Ho73}). It involves defining when a $C^*$-subalgebra $\mB$ of $\mA$ is sufficient with respect to a family of states (Holevo usually sets $\mA=\mB(\Hi)$ with $\Hi$ a Hilbert space, though this is not necessary). We recall this definition and show how this is a special case of the definition of a sufficient statistic (in the sense of~\cite[Definition~14.3]{Fr19}) in the category $\fdCAlgCPU$. 

First, a (finite) \define{family of states} on a $C^*$-algebra $\mA$ parametrized by the finite set $\Theta$ is a collection of states $\{\omega_{\theta}\}_{\theta\in\Theta}$ on $\mA$. Such a collection is equivalent to the datum of a CPU map $\mA\xstoch{\omega}\C^{\Theta}$. To see this, suppose one has such a CPU map $\omega$. Then, given $\theta\in\Theta$, first define the evaluation at $\theta$ map $\C^{\Theta}\xrightarrow{\mathrm{ev}_{\theta}}\C$ by sending a function $\varphi\in\C^{\Theta}$ to $\varphi(\theta)$. Thus, each $\mA\xstoch{\omega}\C^{\Theta}\xrightarrow{\mathrm{ev}_{\theta}}\C$ is a CPU map and hence defines a family of states $\{\omega_{\theta}\}_{\theta\in\Theta}$ on $\mA$ parametrized by $\Theta$. Conversely, given such a family of states, one obtains such a CPU map $\omega$. In this way, a family of states (in the way defined by Holevo) can be identified with a CPU map, which we will proceed to do. Similarly, set $\xi:=\omega\circ F$ so that $\xi$ is the corresponding family of states on $\mB$. 

Second, let $\mB\xrightarrow{F}\mA$ be the $*$-homomorphism corresponding to the inclusion (though we could easily imagine any $*$-homomorphism). Following Umegaki, Holevo then defines $F$ to be \define{sufficient} for a family of states $\omega$ iff there exists a CPU map $\mA\xstoch{G}\mB$ such that 
\[
G\circ F=\id_{\mB}
\qquad\text{ and }\qquad
\xy0;/r.25pc/:
(0,-7.5)*+{\C}="C";
(-12.5,7.5)*+{\mA}="H";
(12.5,7.5)*+{\mB}="K";
{\ar@{~>}"H";"C"_{\mathrm{ev}_{\theta}\circ\w}};
{\ar@{~>}"K";"C"^{\mathrm{ev}_{\theta}\circ\xi}};
{\ar@{~>}"H";"K"^{G}};
{\ar@{=}(-3,0);(5,4.5)};
\endxy
\quad\forall\;\theta\in\Theta.
\]
Since functions are uniquely determined by their values on points, this definition is equivalent to 
\[
G\circ F=\id_{\mB}
\qquad\text{ and }\qquad
\xy0;/r.25pc/:
(0,-7.5)*+{\C^{\Theta}}="C";
(-12.5,7.5)*+{\mA}="H";
(12.5,7.5)*+{\mB}="K";
{\ar@{~>}"H";"C"_{\w}};
{\ar@{~>}"K";"C"^{\xi}};
{\ar@{~>}"H";"K"^{G}};
{\ar@{=}(-3,0);(5,4.5)};
\endxy
,
\]
which is a special case of~\cite[Definition~14.3]{Fr19} (but generalized to the non-commutative setting) due to Theorem~\ref{thm:SPUmodular} and Corollary~\ref{cor:binversiondisint} relating Bayesian inversion to disintegration. In this way, Umegaki's definition of sufficiency can be stated in terms of disintegrations and Bayesian inversion, where the states have been replaced by more general morphisms (cf.\ Remark~\ref{rmk:statestomapsfordisint}).  
\end{exa}

We now move on to another axiom introduced by Fritz~\cite[Definition~13.16]{Fr19}, which includes S-positivity and a.e.\ modularity as a special case. We will see that these are also special cases even in the non-commutative setting (though one must be careful in choosing the axioms appropriately in this more general situation). 

\begin{defn}[Strict positivity]{defn:strictpositive}
An even subcategory $\mC$ of a quantum Markov category is \define{left strictly positive} iff for all $f,g,q$ in $\mC$ such that $f\circ g$ is left $q$-a.e.\ deterministic, then 
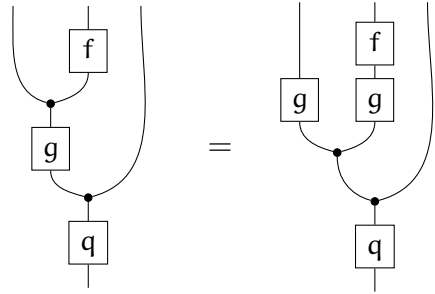
\[
\vcenter{\hbox{
\begin{tikzpicture}[font=\small]
\node[arrow box] (q) at (0,-0.3) {$q$};
\node[copier] (c) at (0,0.3) {};
\node[copier] (c2) at (-0.5,1.55) {};
\node[arrow box] (g) at (-0.5,0.95) {$g$};
\node[arrow box] (f) at (0,2.25) {$f$};
\coordinate (X) at (0.7,2.85);
\coordinate (Y1) at (0,2.85);
\coordinate (Y2) at (-1,2.85);
\draw (0,-0.9) to (q);
\draw (q) to (c);
\draw (c) to[out=15,in=-90] (X);
\draw (c) to[out=165,in=-90] (g);
\draw (c2) to[out=15,in=-90] (f);
\draw (f) to (Y1);
\draw (c2) to[out=165,in=-90] (Y2);
\draw (g) to (c2);
\end{tikzpicture}}}
%%%%%%%%%%%%%%%%%%%
\quad=\quad
%%%%%%%%%%%%%%%%%%%
\vcenter{\hbox{
\begin{tikzpicture}[font=\small]
\node[arrow box] (q) at (0,-0.3) {$q$};
\node[copier] (c) at (0,0.3) {};
\node[copier] (c2) at (-0.5,0.95) {};
\node[arrow box] (g) at (-1,1.65) {$g$};
\coordinate (X) at (0.7,2.95);
\node[arrow box] (g2) at (0,1.65) {$g$};
\node[arrow box] (f) at (0,2.4) {$f$};
\coordinate (Y1) at (0,2.95);
\coordinate (Y2) at (-1,2.95);
\draw (0,-0.9) to (q);
\draw (q) to (c);
\draw (c) to[out=15,in=-90] (X);
\draw (c) to[out=165,in=-90] (c2);
\draw (c2) to[out=15,in=-90] (g2);
\draw (g2) to (f); 
\draw (f) to (Y1); 
\draw (g) to (Y2);
\draw (c2) to[out=165,in=-90] (g);
\end{tikzpicture}}}
\]
must also hold. A similar definition is made for a \define{right strictly positive} subcategory.%
\end{defn}

\vspace{-3mm}
\begin{exa}[Neither $\fdCAlgSPU$ nor $\fdCAlgCPU$ are strictly positive]{exa:strictpositive}
Let $\mM_{2}(\C)\xstoch{F}\mM_{4}(\C)$, $\mM_{4}(\C)\xstoch{G}\mM_{3}(\C)$, and $\mM_{3}(\C)\xstoch{\xi}\C$ be the CPU maps given by 
\[
F\left(\begin{bmatrix}b_{11}&b_{12}\\b_{21}&b_{22}\end{bmatrix}\right):=\begin{bmatrix}b_{11}&b_{12}&0&0\\b_{21}&b_{22}&0&0\\0&0&b_{11}&b_{12}\\0&0&b_{21}&b_{22}\end{bmatrix},\quad
G\left(\begin{bmatrix}a_{11}&a_{12}&a_{13}&a_{14}\\a_{21}&a_{22}&a_{23}&a_{24}\\a_{31}&a_{32}&a_{33}&a_{34}\\a_{41}&a_{42}&a_{43}&a_{44}\end{bmatrix}\right):=\begin{bmatrix}a_{11}&a_{12}&a_{13}\\a_{21}&a_{22}&a_{23}\\a_{31}&a_{32}&a_{33}\end{bmatrix},
\]
(in fact, $F$ here is a $*$-homomorphism) 
and
\[
\xi=\tr(\s\;\cdot\;),\quad\text{ where }\quad\s:=\frac{1}{2}\begin{bmatrix}1&0&0\\0&1&0\\0&0&0\end{bmatrix}.
\]
Then $G\circ F$ is $\xi$-a.e.\ deterministic (both right and left, necessarily so because $F$, $G$, and $\xi$ are $*$-preserving). However, one can check that 
\[
P_{\xi}G\big(AF(B)\big)\ne P_{\xi}G(A)G\Big(\big(F(B)\big)\Big). 
\]
Hence (by Lemma~\ref{thm:ncaeequivalence}), the equality in Definition~\ref{defn:strictpositive} fails. Note, however, that 
\[
P_{\xi}G\big(F(B)A\big)=P_{\xi}G\Big(\big(F(B)\big)\Big)G(A). 
\]
does hold in this example.%
%footnote
\footnote{This is not an indication that we have chosen the incorrect definition of strict positivity. Our definition is chosen so that strict positivity implies a.e.\ modularity and S-positivity (cf.\ Proposition~\ref{prop:strictimplies}). Note that none of these subtleties arise for classical Markov categories. See Remark~\ref{rmk:strictpossubtle} for more details.}
%end footnote
\end{exa}

\vspace{-2mm}
\begin{prop}[Strict positivity implies a.e.\ modularity and S-positivity]{prop:strictimplies}
Suppose $\mC$ is a $*$-preserving, deterministically reasonable (cf.\ Definition~\ref{defn:detreason}), and strictly positive subcategory of a quantum Markov category. Then $\mC$ is a.e.\ modular and S-positive. 
\end{prop}

\bprf
{\color{white}{you found me!}}

\noindent 
(Proof of a.e.\ modularity)
Suppose $f\circ g\aeequals{q}\id$. Then by the deterministically reasonable assumption, $f\circ g$ is $q$-a.e.\ deterministic. Hence
\be
\begin{split}
\vcenter{\hbox{
\begin{tikzpicture}[font=\small]
\node[arrow box] (q) at (0,-0.3) {$q$};
\node[copier] (c) at (0,0.3) {};
\node[copier] (c2) at (-0.5,1.55) {};
\node[arrow box] (g) at (-0.5,0.95) {$g$};
\node[arrow box] (f) at (0,2.25) {$f$};
\coordinate (X) at (0.7,2.85);
\coordinate (Y1) at (0,2.85);
\coordinate (Y2) at (-1,2.85);
\draw (0,-0.9) to (q);
\draw (q) to (c);
\draw (c) to[out=15,in=-90] (X);
\draw (c) to[out=165,in=-90] (g);
\draw (c2) to[out=15,in=-90] (f);
\draw (f) to (Y1);
\draw (c2) to[out=165,in=-90] (Y2);
\draw (g) to (c2);
\end{tikzpicture}}}
%%%%%%%%%%%%%%%%%%%
\;&\overset{\text{Defn~\ref{defn:strictpositive}}}{=\joinrel=\joinrel=\joinrel=\joinrel=\joinrel=}
\;
%%%%%%%%%%%%%%%%%%%
\vcenter{\hbox{
\begin{tikzpicture}[font=\small]
\node[arrow box] (q) at (0,-0.3) {$q$};
\node[copier] (c) at (0,0.3) {};
\node[copier] (c2) at (-0.5,0.95) {};
\node[arrow box] (g) at (-1,1.65) {$g$};
\coordinate (X) at (0.7,2.95);
\node[arrow box] (g2) at (0,1.65) {$g$};
\node[arrow box] (f) at (0,2.4) {$f$};
\coordinate (Y1) at (0,2.95);
\coordinate (Y2) at (-1,2.95);
\draw (0,-0.9) to (q);
\draw (q) to (c);
\draw (c) to[out=15,in=-90] (X);
\draw (c) to[out=165,in=-90] (c2);
\draw (c2) to[out=15,in=-90] (g2);
\draw (g2) to (f); 
\draw (f) to (Y1); 
\draw (g) to (Y2);
\draw (c2) to[out=165,in=-90] (g);
\end{tikzpicture}}}
%%%%%%%%%%%%%%%%%%%
\;
\overset{f\circ g\aeequals{q}\id}{=\joinrel=\joinrel=\joinrel=\joinrel=\joinrel=}
\;
%%%%%%%%%%%%%%%%%%%
\vcenter{\hbox{
\begin{tikzpicture}[font=\small]
\node[arrow box] (q) at (0,-0.3) {$q$};
\node[copier] (c) at (0,0.3) {};
\node[copier] (c2) at (-0.5,0.95) {};
\node[arrow box] (g) at (-1,1.65) {$g$};
\coordinate (X) at (1.0,2.95);
\node[arrow box] (g2) at (0,1.65) {$g$};
\node[arrow box] (f) at (0,2.4) {$f$};
\node[arrow box] (g3) at (1.0,1.65) {$g$};
\node[arrow box] (f3) at (1.0,2.4) {$f$};
\coordinate (Y1) at (0,2.95);
\coordinate (Y2) at (-1,2.95);
\draw (0,-0.9) to (q);
\draw (q) to (c);
\draw (c) to[out=15,in=-90] (g3);
\draw (g3) to (f3);
\draw (f3) to (X);
\draw (c) to[out=165,in=-90] (c2);
\draw (c2) to[out=15,in=-90] (g2);
\draw (g2) to (f); 
\draw (f) to (Y1); 
\draw (g) to (Y2);
\draw (c2) to[out=165,in=-90] (g);
\end{tikzpicture}}}
%%%%%%%%%%%%%%%%%%%
\quad=\quad
%%%%%%%%%%%%%%%%%%%
\vcenter{\hbox{
\begin{tikzpicture}[font=\small]
\node[arrow box] (q) at (0,-0.3) {$q$};
\node[copier] (c) at (0,0.3) {};
\node[copier] (c2) at (0.5,0.95) {};
\node[arrow box] (g) at (-1,1.65) {$g$};
\coordinate (X) at (1.0,2.95);
\node[arrow box] (g2) at (0,1.65) {$g$};
\node[arrow box] (f) at (0,2.4) {$f$};
\node[arrow box] (g3) at (1.0,1.65) {$g$};
\node[arrow box] (f3) at (1.0,2.4) {$f$};
\coordinate (Y1) at (0,2.95);
\coordinate (Y2) at (-1,2.95);
\draw (0,-0.9) to (q);
\draw (q) to (c);
\draw (c) to[out=165,in=-90] (g);
\draw (g3) to (f3);
\draw (f3) to (X);
\draw (c) to[out=15,in=-90] (c2);
\draw (c2) to[out=165,in=-90] (g2);
\draw (g2) to (f); 
\draw (f) to (Y1); 
\draw (g) to (Y2);
\draw (c2) to[out=15,in=-90] (g3);
\end{tikzpicture}}}
\\
%%%%%%%%%%%%%%%%%%%
\;
&\overset{\text{Defn~\ref{defn:detreason}}}{=\joinrel=\joinrel=\joinrel=\joinrel=\joinrel=}
\;
%%%%%%%%%%%%%%%%%%%
%%%%%%%%%%%%%%%%%%%
\vcenter{\hbox{
\begin{tikzpicture}[font=\small]
\node[arrow box] (q) at (0,-0.3) {$q$};
\node[copier] (c) at (0,0.3) {};
\node[copier] (c2) at (0.5,2.35) {};
\node[arrow box] (g) at (-1,1.85) {$g$};
\coordinate (X) at (1.0,2.75);
\node[arrow box] (g3) at (0.5,0.95) {$g$};
\node[arrow box] (f3) at (0.5,1.75) {$f$};
\coordinate (Y1) at (0,2.75);
\coordinate (Y2) at (-1,2.75);
\draw (0,-0.9) to (q);
\draw (q) to (c);
\draw (c) to[out=165,in=-90] (g);
\draw (g3) to (f3);
\draw (f3) to (c2);
\draw (c) to[out=15,in=-90] (g3);
\draw (c2) to[out=165,in=-90] (Y1);
\draw (g) to (Y2);
\draw (c2) to[out=15,in=-90] (X);
\end{tikzpicture}}}
%%%%%%%%%%%%%%%%%%%
\quad\overset{f\circ g\aeequals{q}\id}{=\joinrel=\joinrel=\joinrel=\joinrel=\joinrel=}\quad
%%%%%%%%%%%%%%%%%%%
\vcenter{\hbox{
\begin{tikzpicture}[font=\small,xscale=-1]
\node[arrow box] (q) at (0,-0.3) {$q$};
\node[copier] (c) at (0,0.3) {};
\node[copier] (c2) at (-0.5,1.45) {};
\node[arrow box] (g) at (0.6,1.05) {$g$};
\coordinate (X) at (0.6,1.95);
\coordinate (Y1) at (-1,1.95);
\coordinate (Y2) at (0,1.95);
\draw (0,-0.9) to (q);
\draw (q) to (c);
\draw (c) to[out=165,in=-90] (c2);
\draw (c) to[out=15,in=-90] (g);
\draw (c2) to[out=165,in=-90] (-1,1.95);
\draw (c2) to[out=15,in=-90] (0,1.95);
\draw (g) to (X);
\end{tikzpicture}}}
%%%%%%%%%%%%%%%%%%%
\quad=\quad
%%%%%%%%%%%%%%%%%%%
\vcenter{\hbox{
\begin{tikzpicture}[font=\small]
\node[arrow box] (q) at (0,-0.3) {$q$};
\node[copier] (c) at (0,0.3) {};
\node[copier] (c2) at (-0.5,0.95) {};
\node[arrow box] (g) at (-1,1.65) {$g$};
\coordinate (X) at (0.7,2.35);
\coordinate (Y1) at (0,2.35);
\coordinate (Y2) at (-1,2.35);
\draw (0,-0.9) to (q);
\draw (q) to (c);
\draw (c) to[out=15,in=-90] (X);
\draw (c) to[out=165,in=-90] (c2);
\draw (c2) to[out=15,in=-90] (Y1);
\draw (g) to (Y2);
\draw (c2) to[out=165,in=-90] (g);
\end{tikzpicture}}}
\;\;.
\end{split}
\ee
Note that the $*$-preserving assumption was used on multiple occasions in some of these equalities. For example, the fact that $f\circ g$ is \emph{left} $q$-a.e.\ deterministic was used in the first equality (which followed from the assumption that $f\circ g$ is \emph{left} $q$-a.e.\ equal to $\id$) and the fact that it is \emph{right} $q$-a.e.\ deterministic was used in the fourth equality. In addition, the fact that $f\circ g$ is \emph{right} $q$-a.e.\ equal to $\id$ was used in the second and fifth equalities. 

\noindent
(Proof of S-positivity)
If $f\circ g$ is deterministic, then $f\circ g$ is left (and right) $\id$-a.e.\ deterministic. Therefore, 
\be
\vcenter{\hbox{%
\begin{tikzpicture}[font=\footnotesize]
\coordinate (q) at (0,-0.2) {};
\node[copier] (copier) at (0,0.3) {};
\node[arrow box] (g2) at (0.5,0.95) {$g$};
\node[arrow box] (e2) at (0.5,1.75) {$f$};
\node[arrow box] (g) at (-0.5,0.95) {$g$};
\coordinate (X) at (0.5,2.4);
\coordinate (Y) at (-0.5,2.4);
\draw (q) to (copier);
\draw (copier) to[out=15,in=-90] (g2);
\draw (g2) to (e2);
\draw (e2) to (X);
\draw (copier) to[out=165,in=-90] (g);
\draw (g) to (Y);
\end{tikzpicture}}}
%%%%%%%%%%%%%%%%%%%
\quad=\quad
%%%%%%%%%%%%%%%%%%%
\vcenter{\hbox{
\begin{tikzpicture}[font=\small]
\node[arrow box] (q) at (0,-0.3) {$\id$};
\node[copier] (c) at (0,0.3) {};
\node[copier] (c2) at (-0.5,0.95) {};
\node[arrow box] (g) at (-1,1.65) {$g$};
\node[discarder] (X) at (0.7,2.65) {};
\node[arrow box] (g2) at (0,1.65) {$g$};
\node[arrow box] (f) at (0,2.4) {$f$};
\coordinate (Y1) at (0,2.95);
\coordinate (Y2) at (-1,2.95);
\draw (0,-0.9) to (q);
\draw (q) to (c);
\draw (c) to[out=15,in=-90] (X);
\draw (c) to[out=165,in=-90] (c2);
\draw (c2) to[out=15,in=-90] (g2);
\draw (g2) to (f); 
\draw (f) to (Y1); 
\draw (g) to (Y2);
\draw (c2) to[out=165,in=-90] (g);
\end{tikzpicture}}}
%%%%%%%%%%%%%%%%%%%
\quad
\overset{\text{Defn~\ref{defn:strictpositive}}}{=\joinrel=\joinrel=\joinrel=\joinrel=\joinrel=}
\quad
%%%%%%%%%%%%%%%%%%%
\vcenter{\hbox{
\begin{tikzpicture}[font=\small]
\node[arrow box] (q) at (0,-0.3) {$\id$};
\node[copier] (c) at (0,0.3) {};
\node[copier] (c2) at (-0.5,1.55) {};
\node[arrow box] (g) at (-0.5,0.95) {$g$};
\node[arrow box] (f) at (0,2.25) {$f$};
\node[discarder] (X) at (0.7,2.65) {};
\coordinate (Y1) at (0,2.95);
\coordinate (Y2) at (-1,2.95);
\draw (0,-0.9) to (q);
\draw (q) to (c);
\draw (c) to[out=15,in=-90] (X);
\draw (c) to[out=165,in=-90] (g);
\draw (c2) to[out=15,in=-90] (f);
\draw (f) to (Y1);
\draw (c2) to[out=165,in=-90] (Y2);
\draw (g) to (c2);
\end{tikzpicture}}}
%%%%%%%%%%%%%%%%%%%
\quad=\quad
%%%%%%%%%%%%%%%%%%%
\vcenter{\hbox{
\begin{tikzpicture}[font=\footnotesize]
\node[arrow box] (f) at (0,-0.3) {$g$};
\node[copier] (copier) at (0,0.3) {};
\node[arrow box] (g) at (0.5,0.95) {$f$};
\coordinate (X) at (-0.5,1.6);
\coordinate (Y) at (0.5,1.6);
\draw (0,-0.9) to (f);
\draw (f) to (copier);
\draw (copier) to[out=150,in=-90] (X);
\draw (copier) to[out=15,in=-90] (g);
\draw (g) to (Y);
\end{tikzpicture}}}
\ee
This reproduces the S-positivity axiom (cf.\ Definition~\ref{defn:positivecategory}) since $\mC$ is $*$-preserving.
\eprf

\begin{rmk}[Defining strict positivity and a.e.\ modularity are subtle]{rmk:strictpossubtle}
Notice that the proof of a.e.\ modularity in Proposition~\ref{prop:strictimplies} would fail (S-positivity would still hold) if we had chosen the alternative axiom obtained by replacing the equation in Definition~\ref{defn:strictpositive} with
\be
\label{eq:strictpositivealternative}
\vcenter{\hbox{
\begin{tikzpicture}[font=\small]
\node[arrow box] (q) at (0,-0.3) {$q$};
\node[copier] (c) at (0,0.3) {};
\node[copier] (c2) at (-0.5,1.55) {};
\node[arrow box] (g) at (-0.5,0.95) {$g$};
\node[arrow box] (f) at (-1,2.25) {$f$};
\coordinate (X) at (0.7,2.85);
\coordinate (Y1) at (0,2.85);
\coordinate (Y2) at (-1,2.85);
\draw (0,-0.9) to (q);
\draw (q) to (c);
\draw (c) to[out=15,in=-90] (X);
\draw (c) to[out=165,in=-90] (g);
\draw (c2) to[out=165,in=-90] (f);
\draw (f) to (Y2);
\draw (c2) to[out=15,in=-90] (Y1);
\draw (g) to (c2);
\end{tikzpicture}}}
%%%%%%%%%%%%%%%%%%%
\quad=\quad
%%%%%%%%%%%%%%%%%%%
\vcenter{\hbox{
\begin{tikzpicture}[font=\small]
\node[arrow box] (q) at (0,-0.3) {$q$};
\node[copier] (c) at (0,0.3) {};
\node[copier] (c2) at (-0.5,0.95) {};
\node[arrow box] (g) at (0,1.65) {$g$};
\coordinate (X) at (0.7,2.95);
\node[arrow box] (g2) at (-1,1.65) {$g$};
\node[arrow box] (f) at (-1,2.4) {$f$};
\coordinate (Y1) at (-1,2.95);
\coordinate (Y2) at (0,2.95);
\draw (0,-0.9) to (q);
\draw (q) to (c);
\draw (c) to[out=15,in=-90] (X);
\draw (c) to[out=165,in=-90] (c2);
\draw (c2) to[out=15,in=-90] (g);
\draw (g2) to (f); 
\draw (f) to (Y1); 
\draw (g) to (Y2);
\draw (c2) to[out=165,in=-90] (g2);
\end{tikzpicture}}}
\quad,
\tag{\Radioactivity}
\ee
as suggested by the end of Example~\ref{exa:strictpositive}. 
One might think, then, to change the definition of a.e.\ modularity so that the right-hand-side of Definition~\ref{defn:aemodularity} would read 
\be
\label{eq:aemodularalternative}
\vcenter{\hbox{
\begin{tikzpicture}[font=\small]
\node[arrow box] (q) at (0,-0.3) {$q$};
\node[copier] (c) at (0,0.3) {};
\node[copier] (c2) at (-0.5,1.55) {};
\node[arrow box] (g) at (-0.5,0.95) {$g$};
\node[arrow box] (f) at (-1,2.25) {$f$};
\coordinate (X) at (0.7,2.95);
\coordinate (Y1) at (-1,2.95);
\coordinate (Y2) at (0,2.95);
\draw (0,-0.9) to (q);
\draw (q) to (c);
\draw (c) to[out=15,in=-90] (X);
\draw (c) to[out=165,in=-90] (g);
\draw (c2) to[out=165,in=-90] (f);
\draw (f) to (Y1);
\draw (c2) to[out=15,in=-90] (Y2);
\draw (g) to (c2);
\end{tikzpicture}}}
%%%%%%%%%%%%%%%%%%%
\quad=\quad
%%%%%%%%%%%%%%%%%%%
\vcenter{\hbox{
\begin{tikzpicture}[font=\small]
\node[arrow box] (q) at (0,-0.3) {$q$};
\node[copier] (c) at (0,0.3) {};
\node[copier] (c2) at (-0.5,1.55) {};
\node[arrow box] (g) at (0,2.25) {$g$};
\coordinate (X) at (0.7,2.95);
\coordinate (Y1) at (-1,2.95);
\coordinate (Y2) at (0,2.95);
\draw (0,-0.9) to (q);
\draw (q) to (c);
\draw (c) to[out=15,in=-90] (X);
\draw (c) to[out=165,in=-90] (c2);
\draw (c2) to[out=165,in=-90] (Y1);
\draw (g) to (Y2);
\draw (c2) to[out=15,in=-90] (g);
\end{tikzpicture}}}
\;.
\tag{\Stopsign}
\ee
The reason is because then, if we used~(\ref{eq:strictpositivealternative}) for the definition of strict positivity, we would get \emph{both} S-positivity (the proof is almost the same) and a.e.\ modularity in the sense of~(\ref{eq:aemodularalternative}). Indeed,~(\ref{eq:aemodularalternative}) follows from
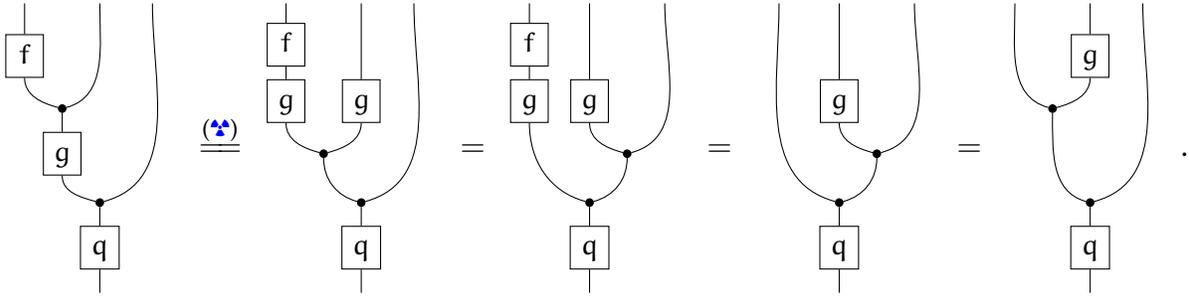
\[
\vcenter{\hbox{
\begin{tikzpicture}[font=\small]
\node[arrow box] (q) at (0,-0.3) {$q$};
\node[copier] (c) at (0,0.3) {};
\node[copier] (c2) at (-0.5,1.55) {};
\node[arrow box] (g) at (-0.5,0.95) {$g$};
\node[arrow box] (f) at (-1,2.25) {$f$};
\coordinate (X) at (0.7,2.95);
\coordinate (Y1) at (-1,2.95);
\coordinate (Y2) at (0,2.95);
\draw (0,-0.9) to (q);
\draw (q) to (c);
\draw (c) to[out=15,in=-90] (X);
\draw (c) to[out=165,in=-90] (g);
\draw (c2) to[out=165,in=-90] (f);
\draw (f) to (Y1);
\draw (c2) to[out=15,in=-90] (Y2);
\draw (g) to (c2);
\end{tikzpicture}}}
%%%%%%%%%%%%%%%%%%%
\;
\overset{\text{(\ref{eq:strictpositivealternative})}}{=\joinrel=}
\;
%%%%%%%%%%%%%%%%%%%
\vcenter{\hbox{
\begin{tikzpicture}[font=\small]
\node[arrow box] (q) at (0,-0.3) {$q$};
\node[copier] (c) at (0,0.3) {};
\node[copier] (c2) at (-0.5,0.95) {};
\node[arrow box] (g) at (0,1.65) {$g$};
\coordinate (X) at (0.7,2.95);
\node[arrow box] (g2) at (-1,1.65) {$g$};
\node[arrow box] (f) at (-1,2.4) {$f$};
\coordinate (Y1) at (-1,2.95);
\coordinate (Y2) at (0,2.95);
\draw (0,-0.9) to (q);
\draw (q) to (c);
\draw (c) to[out=15,in=-90] (X);
\draw (c) to[out=165,in=-90] (c2);
\draw (c2) to[out=15,in=-90] (g);
\draw (g2) to (f); 
\draw (f) to (Y1); 
\draw (g) to (Y2);
\draw (c2) to[out=165,in=-90] (g2);
\end{tikzpicture}}}
%%%%%%%%%%%%%%%%%%%
\;=\;
%%%%%%%%%%%%%%%%%%%
\vcenter{\hbox{
\begin{tikzpicture}[font=\small]
\node[arrow box] (q) at (0,-0.3) {$q$};
\node[copier] (c) at (0,0.3) {};
\node[copier] (c2) at (0.5,0.95) {};
\node[arrow box] (g) at (0,1.65) {$g$};
\coordinate (X) at (1,2.95);
\node[arrow box] (g2) at (-0.8,1.65) {$g$};
\node[arrow box] (f) at (-0.8,2.4) {$f$};
\coordinate (Y1) at (-0.8,2.95);
\coordinate (Y2) at (0,2.95);
\draw (0,-0.9) to (q);
\draw (q) to (c);
\draw (c) to[out=15,in=-90] (c2);
\draw (c) to[out=165,in=-90] (g2);
\draw (c2) to[out=165,in=-90] (g);
\draw (g2) to (f); 
\draw (f) to (Y1); 
\draw (g) to (Y2);
\draw (c2) to[out=15,in=-90] (X);
\end{tikzpicture}}}
%%%%%%%%%%%%%%%%%%%
\;=\;
%%%%%%%%%%%%%%%%%%%
\vcenter{\hbox{
\begin{tikzpicture}[font=\small]
\node[arrow box] (q) at (0,-0.3) {$q$};
\node[copier] (c) at (0,0.3) {};
\node[copier] (c2) at (0.5,0.95) {};
\node[arrow box] (g) at (0,1.65) {$g$};
\coordinate (X) at (1,2.95);
\coordinate (Y1) at (-0.8,2.95);
\coordinate (Y2) at (0,2.95);
\draw (0,-0.9) to (q);
\draw (q) to (c);
\draw (c) to[out=15,in=-90] (c2);
\draw (c) to[out=165,in=-90] (Y1);
\draw (c2) to[out=165,in=-90] (g);
\draw (g) to (Y2);
\draw (c2) to[out=15,in=-90] (X);
\end{tikzpicture}}}
%%%%%%%%%%%%%%%%%%%
\;=\;
%%%%%%%%%%%%%%%%%%%
\vcenter{\hbox{
\begin{tikzpicture}[font=\small]
\node[arrow box] (q) at (0,-0.3) {$q$};
\node[copier] (c) at (0,0.3) {};
\node[copier] (c2) at (-0.5,1.55) {};
\node[arrow box] (g) at (0,2.25) {$g$};
\coordinate (X) at (0.7,2.95);
\coordinate (Y1) at (-1,2.95);
\coordinate (Y2) at (0,2.95);
\draw (0,-0.9) to (q);
\draw (q) to (c);
\draw (c) to[out=15,in=-90] (X);
\draw (c) to[out=165,in=-90] (c2);
\draw (c2) to[out=165,in=-90] (Y1);
\draw (g) to (Y2);
\draw (c2) to[out=15,in=-90] (g);
\end{tikzpicture}}}
\;.
\]
But if we chose our axiom of a.e.\ modularity to be~(\ref{eq:aemodularalternative}), then Theorem~\ref{thm:aemodbayesdisint} would fail! This shows that the axioms we have chosen are quite sensitive. We suspect that for these reasons, we have chosen the most reasonable option among its potential alternatives.%
%footnote
\footnote{Note that none of these subtleties appear for classical Markov categories since all alternatives are equivalent.}
%end footnote 
\end{rmk}

The last class of subcategories we discuss are causal subcategories~\cite[Definition~11.31]{Fr19}. 

\begin{defn}[Causal subcategory]{defn:causal}
An even subcategory $\mC$ of a quantum Markov category is \define{right causal} iff for all quadruples $f,g,h,k$ of morphisms in $\mC$ such that
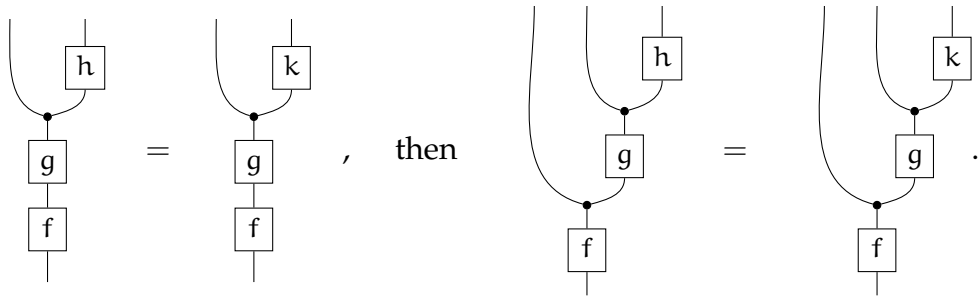
\[
\vcenter{\hbox{%
\begin{tikzpicture}[font=\small]
\node[copier] (copier) at (0,0.3) {};
\node[arrow box] (h) at (0.5,0.95) {$h$};
\node[arrow box] (g) at (0,-0.3) {$g$};
\node[arrow box] (f) at (0,-1.2) {$f$};
\coordinate (X) at (-0.5,1.6);
\coordinate (Y) at (0.5,1.6);
\draw (0,-1.9) to (f);
\draw (f) to (g);
\draw (g) to (copier);
\draw (copier) to[out=165,in=-90] (X);
\draw (copier) to[out=15,in=-90] (h);
\draw (h) to (Y);
\end{tikzpicture}}}
%%%%%%%%%%%%%%%%%%%%
\quad=\quad
%%%%%%%%%%%%%%%%%%%%
\vcenter{\hbox{%
\begin{tikzpicture}[font=\small]
\node[copier] (copier) at (0,0.3) {};
\node[arrow box] (k) at (0.5,0.95) {$k$};
\node[arrow box] (g) at (0,-0.3) {$g$};
\node[arrow box] (f) at (0,-1.2) {$f$};
\coordinate (X) at (-0.5,1.6);
\coordinate (Y) at (0.5,1.6);
\draw (0,-1.9) to (f);
\draw (f) to (g);
\draw (g) to (copier);
\draw (copier) to[out=165,in=-90] (X);
\draw (copier) to[out=15,in=-90] (k);
\draw (k) to (Y);
\end{tikzpicture}}}
%%%%%%%%%%%%%%%%%%%%
\quad,\quad\text{ then }\quad
%%%%%%%%%%%%%%%%%%%%
\vcenter{\hbox{
\begin{tikzpicture}[font=\small]
\node[arrow box] (f) at (0,-0.3) {$f$};
\node[copier] (c) at (0,0.3) {};
\node[copier] (c2) at (0.5,1.55) {};
\node[arrow box] (g) at (0.5,0.95) {$g$};
\node[arrow box] (h) at (1,2.25) {$h$};
\coordinate (X) at (-0.7,2.95);
\coordinate (Y1) at (1,2.95);
\coordinate (Y2) at (0,2.95);
\draw (0,-0.9) to (f);
\draw (f) to (c);
\draw (c) to[out=165,in=-90] (X);
\draw (c) to[out=15,in=-90] (g);
\draw (c2) to[out=15,in=-90] (h);
\draw (h) to (Y1);
\draw (c2) to[out=165,in=-90] (Y2);
\draw (g) to (c2);
\end{tikzpicture}}}
%%%%%%%%%%%%%%%%%%%%
\quad=\quad
%%%%%%%%%%%%%%%%%%%%
\vcenter{\hbox{
\begin{tikzpicture}[font=\small]
\node[arrow box] (f) at (0,-0.3) {$f$};
\node[copier] (c) at (0,0.3) {};
\node[copier] (c2) at (0.5,1.55) {};
\node[arrow box] (g) at (0.5,0.95) {$g$};
\node[arrow box] (k) at (1,2.25) {$k$};
\coordinate (X) at (-0.7,2.95);
\coordinate (Y1) at (1,2.95);
\coordinate (Y2) at (0,2.95);
\draw (0,-0.9) to (f);
\draw (f) to (c);
\draw (c) to[out=165,in=-90] (X);
\draw (c) to[out=15,in=-90] (g);
\draw (c2) to[out=15,in=-90] (k);
\draw (k) to (Y1);
\draw (c2) to[out=165,in=-90] (Y2);
\draw (g) to (c2);
\end{tikzpicture}}}
.
\]
An analogous definition is made for \define{left causal}. 
When $\mC$ is $*$-preserving and right or left causal, it will be called \define{causal}. 
\end{defn}

\begin{prop}[$\fdCAlgSPU$ is causal]{prop:SPUcausal}
$\fdCAlgSPU$ is a causal subcategory of $\fdCAlgUY$. 
\end{prop}

\bprf
Suppose that $F,G,H,K$ are SPU maps that satisfy
\be
\label{eq:causalityassumption}
F\bigg(G\Big(A'\big(H(B)-K(B)\big)\Big)\bigg)=0
\ee
for all $A'$ and $B$ in the appropriate domains. 
By Theorem~\ref{thm:ncaeequivalence}, it suffices to prove 
\be
\label{eq:provingcausality}
F\bigg(\Big(G\big(AH(B)\big)-G\big(AK(B)\big)\Big)^*\Big(G\big(AH(B)\big)-G\big(AK(B)\big)\Big)\bigg)=0
\ee
for all $A$ and $B$. Set $C:=G\big(AH(B)\big)-G\big(AK(B)\big)=G\Big(A\big(H(B)-K(B)\big)\Big)$. Then 
\be
\begin{split}
0\le F(C^*C)&\le F\Bigg(G\bigg(\Big(A\big(H(B)-K(B)\big)\Big)^*\Big(A\big(H(B)-K(B)\big)\Big)\bigg)\Bigg)\quad\text{ by KS for $G$}\\
&=F\bigg(G\Big(\big(H(B)-K(B)\big)^*A^*A\big(H(B)-K(B)\big)\Big)\bigg)=0
\end{split}
\ee
since $A':=\big(H(B)-K(B)\big)^*A^*A$ is just some element of the domain of $G$ and~(\ref{eq:causalityassumption}) applies.  Thus, $F(C^*C)=0$ and~(\ref{eq:provingcausality}) holds.
\eprf

\begin{rmk}[$\fdCAlgPU$ is not causal]{rmk:PUnotcausal}
The category $\fdCAlgPU$ of positive unital maps is not causal. Though there may be a simpler counterexample, consider the maps $H,K:\mM_{3}(\C)\stoch\mM_{2}(\C),$ $G:\mM_{2}(\C)\stoch\mM_{2}(\C)$, and $F:\mM_{2}(\C)\stoch\C$, where $G$ is the transpose map, and the other maps are defined as follows. Set $P$ to be the rank $1$ projection given by 
\[
P:=\frac{1}{2}\begin{bmatrix}1&-i\\i&1\end{bmatrix}
\quad\text{ and let }\quad
B:=\begin{bmatrix}b_{11}&b_{12}&b_{13}\\b_{21}&b_{22}&b_{23}\\b_{31}&b_{32}&b_{33}\end{bmatrix}
\]
be an arbitrary element of $\mM_{3}(\C)$. Set $F:=\tr(P\;\cdot\;)$, 
\[
H(B):=b_{11}P^{\perp}+\frac{1}{2}(b_{22}+b_{33})P
\quad\text{ and }\quad
K(B):=b_{11}P^{\perp}+\big(\l b_{22}+(1-\l)b_{33}\big)P
\]
with $\l\in[0,1]\setminus\{\frac{1}{2}\}$. Hence, $F, H,$ and $K$ are CPU. 
With these definitions, one can check that 
\[
F\Big(G\big(AH(B)\big)\Big)=F\Big(G\big(AK(B)\big)\Big)\quad\forall\;A\in\mM_{2}(\C),\;B\in\mM_{3}(\C).
\]
In fact, since $F\circ G$ is a state with support $P^{T}=P^{\perp}$ (by the first calculation in Remark~\ref{rmk:disintdnibayesforpu}), checking this equality is equivalent to (by Theorem~\ref{thm:ncaeequivalence}~\ref{item:aeP}) checking 
\[
H(B)P^{\perp}=K(B)P^{\perp}\qquad\forall\;B\in\mM_{3}(\C), 
\]
which is much more easily seen to hold. 
However, the condition 
\[
F\Big(CG\big(AH(B)\big)\Big)=F\Big(CG\big(AK(B)\big)\Big)\quad\forall\;A,C\in\mM_{2}(\C),\;B\in\mM_{3}(\C)
\]
fails. One can check this more simply by using Theorem~\ref{thm:ncaeequivalence}~\ref{item:aeP} again, and noting that the projection associated to the state $F$ is $P$ itself. Indeed, 
\[
\xy0;/r.30pc/:
(28,-10)*+{\left(b_{11}P+\big(\l b_{22}+(1-\l)b_{33}\big)P^{\perp}\right)A^{T}P}="6";
(-28,-10)*+{\left(b_{11}P+\frac{1}{2}(b_{22}+b_{33})P^{\perp}\right)A^{T}P}="5";
(-34,3)*+{H(B)^{T}A^{T}P}="4";
(-18,14)*+{G\big(AH(B)\big)P}="3";
(18,14)*+{G\big(AK(B)\big)P}="2";
(34,3)*+{K(B)^{T}A^{T}P}="1";
{\ar@{=}@/_0.75pc/"1";"2"};
{\ar@{=}@/_0.75pc/"3";"4"};
{\ar@{=}@/_0.75pc/"4";"5"};
{\ar@{=}@/_0.75pc/"6";"1"};
\endxy
\]
and the bottom two terms are now more easily seen to not be equal for all $A\in\mM_{2}(\C)$ and $B\in\mM_{3}(\C)$. 
\end{rmk}

\begin{cor}[A.e.\ equivalence classes of state-preserving SPU maps form a category]{cor:aeclassesofstatepreserving}
Given morphisms $\Theta\xstoch{p}X$, $f,f':X\stoch Y$, and $g,g':Y\stoch Z$ in a causal subcategory $\mC$, suppose that $f$ is right (left) $p$-a.e.\ equal to $f'$ and $g$ is right (left) $q$-a.e.\ equal to $g$', where $q:=f\circ p\equiv f'\circ p$. Then $g\circ f$ is right (left) $p$-a.e.\ equal to $g'\circ f'$.%
%footnote
\footnote{If $\mC$ is $*$-preserving, the adjective right (left) can be dropped.}
%end footnote
Furthermore, if $\mC$ is also $*$-preserving, then $\mC$ is a.e.\ well-defined with respect to states (cf.\ Definition~\ref{defn:statepreservingcat}). 
In particular, $\fdCAlgSPU$ is a.e.\ well-defined with respect to states. 
\end{cor}

\bprf
The first claim was proved in~\cite[Proposition~13.9]{Fr19}, but we include the proof to illustrate that the commutativity axiom of a classical Markov category is not needed. Indeed, 
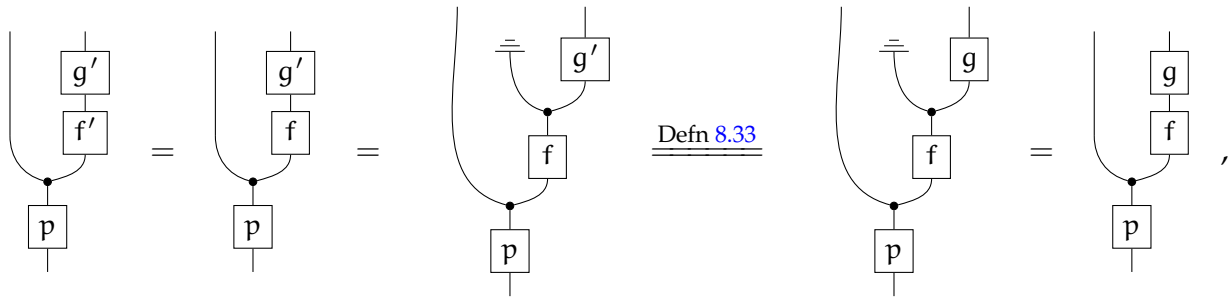
\[
\vcenter{\hbox{%
\begin{tikzpicture}[font=\small]
\node[arrow box] (q) at (0,-0.3) {$p$};
\node[copier] (copier) at (0,0.3) {};
\coordinate (f) at (-0.5,0.91) {};
\node[arrow box] (g) at (0.5,0.95) {$f'$};
\node[arrow box] (e) at (0.5,1.75) {$g'$};
\coordinate (X) at (-0.5,2.3);
\coordinate (Y) at (0.5,2.3);
\draw (0,-0.9) to (q);
\draw (q) to (copier);
\draw (copier) to[out=165,in=-90] (f);
\draw (f) to (X);
\draw (copier) to[out=15,in=-90] (g);
\draw (g) to (e);
\draw (e) to (Y);
\end{tikzpicture}}}
\quad
=
\quad
\vcenter{\hbox{%
\begin{tikzpicture}[font=\small]
\node[arrow box] (q) at (0,-0.3) {$p$};
\node[copier] (copier) at (0,0.3) {};
\coordinate (f) at (-0.5,0.91) {};
\node[arrow box] (g) at (0.5,0.95) {$f$};
\node[arrow box] (e) at (0.5,1.75) {$g'$};
\coordinate (X) at (-0.5,2.3);
\coordinate (Y) at (0.5,2.3);
\draw (0,-0.9) to (q);
\draw (q) to (copier);
\draw (copier) to[out=165,in=-90] (f);
\draw (f) to (X);
\draw (copier) to[out=15,in=-90] (g);
\draw (g) to (e);
\draw (e) to (Y);
\end{tikzpicture}}}
\quad
=
\quad
\vcenter{\hbox{
\begin{tikzpicture}[font=\small]
\node[arrow box] (p) at (0,-0.3) {$p$};
\node[copier] (c) at (0,0.3) {};
\node[copier] (c2) at (0.5,1.55) {};
\node[arrow box] (f) at (0.5,0.95) {$f$};
\node[arrow box] (g) at (1,2.25) {$g'$};
\coordinate (X) at (-0.7,2.95);
\coordinate (Y1) at (1,2.95);
\node[discarder] (Y2) at (0,2.35) {};
\draw (0,-0.9) to (p);
\draw (p) to (c);
\draw (c) to[out=165,in=-90] (X);
\draw (c) to[out=15,in=-90] (f);
\draw (c2) to[out=15,in=-90] (g);
\draw (g) to (Y1);
\draw (c2) to[out=165,in=-90] (Y2);
\draw (f) to (c2);
\end{tikzpicture}}}
\quad
\overset{\text{Defn~\ref{defn:causal}}}{=\joinrel=\joinrel=\joinrel=\joinrel=\joinrel=}
\quad
\vcenter{\hbox{
\begin{tikzpicture}[font=\small]
\node[arrow box] (p) at (0,-0.3) {$p$};
\node[copier] (c) at (0,0.3) {};
\node[copier] (c2) at (0.5,1.55) {};
\node[arrow box] (f) at (0.5,0.95) {$f$};
\node[arrow box] (g) at (1,2.25) {$g$};
\coordinate (X) at (-0.7,2.95);
\coordinate (Y1) at (1,2.95);
\node[discarder] (Y2) at (0,2.35) {};
\draw (0,-0.9) to (p);
\draw (p) to (c);
\draw (c) to[out=165,in=-90] (X);
\draw (c) to[out=15,in=-90] (f);
\draw (c2) to[out=15,in=-90] (g);
\draw (g) to (Y1);
\draw (c2) to[out=165,in=-90] (Y2);
\draw (f) to (c2);
\end{tikzpicture}}}
\quad
=
\quad
\vcenter{\hbox{%
\begin{tikzpicture}[font=\small]
\node[arrow box] (q) at (0,-0.3) {$p$};
\node[copier] (copier) at (0,0.3) {};
\coordinate (f) at (-0.5,0.91) {};
\node[arrow box] (g) at (0.5,0.95) {$f$};
\node[arrow box] (e) at (0.5,1.75) {$g$};
\coordinate (X) at (-0.5,2.3);
\coordinate (Y) at (0.5,2.3);
\draw (0,-0.9) to (q);
\draw (q) to (copier);
\draw (copier) to[out=165,in=-90] (f);
\draw (f) to (X);
\draw (copier) to[out=15,in=-90] (g);
\draw (g) to (e);
\draw (e) to (Y);
\end{tikzpicture}}}
\quad,
\]
where the first equality follows from the assumption that $f'$ is right $p$-a.e.\ equivalent to $f$. Note that causality applies in the third equality because $g'$ is right $q$-a.e.\ equivalent to $g$. 
The last claim of the corollary is a consequence of this together with Proposition~\ref{prop:SPUcausal}. 
\eprf

\begin{ques}[Is there a diagrammatic proof of a.e.\ composition for PU maps?]{ques:diagramaePUcompose}
Corollary~\ref{cor:aeclassesofstatepreserving} was already known~\cite[Proposition~3.106]{PaRu19}. In fact, the original proof using Kadison--Schwarz is more direct than the proof we have given for Corollary~\ref{cor:aeclassesofstatepreserving}. In fact, it is even true for $\fdCAlgPU$, as was shown in~\cite[Theorem~3.113]{PaRu19} (though the proof of this is more involved). 
The fact that $\fdCAlgPU$ is not causal in the sense of Definition~\ref{defn:causal}, but is nevertheless a.e.\ well-defined with respect to states (cf.\  Definition~\ref{defn:statepreservingcat}), shows that causality is not necessary for this latter property. Is there a string diagrammatic axiom weaker than causality (which presumably $\fdCAlgPU$ satisfies) that guarantees that a.e.\ equivalence classes of state-preserving morphisms form a category? 
\end{ques}

\begin{ques}[A non-commutative generalization of the Hewitt--Savage zero-one law]{a022}
Recently, Fritz and Rischel proved an abstract version of the Hewitt--Savage zero-one law in causal Markov categories~\cite[Theorem~5.4]{FrRi20}. Since $\fdCAlgSPU$ is causal, one might wonder if there is a quantum generalization of this theorem. We leave this as an open question but mention that its answer is beyond the formalism of quantum Markov categories as defined here.%
%footnote
\footnote{As for most results in this paper, the finite-dimensionality assumption in Proposition~\ref{prop:SPUcausal} does not \emph{seem} crucial. It only seems needed to make rigorous sense of a monoidal structure (cf.\ Remark~\ref{ex:allCAlg} and Question~\ref{ques:multilinearCAlg}).}
%end footnote
\end{ques}

We end with a summary of the main consequences of earlier results specific to $C^*$-algebras. 

\begin{cor}[Disintegrations and Bayesian inversion in $\fdCAlgCPU$]{cor:summary}
Let $\mA$ and $\mB$ be finite-dimensional $C^*$-algebras, let $\mA\xstoch{\w}\C$ be a state on $\mA$, let $\mB\xstoch{F}\mA$ be a CPU map, and set $\xi:=\w\circ F$. 
\begin{enumerate}[i.]
\itemsep0pt
\item
If a CPU Bayesian inverse of $(F,\w,\xi)$ exists, it is $\xi$-a.e.\ unique. 
\item
Suppose $G$ is a CPU Bayesian inverse of $(F,\w,\xi)$. Then $F$ is $\w$-a.e.\ deterministic if and only if $G$ is a disintegration of $(F,\w,\xi)$.
\item
If $(F,\w,\xi)$ has a CPU disintegration $G$, then $F$ is $\w$-a.e.\ deterministic and $G$ is a Bayesian inverse of $(F,\w,\xi)$.
\item
If $F$ is $\w$-a.e.\ deterministic, then a CPU Bayesian inverse exists if and only if a disintegration exists. 
\item
Any two CPU disintegrations of $(F,\w,\xi)$ are $\xi$-a.e.\ equivalent. 
\item
If $\mA$ is commutative and $F$ is $\omega$-a.e.\ deterministic, a CPU disintegration of $(F,\w,\xi)$ always exists. If either $\mA$ is not commutative or $F$ is not $\omega$-a.e.\ deterministic, a CPU disintegration need not exist. 
\item
If both $\mA$ and $\mB$ are commutative, a CPU Bayesian inverse always exists. If at least one of $\mA$ or $\mB$ is not commutative, then a CPU Bayesian inverse of $(F,\omega,\xi)$ need not exist. 
\end{enumerate}
\end{cor}

Hence, not every deterministic map is invertible nor does every deterministic map equipped with a state have a disintegration (cf.\ \cite[Section~5.2]{PaRu19}). Nevertheless, there are more deterministic morphisms admitting disintegrations than inverses. In the larger category of $C^*$-algebras and (state-preserving) CPU maps (as opposed to just state-preserving deterministic maps like those that appear in~\cite{PaEntropy}), even more morphisms have Bayesian inverses. 

\bprf
[Proof of Corollary~\ref{cor:summary}]
{\color{white}{you found me!}}

\begin{enumerate}[i.]
\item
This follows from Lemma~\ref{lem:bayesaeunique}. 
\item
The forward direction was proved in Proposition~\ref{thm:Bayesianinverseofdeterministicisadisint}. The reverse direction is a consequence of Theorem~\ref{thm:aemodbayesdisint} and Theorem~\ref{thm:SPUmodular}. 
\item
This follows from Theorem~\ref{thm:aemodbayesdisint} and Theorem~\ref{thm:SPUmodular}
\item
This follows from the previous statements. 
\item
This follows from the previous statements. 
\item
First suppose that $\mA$ is commutative and $F$ is $\omega$-a.e.\ deterministic. Since $\mA$ is commutative, it is of the form $\C^{X}$ for some finite set $X$ (up to $*$-isomorphism). Similarly, write $\mB$ in the form $\mB=\bigoplus_{y\in Y}\mM_{n_{y}}(\C)$ for some finite set $Y$. In this case, $\omega$ corresponds to a probability measure $\{\bullet\}\xstoch{p}X$. As such, the support $P_{\omega}=1_{X\setminus N_{p}}$ is the function whose value is $1$ at $x\in X\setminus N_{p}$ and zero elsewhere. Since $F$ is $\omega$-a.e.\ deterministic, $F(B_{1}^*B_{2})P_{\omega}=F(B_{1})^*F(B_{2})P_{\omega}$ for all $B_{1},B_{2}\in\mB$. In particular the composite 
\be
\mM_{n_{y}}(\C)\hookrightarrow\mB\xstoch{F}\mA\xrightarrow{\mathrm{ev}_{x}}\C
\ee
is a (not necessarily unital) $*$-homomorphism for all $x\in X\setminus N_{p}$ and all $y\in Y$. However, the only (not necessarily unital) $*$-homomorphism $\mM_{n_{y}}(\C)\to\C$ is the zero map \emph{unless} $n_{y}=1$ (in which case the map could be the zero map or the identity map). However, since $F$ is unital, for every $x\in X$, there necessarily exists a $y\in Y$ such that $\mM_{n_{y}}(\C)\hookrightarrow\mB\xstoch{F}\mA\xrightarrow{\mathrm{ev}_{x}}\C$ is nonzero. Hence, for $x\in X\setminus N_{p}$, there exists a \emph{unique} $y\in Y$ such that $\mM_{n_{y}}(\C)\hookrightarrow\mB\xstoch{F}\mA\xrightarrow{\mathrm{ev}_{x}}\C$ is nonzero. Let $X\setminus N_{p}\xrightarrow{f}Y$ be the corresponding function sending $x$ to this unique $y$. Therefore, $F$ decomposes as 
\be
\mB\cong\C^{f(X\setminus N_{p})}\oplus\left(\bigoplus_{y\in Y\setminus f(X\setminus N_{p})}\mM_{n_{y}}(\C)\right)\xstoch{F}\C^{X\setminus N_{p}}\oplus\C^{N_{p}}\cong\mA,
\ee
where the image of $\bigoplus_{y\in Y\setminus f(X\setminus N_{p})}\mM_{n_{y}}(\C)$ lands exclusively in $\C^{N_{p}}$. Therefore, the induced functional $\xi$ restricted to $\mM_{n_{y}}(\C)$ vanishes, i.e.\
\be
\xy0;/r.25pc/:
(-15,-7.5)*+{\mM_{n_{y}}(\C)}="1";
(-15,7.5)*+{\mB}="2";
(15,7.5)*+{\mA}="3";
(15,-7.5)*+{\C}="4";
{\ar@{^{(}->}"1";"2"};
{\ar@{~>}"2";"3"^{F}};
{\ar@{~>}"2";"4"^{\;\xi:=\omega\circ F}};
{\ar@{~>}"3";"4"^{\omega}};
{\ar"1";"4"_{0}};
\endxy
\ee
commutes for all $y\in Y\setminus f(X\setminus N_{p})$. This just corresponds to the fact that the map $f$ is surjective onto a set of full measure, namely $f(X\setminus N_{p})$ and the remaining $Y\setminus f(X\setminus N_{p})$ has measure zero. Let $q$ be this associated probability measure on $Y$ so that $N_{q}=Y\setminus f(X\setminus N_{p})$. Since classical disintegrations exist, there is a disintegration $Y\setminus f(X\setminus N_{p})\xstoch{g} X\setminus N_{p}$ of $(f,p,q)$ (where the codomain of $f$ is restricted to the image). Let $G$ be the associated CPU map $\C^{X\setminus N_{p}}\xstoch{G}\C^{f(X\setminus N_{p})}$. Choosing any CPU map $\C^{X}\xstoch{G'}\bigoplus_{y\in Y\setminus f(X\setminus N_{p})}\mM_{n_{y}}(\C)$, one can extend $G$ to a CPU map $\mA\xstoch{\overline{F}}\mB$. This map satisfies $\overline{F}\circ F\aeequals{\xi}\id_{\mB}$ and $\xi\circ\overline{F}=\omega$. 

Finally, if $\mA$ is not commutative, there are examples where CPU disintegrations do not exist~\cite{PaRu19}, and if $F$ is not $\omega$-a.e.\ deterministic, then the previous facts imply a CPU disintegration need not exist. 
\item
When both $\mA$ and $\mB$ are commutative, this is Bayes' theorem (Theorem~\ref{thm:classicalBayestheorem}). When $\mA$ or $\mB$ are not commutative, counterexamples are provided in~\cite{PaRuBayes}. \qedhere
\end{enumerate}
\eprf

\appendix

%%%%%%%%%%%%%%%%%%%%%%%%%%%%%%%%%%%%%%
\section[Notation tables]{a023}
\label{sec:notationtables}
\vspace{-12mm}
\noindent
\begin{tikzpicture}
\coordinate (L) at (-8.75,0);
\coordinate (R) at (8.75,0);
\draw[line width=2pt,orange!20] (L) -- node[anchor=center,rectangle,fill=orange!20]{\strut \Large \textcolor{black}{\textbf{A\;\; Notation tables}}} (R);
\end{tikzpicture}
%\vspace{1mm}
%%%%%%%%%%%%%%%%%%%%%%%%%%%%%%%%%%%%%%

This section contains some tables for reference. 
The first table contains names of certain categories and subcategories used. For many such subcategories, the same notation $\mC$ is used. This should not cause confusion since the specific property is always written out explicitly in all claims. 

%\pagebreak
\begin{center}
{\large Certain subcategories of (or related to) quantum Markov categories}
\begin{longtable}{ccc}
\hline
Condition/axiom&Common Notation&First appearance\\
\hline
Quantum Markov category&$\mM_{\text{\Yinyang}}$&Def'n~\ref{defn:qmc} (page~\pageref{defn:qmc})\\
%\hline
Subcategory of even morphisms&$\mM$ or $\mM_{\mathrm{even}}$&Def'n~\ref{defn:qmc} (page~\pageref{defn:qmc})\\
%\hline
Subcategory of even $*$-preserving morphisms&$\mM_{*}$&Def'n~\ref{defn:selfadjointmorphism} (page~\pageref{defn:selfadjointmorphism})\\
%\hline
S-positive subcategory&$\mC$&Def'n~\ref{defn:positivecategory} (page~\pageref{defn:positivecategory})\\
%\hline
Deterministically reasonable subcategory&$\mC$&Def'n~\ref{defn:detreason} (page~\pageref{defn:detreason})\\
%\hline
A.e.\ modular subcategory&$\mC$&Def'n~\ref{defn:aemodularity} (page~\pageref{defn:aemodularity})\\
%\hline
Category of state-preserving morphisms in $\mC$&$I_{/\mC}$&Def'n~\ref{defn:statepreservingcat} (page~\pageref{defn:statepreservingcat})\\
%\hline
A.e.\ well-defined with respect to states&$\mC$&Def'n~\ref{defn:statepreservingcat} (page~\pageref{defn:statepreservingcat})\\
%\hline
Strictly positive subcategory&$\mC$&Def'n~\ref{defn:strictpositive} (page~\pageref{defn:strictpositive})\\
%\hline
Causal subcategory&$\mC$&Def'n~\ref{defn:causal} (page~\pageref{defn:causal})\\
\hline
\end{longtable}
\end{center}

In the following list of categories, \emph{all} $C^*$-algebras are assumed to be unital and finite-dimensional. If it makes sense to do so, appending a \Yingyang\, symbol as a subscript means that one is also including conjugate-linear maps in the case of $C^*$-algebras. 

%\pagebreak
\begin{center}
{\large Frequently used categories}
\begin{longtable}{cc}
\hline
Notation&Description\\
\hline
$\FinStoch$&finite sets and stochastic maps\\
%\hline
$\FinMeas$&finite sets and transition maps\\
%\hline
$\FinMeasP$&finite sets and non-negative transition maps\\
%\hline
$\Stoch$&measurable spaces and Markov kernels\\
%\hline
$\fdCAlgY$&$C^*$-algebras and linear \& conjugate-linear maps\\
%\hline
$\fdCAlgP$&$C^*$-algebras and positive (linear) maps\\
%\hline
$\fdCAlgU$&$C^*$-algebras and unital linear linear maps\\
%\hline
$\fdCAlgPU$&$C^*$-algebras and positive unital (linear) maps\\
%\hline
$\fdCAlgSPU$&$C^*$-algebras and Schwarz-positive unital (linear) maps\\
%\hline
$\fdCAlgCPU$&$C^*$-algebras and completely positive unital (linear) maps\\
%\hline
$\fdCAlg_{\mathrm{det}}$&$C^*$-algebras and (unital) $*$-homomorphisms\\
\hline
\end{longtable}
\end{center}

\begin{center}
{\large Frequently used morphisms}\\
\begin{longtable}{cc}
\hline
Name/Notation&First appearance/defined\\
\hline
stochastic&Def'n~\ref{defn:stochasticmaps} (page~\pageref{defn:stochasticmaps})\\
%\hline
probability-preserving&Def'n~\ref{defn:stochasticmaps} (page~\pageref{defn:stochasticmaps})\\
%\hline
even&Def'n~\ref{defn:gradedmonoidalcategories} (page~\pageref{defn:gradedmonoidalcategories})\\
%\hline
conjugate-linear&Exa~\ref{ex:linearandantilineartensor} (page~\pageref{ex:linearandantilineartensor})\\
%\hline
copy $\Delta_{X}$, ground $!_{X}$, involve $*_{X}$&Def'n~\ref{defn:qmc} (page~\pageref{defn:qmc})\\
%\hline
unital&Def'n~\ref{defn:qmc} (page~\pageref{defn:qmc})\\
%\hline
positive, Schwarz-positive, completely positive&Def'n~\ref{defn:positivemaps} (page~\pageref{defn:positivemaps})\\
%\hline
$*$-preserving&Def'n~\ref{defn:selfadjointmorphism} (page~\pageref{defn:selfadjointmorphism})\\
%\hline
deterministic&Def'n~\ref{eq:deterministicmap} (page~\pageref{eq:deterministicmap})\\
%\hline
states and effects&Def'n~\ref{defn:state} (page~\pageref{defn:state})\\
%\hline
a.e.\ equivalent&Def'n~\ref{defn:aeequivalence} (page~\pageref{defn:aeequivalence})\\
%\hline
a.e.\ deterministic&Def'n~\ref{defn:aeequivalentdeterministic} (page~\pageref{defn:aeequivalentdeterministic})\\
%\hline
a.e.\ unital&Def'n~\ref{eq:aecausal} (page~\pageref{eq:aecausal})\\
%\hline
state-preserving&Def'n~\ref{defn:disintegration} (page~\pageref{defn:disintegration})\\
%\hline
disintegration&Def'n~\ref{defn:disintegration} (page~\pageref{defn:disintegration})\\
%\hline
Bayes map and Bayesian inverse&Def'n~\ref{defn:bayesianinverse} (page~\pageref{defn:bayesianinverse})\\
%\hline
a.e.\ modular&Def'n~\ref{defn:aemodularity} (page~\pageref{defn:aemodularity})\\
\hline
%\hline
%XXXX&Definition~\ref{} (page~\pageref{})\\
\end{longtable}
\end{center}

\section*{Acknowledgements}

I gratefully acknowledge support from the Simons Center for Geometry and Physics, Stony Brook University and for the opportunity to participate in the workshop ``Operator Algebras and Quantum Physics'' in June 2019. Specifically, I thank Luca Giorgetti for encouraging me to pursue what happens to the notion of a disintegration when the morphism you are disintegrating is stochastic as opposed to deterministic. 
I thank Tobias Fritz for additional insight and for answering many of my questions on Markov categories and I thank Kenta Cho for informing me of Fritz' work on Markov categories. 
I thank Chris Heunen for pointing out an inconsistency in an earlier version of the definition of a quantum Markov category. 
I thank Christian Carmellini, Aaron Fenyes, Philip Parzygnat, Alessio Ranallo, and Ir\`ene Ren for discussions.
I thank St\'ephan Attal, whose notes on quantum channels~\cite{Attal} (particularly Theorem~6.38) have influenced some of the most important proofs in this work. 
I also thank Kenta Cho, Chris Heunen, and Bart Jacobs for sharing their \LaTeX\;\! code for the string diagrams used in this paper and I thank Tom for his \LaTeX\, code on frames, which was accessed from
\url{https://texblog.org/2015/09/30/fancy-boxes-for-theorem-lemma-and-proof-with-mdframed/} on November 1st, 2018.
Most importantly, I thank Benjamin Russo for many years of fruitful discussions, and I thank two anonymous referees whose insight, comments, questions, and suggestions have substantially improved the quality and scope of this work. 
This project began while 
I was an Assistant Research Professor at the University of Connecticut. 
This research has also received funding from the European Research Council (ERC) under the European Union's Horizon 2020 research and innovation program (QUASIFT grant agreement 677368).

%%%%%%%%BIBLIOGRAPHY%%%%%%%%%%%%
\addcontentsline{toc}{section}{\numberline{}Bibliography}
\bibliographystyle{hplainParzygnatv1}
\bibliography{Bayesbib}

\Addresses
%%based on egreg's code (see after title and authors in preamble)
\end{document}